%% file: clifford.tex
\documentclass[10pt,a4paper,notitlepage]{article}
\usepackage[latin1]{inputenc}
\usepackage{graphicx,psfrag}
\usepackage{amsmath,amsfonts,amssymb,amsthm,verbatim}
\usepackage[mathscr]{euscript}
\usepackage{rotating}

\date{\scriptsize{Department of Mathematics, KTH}\\ \scriptsize{SE-100 44 \ Stockholm, Sweden}}
\author{Douglas Lundholm and Lars Svensson}
\title{Clifford algebra, geometric algebra,\\and applications}

\numberwithin{equation}{section}
\numberwithin{table}{section}
\numberwithin{figure}{section}

\theoremstyle{plain}
\newtheorem{thm}{Theorem}[section]
\newtheorem{lem}[thm]{Lemma}
\newtheorem{prop}[thm]{Proposition}
\newtheorem*{cor}{Corollary}

\theoremstyle{definition}
\newtheorem{defn}{Definition}[section]

\newtheorem{exmp}{Example}[section]
\newtheorem{exc}{Exercise}[section]

\theoremstyle{remark}
\newtheorem*{rem}{Remark}

\newcommand{\cl}{\mathcal{C}l}
\newcommand{\liprod}{\ \raisebox{.2ex}{$\llcorner$}\ }
\newcommand{\riprod}{\ \raisebox{.2ex}{$\lrcorner$}\ }
\newcommand{\iprod}{\ \raisebox{.3ex}{$\scriptscriptstyle \bullet$}\ }
\newcommand{\cliffconj}{{\scriptscriptstyle \square}}
\newcommand{\symdiff}{\!\bigtriangleup\!}
\newcommand{\dual}{^\mathbf{c}}
\DeclareMathOperator{\ad}{\textup{ad}}
\DeclareMathOperator{\adj}{\textup{adj}}
\DeclareMathOperator{\Ad}{\textup{Ad}}
\DeclareMathOperator{\tAd}{\widetilde{\Ad}}
\DeclareMathOperator{\charop}{\textup{char}}
\DeclareMathOperator{\Cont}{\textup{Cont}}

\DeclareMathOperator{\End}{\textup{End}}
\DeclareMathOperator{\GL}{\textup{GL}}
\DeclareMathOperator{\Hom}{\textup{Hom}}
\DeclareMathOperator{\id}{\textup{id}}
\DeclareMathOperator{\im}{\textup{im}}
\DeclareMathOperator{\Ind}{\textup{Ind}}
\DeclareMathOperator{\List}{\textup{List}}
\DeclareMathOperator{\Ogrp}{\textup{O}}
\DeclareMathOperator{\Pin}{\textup{Pin}}
\DeclareMathOperator{\Pol}{\textup{Pol}}

\DeclareMathOperator{\sgn}{\textup{sgn}}
\DeclareMathOperator{\Spin}{\textup{Spin}}
\DeclareMathOperator{\SO}{\textup{SO}}
\DeclareMathOperator{\so}{\mathfrak{so}}
\DeclareMathOperator{\SU}{\textup{SU}}
\DeclareMathOperator{\Span}{\textup{Span}}
\DeclareMathOperator{\supp}{\textup{supp}}
\DeclareMathOperator{\tr}{\textup{tr}}
\DeclareMathOperator{\transp}{\textup{T}}
\DeclareMathOperator{\vecmeet}{{^\wedge}}
\DeclareMathOperator{\vecjoin}{{^\vee}}

\begin{document}

\pagestyle{empty}

\maketitle
\thispagestyle{empty}

\begin{abstract}
	These are
	lecture notes for a course on the theory of Clifford algebras, 
	with special emphasis on their wide range of applications
	in mathematics and physics.
	Clifford algebra is introduced both 
	through a conventional tensor algebra construction
	(then called geometric algebra)
	with geometric applications in mind,
	as well as in an algebraically more general form 
	which is well suited for combinatorics, 
	and for defining and understanding the numerous 
	products and operations of the algebra.
	The various applications presented include 
	vector space and projective geometry,
	orthogonal maps and spinors, 
	normed division algebras,
	as well as 
	simplicial complexes and graph theory.
\end{abstract}

\newpage

\section*{Preface}

	These lecture notes were prepared for a course 
	held at the Department of Mathematics, KTH, during the spring of 2009,
	and intended for advanced undergraduates and Ph.D. students
	in mathematics and mathematical physics.
	They are primarily based on the lecture notes of L.S. from a similar 
	course in 2000 \cite{svensson}, and the M. Sc. Thesis of D.L. in 2006 \cite{lundholm_thesis}.
	Additional material 
	has 
	been adopted from 
	\cite{lawson_michelsohn,lounesto,doran_lasenby,hestenes_sobczyk}
	as well as from numerous other sources that can be found in the references.
	
	When preparing these notes, we have tried to aim for
\begin{itemize}
	\item[--] efficiency in the presentation, made possible by the introduction,
	already at an early stage,
	of a combinatorial approach to 
	Clifford algebra along with its standard operations.

	\item[--] completeness, in the sense that 
	we have tried to include
	most of the standard theory of Clifford algebras,
	along with proofs of all stated theorems, or (when in risk
	of a long digression) reference to where a proof can be found.
	
	\item[--] diversity in the range of applications presented 
	-- from pure algebra and combinatorics to geometry and physics.
\end{itemize}

	In order to make 
	the presentation
	rather self-contained and accessible also to undergraduate students,
	we have included an appendix with some basic notions in algebra, as well as
	some additional material which may be unfamiliar also to graduate students.
	Statements requiring a wider mathematical knowledge are placed as
	remarks in the text, and are not required for an understanding of the main material.
	There are also certain sections and proofs, marked with an asterisk,
	that go slightly outside the scope of the
	course and consequently could be 
	omitted by the student
	(but which could be of interest for a researcher in the field).
	
	We are grateful to our students for valuable input, as well as to
	John Baez for 
	correspondence.
	D.L. would also like to thank the Swedish Research Council and
	the Knut and Alice Wallenberg Foundation (grant KAW 2005.0098) for financial support.

\newpage

\tableofcontents
\newpage

\pagestyle{plain}
\setcounter{page}{1}

\section{Introduction}

	Let us start with an introducton, in terms of modern language,
	to the ideas of Hermann G\"unther Grassmann (1809-1877)
	and William Kingdon Clifford (1845-1879).
	For a detailed account of the history of the subject, see e.g. \cite{lounesto}.

\subsection{Grassmann's idea of extensive quantities}

	Let us consider our ordinary physical space which
	we represent as $\mathbb{R}^3$.
	An ordered pair $(P,Q)$ in $\mathbb{R}^3 \times \mathbb{R}^3$
	is called a \emph{line segment}, and we think of it as directed from 
	the point $P$ to the point $Q$.
	We call two line segments equivalent if one
	can be transformed into the other by a translation.
	This leads to the original picture of a vector as an equivalence
	class of directed line segments. 
	Note that such a vector can be characterized by
	a \emph{direction}, in the form of a line through the origin and
	parallel to the line segment, an \emph{orientation}, 
	i.e. the direction along the line in which it points,
	and a \emph{magnitude}, or length of the segment, 
	represented by a positive number.
	We are very familiar with how to add, subtract and rescale vectors.
	But can they be multiplied? We of course know about the
	scalar product and the cross product, but we shall introduce
	a new product called the \emph{outer product}.

	\begin{figure}[ht]
		\centering
		\psfrag{T_P}{$P$}
		\psfrag{T_Q}{$Q$}
		\psfrag{T_Pp}{$P'$}
		\psfrag{T_A}{$A$}
		\psfrag{T_Ap}{$A'$}
		\includegraphics{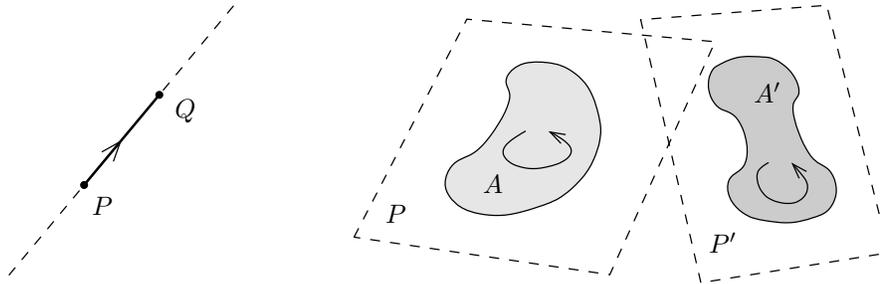}
		\caption{A directed line segment $(P,Q)$ and 
		a pair of plane surface areas $A$ and $A'$ in the planes 
		$P$ and $P'$, respectively.}
		\label{fig_segment_areas}
	\end{figure}

	Let $A$ denote a plane surface area, in a plane $P$, endowed with
	an ``orientation" or ``signature" represented by a rotation arrow
	as in Figure \ref{fig_segment_areas}.
	If $A'$ denotes another plane surface area in a plane $P'$,
	also endowed with a rotation arrow, then we shall say that
	$A$ is equivalent to $A'$ if and only if $P$ and $P'$ are parallel,
	the areas of $A$ and $A'$ are the same, and if the rotation
	arrow is directed in the same way after translating $A'$ to $A$,
	i.e. $P'$ to $P$.
	Equivalence classes of such directed surface areas are to be called 
	\emph{2-blades} (or \emph{2-vectors}).

	\begin{figure}[t]
		\centering
		\psfrag{T_P}{$P$}
		\psfrag{T_A}{$A$}
		\psfrag{T_a}{$a$}
		\psfrag{T_b}{$b$}
		\psfrag{T_awb}{$a \wedge b$}
		\includegraphics{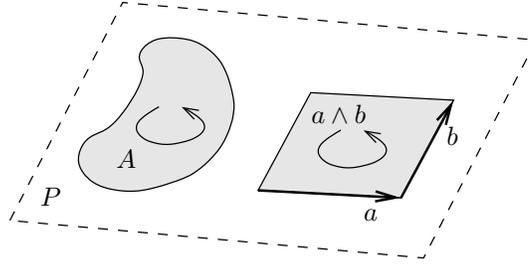}
		\caption{A 2-blade $A$ represented by an outer product $a \wedge b$.}
		\label{fig_area_blade}
	\end{figure}
	
	We observe that a 2-blade $A$ can be represented as a
	parallelogram spanned by two vectors $a$ and $b$ in 
	the same plane as $A$, according to Figure \ref{fig_area_blade},
	and call it the outer product of $a$ and $b$ and 
	denote it by $a \wedge b$.
	If the area of $A$ is zero then we write $A = 0$.
	Hence, $a \wedge a = 0$.
	By $-A$ we denote the equivalence class of plane surface areas
	with the same area and in the same plane as $A$,
	but with an oppositely directed rotation arrow.
	Thus, e.g. $-(a \wedge b) = b \wedge a$.

	\begin{figure}[ht]
		\centering
		\psfrag{T_a}{$a$}
		\psfrag{T_b}{$b$}
		\psfrag{T_awb}{$a \wedge b$}
		\psfrag{T_bwa}{$b \wedge a$}
		\includegraphics{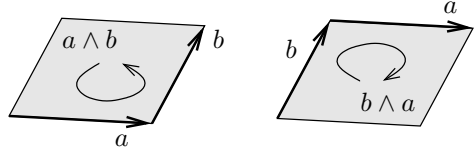}
		\caption{Geometric interpretation of $b \wedge a = -a \wedge b$.}
		\label{fig_minus_blade}
	\end{figure}

	We would now like to define the sum of two 2-blades	$A$ and $B$.
	If $A$ and $B$ are represented by plane surface areas in the
	planes $P$ and $Q$, respectively, then we can translate
	$Q$ such that $P \cap Q$ contains a line $L$.
	We then realize that we can choose a vector $c$,
	parallel to $L$, such that $A = a \wedge c$ and $B = b \wedge c$
	for some vectors $a$ in $P$, and $b$ in $Q$.
	Now, we define
	$$
		A + B = a \wedge c + b \wedge c := (a + b) \wedge c,
	$$
	which has a nice geometric representation according to Figure \ref{fig_blade_addition}.

	\begin{figure}[t]
		\centering
		\psfrag{T_P}{$P$}
		\psfrag{T_Q}{$Q$}
		\psfrag{T_L}{$L$}
		\psfrag{T_a}{$a$}
		\psfrag{T_b}{$b$}
		\psfrag{T_c}{$c$}
		\psfrag{T_apb}{$a + b$}
		\psfrag{T_A}{$A$}
		\psfrag{T_B}{$B$}
		\psfrag{T_ApB}{$A + B$}
		\includegraphics{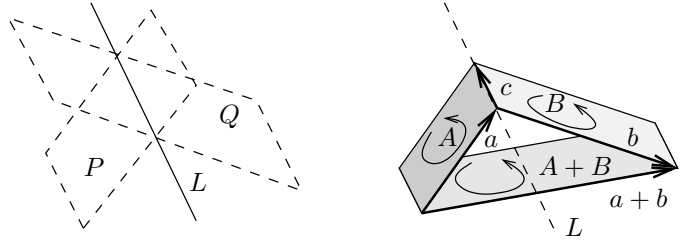}
		\caption{Geometric interpretation of the sum $A + B$ of a pair of 2-blades.}
		\label{fig_blade_addition}
	\end{figure}
	
	Analogously,
	outer products of three vectors, $a \wedge b \wedge c$, 
	or \emph{3-blades}, are defined as equivalence
	classes of ``directed" volumes with the same volume as
	the parallel epiped generated by $a$, $b$ and $c$, 
	and with the orientation given by the handedness of
	the frame $\{a,b,c\}$ (see Figure \ref{fig_abc_frame}).
	Furthermore, although we have chosen to work with 
	ordinary three-dimensional space for simplicity, 
	the notion of a blade generalizes to arbitrary dimensions,
	where a $k$-blade can be thought of as representing a $k$-dimensional 
	linear subspace equipped with an orientation and a magnitude.

	\begin{figure}[t]
		\centering
		\psfrag{T_a}{$a$}
		\psfrag{T_b}{$b$}
		\psfrag{T_c}{$c$}
		\includegraphics{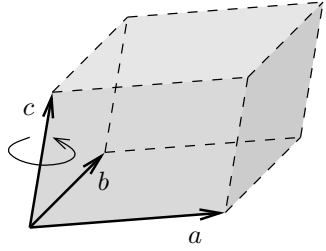}
		\caption{A 3-blade $a \wedge b \wedge c$.}
		\label{fig_abc_frame}
	\end{figure}

\subsection{Clifford's application of the extensive algebra}

	Say that we are given three symbols $\{e_1,e_2,e_3\}$ and that
	we construct a game with these symbols by
	introducing the following set of rules:
\begin{trivlist}
	\item[\textperiodcentered]
	We may form words of the symbols by writing them
	next to each other (the empty word is denoted 1).
	\item[\textperiodcentered]
	We may also form linear combinations of such words, i.e.
	multiply/scale each word by a number, and add scaled words together
	-- writing a plus sign between different words,
	and summing up the numbers in front of equal words.
	\item[\textperiodcentered]
	Furthermore, if two different symbols appear next to each other in a word
	then we may swap places and scale the word by $-1$, 
	i.e. $e_ie_j = -e_je_i$, $i\neq j$.
	\item[\textperiodcentered]
	Lastly, we impose a rule that if
	a symbol appears next to the same symbol
	then we may strike out them both, i.e. $e_1e_1 = e_2e_2 = e_3e_3 = 1$.
\end{trivlist}
	For example,
	$$
		e_1e_2e_1e_3 = -e_1e_1e_2e_3 = -e_2e_3.
	$$
	Similarly, we see that the only words that can appear
	after applying the reduction rules are
	$$
		1,\ e_1,\ e_2,\ e_3,\ e_1e_2,\ e_1e_3,\ e_2e_3,\ \ \textrm{and}\ \ e_1e_2e_3,
	$$
	and that what we have obtained is just a game of adding and multiplying
	linear combinations of such words.
	
	As an example, consider a linear combination of the original symbols,
	$$
		a = a_1e_1 + a_2e_2 + a_3e_3,
	$$
	where $a_1,a_2,a_3$ are ordinary (say, real) numbers. 
	Then the square of this linear combination is
	\begin{eqnarray*}
		a^2 &:=& aa = (a_1e_1 + a_2e_2 + a_3e_3)(a_1e_1 + a_2e_2 + a_3e_3) \\
		&=& a_1^2 e_1e_1 + a_2^2 e_2e_2 + a_3^2 e_3e_3 \\
		&& +\ (a_1a_2 - a_2a_1)e_1e_2 + (a_1a_3 - a_3a_1)e_1e_3 + (a_2a_3 - a_3a_2)e_2e_3 \\
		&=& (a_1^2 + a_2^2 + a_3^2)1.
	\end{eqnarray*}
	Similarly, if $b = b_1e_1 + b_2e_2 + b_3e_3$ is another
	linear combination, then
	\begin{eqnarray}
		ab &=& (a_1b_1 + a_2b_2 + a_3b_3)1 \label{ab_expansion} \\ 
		&& +\ (a_1b_2 - a_2b_1)e_1e_2 + (a_1b_3 - a_3b_1)e_1e_3 + (a_2b_3 - a_3b_2)e_2e_3, \nonumber 
	\end{eqnarray}
	where we recognize the number in front of the empty word 
	(the \emph{scalar part} of $ab$)
	as the ordinary scalar product $a \iprod b$ of spatial vectors
	$a = (a_1,a_2,a_3)$, $b = (b_1,b_2,b_3)$
	in a cartesian coordinate representation, 
	while the remaining part of the expression resembles the cross
	product $a \times b$, but in terms of the words $e_1e_2$ instead of $e_3$ etc.

	\begin{figure}[bt]
		\centering
		\psfrag{T_e1}{$e_1$}
		\psfrag{T_e2}{$e_2$}
		\psfrag{T_e3}{$e_3$}
		\psfrag{T_e1e2}{$e_1e_2$}
		\psfrag{T_e2e3}{$e_2e_3$}
		\psfrag{T_e3e1}{$e_3e_1$}
		\includegraphics{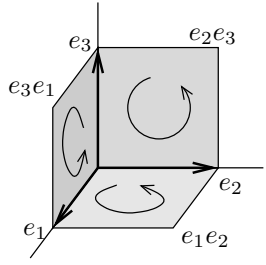}
		\caption{The standard unit blades of $\mathbb{R}^3$.}
		\label{fig_std_blades}
	\end{figure}
	
	The geometric interpretation of the above construction
	in terms of our ordinary space $\mathbb{R}^3$
	is that $e_1,e_2,e_3$ represent
	orthonormal basis vectors, while $e_1e_2, e_2e_3, e_3e_1$
	represent orthonormal basis blades corresponding to the coordinate
	planes, with a right-handed orientation and unit area (see Figure \ref{fig_std_blades}).
	The full word $e_1e_2e_3$ represents a unit 3-blade, i.e. a unit volume
	with right-handed oriententation,
	while the empty word could be thought of as representing
	an empty blade, corresponding to the trivial zero-dimensional subspace.
	From this interpretation follows that the non-scalar part of \eqref{ab_expansion}
	is a sum of coordinate 2-blades, with each component given
	by the corresponding component of the cross product.
	Recalling that the direction, orientation and magnitude of the vector
	$a \times b$ is determined by the direction and orientation of the vectors $a,b$ 
	and the area of the parallelogram spanned by $a$ and $b$
	-- similarly to how we defined a 2-blade above --
	it follows that the sum of these blades actually is equal to the
	blade $a \wedge b$, and the expression \eqref{ab_expansion} 
	becomes simply
	$$
		ab = a \iprod b + a \wedge b.
	$$

	Let us see an example of why this representation 
	of objects in space could be useful. 
	Consider the $e_1e_2$-plane spanned by the 
	orthogonal unit vectors $e_1$ and $e_2$.
	A rotation in this plane by an angle $\theta$ 
	can of course be understood as a linear map $\textsf{R}$ determined by
	\begin{eqnarray} 
		e_1 &\mapsto& \textsf{R}(e_1) = \phantom{-}\cos \theta \ e_1 + \sin \theta \ e_2, \label{plane_rotation}\\
		e_2 &\mapsto& \textsf{R}(e_2) = 		  -\sin \theta \ e_1 + \cos \theta \ e_2. \nonumber
	\end{eqnarray}
	Now, note that
	$$
		(e_1 e_2)^2 = e_1 e_2 e_1 e_2 = - e_1 e_1 e_2 e_2 = -1,
	$$
	which means that the unit blade $e_1e_2$ 
	behaves like a complex imaginary unit in this regard,
	and we can rewrite \eqref{plane_rotation} as\footnote{Just like for complex
	numbers, we could define the exponential expression in terms of sine
	and cosine, or alternatively, by the usual power series expansion.}
	$$
		\textsf{R}(e_1) = \cos \theta \ e_1 + \sin \theta \ e_1 e_1 e_2 
		= e_1 (\cos \theta + \sin \theta \ e_1 e_2) = e_1 e^{\theta e_1 e_2},
	$$
	and
	$$
		\textsf{R}(e_2) = -\sin \theta \ e_2^2 e_1 + \cos \theta \ e_2
		= e_2 (\sin \theta \ e_1 e_2 + \cos \theta) = e_2 e^{\theta e_1 e_2}.
	$$
	This is of course not much more insightful than the 
	representation of a rotation in the plane in terms of complex numbers.
	However, note that we can also write
	\begin{eqnarray*}
		e_1 e^{\theta e_1 e_2} &=& e_1 e^{\frac{\theta}{2} e_1 e_2} e^{\frac{\theta}{2} e_1 e_2} = 
			e_1 \left( \cos \mbox{$\frac{\theta}{2}$} + \sin \mbox{$\frac{\theta}{2}$} \ e_1 e_2 \right) e^{\frac{\theta}{2} e_1 e_2} \\
		&=& \left( \cos \mbox{$\frac{\theta}{2}$} - \sin \mbox{$\frac{\theta}{2}$} \ e_1 e_2 \right) e_1 e^{\frac{\theta}{2} e_1 e_2} =
			e^{- \frac{\theta}{2} e_1 e_2} e_1 e^{\frac{\theta}{2} e_1 e_2},
	\end{eqnarray*}
	and similarly,
	$$
		e_2 e^{\theta e_1 e_2} = 
			e^{- \frac{\theta}{2} e_1 e_2} e_2 e^{\frac{\theta}{2} e_1 e_2}.
	$$
	The point with the above rewriting is that the resulting expression also 
	applies to the basis vector $e_3$, which is orthogonal to the
	plane and thus unaffected by the rotation, i.e.
	$$
		\textsf{R}(e_3) = e_3 = e_3 e^{- \frac{\theta}{2} e_1 e_2} e^{\frac{\theta}{2} e_1 e_2} = 
		e^{- \frac{\theta}{2} e_1 e_2} e_3 e^{\frac{\theta}{2} e_1 e_2},
	$$
	where we used that $e_3e_1e_2 = e_1e_2e_3$.
	By linearity, the rotation \textsf{R} acting on \emph{any} vector 
	$x \in \mathbb{R}^3$ can then be written
	\begin{equation} \label{plane_rotation_rotor}
		x \mapsto \textsf{R}(x) = e^{- \frac{\theta}{2} e_1 e_2} x e^{\frac{\theta}{2} e_1 e_2} = R x R^{-1}.
	\end{equation}
	The object $R := e^{- \frac{\theta}{2} e_1 e_2}$ is called a \emph{rotor},
	and it encodes in a compact way both the plane of rotation and the angle.
	Also note that by the generality of the construction, and by choosing
	a basis accordingly, \emph{any} rotation in $\mathbb{R}^3$ can be expressed 
	in this form by 
	some rotor $R = e^{a \wedge b}$, where the blade $a \wedge b$ represents
	both the plane in which the rotation is performed, 
	as well as the direction and angle of rotation.
	As we will see, similar expressions hold true for general orthogonal
	transformations in arbitrary dimensions, and in euclidean
	space as well as in lorentzian spacetime.
	
	\newpage
	
\section{Foundations} \label{sec_foundations}

\input{clifford_foundations}
	\newpage

\section{Vector space geometry} \label{sec_vector_geom}

\input{clifford_vector}

	\newpage

\section{Discrete geometry} \label{sec_discrete}

\input{clifford_discrete}
	\newpage

\section{Classification of real and complex geometric algebras} \label{sec_isomorphisms}

\input{clifford_isomorphisms}
	\newpage

\section{Groups} \label{sec_groups}

\input{clifford_groups}

\input{clifford_examples}

	\newpage

\section{Euclidean and conformal geometry} \label{sec_conformal}

\input{clifford_conformal}
	\newpage
	
\section{Representation theory} \label{sec_reps}

\input{clifford_reps}
	\newpage

\section{Spinors} \label{sec_spinors}

\input{clifford_spinors}

	\newpage

\section{Some Clifford analysis on $\mathbb{R}^n$} \label{sec_analysis}

	\input{clifford_analysis}

	\newpage

\section{The algebra of Minkowski spacetime} \label{sec_STA}

	\input{clifford_minkowski}
	\newpage

\appendix
\section{Appendix}

\input{clifford_appendix}

\end{document}

%% file: clifford_foundations.tex
	In this section we define geometric algebra and work out a number of
	its basic properties. We first consider the definition 
	that is most common in the mathematical literature, 
	where it is introduced as a quotient space 
	on the tensor algebra of a vector space with a quadratic form.
	We see that this leads, in the finite-dimensional real or complex case, to the
	equivalent definition of an algebra with generators $\{e_1,\ldots,e_n\}$ satisfying
	the so-called \emph{anti-commutation relations}
	$$e_i e_j + e_j e_i = 2 g_{ij}$$
	for some given metric tensor (symmetric matrix) $g$.
	For physicists, this is perhaps the most well-known definition.
	
	We go on to consider an alternative definition of geometric algebra
	based on its purely algebraic and combinatorial features. 
	The resulting algebra, which we call \emph{Clifford algebra} due to
	its higher generality but less direct connection to geometry, 
	allows us to introduce common operations and
	prove fundamental identities in a 
	simple and straightforward way compared to the tensor algebraic approach.

\subsection{Geometric algebra $\mathcal{G}(V,q)$}

	The conventional definition of geometric algebra is carried out in the context
	of vector spaces endowed with an inner product, or more generally a quadratic form.
	We consider here a vector space $V$ of arbitrary dimension over 
	some field\footnote{See Appendix \ref{app_algebra} 
	if the notion of a field, ring or tensor is unfamiliar.}
	$\mathbb{F}$.
	
	\begin{defn} \label{def_quad_form}
		A \emph{quadratic form} $q$ on a vector space $V$ is a map 
		$q\!: V \to \mathbb{F}$ such that
		\begin{displaymath}
		\begin{array}{rl}
			i)  & q(\alpha v) = \alpha^2 q(v) \quad \forall \ \alpha \in \mathbb{F}, v \in V, \quad \textrm{and} \\[5pt]
			ii) & \textrm{the map \ $(v,w) \mapsto q(v+w) - q(v) - q(w)$ \ is linear in both $v$ and $w$.}
		\end{array}
		\end{displaymath}
		The corresponding bilinear form 
		$\beta_q (v,w) := \frac{1}{2} \big( q(v+w) - q(v) - q(w) \big)$ 
		is called the \emph{polarization} of $q$.
	\end{defn}

	Many vector spaces are naturally endowed with quadratic forms, 
	as the following examples show.

	\begin{exmp}
		If $V$ has a bilinear form $\langle \cdot , \cdot \rangle$
		then $q(v) := \langle v, v \rangle$ is a quadratic form
		and $\beta_q$ is the symmetrization of $\langle \cdot , \cdot \rangle$.
		This could be positive definite (i.e. an inner product), 
		or indefinite (a metric tensor of arbitrary signature).
	\end{exmp}
	
	\begin{exmp} \label{exmp_JvN}
		If $V$ is a normed vector space over $\mathbb{R}$, 
		with norm denoted by $| \cdot |$, and where the parallelogram 
		identity $|x+y|^2 + |x-y|^2 = 2|x|^2 + 2|y|^2$ holds,
		then $q(v) := |v|^2$ is a quadratic form and $\beta_q$ is an inner 
		product on $V$. This classic fact is 
		sometimes called the Jordan-von Neumann theorem.
	\end{exmp}
	
	Let 
	$$
		\mathcal{T}(V) := \bigoplus_{k=0}^{\infty} \bigotimes\nolimits^{k} V
	$$ 
	denote the tensor algebra over $V$, the elements of which are
	finite sums of tensors of arbitrary finite grades on $V$.
	Consider the two-sided ideal generated by all elements of the form\footnote{It is common 
	(mainly in the mathematical literature) to choose a different sign convention here, 
	resulting in reversed signature in many of the following expressions. 
	One argument for the convention we use here is that
	e.g. squares of vectors in euclidean spaces become positive instead of negative.}
	$v \otimes v - q(v)$ for vectors $v$,
	\begin{equation}
		\mathcal{I}_q(V) := \Big\{ \sum_k A_k \otimes \big(v \otimes v - q(v)\big) \otimes B_k \quad : \quad v \in V,\ A_k,B_k \in \mathcal{T}(V) \Big\}.
	\end{equation}
	We define the geometric algebra over $V$ by quoting out this ideal 
	from $\mathcal{T}(V)$, so that, in the resulting algebra, 
	the square of a vector $v$ will be equal to the scalar $q(v)$.
	
	\begin{defn} \label{def_g}
		The \emph{geometric algebra} $\mathcal{G}(V,q)$ over the 
		vector space $V$ with quadratic form $q$ is defined by
		\begin{displaymath}
			\mathcal{G}(V,q) := \mathcal{T}(V) / \mathcal{I}_q(V).
		\end{displaymath}
		When it is clear from the context what vector space or quadratic form we are working with,
		we will often denote $\mathcal{G}(V,q)$ by $\mathcal{G}(V)$, or just $\mathcal{G}$.
	\end{defn}
	
	\noindent
	The product in $\mathcal{G}$, called the \emph{geometric} or \emph{Clifford product}, 
	is inherited from the tensor product in $\mathcal{T}(V)$ and
	we denote it by juxtaposition (or $\cdot$ if absolutely necessary),
	\begin{displaymath}
	\setlength\arraycolsep{2pt}
	\begin{array}{ccl}
		\mathcal{G} \times \mathcal{G} 	& \to 		& \mathcal{G}, \\
		(A,B) 							& \mapsto 	& AB := [A \otimes B] = A \otimes B + \mathcal{I}_q.
	\end{array}
	\end{displaymath}
	Note that this product is bilinear and associative.
	Furthermore,
	$$
		v^2 = [v \otimes v] = [v \otimes v - q(v)1] + q(v)1_{\mathcal{G}} = q(v),
	$$
	and 
	$$
		q(v+w) = (v+w)^2 = v^2 + vw + wv + w^2 = q(v) + vw + wv + q(w),
	$$
	so that, together with
	the definition of $\beta_q$, we immediately find the following 
	identities on $\mathcal{G}$ for all $v,w \in V$ 
	-- characteristic for a Clifford algebra:
	\begin{equation} \label{generator_relations}
		v^2 = q(v) \qquad \textrm{and} \qquad vw + wv = 2\beta_q(v,w).
	\end{equation}
	
	One of the most important consequences of this definition of the 
	geometric algebra is the following 
	
	\begin{prop}[Universality] \label{prop_universality}
		Let $\mathcal{A}$ be an associative algebra over $\mathbb{F}$ with a unit denoted by $1_\mathcal{A}$. 
		If $f\!: V \to \mathcal{A}$ is linear and
		\begin{equation} \label{universality_property}
			f(v)^2 = q(v)1_\mathcal{A} \quad \forall\ v \in V
		\end{equation}
		then $f$ extends uniquely to an $\mathbb{F}$-algebra homomorphism
		$F\!: \mathcal{G}(V,q) \to \mathcal{A}$, i.e.
		\begin{displaymath}
		\setlength\arraycolsep{2pt}
		\begin{array}{rcll}
			F(\alpha) &=& \alpha 1_\mathcal{A},	&\quad \forall\ \alpha \in \mathbb{F}, \\
			F(v) &=& f(v),  					&\quad \forall\ v \in V, \\
			F(xy) &=& F(x)F(y),					\\
			F(x+y) &=& F(x)+F(y),				&\quad \forall\ x,y \in \mathcal{G}.
		\end{array}
		\end{displaymath}
		Furthermore, $\mathcal{G}$ is the unique associative $\mathbb{F}$-algebra with this property.
	\end{prop}
	\begin{proof}
		Any linear map $f\!: V \to \mathcal{A}$ extends to a unique
		algebra homomorphism $\hat{f}\!: \mathcal{T}(V) \to \mathcal{A}$
		defined by $\hat{f}(u \otimes v) := f(u)f(v)$ etc.
		Property \eqref{universality_property} implies that $\hat{f}=0$ on the
		ideal $\mathcal{I}_q(V)$ and so $\hat{f}$ descends to a well-defined map $F$ 
		on $\mathcal{G}(V,q)$ which has the required properties.
		Suppose now that $\mathcal{C}$ is an associative $\mathbb{F}$-algebra
		with unit and that $i\!: V \hookrightarrow \mathcal{C}$ is an embedding
		with the property that any linear map $f\!: V \to \mathcal{A}$
		with property \eqref{universality_property} extends uniquely to an
		algebra homomorphism $F\!: \mathcal{C} \to \mathcal{A}$. Then the isomorphism
		from $V \subseteq \mathcal{G}$ to $i(V) \subseteq \mathcal{C}$
		clearly induces an algebra isomorphism $\mathcal{G} \to \mathcal{C}$.
	\end{proof}

	So far we have not made any assumptions on the dimension of $V$.
	We will come back to the infinite-dimensional case after discussing
	the more general Clifford algebra. 
	Here we will familiarize ourselves with the properties of quadratic forms
	on finite-dimensional spaces, which will lead to a better understanding
	of the structure of geometric algebras.
	For the remainder of this section
	we will therefore assume that $\dim V = n < \infty$.

	\begin{defn} \label{def_o_basis}
		A basis $\{e_1,\ldots,e_n\}$ of $(V,q)$ is said to be \emph{orthogonal} or \emph{canonical}
		if
		$$
			q(e_i + e_j) = q(e_i) + q(e_j) \quad \forall i \neq j,
		$$
		i.e. (for $\charop \mathbb{F} \neq 2$) if
		$\beta_q(e_i,e_j) = 0$ for all $i \neq j$. 
		The basis is called \emph{orthonormal} if we also have that $q(e_i) \in \{-1,0,1\}$ for all $i$.
	\end{defn}

	Orthogonal bases are essential for our understanding of geometric algebra,
	and the following theorem asserts their existence.
	
	\begin{thm} \label{thm_o_basis}
		If 
		$\charop \mathbb{F} \neq 2$ then 
		there exists an orthogonal basis of $(V,q)$.
	\end{thm}
	\begin{proof}
		If $q(x) = 0\ \forall x \in V$ 
		then every basis is orthogonal.
		Hence, we can assume that there exists $x \in V$ with $q(x) \neq 0$.
		Let $e$ be such an element and define
		\begin{eqnarray*}
			W = e^\perp &=& \{ x \in V : \beta_q(x,e) = 0 \}, \\
			\mathbb{F}e &=& \{ te : t \in \mathbb{F} \}.
		\end{eqnarray*}
		Since $\beta_q(te,e) = t\beta_q(e,e)$, it follows that $W \cap \mathbb{F}e = 0$
		and that $\dim W < \dim V$.
		
		Take an arbitrary $x \in V$ and form
		$$
			\tilde{x} = x - \frac{\beta_q(x,e)}{q(e)}e.
		$$
		Then we have $\tilde{x} \in W$ because
		$$
			\beta_q(\tilde{x},e) = \beta_q(x,e) - \frac{\beta_q(x,e)}{q(e)}\beta_q(e,e) = 0.
		$$
		This shows the orthogonal decomposition 
		$V = W \oplus \mathbb{F}e$ with respect to $q$.
		
		We now use induction over the dimension of $V$.
		For $\dim V = 1$ the theorem is trivial, so let us assume that
		the theorem is proved for all vector spaces of strictly lower
		dimension than $\dim V = n$. But then there exists an orthogonal
		basis $\{e_1,\ldots,e_{n-1}\}$ of $(W,q|_W)$. 
		Hence, $\{e_1,\ldots,e_{n-1},e\}$ is an orthogonal basis of $(V,q)$.
	\end{proof}
	
	\noindent
	Because this rather fundamental theorem is not valid for fields of characteristic
	two (such as $\mathbb{Z}_2)$, we will always assume that $\charop \mathbb{F} \neq 2$ when
	talking about geometric algebra. General fields and rings will be treated by
	the general Clifford algebra $\cl$, 
	which is introduced in the next subsection.
	
	\begin{prop}
		If $\mathbb{F} = \mathbb{R}$ and $q(e) > 0$ (or $q(e) < 0$) 
		for all elements $e$ of an orthogonal basis $E$, 
		then $q(v) = 0$ implies $v = 0$.
	\end{prop}
	\begin{proof}
		This is obvious when expanding $v$ in the basis $E$.
	\end{proof}
	
	\begin{thm}[Sylvester's Law of Inertia] \label{thm_sylvester}
		Assume that 
		$\mathbb{F}=\mathbb{R}$.
		If $E$ and $F$ are two orthogonal bases of $V$ and we set
		\begin{displaymath}
		\setlength\arraycolsep{2pt}
		\begin{array}{lcl}
			E^+ &:=& \{ e \in E : q(e)>0 \}, \\
			E^- &:=& \{ e \in E : q(e)<0 \}, \\
			E^0 &:=& \{ e \in E : q(e)=0 \},
		\end{array}
		\end{displaymath}
		and similarly for $F^{+,-,0}$,
		then
		\begin{displaymath}
		\setlength\arraycolsep{2pt}
		\begin{array}{c}
			|E^+| = |F^+|, \\
			|E^-| = |F^-|, \\
			\Span E^0 = \Span F^0.
		\end{array}
		\end{displaymath}
	\end{thm}
	\begin{proof}
		Let $V^\perp = \{ x \in V : \beta_q(x,y) = 0\ \forall y \in V \}$.
		It is obvious that $\Span E^0 \subseteq V^\perp$ because $E$ is
		an orthogonal basis.
		In order to show that $V^\perp \subseteq \Span E^0$, we take an
		arbitrary $x \in V^\perp$ and write
		$$
			x = x^0 + x^+ + x^-, \qquad x^{0,+,-} \in \Span E^{0,+,-}.
		$$
		Now, $0 = \beta_q(x,x^+) = \beta_q(x^+,x^+) \Rightarrow x^+ = 0$,
		and $0 = \beta_q(x,x^-) = \beta_q(x^-,x^-) \\ \Rightarrow x^- = 0$,
		so that $x = x^0$, which shows that $V^\perp = \Span E^0$.
		Analogously, we find that $V^\perp = \Span F^0$.
		
		In order to prove that $|E^+| = |F^+|$, we assume that $|E^+| < |F^+|$.
		But then we must have $|E^-| > |F^-|$ in which case
		$E^0 \cup F^+ \cup E^-$ has to be linearly dependent. Thus, there exist
		$x^0 \in \Span E^0$, $y^+ \in \Span F^+$, $x^- \in \Span E^-$, 
		not all zero, such that $x^0 + y^+ + x^- = 0$. This implies
		\begin{eqnarray*}
			\lefteqn{ \beta_q(x^0 + x^-,x^0 + x^-) = \beta_q(-y^+,-y^+) }\\
			&\Rightarrow& \beta_q(x^-,x^-) = \beta_q(y^+,y^+) 
			\ \Rightarrow \ x^- = y^+ = 0 \ \Rightarrow \ x^0=0
		\end{eqnarray*}
		This contradiction shows that $|E^+| \ge |F^+|$.
		Analogously, we find $|E^+| \le |F^+|$, which proves the theorem.
	\end{proof}
	
	\noindent
	This means that there is a unique \emph{signature} 
	$(s,t,u) := (|E^+|,|E^-|,|E^0|)$ associated to $(V,q)$.
	For the complex (and in general non-real) case we have the following simpler version:
	
	\begin{thm} \label{thm_complex_sylvester}
		If $E$ and $F$ are orthogonal bases of $V$ with $\mathbb{F}=\mathbb{C}$ and
		\begin{displaymath}
		\setlength\arraycolsep{2pt}
		\begin{array}{lcl}
			E^\times &:=& \{ e \in E : q(e) \neq 0 \}, \\
			E^0 &:=& \{ e \in E : q(e)=0 \},
		\end{array}
		\end{displaymath}
		and similarly for $F^{\times,0}$,
		then
		\begin{displaymath}
		\setlength\arraycolsep{2pt}
		\begin{array}{c}
			\Span E^\times = \Span F^\times, \\
			\Span E^0 = \Span F^0.
		\end{array}
		\end{displaymath}
		Furthermore, if $E^0 = \varnothing$ 
		(i.e. $q$ is nondegenerate) then there exists an orthonormal basis
		$\hat{E}$ with $q(e)=1 \ \forall\ e \in \hat{E}$.
	\end{thm}
	
	\noindent
	From the above follows that we can talk about \emph{the} signature of a quadratic form 
	or a metric tensor without ambiguity.
	We use the short-hand notation $\mathbb{R}^{s,t,u}$ to denote the $(s+t+u)$-dimensional real vector space
	with a quadratic form of signature $(s,t,u)$, while $\mathbb{C}^n$ is
	understood to be the complex $n$-dimensional space with a nondegenerate quadratic form.
	When $u=0$ or $t=u=0$ we may simply write $\mathbb{R}^{s,t}$ or $\mathbb{R}^{s}$.
	A space of type $\mathbb{R}^{n,0}$ is called \emph{euclidean} and $\mathbb{R}^{0,n}$
	\emph{anti-euclidean}, while the spaces $\mathbb{R}^{1,n}$ ($\mathbb{R}^{n,1}$) are called
	(\emph{anti-})\emph{lorentzian}.
	Note that, within real and complex spaces we can always find bases that are orthonormal.

	\begin{rem}
		The general condition for orthonormal bases to exist 
		is that the field $\mathbb{F}$ is a so called
		\emph{spin field}. This means that every $\alpha \in \mathbb{F}$ 
		can be written as $\alpha = \beta^2$ or $-\beta^2$ for some $\beta \in \mathbb{F}$.
		The fields $\mathbb{R}$, $\mathbb{C}$ and $\mathbb{Z}_p$ 
		for $p$ a prime with $p \equiv 3 \pmod{4}$, are spin,
		but e.g. $\mathbb{Q}$ is not.
	\end{rem}
	
	Consider now the geometric algebra $\mathcal{G}$ over a real or complex space $V$.
	If we pick an orthonormal basis $E = \{e_1,\ldots,e_n\}$ of $V$ then 
	it follows from Definition \ref{def_g} and Equation \eqref{generator_relations}
	that $\mathcal{G}$ is the free associative algebra generated by $E$ modulo the relations
	\begin{equation}
		e_i^2 = q(e_i) \in \{-1,0,1\} 
		\qquad \textrm{and} \qquad e_i e_j = - e_j e_i,\ i \neq j.
	\end{equation}

	\begin{exmp}
		Consider the case $V = \mathbb{R}^2$, i.e. the euclidean plane.
		The corresponding geometric algebra $\mathcal{G}(\mathbb{R}^2)$, 
		also called the \emph{plane algebra},
		is generated by an orthonormal basis $\{e_1,e_2\}$, with
		$e_1^2 = e_2^2 = 1$ and $e_1 e_2 = -e_2 e_1$.
		Starting with the unit $1$ and multiplying from the left or right
		with $e_1$ and $e_2$ we see that the process stops at $e_1e_2$ and
		it follows that $\mathcal{G}(\mathbb{R}^2)$ is spanned by four elements:
		$$
			1,\ e_1,\ e_2,\ e_1e_2.
		$$
		Similarly, (as seen already in the introduction)
		the \emph{space algebra} $\mathcal{G}(\mathbb{R}^3)$ is
		span\-ned by eight elements, 
		$$
			1,\ e_1,\ e_2,\ e_3,\ e_1e_2,\ e_2e_3,\ e_1e_3,\ e_1e_2e_3,
		$$
		generated by an orthonormal basis $\{e_1,e_2,e_3\}$
		of vectors
		which are mutually anticommuting and squaring to 1.
	\end{exmp}

	In the same manner we conclude that a general 
	geometric algebra $\mathcal{G}(V)$ is spanned by $2^n$
	elements of the form
	\begin{equation} \label{basis_of_g}
		\{E_{i_1 i_2 \ldots i_k}\}_{\genfrac{}{}{0pt}{}{k = 1,\ldots,n}{1 \le i_1 < i_2 < \ldots <i_k \le n}}, 
		\quad \textrm{with} \quad
		E_{i_1 i_2 \ldots i_k} := e_{i_1} e_{i_2} \ldots e_{i_k},
	\end{equation}
	and it will follow from the construction of $\cl$
	in the next subsection that there cannot be any linear relations
	among these elements, so that \eqref{basis_of_g} 
	is actually a basis of $\mathcal{G}$.
	
	\begin{rem}
		In the case that $q=0$, the resulting algebra of anticommuting, 
		nilpotent elements is called a \emph{Grassmann algebra}. 
		When associated to a vector space $V$ it is also commonly
		referred to as the \emph{exterior algebra} of $V$ 
		and is often denoted $\wedge^* V$ and the multiplication by $\wedge$.
		Note that $\mathcal{G}(V) \cong \wedge^* V$ holds as an 
		$\mathbb{F}$-algebra isomorphism only when $q=0$, 
		but remains as an isomorphism of (graded) vector spaces 
		also for non-trivial quadratic forms.
	\end{rem}

	One element in $\mathcal{G}$ deserves special attention, namely the so-called \emph{pseudoscalar} 
	(sometimes called \emph{the (oriented) volume element}) 
	formed by the product of all elements of an orthonormal basis,
	\begin{equation}
		I := e_1 e_2 \ldots e_n.
	\end{equation}
	Note that this definition is basis independent up to orientation when $q$ is nondegenerate.
	Indeed, let $\{Re_1,\ldots,Re_n\}$ be another orthonormal basis with the same orientation, 
	where $R \in \SO(V,q)$,
	the group of linear orientation-preserving transformations which leave $q$ invariant\footnote{Such
	transformations will be introduced and discussed in Section \ref{sec_groups}.}.
	Then,
	due to the anticommutativity of the $e_i$:s,
	and orthogonality of the rows and columns of the matrix of $R$,
	\begin{eqnarray*}
		\lefteqn{ Re_1 Re_2 \ldots Re_n 
		= \sum_{j_1,\ldots,j_n} R_{j_1 1} R_{j_2 2} \ldots R_{j_n n}\ e_{j_1} e_{j_2} \ldots e_{j_n} }\\
		&=& \sum_{\pi \in S_n} \textrm{sign}(\pi)\ R_{1\pi(1)} R_{2\pi(2)} \ldots R_{n\pi(n)}\ e_1 e_2 \ldots e_n \\
		&=& \det R \ e_1 e_2 \ldots e_n = I,
	\end{eqnarray*}
	where $S_n$ denotes the symmetric group of order $n$, 
	i.e. all permutations of $n$ elements.
	Note that, by selecting a certain pseudoscalar for $\mathcal{G}$ we also 
	impose a certain orientation on $V$. There is no such thing as an absolute orientation;
	instead all statements concerning orientation will be made relative to the chosen one.
	
	The square of the pseudoscalar is given by (and gives information about) 
	the signature and dimension of $(V,q)$.
	For $\mathcal{G}(\mathbb{R}^{s,t,u})$ we have that
	\begin{equation} \label{pseudoscalar_squared}
		I^2 = (-1)^{\frac{1}{2}n(n-1) + t} \delta_{u,0}, \quad \textrm{where} \ \ n=s+t+u.
	\end{equation}
	We say that $\mathcal{G}$ is \emph{degenerate} if the quadratic form 
	is degenerate, or equivalently if $I^2=0$.
	For odd $n$, $I$ commutes with all elements in $\mathcal{G}$, 
	while for even $n$, $I$ anticommutes with all vectors $v \in V$.

	\begin{exmp}
		The pseudoscalar in $\mathcal{G}(\mathbb{R}^{2})$ is $I = e_1e_2$.
		Here, $I^2 = -1$, so that the subalgebra spanned by
		$\{1,I\}$ is isomorphic to $\mathbb{C}$.
		Also, any vector $\boldsymbol{v} \in \mathbb{R}^2$ can be written 
		$$
			\boldsymbol{v} = ae_1 + be_2 = e_1 (a + bI) = |\boldsymbol{v}|e_1 e^{\varphi I},
		$$
		for some $a,b,\varphi \in \mathbb{R}$,
		which is similar to the polar form of a complex number.
	\end{exmp}
	
	\begin{exmp}
		In $\mathcal{G}(\mathbb{R}^{3})$ the pseudoscalar is
		given by $I = e_1e_2e_3$. Again, $I^2 = -1$ and
		$\{1,I\}$ forms a subalgebra isomorphic to the complex
		numbers. However, there are infinitely many other such subalgebras
		embedded in $\mathcal{G}$ since for every choice of
		$\boldsymbol{e} = \alpha e_1 + \beta e_2 + \gamma e_3$ on the unit sphere
		$\alpha^2 + \beta^2 + \gamma^2 = 1$,
		the element
		\begin{equation} \label{unit_bivector}
			\boldsymbol{e}I = \alpha e_2e_3 + \beta e_3e_1 + \gamma e_1e_2
		\end{equation}
		also satisfies 
		$(\boldsymbol{e}I)^2 = \boldsymbol{e}I\boldsymbol{e}I = \boldsymbol{e}^2I^2 = -1$.
	\end{exmp}
	
	\begin{exmp} \label{exmp_complex_numbers}
		Let us also consider an example with non-euclidean
		signature. The simplest case is the anti-euclidean line,
		i.e. $\mathbb{R}^{0,1}$. A unit basis element (also pseudoscalar) 
		$i \in \mathbb{R}^{0,1}$ satisfies $i^2 = -1$, and
		$$
			\mathcal{G}(\mathbb{R}^{0,1}) = \Span_\mathbb{R} \{1,i\} = \mathbb{C}.
		$$
	\end{exmp}
	
	\begin{exmp} \label{exmp_quaternions}
		An even more interesting case is the anti-euclidean 
		plane, $\mathbb{R}^{0,2}$. Let $\{i,j\}$ denote an orthonormal
		basis, i.e. $i^2 = j^2 = -1$ and $ij = -ji$.
		Denoting the pseudoscalar by $k$, we have $k = ij$ and
		$$
			i^2 = j^2 = k^2 = ijk = -1,
		$$
		which are the defining relations of the quaternion algebra
		$\mathbb{H}$. It follows that $\mathcal{G}(\mathbb{R}^{0,2})$ 
		and $\mathbb{H}$ are isomorphic as $\mathbb{R}$-algebras.
	\end{exmp}
	
	\begin{exc}
		Prove the Jordan-von Neumann theorem (see Example \ref{exmp_JvN}).
		\emph{Hint:} First prove additivity and linearity over $-1$, 
		then linearity over $\mathbb{Z}$, $\mathbb{Q}$, and finally $\mathbb{R}$.
	\end{exc}

	\begin{exc}
		Prove Theorem \ref{thm_complex_sylvester}.
	\end{exc}

	\begin{exc}
		Show that $\mathbb{Z}_p$ is a spin field when $p$ is a prime
		such that $p \equiv 3 \pmod{4}$.
	\end{exc}
	
	\begin{exc}
		Verify formula \eqref{pseudoscalar_squared} for the square
		of the pseudoscalar.
	\end{exc}

	\begin{exc} \label{exc_twodim_isomorphism}
		Find an $\mathbb{R}$-algebra isomorphism
		$$
			\mathcal{G}(\mathbb{R}^2) \to \mathbb{R}^{2 \times 2} = \{ \textrm{$2 \times 2$ matrices with entries in $\mathbb{R}$} \}.
		$$
		\emph{Hint:} Find \emph{simple} $2 \times 2$ matrices $\gamma_1, \gamma_2$ 
		satisfying $\gamma_1^2 = \gamma_2^2 = 1_{2 \times 2}$.
	\end{exc}

	\begin{exc}
		The \emph{center} $Z(\mathcal{A})$ of an algebra $\mathcal{A}$ consists of those
		elements which commute with all of $\mathcal{A}$.
		Show that, for odd dimensional $V$, 
		the center of $\mathcal{G}(V)$ is 
		$Z(\mathcal{G}) = \Span_\mathbb{F} \{1,I\}$,
		while for even dimensions, the center consists of the scalars $\mathbb{F}$ only.
	\end{exc}

	\begin{exc} \label{exc_exponential}
		Assume (as usual) that $\dim V < \infty$ and
		let $|\cdot|$ be any norm on $\mathcal{G}(V)$.
		Prove that $e^x := \sum_{k=0}^{\infty} \frac{x^k}{k!}$ converges
		and that $e^x e^{-x} = e^{-x} e^x = 1$ for all $x \in \mathcal{G}$.
	\end{exc}

\subsection{Combinatorial Clifford algebra $\cl(X,R,r)$}
	
	We now take a temporary step away from the comfort of fields and vector spaces
	and instead consider the purely algebraic features of geometric algebra that were
	uncovered in the previous subsection.
	
	Note that we could roughly write 
	\begin{equation}
		\mathcal{G}(V) = \textrm{Span}_\mathbb{F} \thinspace \{E_A\}_{A \subseteq \{1,2,\ldots,n\}}
	\end{equation}
	for an $n$-dimensional space $V$ over $\mathbb{F}$, 
	and that the geometric product of these basis elements behaves as
	\begin{equation}
		E_A E_B = \tau(A,B) \ E_{A \bigtriangleup B}, \quad \textrm{where} \quad \tau(A,B)=1,-1, \ \textrm{or}\ 0,
	\end{equation}
	and $A \symdiff B := (A \cup B) \!\smallsetminus\! (A \cap B)$ is 
	the \emph{symmetric difference} between the sets $A$ and $B$.
	Motivated by this we consider the following generalization.
	
	\begin{defn} \label{def_cl}
		Let $X$ be a finite set and $R$ a commutative ring with unit. 
		Let $r\!: X \to R$ be an arbitrary function which 
		is to be thought of as a \emph{signature} on $X$.
		The \emph{Clifford algebra} over $(X,R,r)$ is defined as the set\footnote{Again,
		see Appendix \ref{app_algebra} if the notation is unfamiliar.}
		\begin{displaymath}
			\cl(X,R,r) := \bigoplus_{\mathscr{P}(X)} R,
		\end{displaymath}
		i.e. the free $R$-module generated by $\mathscr{P}(X)$, the set of all subsets of $X$.
		We may use the shorter notation $\cl(X)$, or just $\cl$, when 
		the current choice of $X$, $R$ and $r$ is clear from the context.
		Furthermore, we call $R$ the \emph{scalars} of $\cl$.
	\end{defn}
	
	\begin{exmp}
		A typical element of $\cl(\{x,y,z\},\mathbb{Z},r)$, i.e.
		the Clifford algebra over a set of three elements with integer scalars,
		could for example look like
		\begin{equation}
			5 \varnothing + 3 \{x\} + 2 \{y\} - \{x,y\} + 12 \{x,y,z\}.
		\end{equation}
	\end{exmp}
	
	We have not yet defined a product on $\cl$.
	In addition to being $R$-bilinear and associative, we would like the product to satisfy
	$\{x\}^2 = r(x) \varnothing$ for $x \in X$, 
	$\{x\}\{y\} = -\{y\}\{x\}$ for $x \neq y \in X$ and 
	$\varnothing A = A \varnothing = A$	for all subsets $A \in \mathscr{P}(X)$.
	In order to arrive at such a product we make use of the following

	\begin{lem} \label{lem_finite_tau}
		There exists a map $\tau\!: \mathscr{P}(X) \times \mathscr{P}(X) \to R$ such that
		\begin{displaymath}
		\begin{array}{rl}
			i)  & \tau(\{x\},\{x\}) = r(x) \quad \forall\ x \in X, \\[5pt]
			ii) & \tau(\{x\},\{y\}) = -\tau(\{y\},\{x\}) \quad \forall\ x,y \in X : x \neq y, \\[5pt]
			iii)& \tau(\varnothing,A) = \tau(A,\varnothing) = 1 \quad \forall\ A \in \mathscr{P}(X), \\[5pt]
			iv) & \tau(A,B) \tau(A \symdiff B,C) = \tau(A,B \symdiff C) \tau(B,C) \quad \forall\ A,B,C \in \mathscr{P}(X), \\[5pt]
			v)  & \tau(A,B) \in \{-1,1\} \quad \textrm{if} \quad A \cap B = \varnothing.
		\end{array}
		\end{displaymath}
	\end{lem}
	\begin{proof}
		We proceed by induction on the cardinality $|X|$ of $X$.
		For $X=\varnothing$ the lemma is trivial, so let $z \in X$
		and assume the lemma holds for $Y := X \!\smallsetminus\! \{z\}$.
		Hence, there is a $\tau'\!: \mathscr{P}(Y) \times \mathscr{P}(Y) \to R$ 
		which has the properties (\emph{i})-(\emph{v}) above. If $A \subseteq Y$ we
		write $A' = A \cup \{z\}$ and, for $A,B$ in $\mathscr{P}(Y)$ we extend
		$\tau'$ to $\tau\!: \mathscr{P}(X) \times \mathscr{P}(X) \to R$ in the
		following way:
		\begin{displaymath}
		\setlength\arraycolsep{2pt}
		\begin{array}{rcl}
			\tau(A,B)	&:=& \tau'(A,B) \\[3pt]
			\tau(A',B)	&:=& (-1)^{|B|} \tau'(A,B) \\[3pt]
			\tau(A,B')	&:=& \tau'(A,B) \\[3pt]
			\tau(A',B') &:=& r(z) (-1)^{|B|} \tau'(A,B)
		\end{array}
		\end{displaymath}
		Now it is straightforward to verify that (\emph{i})-(\emph{v}) holds for
		$\tau$, which completes the proof.
	\end{proof}

	\begin{defn} \label{def_cl_prod}
		Define the \emph{Clifford product}
		\begin{displaymath}
		\setlength\arraycolsep{2pt}
		\begin{array}{ccc}
			\cl(X) \times \cl(X) & \to & \cl(X) \\
			(A,B) & \mapsto & AB
		\end{array}
		\end{displaymath}
		by taking $AB := \tau(A,B) A \symdiff B$ for $A,B \in \mathscr{P}(X)$ and extending linearly.
		We choose to use the $\tau$ which is constructed as in the proof of Lemma \ref{lem_finite_tau} 
		by consecutively adding elements from the set $X$.
		A unique such $\tau$ may only be selected after imposing a certain order (orientation) on the set $X$.
	\end{defn}
	
	Using Lemma \ref{lem_finite_tau} one easily verifies that this product has all
	the properties that we asked for above. For example, in order to verify associativity
	we note that 
	\begin{equation}
		A(BC) = A\big(\tau(B,C) B \symdiff C\big) = \tau(A, B \symdiff C) \tau(B,C) A \symdiff (B \symdiff C),
	\end{equation}
	while 
	\begin{equation}
		(AB)C = \tau(A,B) (A \symdiff B)C = \tau(A,B) \tau(A \symdiff B,C) (A \symdiff B) \symdiff C.
	\end{equation}
	Associativity now follows from (\emph{iv}) and the associativity of the symmetric difference.
	As is expected from the analogy with $\mathcal{G}$, we also have the property that
	different basis elements of $\cl$ commute up to a sign.

	\begin{prop} \label{prop_reverse}
		If $A,B \in \mathscr{P}(X)$ then
		\begin{displaymath}
			AB = (-1)^{\frac{1}{2}|A|(|A|-1) \ + \ \frac{1}{2}|B|(|B|-1) \ + \ \frac{1}{2}|A \bigtriangleup B|(|A \bigtriangleup B|-1)} BA.
		\end{displaymath}
	\end{prop}
	\begin{proof}
		By the property (\emph{v}) in Lemma \ref{lem_finite_tau} it is sufficient
		to prove this for $A = \{a_1\} \{a_2\} \ldots \{a_k\}$, $B = \{b_1\} \{b_2\} \ldots \{b_l\}$,
		where $a_i$ are disjoint elements in $X$ and likewise for $b_i$.
		If $A$ and $B$ have $m$ elements in common then 
		$$
			AB = (-1)^{(k-m)l\ +\ m(l-1)}BA = (-1)^{kl-m}BA
		$$
		by property (\emph{ii}).
		But then we are done, since 
		\begin{equation*}
			\frac{1}{2} \big(-k(k-1) - l(l-1) + (k+l-2m)(k+l-2m-1) \big) \equiv kl + m \pmod{2}. 
		\end{equation*}
	\end{proof}
	
	We are now ready to make the formal connection between $\mathcal{G}$ and $\cl$.
	Let $(V,q)$ be a finite-dimensional vector space over $\mathbb{F}$ with a quadratic form.
	Pick an orthogonal basis $E = \{e_1,\ldots,e_n\}$ of $V$ 
	and consider the Clifford algebra $Cl(E,\mathbb{F},q|_E)$.
	Define $f\!: V \to \cl$ by $f(e_i) := \{e_i\}$ for $i=1,\ldots,n$ and extend linearly.
	We then have 
	$$
	\setlength\arraycolsep{2pt}
	\begin{array}{rl}
		f(v)^2 	&= f(\sum_i v_i e_i) f(\sum_j v_j e_j) = \sum_{i,j} v_i v_j f(e_i)f(e_j) \\[10pt]
				&= \sum_{i,j} v_i v_j \{e_i\}\{e_j\} = \sum_i v_i^2 \{e_i\}^2 \\[10pt]
				&= \sum_i v_i^2 q|_E(e_i) \varnothing = \sum_i v_i^2 q(e_i) \varnothing = q(\sum_i v_i e_i) \varnothing = q(v) \varnothing.
	\end{array}
	$$
	By Proposition \ref{prop_universality}, $f$ extends uniquely to a homomorphism 
	$F\!: \mathcal{G} \to \cl$
	which is easily seen to be surjective from the property (\emph{v}).
	Moreover, since $V$ is vector space isomorphic to the linear span of
	the singleton sets $\{e_i\}$, and because $\cl$ solves the same
	universal problem (Proposition \ref{prop_universality}) as $\mathcal{G}(V,q)$,
	we arrive at an $\mathbb{F}$-algebra isomorphism
	\begin{equation} \label{g_isomorphic_to_cl}
		\mathcal{G}(V,q) \cong \cl(E,\mathbb{F},q|_E).
	\end{equation}
	We make this equivalence between $\mathcal{G}$ and $\cl$ even more transparent by 
	suppressing the unit $\varnothing$ in expressions and writing simply $e$ instead
	of $\{e\}$ for singletons $e \in E$. For example, with an orthonormal
	basis $\{e_1,e_2,e_3\}$ in $\mathbb{R}^3$, both $\mathcal{G}$ and $\cl$ are then spanned by
	\begin{equation}
		\{ 1, \ e_1, e_2, e_3, \ e_1 e_2, e_1 e_3, e_2 e_3, \ e_1 e_2 e_3 \}.
	\end{equation}
	
	There is a natural grade structure on $Cl$ given by the cardinality of the subsets of $X$.
	Consider the following

	\begin{defn} \label{def_grade_space}
		The \emph{subspace of $k$-vectors} in $\cl$, or the \emph{grade-$k$ part} of $\cl$, is defined by
		\begin{displaymath}
			\cl^k(X,R,r) := \bigoplus_{A \in \mathscr{P}(X)\ :\ |A| = k} R.
		\end{displaymath}
		Of special importance are the \emph{even} and \emph{odd subspaces},
		\begin{displaymath}
			\cl^{\pm}(X,R,r) := \bigoplus_{k \ \text{is}\ \substack{\text{\tiny even}\\ \text{\tiny odd}}} \cl^k(X,R,r).
		\end{displaymath}
		This notation carries over to the corresponding subspaces of $\mathcal{G}$ 
		and we write $\mathcal{G}^k$, $\mathcal{G}^\pm$ etc. where for example 
		$\mathcal{G}^0 = \mathbb{F}$ and $\mathcal{G}^1 = V$.
		The elements of $\mathcal{G}^2$ are also called \emph{bivectors}, while
		arbitrary elements of $\mathcal{G}$ are conventionally called \emph{multivectors}.
	\end{defn}
	
	\noindent
	We then have a split of $\cl$ into graded subspaces as
	\begin{equation} \label{cl_grades}
	\setlength\arraycolsep{2pt}
	\begin{array}{rcl}
		\cl(X)
			&=& \cl^{+} \!\oplus \cl^{-} \\[5pt]
			&=& \cl^0 \oplus \cl^1 \oplus \cl^2 \oplus \ldots \oplus \cl^{|X|}.
	\end{array}
	\end{equation}
	Note that, under the Clifford product, $\cl^{\pm} \cdot \cl^{\pm} \subseteq \cl^{+}$ 
	and $\cl^{\pm} \cdot \cl^{\mp} \subseteq \cl^{-}$
	(by the properties of the symmetric difference).
	Hence, the even-grade elements $\cl^{+}$ form a subalgebra of $\cl$.

	\begin{rem}	
		Strictly speaking, although the algebra $\cl$ is $\mathbb{Z}_2$-graded
		in this way, it is not $\mathbb{Z}$-graded
		but \emph{filtered} in the sense that
		$$
			\left( \bigoplus_{i \le k} \cl^i \right) \cdot \left( \bigoplus_{j \le l} \cl^j \right) 
			\subseteq \bigoplus_{m \le k+l} \cl^m.
		$$
	\end{rem}
	
	In $\cl(X,R,r)$ we have the possibility of defining a unique pseudoscalar
	independently of the signature $r$, namely the set $X$ itself.
	Note, however, that it can only be normalized if $X^2 = \tau(X,X) \in R$ is invertible,
	which requires that $r$ is nondegenerate
	(i.e. that its values in $R$ are invertible).
	We will almost always talk about pseudoscalars in the setting 
	of nondegenerate vector spaces, so this will not be a problem.

	\begin{exmp} \label{exmp_space_algebra_decomp}
		The space algebra
		$$
			\mathcal{G}(\mathbb{R}^3) = \mathcal{G}^0 \oplus \mathcal{G}^1 \oplus \mathcal{G}^2 \oplus \mathcal{G}^3
		$$
		can be decomposed into scalars, vectors, bivectors and trivectors
		(multiples of the pseudoscalar). Using the observation \eqref{unit_bivector}
		for the relation between a vector and a bivector in the space algebra,
		we write an arbitrary multivector
		$x \in \mathcal{G}(\mathbb{R}^3)$ as
		$$
			x = \alpha + \boldsymbol{a} + \boldsymbol{b}I + \beta I,
		$$
		with $\alpha,\beta \in \mathbb{R}$, $\boldsymbol{a},\boldsymbol{b} \in \mathbb{R}^3$.
		The subalgebra $\mathcal{G}^+(\mathbb{R}^3)$ 
		of scalars and bivectors is actually isomorphic 
		to the quaternion algebra, as is easily verified by taking e.g.
		$i := -e_1I, j := -e_2I, k := -e_3I$.
	\end{exmp}
	
	\begin{exc}
		Prove that the symmetric difference $\ \symdiff\ $ is associative
		(for arbitrary sets -- finite or infinite).
	\end{exc}

	\begin{exc}
		Verify that conditions (\emph{i})-(\emph{v}) hold for $\tau$ 
		constructed in the proof of Lemma \ref{lem_finite_tau}.
	\end{exc}

	\begin{exc}
		Let $\mathcal{G} = \mathcal{G}(\mathbb{R}^3)$,
		$C_0 = \mathcal{G}^0 \oplus \mathcal{G}^3$, and
		$C_1 = \mathcal{G}^1 \oplus \mathcal{G}^2$.
		The complex numbers are as usual denoted by $\mathbb{C}$.
		Find (natural) maps
		\begin{eqnarray*}
			&\alpha: &C_0 \to \mathbb{C} \\
			&\beta:  &C_1 \to \mathbb{C}^3 \\
			&\gamma: &C_0 \times C_1 \to C_1
		\end{eqnarray*}
		such that $\alpha$ is an algebra isomorphism,
		$\beta$ is a bijection, and such that the diagram below commutes
		$$
			\begin{array}{rclcl}
			C_0 & \times & C_1 & \stackrel{\gamma}{\to} & C_1 \\
			\alpha \downarrow & & \downarrow \beta & & \downarrow \beta \\
			\mathbb{C} & \times & \mathbb{C}^3 & \stackrel{\delta}{\to} & \mathbb{C}^3
			\end{array}
		$$
		where $\delta$ is the usual scalar multiplication on $\mathbb{C}^3$, 
		i.e. $$\delta(z,(z_1,z_2,z_3)) = (zz_1,zz_2,zz_3).$$
	\end{exc}

	\begin{rem}
		Hence, we can think of $\mathbb{C}^3$ (as complex vector space)
		as $C_1$ interpreted as a vector space over the field $C_0$.
	\end{rem}

	\begin{exc}
		Verify that there is an $\mathbb{R}$-algebra isomorphism
		$\varphi: \mathcal{G}(\mathbb{R}^4) \to \mathbb{H}^{2 \times 2}$ such that
		$$
			\varphi(e_1) = \left[ \begin{smallmatrix}
			0 & -i  \\[3pt] i & 0
			\end{smallmatrix} \right], \ 
			\varphi(e_2) = \left[ \begin{smallmatrix}
			0 & -j  \\[3pt] j & 0
			\end{smallmatrix} \right], \ 
			\varphi(e_3) = \left[ \begin{smallmatrix}
			0 & -k  \\[3pt] k & 0
			\end{smallmatrix} \right], \ 
			\varphi(e_4) = \left[ \begin{smallmatrix}
			1 & 0  \\[3pt] 0 & -1
			\end{smallmatrix} \right].
		$$
	\end{exc}

\subsection{Standard operations}
	
	A key feature of Clifford algebras is that they contain a surprisingly large
	amount of structure. In order to conveniently access the power of this structure
	we need to introduce powerful notation. 
	Most of the following definitions will be made on $\cl$ for simplicity,
	but because of the equivalence between $\mathcal{G}$ and $\cl$ they
	carry over to $\mathcal{G}$ in a straightforward manner.
	The amount of new notation might seem heavy, and we therefore
	recommend the reader not to get stuck trying to learn everything at once,
	but to come back to this subsection later as a reference when the notation is needed.
	
	We will find it convenient to introduce the notation that for any proposition $P$, 
	$(P)$ will denote the number $1$ if $P$ is true and $0$ if $P$ is false.
	
	\begin{defn} \label{def_operations}
		For $A,B \in \mathscr{P}(X)$ we define
		\begin{displaymath}
		\setlength\arraycolsep{2pt}
		\begin{array}{ccll}
			A \wedge B &:=& (A \cap B = \varnothing) \thinspace AB	& \quad\textrm{\emph{outer product}} \\[5pt]
			A \liprod B &:=& (A \subseteq B) \thinspace AB			& \quad\textrm{\emph{left inner product}} \\[5pt]
			A \riprod B &:=& (A \supseteq B) \thinspace AB			& \quad\textrm{\emph{right inner product}} \\[5pt]
			A * B &:=& (A = B) \thinspace AB						& \quad\textrm{\emph{scalar product}} \\[5pt]
			\langle A \rangle_k &:=& (|A|=k) \thinspace A			& \quad\textrm{\emph{projection on grade $k$}} \\[5pt]
			A^\star &:=& (-1)^{|A|} \thinspace A					& \quad\textrm{\emph{grade involution}} \\[3pt]
			A^\dagger &:=& (-1)^{\binom{|A|}{2}} \thinspace A		& \quad\textrm{\emph{reversion}}
		\end{array}
		\end{displaymath}
		and extend linearly to $\cl(X,R,r)$.
	\end{defn}
	
	\noindent
	The grade involution is also called the (\emph{first}) \emph{main involution}. It has the property
	\begin{equation} \label{grade_inv_property}
		(xy)^\star = x^\star y^\star, \quad v^\star = -v
	\end{equation}
	for all $x,y \in Cl(X)$ and $v \in Cl^1(X)$, as is easily verified 
	by expanding in linear combinations of elements in $\mathscr{P}(X)$
	and using that $|A \symdiff B| \equiv |A| + |B|\ (\textrm{mod}\ 2)$.
	The reversion earns its name from the property
	\begin{equation} \label{reversion_property}
		(xy)^\dagger = y^\dagger x^\dagger, \quad v^\dagger = v,
	\end{equation}
	and it is sometimes called the \emph{second main involution} or the \emph{principal antiautomorphism}.
	This reversing behaviour follows directly from Proposition \ref{prop_reverse}.
	We will find it convenient to have a name for the composition of these two involutions.
	Hence, we define the \emph{Clifford conjugate} $x^\cliffconj$ of $x \in Cl(X)$ by $x^\cliffconj := x^{\star\dagger}$
	and observe the property
	\begin{equation} \label{cliffconj_property}
		(xy)^\cliffconj = y^\cliffconj x^\cliffconj, \quad v^\cliffconj = -v.
	\end{equation}
	Note that all the above involutions act by changing sign on some of the graded subspaces.
	We summarize the action of these involutions in Table \ref{table_involutions}.
	Note the periodicity.
	
	\begin{table}[ht]
		\begin{displaymath}
		\begin{array}{c|cccccccc}
				& \cl^0 & \cl^1 & \cl^2 & \cl^3 & \cl^4 & \cl^5 & \cl^6 & \cl^7 \\
			\hline
			\\[-1.5ex]
			\star				& + & - & + & - & + & - & + & - \\
			\dagger				& + & + & - & - & + & + & - & - \\
			\cliffconj			& + & - & - & + & + & - & - & + \\
		\end{array}
		\end{displaymath}
		\caption{The action of involutions on graded subspaces of $\cl$. \label{table_involutions}}
	\end{table}

	\begin{exmp}
		If $x \in \mathcal{G}(\mathbb{R}^2)$ and $x^\star = -x$,
		then $x$ has to be a vector since those are the only odd 
		elements in this case.
		Similarly, we see that if $x$ is a multivector in 
		$\mathcal{G}(\mathbb{R}^3)$ then the expression $x^\dagger x$ 
		cannot contain any bivector or trivector parts since
		the expression is self-reversing, i.e.
		$(x^\dagger x)^\dagger = x^\dagger x^{\dagger\dagger} = x^\dagger x$.
	\end{exmp}
	
	\begin{exmp}
		In the case $\mathcal{G}(\mathbb{R}^{0,1}) \cong \mathbb{C}$ we see that
		the grade involution (or Clifford conjugate) corresponds to the complex conjugate,
		while in the quaternion algebra $\mathcal{G}(\mathbb{R}^{0,2}) \cong \mathbb{H}$
		(see Example \ref{exmp_quaternions}), 
		the Clifford conjugate corresponds to the quaternion conjugate
		since it changes sign on all non-scalar (i.e. imaginary) grades,
		$$
			(\alpha + \beta i + \gamma j + \delta k)^\cliffconj = \alpha - \beta i - \gamma j - \delta k.
		$$
	\end{exmp}
	
	The scalar product has the symmetric property $x * y = y * x$ for all $x,y \in \cl$.
	Therefore, it forms a symmetric bilinear map $\cl \times \cl \to R$
	which is seen to be degenerate if and only if $\cl$ 
	(i.e. the signature $r$) is degenerate.
	Note that this map coincides with the bilinear form $\beta_q$ 
	when restricted to $V = \mathcal{G}^1(V,q)$,
	and that subspaces of different grade are orthogonal with respect to the scalar product.
	
	The following identities relating the inner, outer, and scalar products\footnote{Another
	product that is often seen in the context of geometric algebra is
	the \emph{inner product}, defined by 
	$$
		A \iprod B := (A \subseteq B \ \textrm{or} \ A \supseteq B)\thinspace AB = A \liprod B + A \riprod B - A * B.
	$$
	We will stick to the left and right inner products, however, 
	because they admit a simpler handling of grades. 
	For example, the corresponding identities in Proposition \ref{prop_trix} 
	would need to be supplied with grade restrictions.
	Also beware that the meaning of the symbols $\liprod$ and $\riprod$
	is sometimes reversed in the literature.}
	will turn out to be extremely useful:

	\begin{prop} \label{prop_trix}
		For all $x,y,z \in \cl(X)$ we have
		$$
			\setlength\arraycolsep{2pt}
			\begin{array}{ccc}
				x \wedge (y \wedge z) &=& (x \wedge y) \wedge z, \\[2pt]
				x \liprod (y \riprod z) &=& (x \liprod y) \riprod z, \\[2pt]
				x \liprod (y \liprod z) &=& (x \wedge y) \liprod z, \\[2pt]
				x * (y \liprod z) &=& (x \wedge y) * z,
			\end{array}
		$$
		and
		$$
			1 \wedge x = x \wedge 1 = 1 \liprod x = x \riprod 1 = x.
		$$
	\end{prop}
	\begin{proof}
		This follows directly from Definition \ref{def_operations} and basic set logic.
		For example, taking $A,B,C \in \mathscr{P}(X)$ we have
		(consider drawing a Venn diagram)
		$$
		\setlength\arraycolsep{2pt}
		\begin{array}{rl}
			A \liprod (B \liprod C)
				&= (B \subseteq C)(A \subseteq B \symdiff C)ABC \\[5pt]
				&= (B \subseteq C \ \textrm{and} \ A \subseteq C \!\smallsetminus\! B)ABC \\[5pt]
				&= (A \cap B = \varnothing \ \textrm{and} \ A \cup B \subseteq C)ABC \\[5pt]
				&= (A \cap B = \varnothing)(A \symdiff B \subseteq C)ABC \\[5pt]
				&= (A \wedge B) \liprod C.
		\end{array}
		$$
		The other identities are proven in an equally simple way.
	\end{proof}
	
	\noindent
	Note that the first identity in the above proposition states that
	the wedge product $\wedge$ is associative.
	It is not difficult to see that the algebra $(\mathcal{G},\wedge)$ 
	with this product is isomorphic (as graded algebras) 
	to the exterior algebra $\wedge^* V$.
	
	In order to be able to work 
	efficiently with Clifford algebras it is crucial to understand
	how vectors behave under these operations.
	
	\begin{prop} \label{prop_trix2}
		For all $x,y \in \cl(X)$ and $v \in \cl^1(X)$ we have
		\begin{displaymath}
		\setlength\arraycolsep{2pt}
		\begin{array}{rcl}
			vx				&=& v \liprod x + v \wedge x, \\[3pt]
			v \liprod x		&=& \frac{1}{2}(vx - x^\star v) = - x^\star \!\riprod v, \\[3pt]
			v \wedge x		&=& \frac{1}{2}(vx + x^\star v) = \phantom{-} x^\star \!\wedge v, \\[3pt]
			v \liprod (xy)	&=& (v \liprod x)y + x^\star (v \liprod y).
		\end{array}
		\end{displaymath}
	\end{prop}
	
	\noindent
	The first three identities are shown simply by using linearity and set relations,
	while the fourth follows immediately from the second.
	
	\begin{exmp}
		For 1-vectors $u,v \in \cl^1$ we have the basic relations
		\begin{equation}
			uv = \langle uv \rangle_0 + \langle uv \rangle_2 = u * v + u \wedge v,
		\end{equation}
		\begin{equation}
			u \liprod v = v \liprod u = u \riprod v = u * v = \frac{1}{2}(uv + vu),
		\end{equation}
		and
		\begin{equation} \label{vector_wedge}
			u \wedge v = - v \wedge u  = \frac{1}{2}(uv - vu),
		\end{equation}
		while for bivectors $B \in \cl^2$ and $v \in \cl^1$ we have
		\begin{equation}
			vB = \langle vB \rangle_1 + \langle vB \rangle_3 = v \liprod B + v \wedge B,
		\end{equation}
		\begin{equation}
			v \liprod B = \frac{1}{2}(vB - Bv) = -B \riprod v,
		\end{equation}
		\begin{equation}
			v \wedge B = \frac{1}{2}(vB + Bv) = B \wedge v,
		\end{equation}
		and $v \riprod B = B \liprod v = v * B = 0$. 
		In the special case that $B = u \wedge w$ 
		is an outer product of two vectors $u,w$ (i.e. a 2-blade),
		we find from taking the vector part of the the last identity in
		Proposition \eqref{prop_trix2} the useful expansion
		\begin{equation} \label{vec_2-blade_expansion}
			v \liprod (u \wedge w) = (v * u)w - (v*w)u.
		\end{equation}
		This expansion is generalized in Exercise \ref{exc_anti_derivation}.
	\end{exmp}

	It is sometimes useful to expand the various products and involutions 
	in terms of the grades involved.
	The following identities are left as exercises.
	
	\begin{prop} \label{prop_grade_def}
		For all $x,y \in Cl(X)$ we have
		\begin{displaymath}
		\setlength\arraycolsep{2pt}
		\begin{array}{ccl}
			x \wedge y		&=& \sum_{n,m \geq 0} \big\langle \langle x \rangle_n \langle y \rangle_m \big\rangle_{n+m}, \\[5pt]
			x \liprod y		&=& \sum_{0 \leq n \leq m} \big\langle \langle x \rangle_n \langle y \rangle_m \big\rangle_{m-n}, \\[5pt]
			x \riprod y		&=& \sum_{n \geq m \geq 0} \big\langle \langle x \rangle_n \langle y \rangle_m \big\rangle_{n-m}, \\[5pt]
			x \iprod y		&=& \sum_{n,m \geq 0} \big\langle \langle x \rangle_n \langle y \rangle_m \big\rangle_{|n-m|}, \\[5pt]
			x * y			&=& \langle xy \rangle_0, \\[5pt]
			x^\star			&=& \sum_{n \geq 0} (-1)^n \langle x \rangle_n, \\[3pt]
			x^\dagger		&=& \sum_{n \geq 0} (-1)^{\binom{n}{2}} \langle x \rangle_n.
		\end{array}
		\end{displaymath}
	\end{prop}
	
	In the general setting of a Clifford algebra with scalars in a ring $R$, we need
	to be careful about the notion of linear (in-)dependence. A subset
	$\{x_1,x_2,\ldots,x_m\}$ of $\cl$ is called \emph{linearly dependent} if 
	there exist $r_1,\ldots,r_m \in R$, not all zero, such that
	\begin{equation}
		r_1 x_1 + r_2 x_2 + \ldots + r_m x_m = 0.
	\end{equation}
	Note that a single nonzero 1-vector could be linearly dependent in this context
	(e.g. with $R = \mathbb{Z}_4$, $x_1 = 2e_1$ and $r_1 = 2$).
	We will prove an important theorem concerning linear dependence where we use the following
	
	\begin{lem} \label{lem_blade_det}
		If $u_1,u_2,\ldots,u_k$ and $v_1,v_2,\ldots,v_k$ are 1-vectors then
		\begin{displaymath}
			(u_1 \wedge u_2 \wedge \cdots \wedge u_k) * (v_k \wedge v_{k-1} \wedge \cdots \wedge v_1)
				= \det\thinspace \left[ \begin{array}{ccc} 
					u_1*v_1 & \cdots & u_1*v_k \\ 
					\vdots & & \vdots \\
					u_k*v_1 & \cdots & u_k*v_k
				\end{array} \right].
		\end{displaymath}
	\end{lem}
	\begin{proof}
		Since both sides of the expression are multilinear and alternating in both
		the $u_i$:s and the $v_i$:s, we need only consider ordered disjoint elements 
		$\{e_i\}$ in the basis of singleton sets in $X$.
		Both sides are zero, except in the case
		$$
		\setlength\arraycolsep{2pt}
		\begin{array}{l}
			(e_{i_1} e_{i_2} \ldots e_{i_k}) * (e_{i_k} e_{i_{k-1}} \ldots e_{i_1}) = \\[5pt]
			\qquad = r(e_{i_1}) r(e_{i_2}) \ldots r(e_{i_k}) = \det\thinspace \big[ r(e_{i_p}) \delta_{p,q} \big]_{1 \leq p,q \leq k} \\[5pt]
			\qquad = \det\thinspace \big[ e_{i_p}*e_{i_q} \big]_{1 \leq p,q \leq k},
		\end{array}
		$$
		so we are done.
	\end{proof}
	
	\begin{thm} \label{thm_linear_indep}
		The 1-vectors $\{x_1,x_2,\ldots,x_m\}$ are linearly independent if and only if
		the m-vector $\{x_1 \wedge x_2 \wedge \cdots \wedge x_m\}$ is linearly independent.
	\end{thm}

	\begin{proof}[*Proof]
		Assume that $r_1 x_1 + \ldots + r_m x_m = 0$, where, say, $r_1 \neq 0$.
		Then
		$$
		\setlength\arraycolsep{2pt}
		\begin{array}{l}
			r_1 (x_1 \wedge \cdots \wedge x_m) = (r_1 x_1) \wedge x_2 \wedge \cdots \wedge x_m \\[5pt]
			\qquad	= (r_1 x_1 + \ldots + r_m x_m) \wedge x_2 \wedge \cdots \wedge x_m = 0,
		\end{array}
		$$
		since $x_i \wedge x_i = 0$.
		
		Conversely, assume that $rX = 0$ for
		$r \neq 0$ in $R$ and $X = x_1 \wedge \cdots \wedge x_m$.
		We will use the basis minor theorem for arbitrary rings which
		can be found in the appendix (see Appendix \ref{app_matrix_theorems}).
		Assume that $x_j = x_{1j} e_1 + \ldots + x_{nj} e_n$, where $x_{ij} \in R$
		and $e_i \in X$ are basis elements such that $e_i^2 = 1$. This assumption on
		the signature is no loss in generality, since this theorem only concerns
		the exterior algebra associated to the outer product. It will only serve to
		simplify our reasoning below. Collect the coordinates in a matrix
		$$
			A := \left[
			\begin{array}{cccc}
				rx_{11} & x_{12} & \cdots & x_{1m} \\
				rx_{21} & x_{22} & \cdots & x_{2m} \\
				\vdots  & \vdots & 		  & \vdots \\
				rx_{n1} & x_{n2} & \cdots & x_{nm} \\
			\end{array}
			\right] \in R^{n \times m},\ m \leq n
		$$
		and note that we can expand $rX$ in a grade-$m$ basis as
		$$
			rX = \sum_{E \subseteq X : |E| = m} (rX * E^\dagger) E = \sum_{E \subseteq X : |E| = m} (\det A_{E,\{1,\ldots,m\}}) E,
		$$
		where we used Lemma \ref{lem_blade_det}.
		We find that the determinant of each $m \times m$ minor of $A$ is zero.
		
		Now, let $k$ be the rank of $A$. Then we must have $k<m$, and if $k=0$ then
		$rx_{i1}=0$ and $x_{ij}=0$ for all $i=1,\ldots,n$, $j>1$.
		But that would mean that $\{x_1,\ldots,x_m\}$ are linearly dependent.
		Therefore we assume that $k>0$ and, without loss of generality, that
		$$
			d := \det \left[
			\begin{array}{cccc}
				rx_{11} & x_{12} & \cdots & x_{1k} \\
				\vdots  & \vdots & 		  & \vdots \\
				rx_{k1} & x_{k2} & \cdots & x_{kk} \\
			\end{array}
			\right] \neq 0.
		$$
		By Theorem \ref{thm_basis_minor} (Basis minor) there exist $r_1,\ldots,r_k \in R$
		such that
		$$
			r_1 rx_1 + r_2 x_2 + \ldots + r_k x_k + dx_m = 0.
		$$
		Hence, $\{x_1,\ldots,x_m\}$ are linearly dependent.
	\end{proof}
	
	For our final set of operations,
	we will consider a nondegenerate geometric algebra $\mathcal{G}$ with
	pseudoscalar $I$. The nondegeneracy implies that there exists a natural
	duality between the inner and outer products.
	
	\begin{defn} \label{def_dual}
		We define the \emph{dual} of $x \in \mathcal{G}$ by $x\dual := xI^{-1}$.
		Furthermore, the \emph{dual outer product} or \emph{meet}, denoted $\vee$,
		is defined such that the diagram
		\begin{displaymath}
		\setlength\arraycolsep{2pt}
		\begin{array}{rcccccl}
							& \mathcal{G} 	& \times 	& \mathcal{G} 	& \xrightarrow{\vee} 	& \mathcal{G} \\
			(\cdot)\dual 	& \downarrow 	& 			& \downarrow	&						& \downarrow	& (\cdot)\dual \\
							& \mathcal{G} 	& \times	& \mathcal{G}	& \xrightarrow{\wedge}	& \mathcal{G}
		\end{array}
		\end{displaymath}
		commutes, i.e. $(x \vee y)\dual := x\dual \wedge y\dual$,
		implying $x \vee y = ((xI^{-1}) \wedge (yI^{-1}))I$.
	\end{defn}
	
	\begin{rem}
		In $\cl(X)$, the corresponding dual of $A \in \mathscr{P}(X)$ is 
		$$
			A\dual = AX^{-1} = \tau(X,X)^{-1} \tau(A,X) A \symdiff X \ \propto \ A^c,
		$$
		the complement of the set $A$.
		Hence, we actually find that the dual is a linearized version of a 
		sign (or orientation) -respecting complement. 
		This motivates our choice of notation.
	\end{rem}
	
	\begin{prop} \label{prop_dual}
		For all $x,y \in \mathcal{G}$ we have
		\begin{equation} \label{dual_products}
		\setlength\arraycolsep{2pt}
		\begin{array}{ccc}
			x \liprod y\dual &=& (x \wedge y)\dual, \\
			x \wedge y\dual &=& (x \liprod y)\dual, \\
		\end{array}
		\end{equation}
		and if $x \in \mathcal{G}^k$ then $x\dual \in \mathcal{G}^{\dim V - k}$.
	\end{prop}
	\begin{proof}
		Using Proposition \ref{prop_trix} and the fact that 
		$xI = x \liprod I \ \forall x$, we obtain
		$$
			x \liprod (yI^{-1}) = x \liprod (y \liprod I^{-1}) = (x \wedge y) \liprod I^{-1} = (x \wedge y)I^{-1},
		$$
		and from this follows also the second identity
		$$
			(x \wedge y\dual)I^{-1}I = (x \liprod y^\mathbf{cc})I = (x \liprod y)I^{-2}I.
		$$
		Lastly, the grade statement is obvious from the above definition and remark.
	\end{proof}

	\begin{exmp}
		The cross product in $\mathbb{R}^3$ can be defined as the dual
		of the outer product,
		\begin{equation} \label{cross_product}
			\boldsymbol{a} \times \boldsymbol{b} := (\boldsymbol{a} \wedge \boldsymbol{b})\dual,
		\end{equation}
		for $\boldsymbol{a},\boldsymbol{b} \in \mathbb{R}^3$. 
		From Proposition \ref{prop_dual} and e.g. \eqref{vec_2-blade_expansion} 
		we easily obtain the familiar relations
		$$
			\boldsymbol{a} \times (\boldsymbol{b} \times \boldsymbol{c}) = \big( \boldsymbol{a} \wedge (\boldsymbol{b} \wedge \boldsymbol{c})\dual \big)\dual 
			= \big( \boldsymbol{a} \liprod (\boldsymbol{b} \wedge \boldsymbol{c}) \big)^\mathbf{cc} = -(\boldsymbol{a}*\boldsymbol{b})\boldsymbol{c} + (\boldsymbol{a}*\boldsymbol{c})\boldsymbol{b},
		$$
		and
		$$
			\boldsymbol{a} * (\boldsymbol{b} \times \boldsymbol{c}) = \boldsymbol{a} \liprod (\boldsymbol{b} \wedge \boldsymbol{c})\dual = (\boldsymbol{a} \wedge \boldsymbol{b} \wedge \boldsymbol{c})*I^\dagger,
		$$
		which by Lemma \ref{lem_blade_det} and a choice of basis becomes the
		usual determinant expression.
	\end{exmp}
	
	\begin{rem}
		It is instructive to compare the exterior algebra $(\mathcal{G},\wedge)$ 
		together with this duality operation to the
		language of differential
		forms and the Hodge $*$ duality operation, which are completely equivalent
		(through the isomorphism $(\mathcal{G},\wedge) \cong \wedge^* V$). 
		In that setting one	often starts with the outer product 
		and then uses a given metric tensor to define a dual.
		The inner product is then defined from the outer product and dual according to \eqref{dual_products}.
		The exact definition of the dual varies, but a common choice in the literature for differential forms
		is $*x := (I^{-1}x)^\dagger$, so that $**x = IxI^\dagger$
		(cp. e.g. \cite{nakahara}).
	\end{rem}

	\begin{exc}
		Show that $|A \symdiff B| \equiv |A| + |B| \pmod{2}$ and
		verify Eqs. \eqref{grade_inv_property}-\eqref{cliffconj_property} 
		from the definitions of these involutions.
	\end{exc}

	\begin{exc}
		Prove the remaining identities in Proposition \ref{prop_trix}.
	\end{exc}
	
	\begin{exc}
		Prove Proposition \ref{prop_trix2}.
	\end{exc}
	
	\begin{exc}
		Prove Proposition \ref{prop_grade_def}.
	\end{exc}
	
	\begin{exc} \label{exc_meet_algebra}
		Show that the algebra $(\mathcal{G},\vee)$ with the meet product
		is associative with unit $I$.
		(The combined algebra $(\mathcal{G},\wedge,\vee)$ 
		is sometimes called \emph{the double Cayley algebra}.)
	\end{exc}
	
	\begin{exc}
		Show that the cross product as defined in \eqref{cross_product}
		gives rise to the familiar cartesian coordinate expression 
		(and has the correct orientation)
		when expanding in an orthonormal basis.
	\end{exc}

	\begin{exc} \label{exc_anti_derivation}
		For $v \in \cl^1$ define the map
		$$
			\setlength\arraycolsep{2pt}
			\begin{array}{rccl}
				\partial_v\!: 	& \cl 	& \to 		& \cl, \\
								& x 	& \mapsto 	& \partial_v(x) = v \liprod x.
			\end{array}
		$$		
		Verify that $\partial_v \circ \partial_v = 0$ and that
		$$\partial_v(xy) = \partial_v(x)y + x^\star \partial_v(y)$$
		as well as
		$$\partial_v(x \wedge y) = \partial_v(x) \wedge y + x^\star \wedge \partial_v(y)$$
		for all $x,y \in \cl$
		(such $\partial_v$ is called an \emph{anti-derivation}).
		Furthermore, note that this leads to the following useful formula:
		\begin{equation} \label{vec_blade_expansion}
			v \liprod (a_1 \wedge a_2 \wedge \cdots \wedge a_m) 
			= \sum_{k=1}^m (-1)^{k-1} (v * a_k) a_1 \wedge \cdots \wedge \check{a}_{k} \wedge \cdots \wedge a_m,
		\end{equation}
		(where $\check{\phantom{a}}$ denotes deletion) for any $v,a_1,\ldots,a_m \in \cl^1$.
		A generalized expansion for higher grades can be found in Appendix \ref{app_inner_prod_expansion}.
	\end{exc}

	\begin{exc}
		Let $\boldsymbol{v} \in \mathbb{R}^3$ and $B \in \mathcal{G}^2(\mathbb{R}^3)$.
		In each of the following cases, find $w \in \mathcal{G}(\mathbb{R}^3)$
		that satisfies
		\item{a)} $w \wedge w = 1 + \boldsymbol{v} + B$
		\item{b)} $w \wedge (1 + \boldsymbol{v} + B) = 1$
	\end{exc}

	\begin{exc} \label{exc_bivectors_closed}
		Show that $\cl^2$ is closed under the commutator bracket, i.e.
		$[A,B] := AB-BA \in \cl^2$
		for all $A,B \in \cl^2$.
	\end{exc}

	\begin{exc} \label{exc_wedge_commutativity}
		Show that $x_p \wedge x_r = (-1)^{pr} x_r \wedge x_p$
		for all $x_k \in \cl^k$.
	\end{exc}

%% file: clifford_vector.tex
	We will now leave the general setting of combinatorial 
	Clifford algebra for a moment and instead
	focus on the geometric properties of $\mathcal{G}$ and its newly
	defined operations in the context of vector spaces and linear transformations.
	
\subsection{Blades and subspaces}

	The concept of blades is central for understanding the 
	geometry encoded in a geometric algebra $\mathcal{G}(V)$.
	
	\begin{defn} \label{def_blades}
		A \emph{blade}, or \emph{simple multivector}, 
		is an outer product of 1-vectors. We define the following:
		\begin{displaymath}
		\setlength\arraycolsep{2pt}
		\begin{array}{lcll}
			\mathcal{B}_k &:=& \{ v_1 \wedge v_2 \wedge \cdots \wedge v_k \in \mathcal{G} : v_i \in V \}
				& \quad\textrm{\emph{the set of $k$-blades}} \\[5pt]
			\mathcal{B} &:=& \bigcup_{k=0}^{\infty} \mathcal{B}_k
				& \quad\textrm{\emph{the set of all blades}} \\[5pt]
			\mathcal{B}^* &:=& \mathcal{B} \!\smallsetminus\! \{0\}
				& \quad\textrm{\emph{the nonzero blades}} \\[5pt]
			\mathcal{B}^\times &:=& \{ B \in \mathcal{B} : B^2 \neq 0\}
				& \quad\textrm{\emph{the invertible blades}}
		\end{array}
		\end{displaymath}
		The \emph{orthogonal basis blades} associated to an orthogonal basis 
		$E = \{e_i\}_{i=1}^{\dim V}$ consist of 
		the basis of $\mathcal{G}$ generated by $E$, i.e.
		\begin{displaymath}
			\mathcal{B}_E := \{ e_{i_1} \wedge e_{i_2} \wedge \cdots \wedge e_{i_k} \in \mathcal{G} : i_1 < i_2 < \ldots < i_k \} 
		\end{displaymath}
		(corresponding to $\mathscr{P}(E)$ in $\cl$).
		We also include the unit 1 among the blades and call it the \emph{$0$-blade}.
	\end{defn}
	
	\noindent
	Note that $\mathcal{B}_k \subseteq \mathcal{G}^k$ and that 
	(e.g. by expanding in an orthogonal basis; see Exercise \ref{exc_blade_expansion})
	we can expand a blade as a sum of geometric products,
	\begin{equation} \label{blade_expansion}
		a_1 \wedge a_2 \wedge \cdots \wedge a_k 
			= \frac{1}{k!} \sum_{\pi \in S_k} \textrm{sign}(\pi)\ a_{\pi(1)} a_{\pi(2)} \ldots a_{\pi(k)}.
	\end{equation}
	This expression is clearly similar to a determinant, except that this is a product of \emph{vectors}
	instead of scalars.
	
	The key property of blades is that they represent linear subspaces of $V$.
	This is made precise by the following

	\begin{prop} \label{prop_blade_subspace}
		If $A = a_1 \wedge a_2 \wedge \cdots \wedge a_k \neq 0$ is a 
		nonzero $k$-blade and $a \in V$ then
		\begin{displaymath}
			a \wedge A = 0 \quad \Leftrightarrow \quad a \in \textrm{\emph{Span}} \{a_1,a_2,\ldots,a_k\}.
		\end{displaymath}
	\end{prop}
	\begin{proof}
		This follows directly from Theorem \ref{thm_linear_indep} since 
		$\{a_1,\ldots,a_k\}$ are linearly independent, and $a \wedge A=0$ if and only if
		$\{a,a_1,\ldots,a_k\}$ are linearly dependent.
	\end{proof}
	
	\noindent
	Hence, to every nonzero $k$-blade $A = a_1 \wedge a_2 \wedge \cdots \wedge a_k$ 
	there corresponds a unique $k$-dimensional subspace
	\begin{equation} \label{blade_subspace}
		\bar{A} := \{ a \in V : a \wedge A = 0 \} = \textrm{Span} \{a_1,a_2,\ldots,a_k\}.
	\end{equation}
	Conversely, if $\bar{A} \subseteq V$ is a $k$-dimensional subspace of $V$, 
	then we can find a nonzero $k$-blade $A$ representing $\bar{A}$ by simply taking a basis
	$\{a_i\}_{i=1}^k$ of $\bar{A}$ and forming
	\begin{equation} \label{subspace_blade}
		A := a_1 \wedge a_2 \wedge \cdots \wedge a_k.
	\end{equation}
	We thus have the geometric interpretation of blades as subspaces with an associated
	orientation (sign) and magnitude.
	Since every element in $\mathcal{G}$ is a linear combination of orthogonal basis blades,
	we can think of every element as representing a linear combination of orthogonal subspaces.
	In the case of a nondegenerate algebra these basis subspaces are 
	nondegenerate as well.
	On the other hand, any blade which represents a nondegenerate subspace can also be 
	treated as a basis blade associated to an orthogonal basis.
	This will follow in the discussion below.
	
	\begin{prop} \label{prop_blade_product}
		Every $k$-blade can be written as a geometric product of $k$ vectors.
	\end{prop}
	\begin{proof}
		Take a nonzero $A = a_1 \wedge \cdots \wedge a_k \in \mathcal{B}^*$. 
		Pick an orthogonal basis $\{e_i\}_{i=1}^k$ of the subspace $(\bar{A},q|_{\bar{A}})$.
		Then we can write $a_i = \sum_j \beta_{ij} e_j$ for some $\beta_{ij} \in \mathbb{F}$,
		and $A = \det\ [\beta_{ij}]\ e_1 e_2 \ldots e_k$ by \eqref{blade_expansion}.
	\end{proof}
	
	\noindent
	There are a number of useful consequences of this result.
	
	\begin{cor}
		For any blade $A \in \mathcal{B}$, $A^2$ is a scalar.
	\end{cor}
	\begin{proof}
		Use the expansion of $A$ above to obtain 
		\begin{equation} \label{blade_squared}
			A^2 = (\det\ [\beta_{ij}])^2\ (-1)^{\frac{1}{2}k(k-1)} q(e_1) q(e_2) \ldots q(e_k) \in \mathbb{F}. \qedhere
		\end{equation}
	\end{proof}
	
	\begin{cor}
		If $A \in \mathcal{B}^\times$ then $A$ has an inverse $A^{-1} = \frac{1}{A^2}A$.
	\end{cor}
	
	\begin{exmp}
		In $\mathbb{R}^3$, every bivector is also a 2-blade.
		This follows by duality, since if $B$ is a bivector then
		$B\dual = BI^{-1} =: b$ is a vector. 
		Choosing an orthonormal basis $\{e_1,e_2,e_3\}$ such that,
		say, $b = \beta e_3$, we have $B = bI = \beta e_1 \wedge e_2$.
		The same situation is not true in $\mathbb{R}^4$, however,
		where we e.g. have a bivector 
		$$
			B := e_1 \wedge e_2 + e_3 \wedge e_4 = e_1e_2(1 - I)
		$$
		such that $B^2 = -(1 - I)^2 = -2(1 - I)$,
		which is not a scalar.
		(Actually, one can show in general (see Exercise \ref{exc_bivector_blade})
		that $B^2 \in \mathbb{F}$ is both a necessary and sufficient condition
		for a bivector $B$ to be a blade.)
	\end{exmp}
	
	Another useful consequence of Proposition 
	\ref{prop_blade_product} is the following
	
	\begin{cor}
		If $A \in \mathcal{B}^\times$ then $q$ is nondegenerate on $\bar{A}$
		and there exists an
		orthogonal basis $E$ of $V$ such that $A \in \mathcal{B}_E$.
	\end{cor}
	\begin{proof}
		The first statement follows directly from \eqref{blade_squared}.
		For the second statement note that, since $q$ is nondegenerate on $\bar{A}$,
		we have $\bar{A} \cap \bar{A}^\perp = 0$. 
		Take an orthogonal basis $\{e_i\}_{i=1}^k$ of $\bar{A}$.
		For any $v \in V$ we have that 
		$$
			v - \sum_i \beta_q(v,e_i) q(e_i)^{-1} e_i \in \bar{A}^\perp.
		$$
		Thus, $V = \bar{A} \oplus \bar{A}^\perp$ and
		we can extend $\{e_i\}_{i=1}^k$ to an orthogonal basis of $V$ consisting of 
		one part in $\bar{A}$ and one part in $\bar{A}^\perp$.
		By rescaling this basis we have $A = e_1 \wedge \cdots \wedge e_k$.
	\end{proof}
	
	\begin{rem}
		Note that if we have an orthogonal basis of a subspace of $V$ where $q$ is degenerate, 
		then it may not be possible to extend this basis to an orthogonal basis for all of $V$.
		Another way to state this is that one or more vectors in the 
		corresponding product expansion may fail to be invertible.
		If the space is euclidean or anti-euclidean, though, orthogonal bases can always be extended
		(e.g. using the Gram-Schmidt algorithm; see Exercise \ref{exc_gram_schmidt}).
	\end{rem}
	
	\begin{exmp} \label{exmp_twodim_lorentz}
		The simplest lorentzian space is $\mathbb{R}^{1,1}$,
		with an orthonormal basis basis conventionally denoted $\{e_0,e_1\}$, $e_0^2 = 1$, $e_1^2 = -1$.
		It has precisely two one-dimensional degenerate subspaces $\mathbb{R}n_\pm$, where
		$$
			n_\pm := e_0 \pm e_1, \quad n_\pm^2 = 0,
		$$
		but these cannot be orthogonal. In fact,
		$$
			n_+ * n_- = e_0^2 - e_1^2 = 2,
		$$
		but they of course span the space since they are linearly independent, 
		and $n_+ \wedge n_- = 2e_1e_0$. 
		In general, the degenerate part of a $n$-dimensional lorentzian space
		will be a $(n-1)$-dimensional cone, called the \emph{null-} or \emph{light-cone}.
	\end{exmp}
	
	It is useful to be able to work efficiently also with general bases of $V$ 
	and $\mathcal{G}$ which	need not be orthogonal.
	Let $\{e_1,\ldots,e_n\}$ be \emph{any} basis of $V$.
	Then the set set of blades formed out of this basis, $\{e_\mathbf{i}\}_\mathbf{i}$,
	is a basis of $\mathcal{G}(V)$, where we
	use a multi-index notation
	\begin{equation} \label{genbase_index}
		\mathbf{i} = (i_1,i_2,\ldots,i_k), \quad i_1 < i_2 < \ldots < i_k, \quad 0 \leq k \leq n
	\end{equation}
	and
	\begin{equation} \label{genbase_blade}
		e_{()} := 1, \quad e_{(i_1,i_2,\ldots,i_k)} := e_{i_1} \wedge e_{i_2} \wedge \cdots \wedge e_{i_k}.
	\end{equation}
	Sums over $\mathbf{i}$ are understood to be performed over all allowed such indices.
	If $\mathcal{G}$ is nondegenerate then the scalar product $(A,B) \mapsto A * B$
	is also nondegenerate and we can find a so-called \emph{reciprocal basis} $\{e^1,\ldots,e^n\}$
	of $V$ such that
	\begin{equation} \label{reciprocal_def}
		e^i * e_j = \delta^i_j.
	\end{equation}
	The reciprocal basis is easily verified (Exercise \ref{exc_reciprocal}) to be given by
	\begin{equation} \label{reciprocal_exp}
		e^i = (-1)^{i-1} (e_1 \wedge \cdots \wedge \check{e}_i \wedge \cdots \wedge e_n) e_{(1,\ldots,n)}^{-1},
	\end{equation}
	where $\check{}$ denotes a deletion.
	Furthermore, we have that $\{e^\mathbf{i}\}_\mathbf{i}$ is a reciprocal basis of $\mathcal{G}$,
	where $e^{(i_1,\ldots,i_k)} := e^{i_k} \wedge \cdots \wedge e^{i_1}$
	(note the order).
	This follows since by Lemma \ref{lem_blade_det} and \eqref{reciprocal_def},
	\begin{equation} \label{reciprocal_g}
		e^\mathbf{i} * e_\mathbf{j} 
			= (e^{i_k} \wedge \cdots \wedge e^{i_1}) * (e_{j_1} \wedge \cdots \wedge e_{j_l})
			= \delta^k_l \det\ \big[ e^{i_p} * e_{j_q} \big]_{1 \le p,q \le k} = \delta^\mathbf{i}_\mathbf{j}.
	\end{equation}
	We now have the coordinate expansions
	\begin{equation} \label{coordinate_expansions}
	\setlength\arraycolsep{2pt}
	\begin{array}{rcll}
		v &=& \sum_i (v * e^i) e_i = \sum_i (v * e_i) e^i & \quad \forall\ v \in V, \\[5pt]
		x &=& \sum_\mathbf{i} (x * e^\mathbf{i}) e_\mathbf{i} = \sum_\mathbf{i} (x * e_\mathbf{i}) e^\mathbf{i} & \quad \forall\ x \in \mathcal{G}(V),
	\end{array}
	\end{equation}
	which also gives a
	geometric understanding of an arbitrary
	multivector in $\mathcal{G}(V)$ as a linear combination of general 
	(not necessarily orthogonal) subspaces of $V$.

	\begin{exmp}
		In the euclidean case an orthonormal basis $\{e_i\}_i$,
		satisfying $e_i^2 = 1$ for all $i$, is its own reciprocal basis, i.e. $e^i = e_i$.
		In the lorentzian case $\mathbb{R}^{1,n}$ 
		the standard orthonormal basis $\{e_0,e_1,\ldots,e_n\}$, 
		with $e_0^2 = 1$, has as reciprocal basis
		$\{e^0,e^1,\ldots,e^n\}$ where $e^0 = e_0$ and $e^j = -e_j$, $j=1,2,\ldots,n$.
		The non-orthogonal basis $\{n_+,n_-\}$ in Example \ref{exmp_twodim_lorentz}
		has the reciprocal basis $\{\frac{1}{2}n_-,\frac{1}{2}n_+\}$.
	\end{exmp}
	
	In addition to being useful in coordinate expansions, the general and reciprocal bases
	also provide a geometric understanding of the dual operation because of the following
	
	\begin{thm} \label{thm_reciprocal_blade}
		Assume that $\mathcal{G}$ is nondegenerate.
		If $A = a_1 \wedge \cdots \wedge a_k \in \mathcal{B}^*$ and we extend $\{a_i\}_{i=1}^k$
		to a basis $\{a_i\}_{i=1}^n$ of $V$ then
		\begin{displaymath}
			A\dual \propto a^{k+1} \wedge a^{k+2} \wedge \cdots \wedge a^n,
		\end{displaymath}
		where $\{a^i\}_i$ is the reciprocal basis of $\{a_i\}_i$.
	\end{thm}
	\begin{proof}
		We obtain by induction on $k$ that
		\begin{eqnarray*}
			\lefteqn{ (a_k \wedge a_{k-1} \wedge \cdots \wedge a_1) \liprod (a^1 \wedge \cdots \wedge a^k \wedge a^{k+1} \wedge \cdots \wedge a^n) }\\
			&=& a_k \liprod \left( (a_{k-1} \wedge \cdots \wedge a_1) \liprod (a^1 \wedge \cdots \wedge a^{k-1} \wedge a^k \wedge \cdots \wedge a^n) \right) \\
			&=& \{\textrm{induction assumption}\} = a_k \liprod (a^k \wedge \cdots \wedge a^n) \\
			&=& a^{k+1} \wedge \cdots \wedge a^n,
		\end{eqnarray*}
		where in the last step we used the expansion formula \eqref{vec_blade_expansion},
		plus orthogonality \eqref{reciprocal_def}. 
		It follows that
		$$
			A\dual = A \liprod I^{-1} 
			\propto (a_k \wedge \cdots \wedge a_1) \liprod (a^1 \wedge \cdots \wedge a^k \wedge a^{k+1} \wedge \cdots \wedge a^n)
			= a^{k+1} \wedge \cdots \wedge a^n.
		$$
		This can also be proved using an expansion of the 
		inner product into sub-blades; see Appendix \ref{app_inner_prod_expansion}.
	\end{proof}
	
	\begin{cor}
		If $A$ and $B$ are blades then $A\dual$, $A \wedge B$, $A \vee B$ and $A \liprod B$ are blades as well.
	\end{cor}
	
	\noindent
	The blade-subspace correspondence then gives us a geometric interpretation of
	these operations.
	
	\begin{prop} \label{prop_subspace_ops}
		If $A,B \in \mathcal{B}^*$ are nonzero blades then $\overline{A\dual} = \bar{A}^\perp$ and
		\begin{displaymath}
		\setlength\arraycolsep{2pt}
		\begin{array}{rcl}
			A \wedge B \neq 0 
				&\Rightarrow& 
				\overline{A \wedge B} = \bar{A} + \bar{B} \ \textrm{and}\ \bar{A} \cap \bar{B} = 0, \\[5pt]
			\bar{A} + \bar{B} = V
				 &\Rightarrow& 
				\overline{A \vee B} = \bar{A} \cap \bar{B}, \\[5pt]
			A \liprod B \neq 0 
				&\Rightarrow& 
				\overline{A \liprod B} = \bar{A}^\perp \cap \bar{B}, \\[5pt]
			\bar{A} \subseteq \bar{B} 
				&\Rightarrow& 
				A \liprod B = AB, \\[5pt]
			\bar{A} \cap \bar{B}^\perp \neq 0 
				&\Rightarrow& 
				A \liprod B = 0.
		\end{array}
		\end{displaymath}
	\end{prop}
	
	\noindent
	The proofs of the statements in the above corollary and proposition are left as exercises. 
	Some of them can be found in \cite{svensson} and \cite{hestenes_sobczyk}.
	
	\begin{exc} \label{exc_blade_expansion}
		Prove Eqn. \eqref{blade_expansion}, e.g. by noting that both
		sides are multilinear and alternating.
	\end{exc}

	\begin{exc} \label{exc_reciprocal}
		Verify that $\{e^i\}_i$ defined in \eqref{reciprocal_exp} is a 
		reciprocal basis w.r.t. $\{e_i\}_i$.
	\end{exc}

	\begin{exc}
		Prove the corollary to Theorem \ref{thm_reciprocal_blade}.
	\end{exc}

	\begin{exc}
		Prove Proposition \ref{prop_subspace_ops}.
	\end{exc}

	\begin{exc}
		Show that the square of an arbitrary blade 
		$A = a_1 \wedge \ldots \wedge a_k \in \mathcal{B}_k$ is
		$A^2 = (-1)^{\frac{1}{2}k(k-1)} \det [a_i * a_j]_{1\le i,j \le k}$.
		Also, show in two different ways that
		\begin{equation} \label{blade_square}
			(a \wedge b)^2 = (a * b)^2 - a^2b^2
		\end{equation}
		for $a,b \in V$.
	\end{exc}

	\begin{exc} \label{exc_basis_identities}
		Let $\{e_i\}_{i=1}^n$ be an arbitrary basis of a nondegenerate space.
		Show that $\sum_i e_i e^i = n$ and
		that for an arbitrary grade-$r$ multivector $A_r \in \cl^r$
		$$
			\sum_i e_i (e^i \liprod A_r) = rA_r,
		$$
		$$
			\sum_i e_i (e^i \wedge A_r) = (n-r)A_r,
		$$
		and
		$$
			\sum_i e_i A_r e^i = (-1)^r (n-2r) A_r.
		$$
	\end{exc}

	\begin{exc}
		Let $v \in V^\times = \{u \in V : u^2 \neq 0\}$. Show that the map
		$$
			V \ni x \ \mapsto \ \frac{1}{2}(x + v^\star x v^{-1}) \in V
		$$
		is an orthogonal projection on $\bar{v}^\perp = \{u \in V : u*v=0\}$.
	\end{exc}

	\begin{exc} \label{exc_gram_schmidt}
		Consider the real euclidean geometric algebra $\mathcal{G}(\mathbb{R}^n)$.
		Let $a_1, a_2,\\ \ldots$ be a sequence of 1-vectors in $\mathcal{G}$,
		and form the blades $A_0 = 1$, $A_k = a_1 \wedge a_2 \wedge \cdots \wedge a_k$.
		Now, let $b_k = A_{k-1}^\dagger A_k$.
		Show that $b_1,b_2,\ldots$ are the vectors obtained using the Gram-Schmidt
		orthogonalization procedure from $a_1,a_2,\ldots$.
	\end{exc}

\subsection{Linear functions}
	
	Since $\mathcal{G}$ is itself a vector space which embeds $V$, it is natural to consider the
	properties of linear functions on $\mathcal{G}$. There is a special class of such
	functions, called outermorphisms, which can be said to respect the graded/exterior structure
	of $\mathcal{G}$ in a natural way. We will see that, just as the geometric algebra $\mathcal{G}(V,q)$
	is completely determined by the underlying vector space $(V,q)$, an outermorphism
	is completely determined by its behaviour on $V$.
	
	\begin{defn} \label{def_outermorphism}
		A linear map $F\!: \mathcal{G} \to \mathcal{G}'$ 
		is called an \emph{outermorphism} or \emph{$\wedge$-morphism} if 
		\begin{displaymath}
		\setlength\arraycolsep{2pt}
		\begin{array}{rl}
			i)   & F(1) = 1, \\[5pt]
			ii)  & F(\mathcal{G}^m) \subseteq \mathcal{G}'^m \quad \forall\ m \geq 0, \quad \textrm{(grade preserving)} \\[5pt]
			iii) & F(x \wedge y) = F(x) \wedge F(y) \quad \forall\ x,y \in \mathcal{G}. \phantom{lagg pa samma tab}
		\end{array}
		\end{displaymath}
		A linear transformation $F\!: \mathcal{G} \to \mathcal{G}$ 
		is called a \emph{dual outermorphism} or \emph{$\vee$-morphism} if 
		\begin{displaymath}
		\setlength\arraycolsep{2pt}
		\begin{array}{rl}
			i)   & F(I) = I, \\[5pt]
			ii)  & F(\mathcal{G}^m) \subseteq \mathcal{G}^m \quad \forall\ m \geq 0, \\[5pt]
			iii) & F(x \vee y) = F(x) \vee F(y) \quad \forall\ x,y \in \mathcal{G}. \phantom{lagg pa samma tab}
		\end{array}
		\end{displaymath}
	\end{defn}
	
	\begin{thm} \label{thm_outermorphism}
		For every linear map $f\!: V \to W$ there exists a unique outermorphism
		$f_\wedge\!: \mathcal{G}(V) \to \mathcal{G}(W)$ such that $f_\wedge(v) = f(v) \ \forall\ v \in V$.
	\end{thm}
	\begin{proof}
		Take a general basis $\{e_1,\ldots,e_n\}$ of $V$ and define,
		for $1 \leq i_1 < i_2 < \ldots < i_m \leq n$,
		\begin{equation}
			f_\wedge(e_{i_1} \wedge \cdots \wedge e_{i_m}) := f(e_{i_1}) \wedge \cdots \wedge f(e_{i_m}),
		\end{equation}
		and extend $f_\wedge$ to the whole of $\mathcal{G}(V)$ by linearity.
		We also define $f_\wedge(\alpha) := \alpha$ for $\alpha \in \mathbb{F}$.
		Hence, (\emph{i}) and (\emph{ii}) are satisfied. (\emph{iii}) is easily
		verified by expanding in the induced basis $\{e_\mathbf{i}\}$ of $\mathcal{G}(V)$.
		Uniqueness is obvious since our definition was necessary.
	\end{proof}
	
	\noindent
	Uniqueness immediately implies the following.
	
	\begin{cor}
		If $f\!: V \to V'$ and $g\!: V' \to V''$
		are linear then $(g \circ f)_\wedge = g_\wedge \circ f_\wedge$.
	\end{cor}

	\begin{cor}
		If $F\!: \mathcal{G}(V) \to \mathcal{G}(W)$ is an outermorphism
		then $F = (F|_V)_\wedge$.
	\end{cor}
	
	\begin{rem}
		In the setting of $\cl$ this means that an outermorphism $F\!: \cl(X,R,r) \to \cl(X',R,r')$ 
		is completely determined by its values on the elements of $X$.
		Also note that, if $E,F$ are two orthonormal bases of $V$ 
		and $F = \{f_j = g(e_j)\}$ then 
		$$
			\mathcal{P}(E) \ni A_E \stackrel{g_\wedge}{\mapsto} A_F \in \mathcal{P}(F)
		$$
		so that $g_\wedge: \cl(E) \to \cl(F)$ is an isomorphism of $R$-algebras.
	\end{rem}
	
	We have seen that a nondegenerate $\mathcal{G}$ results in a nondegenerate bilinear form $x*y$.
	This gives a canonical isomorphism 
	$$
		\theta\!: \mathcal{G} \to \mathcal{G}^* = \textrm{Lin}(\mathcal{G},\mathbb{F})
	$$
	between the elements of $\mathcal{G}$ and the linear functionals on	$\mathcal{G}$ as follows.
	For every $x \in \mathcal{G}$ we define a linear functional $\theta(x)$ by $\theta(x)(y) := x*y$.
	Taking a general basis $\{e_\mathbf{i}\}_\mathbf{i}$ of $\mathcal{G}$ and using \eqref{reciprocal_g} 
	we obtain a dual basis $\{\theta(e^\mathbf{i})\}_\mathbf{i}$ such that 
	$\theta(e^\mathbf{i})(e_\mathbf{j}) = \delta_\mathbf{j}^\mathbf{i}$.
	This shows that $\theta$ is an isomorphism.

	Now that we have a canonical way of moving between $\mathcal{G}$ and its dual space $\mathcal{G}^*$, 
	we can for every linear map $F\!: \mathcal{G} \to \mathcal{G}$ define an \emph{adjoint map}
	$F^*\!: \mathcal{G} \to \mathcal{G}$ by
	\begin{equation} \label{adjoint_def}
		F^*(x) := \theta^{-1}\big( \theta(x) \circ F \big).
	\end{equation}
	By definition, this has the expected and unique property
	\begin{equation} \label{adjoint_property}
		F^*(x) * y = x * F(y)
	\end{equation}
	for all $x,y \in \mathcal{G}$.
	Note that if we restrict our attention to 
	$V = \mathcal{G}^1$ then this construction results 
	in the usual adjoint (transpose) $f^*$ of a linear map $f\!: V \to V$.
	
	\begin{thm}[Hestenes' Theorem] \label{thm_hestenes}
		Assume that $\mathcal{G}$ is nondegenerate and let $F\!: \mathcal{G} \to \mathcal{G}$ be an outermorphism. 
		Then the adjoint $F^*$ is also an outermorphism and
		\begin{displaymath}
		\setlength\arraycolsep{2pt}
		\begin{array}{ccc}
			x \liprod F(y) 	&=& F\big( F^*(x) \liprod y \big), \\[5pt]
			F(x) \riprod y 	&=& F\big( x \riprod F^*(y) \big),
		\end{array}
		\end{displaymath}
		for all $x,y \in \mathcal{G}$.
	\end{thm}
	\begin{proof}
		We first prove that $F^*$ is an outermorphism. 
		The fact that $F^*$ is grade preserving follows from \eqref{adjoint_property} and
		the grade preserving property of $F$. Now take basis blades
		$x = x_1 \wedge \cdots \wedge x_m$ and $y = y_m \wedge \cdots \wedge y_1$ 
		with $x_i,y_j \in V$. Then
		\begin{displaymath}
		\setlength\arraycolsep{2pt}
		\begin{array}{rcl}
			F^*(x_1 \wedge \cdots \wedge x_m) * y
				&=& (x_1 \wedge \cdots \wedge x_m) * F(y_m \wedge \cdots \wedge y_1) \\[5pt]
				&=& (x_1 \wedge \cdots \wedge x_m) * \big(F(y_m) \wedge \cdots \wedge F(y_1)\big) \\[5pt]
				&=& \det\thinspace \big[ x_i * F(y_j) \big]_{i,j} = \det\thinspace \big[ F^*(x_i) * y_j \big]_{i,j} \\[5pt]
				&=& \big(F^*(x_1) \wedge \cdots \wedge F^*(x_m)\big) * (y_m \wedge \cdots \wedge y_1) \\[5pt]
				&=& \big(F^*(x_1) \wedge \cdots \wedge F^*(x_m)\big) * y,
		\end{array}
		\end{displaymath}
		where we have used Lemma \ref{lem_blade_det}.
		By linearity and nondegeneracy it follows that $F^*$ is an outermorphism.
		The first identity stated in the therorem now follows quite easily from 
		Proposition \ref{prop_trix}. For any $z \in \mathcal{G}$ we have
		\begin{displaymath}
		\setlength\arraycolsep{2pt}
		\begin{array}{rcl}
			z * \big(x \liprod F(y)\big)
				&=& (z \wedge x) * F(y) = F^*(z \wedge x) * y \\[5pt]
				&=& \big(F^*(z) \wedge F^*(x)\big) * y = F^*(z) * \big(F^*(x) \liprod y\big) \\[5pt]
				&=& z * F\big(F^*(x) \liprod y\big).
		\end{array}
		\end{displaymath}
		The nondegeneracy of the scalar product then gives the first identity.
		The second identity is proven similarly, using that $(x \riprod y) * z = x * (y \wedge z)$.
	\end{proof}
	
	\noindent
	From uniqueness of outermorphisms we also obtain the following
	
	\begin{cor}
		If $f\!: V \to V$ is a linear transformation then
		$(f^*)_\wedge = (f_\wedge)^*$.
	\end{cor}

	\noindent
	This means that we can simply write $f^*_\wedge$ for the adjoint outermorphism of $f$.
	
	Another powerful concept in geometric algebra (or exterior algebra) is the generalization of eigenvectors
	to so called \emph{eigenblades}. For a function $f\!: V \to V$, a $k$-eigenblade
	with eigenvalue $\lambda \in \mathbb{F}$ is a blade $A \in \mathcal{B}_k$ such that
	\begin{equation} \label{eigenblade_def}
		f_\wedge(A) = \lambda A.
	\end{equation}
	Just as eigenvectors can be said to represent 
	invariant 1-dimensional subspaces of a function, a $k$-blade with nonzero
	eigenvalue represents an invariant $k$-dimensional subspace.
	One important example of an eigenblade is the pseudoscalar $I$, which represents the whole 
	invariant vector space $V$.
	Since $f_\wedge$ is
	grade preserving, we must have $f_\wedge(I) = \lambda I$ for some $\lambda \in \mathbb{F}$
	which we call the \emph{determinant} of $f$, i.e.
	\begin{equation} \label{determinant_def}
		f_\wedge(I) = (\det f) I.
	\end{equation}
	Expanding $\det f = f_\wedge(I) * I^{-1}$ in a basis using Lemma \ref{lem_blade_det}, one finds
	that this agrees with the usual definition of the determinant of a linear function
	(see Exercise \ref{exc_determinant}).

	\begin{exmp}
		Since
		$$
			(\det f) I^2 = f_\wedge(I)*I = I*f_\wedge^*(I) = (\det f^*) I^2,
		$$
		we immediately find that $\det f = \det f^*$.
		Also, by the corollary to Theorem \ref{thm_outermorphism}
		we have
		$$
			(f \circ g)_\wedge(I) = f_\wedge\big( g_\wedge(I) \big) = f_\wedge\big( (\det g)I \big) = (\det g)(\det f)I,
		$$
		so that $\det(fg) = \det f \cdot \det g$.
	\end{exmp}

	\begin{defn} \label{def_dual_map}
		For linear $F\!: \mathcal{G} \to \mathcal{G}$ we define the \emph{dual map} 
		$F\dual\!: \mathcal{G} \to \mathcal{G}$ by
		$F\dual(x) := F(xI)I^{-1}$, so that the following diagram commutes:
		\begin{displaymath}
		\setlength\arraycolsep{2pt}
		\begin{array}{rcccl}
							& \mathcal{G} 	& \xrightarrow{F} 		& \mathcal{G} \\
			(\cdot)\dual 	& \downarrow	&						& \downarrow	& (\cdot)\dual \\
							& \mathcal{G}	& \xrightarrow{F\dual}	& \mathcal{G}
		\end{array}
		\end{displaymath}
	\end{defn}
	
	\begin{prop} \label{prop_dual_map}
		We have the following properties of the dual map:
		\begin{displaymath}
		\setlength\arraycolsep{2pt}
		\begin{array}{rrcl}
			i) 	& F^\mathbf{cc} &=& F, \\[5pt]
			ii) & (F \circ G)\dual &=& F\dual \circ G\dual, \\[5pt]
			iii)& \id\dual &=& \id, \\[5pt]
			iv) & F(\mathcal{G}^s) \subseteq \mathcal{G}^t &\Rightarrow& F\dual(\mathcal{G}^{\dim V -s}) \subseteq \mathcal{G}^{\dim V -t}, \\[5pt]
			v) 	& F \ \textrm{$\wedge$-morphism} &\Rightarrow& F\dual \ \textrm{$\vee$-morphism}, \\[5pt]
			vi) & (F^*)\dual &=& (F\dual)^* \ \textrm{if $F$ is grade preserving},
		\end{array}
		\end{displaymath}
		for all linear $F,G\!: \mathcal{G} \to \mathcal{G}$.
	\end{prop}
	
	\noindent
	The proofs are straightforward and left as exercises to the reader.
	As a special case of Theorem \ref{thm_hestenes} we obtain, with $y=I$
	and a linear map $f\!: V \to V$,
	\begin{equation}
		(\det f) xI = f_\wedge \big( f_\wedge^*(x)I \big),
	\end{equation}
	so that
	\begin{equation}
		(\det f) \id = f_\wedge\dual \circ f_\wedge^* = f_\wedge \circ f_\wedge^{*\mathbf{c}}.
	\end{equation}
	If $\det f \neq 0$ we then have a simple expression for the inverse;
	\begin{equation} \label{inverse_outermorphism}
		f_\wedge^{-1} = (\det f)^{-1} f_\wedge^{*\mathbf{c}},
	\end{equation}
	which is essentially the dual of the adjoint
	($F^{*\mathbf{c}}$ is sometimes called the \emph{adjugate} $F^{\adj}$ of $F$).
	$f^{-1}$ is obtained by simply restricting to $V$.

	\begin{exmp}
		An orthogonal transformation $f \in O(V,q)$ 
		satisfies $f^{-1} = f^*$ and $\det f = \pm 1$, 
		so in this case \eqref{inverse_outermorphism} gives
		$f_\wedge\dual = \pm f_\wedge$.
	\end{exmp}

	\begin{exc}
		Prove the second identity in Theorem \ref{thm_hestenes}.
	\end{exc}

	\begin{exc}
		Prove Proposition \ref{prop_dual_map}.
	\end{exc}

	\begin{exc} \label{exc_bivec_adj}
		Let $B \in \mathcal{G}^2$ be a bivector, and define
		the map $\ad_B\!: \mathcal{G} \to \mathcal{G}$ by
		$\ad_B(x) := [B,x] = Bx - xB$.
		Show that $\ad_B(\mathcal{G}^m) \subseteq \mathcal{G}^m$
		and that 
		$$
			\ad_B(xy) = \ad_B(x)y + x\ad_B(y),
		$$
		i.e. that $\ad_B$ is a grade preserving derivation. 
		Is $\ad_B$ an outermorphism?
	\end{exc}

	\begin{exc} \label{exc_func_bivec_corresp}
		Let $f\!: V \to V$ be linear
		and antisymmetric, i.e.
		$$
			f(u) * v = - u * f(v) \quad \forall u,v \in V.
		$$
		Show that there is a unique bivector $B \in \mathcal{G}^2(V)$
		such that $f(v) = \ad_B(v) \ \forall v \in V$.
	\end{exc}

	\begin{exc} \label{exc_determinant}
		Let $\{a_1,\ldots,a_n\}$ be an arbitrary basis of $V$
		and use the reciprocal basis to verify that $\det f = f(I)*I^{-1}$
		is the usual expression for the determinant.
		Use \eqref{blade_expansion} to verify that the definition makes sense also
		for degenerate spaces.
	\end{exc}

\subsection{Projections and rejections} \label{subsec_proj_rej}

	Let $A$ be an invertible blade in a geometric algebra $\mathcal{G}$,
	and define a linear map $P_A: \mathcal{G} \to \mathcal{G}$ through
	$$
		P_A(x) := (x \liprod A)A^{-1}.
	$$
	
	\begin{prop}
		$P_A$ is an outermorphism.
	\end{prop}
	\begin{proof}
		We need to prove that $P_A(x \wedge y) = P_A(x) \wedge P_A(y)$ 
		$\forall x,y \in \mathcal{G}$. 
		We observe that both sides are linear in both $x$ and $y$.
		Hence, we can assume that $x$ and $y$ are basis blades. 
		Also, since $A$ was assumed to be invertible, 
		we can by the corollary to Proposition \ref{prop_blade_product},
		without loss of generality, assume 
		that $A$ is also among the orthogonal basis blades.
		We obtain
		$$
			P_A(x \wedge y) = ((x \wedge y) \liprod A)A^{-1} = (x \cap y = \varnothing)(x \cup y \subseteq A)xyAA^{-1}
		$$
		and
		\begin{eqnarray*}
			P_A(x) \wedge P_A(y) 
				&=& \left( (x \liprod A)A^{-1} \right) \wedge \left( (y \liprod A)A^{-1} \right) \\
				&=& \left( (x \subseteq A)xAA^{-1} \right) \wedge \left( (y \subseteq A)yAA^{-1} \right) \\
				&=& (x \subseteq A)(y \subseteq A) x \wedge y \\
				&=& (x \subseteq A)(y \subseteq A)(x \cap y = \varnothing)xy.
		\end{eqnarray*}
		Furthermore, we obviously have $P_A(1) = 1$, 
		and (e.g. by again expanding in a basis) we also observe that $P_A(V) \subseteq V$,
		so $P_A$ preserves grades.
	\end{proof}
	
	We shall now show that $P_A$ is a projection. 
	First, note the following
	
	\begin{prop}
		$P_A \circ P_A = P_A$
	\end{prop}
	\begin{proof}
		We would like to prove
		$$
			\left( ((x \liprod A)A^{-1}) \liprod A \right)A^{-1} 
			= (x \liprod A)A^{-1}  \quad \forall x \in \mathcal{G}.
		$$
		As both sides are linear in $x$, we assume that $x$ 
		is a basis blade and obtain
		$$
			(x \liprod A)A^{-1} = (x \subseteq A)xAA^{-1} = (x \subseteq A)x
		$$
		and
		\begin{eqnarray*}
			\lefteqn{ \left( ((x \liprod A)A^{-1}) \liprod A \right)A^{-1} = \left( ((x \subseteq A)x) \liprod A \right)A^{-1} }\\
				&&= (x \subseteq A)(x \liprod A)A^{-1} = (x \subseteq A)(x \subseteq A)x 
				= (x \subseteq A)x. 
		\end{eqnarray*}
	\end{proof}
	
	\noindent
	Also note that, for all $v \in V$, we have by Proposition \ref{prop_trix2}
	\begin{equation} \label{proj_rej_expansion}
		v = vAA^{-1} = (vA)A^{-1} = (v \liprod A + v \wedge A)A^{-1}.
	\end{equation}
	Now, define the \emph{rejection} $R_A$ of $A$ through
	$$
		R_A(x) := (x \wedge A)A^{-1}, \quad \textrm{for}\ x \in \mathcal{G}.
	$$
	Then \eqref{proj_rej_expansion} becomes
	$$
		v = P_A(v) + R_A(v) \qquad \forall v \in V.
	$$
	If $v \in \bar{A} = \{u \in V : u \wedge A = 0\}$,
	we obviously obtain 
	$$
		P_A(v) = v \qquad \textrm{and} \qquad R_A(v) = 0,
	$$
	and if $v \in \bar{A}^\perp = \{u \in V : u \liprod A = 0\}$, we find
	$$
		P_A(v) = 0 \qquad \textrm{and} \qquad R_A(v) = v.
	$$
	We therefore see that 
	$P_A$ is the outermorphism of an orthogonal projection on $\bar{A}$,
	while $R_A$ corresponds to an orthogonal projection on $\bar{A}^\perp$.
	
	\begin{exc}
		Show that $R_A$ is an outermorphism and that $P_A$ and $R_A$ are self-adjoint,
		i.e. $P_A^* = P_A$ and $R_A^* = R_A$.
	\end{exc}

\subsection{Some projective geometry}

	We finish this section by considering some applications 
	in classical projective geometry.
	
	Given a vector space $V$, the \emph{projective geometry $P(V)$ on $V$}
	is defined to be the set of all subspaces of $V$.
	The set of $k$-dimensional subspaces is denoted $P_k(V)$.
	Traditionally, $P_1(V)$ is called the \emph{projective space}
	and is usually denoted $\mathbb{P}^{n-1}(\mathbb{F})$ if $V$ is an
	$n$-dimensional vector space over the field $\mathbb{F}$.
	
	$P(V)$ forms a natural so-called \emph{lattice} with a \emph{meet} $\vecmeet$
	and \emph{join} $\vecjoin$ operation defined as follows:
	Let $U$ and $W$ be elements in $P(V)$, i.e. subspaces of $V$.
	The largest subspace of $V$ which is contained in both
	$U$ and $W$ is denoted $U \vecmeet W$,
	and the smallest subspace of $V$ which contains both 
	$U$ and $W$ is denoted $U \vecjoin W$.
	Hence, $U \vecmeet W = U \cap W$ and $U \vecjoin W = U + W$.

	We saw earlier that when $V$ is considered as embedded in
	a geometric algebra $\mathcal{G}(V)$, then to each nonzero
	$k$-blade $A \in \mathcal{G}$ there is associated a unique 
	$k$-dimensional subspace $\bar{A} \in P_k(V)$.
	Furthermore, we found that the outer and meet products
	of blades corresponded to the subspace operations\footnote{In
	this context the choice of notation is slightly unfortunate.
	However, both the notation for meet and join of subspaces
	(which matches the logical operations) and the outer and meet products of
	multivectors are strictly conventional.}
	$$
		\overline{A \wedge B} = \bar{A} \vecjoin \bar{B} \quad \textrm{if} \quad \bar{A} \vecmeet \bar{B} = 0
	$$
	and
	$$
		\overline{A \vee B} = \bar{A} \vecmeet \bar{B} \quad \textrm{if} \quad \bar{A} \vecjoin \bar{B} = V.
	$$

	Let us now get acquainted with some concrete classical 
	projective geometry in $\mathbb{R}^n$.
	
\subsubsection{The cross ratio in $\mathbb{P}^1(\mathbb{R})$}
	
	Let $V = \mathbb{R}^2$ and choose $a,b,c,d \in V$
	such that each pair of these vectors is linearly independent.
	Define the so-called \emph{cross ratio of four points}
	$$
		D(a,b,c,d) = \frac{(a \wedge b)(c \wedge d)}{(a \wedge c)(b \wedge d)}.
	$$
	Note that each of the factors above is a nonzero multiple of
	the pseudoscalar, so that the above expression makes sense.
	The expression is a \emph{projective invariant} in the following sense:
	
	\begin{enumerate}
		\item[\emph{i})]
		$D(a,b,c,d)$ only depends on $\bar{a},\bar{b},\bar{c},\bar{d}$,
		i.e. on the four one-dimensional subspaces in $\mathbb{P}^1(\mathbb{R})$
		which correspond to the respective vectors.
		
		\item[\emph{ii})]
		$D(Ta,Tb,Tc,Td) = D(a,b,c,d)$ for every linear bijection 
		$T: \mathbb{R}^2 \to \mathbb{R}^2$.
	\end{enumerate}
	
	We can reinterpret the projectively invariant quantity $D(a,b,c,d)$ 
	as a universal property of four points in $\mathbb{R}^1$.
	Namely, choose a line $L$ away from the origin in $V$. 
	By applying a projective transformation, i.e. a linear bijection,
	we can without loss of generality assume that the line or basis
	is such that $L = \mathbb{R}e_1 + e_2$. Then,
	representing a point $A \in \mathbb{R}$ as $a = Ae_1 + e_2 \in L$ etc.,
	we find $a \wedge b = (A-B)I$ and that the cross ratio is
	\begin{equation} \label{cross_ratio_rep}
		D(a,b,c,d) = \frac{(A - B)(C - D)}{(A - C)(B - D)}.
	\end{equation}
	Projective invariance, i.e. the possibility of mapping $L$
	to any other line away from the origin, then means that this 
	quantity is invariant under 
	both translations, rescalings and inversions in $\mathbb{R}^1$.

	\begin{exc}
		Verify properties (\emph{i}) and (\emph{ii}) above.
		Derive the general form of the transformations of $\mathbb{R}^1$
		which leave the cross ratio invariant,
		and verify explicitly that the expression \eqref{cross_ratio_rep}
		is invariant under such a transformation.
	\end{exc}

	\begin{exc} \label{exc_projective_cross}
		Let us consider $\mathcal{G}(\mathbb{R}^3 = \Span\{e_1,e_2,e_3\})$ 
		as a part of $\mathcal{G}(\mathbb{R}^4 = \Span\{e_1,e_2,e_3,e_4\})$ 
		via the map
		$$
		\begin{array}{rcl}
			\mathbb{R}^3 & \to & \mathbb{R}^4 \\
			\boldsymbol{x} & \mapsto & X := \boldsymbol{x} + e_4
		\end{array}
		$$
		which is represented in Figure \ref{fig_proj_R4}.
		We denote the respective pseudoscalars by
		$I = e_1e_2e_3$ and $J = e_1e_2e_3e_4$.
		Let $\times$ denote the usual cross product in 
		$\mathbb{R}^3$ 
		(as defined in \eqref{cross_product}) and verify that
		\item{a)} $\boldsymbol{x} \times \boldsymbol{y} = (e_4 \wedge X \wedge Y) J$
		\item{b)} $\boldsymbol{y} - \boldsymbol{x} = e_4 \liprod (X \wedge Y)$
		
		\begin{figure}[t]
			\centering
			\psfrag{T_e1}{$e_1$}
			\psfrag{T_e3}{$e_3$}
			\psfrag{T_e4}{$e_4$}
			\psfrag{T_x}{$\boldsymbol{x} = x_1e_1 + x_2e_2 + x_3e_3$}
			\psfrag{T_X}{$X = \boldsymbol{x} + e_4$}
			\psfrag{T_copy}{Copy of $\mathbb{R}^3$}
			\includegraphics{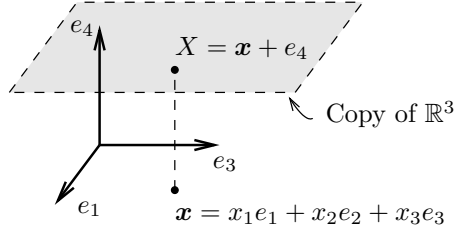}
			\caption{Mapping a point $\boldsymbol{x} \in \mathbb{R}^3$ to its
			projective representative $X \in \mathbb{R}^4$.}
			\label{fig_proj_R4}
		\end{figure}
	\end{exc}

	\begin{exc} \label{exc_projective_intersection}
		Let $\boldsymbol{a},\boldsymbol{b},\boldsymbol{c}$ and $\boldsymbol{d}$ 
		be distinct vectors in $\mathbb{R}^2$.
		We are interested in determining the intersection of the lines
		through $\boldsymbol{a},\boldsymbol{b}$ and $\boldsymbol{c},\boldsymbol{d}$ respectively.
		As in the previous exercise, we consider a map (with a basis $\{e_1,e_2,e_3\}$)
		$$
		\begin{array}{rcl}
			\mathbb{R}^2 & \to & \mathbb{R}^3 \\
			\boldsymbol{x} & \mapsto & X := \boldsymbol{x} + e_3
		\end{array}
		$$
		Introduce $L = A \wedge B$ and $M = C \wedge D$ and show that
		$$
			L \vee M = [A,B,C]D - [A,B,D]C,
		$$
		where $[X,Y,Z] := (X \wedge Y \wedge Z)*(e_3 \wedge e_2 \wedge e_1)$
		for $X,Y,Z \in \mathbb{R}^3$.
		Interpret this result geometrically.
	\end{exc}

\subsubsection{Cross ratios and quadrics in $\mathbb{P}^2(\mathbb{R})$}

	Now, let $V = \mathbb{R}^3$ (we could also choose $\mathbb{R}^{s,t}$
	with $s+t=3$). Choose $a,b,c,d,e \in V$, pairwise linearly
	independent, such that $a,b,c,d$ all lie in a two-dimensional
	subspace not containing $e$.
	In the projective jargon one then says that the projective
	points $\bar{a},\bar{b},\bar{c},\bar{d}$ lie on a line which
	does not contain $\bar{e}$.
	Let us think of the plane of this page as an affine\footnote{i.e. constant plus linear} 
	subspace of $V$ which does not contain the origin, 
	and which intersects the lines $\bar{a},\bar{b},\bar{c},$ and $\bar{d}$
	at points which we represent by the names of those lines.
	Let us imagine that we have one eye in the origin and that
	the line through the eye and the point in the paper denoted
	$\bar{a}$ is the line $\bar{a}$.
	We then have a picture like in Figure \ref{fig_proj_line}.

	\begin{figure}[ht]
		\centering
		\psfrag{T_a}{$\bar{a}$}
		\psfrag{T_b}{$\bar{b}$}
		\psfrag{T_c}{$\bar{c}$}
		\psfrag{T_d}{$\bar{d}$}
		\psfrag{T_e}{$\bar{e}$}
		\includegraphics{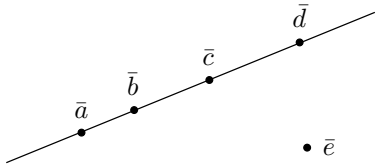}
		\caption{Projective points $\bar{a},\bar{b},\bar{c},$ and $\bar{d}$
		on a common projective line which does not contain $\bar{e}$.}
		\label{fig_proj_line}
	\end{figure}

	Define
	$$
		F(a,b,c,d,e) := \frac{(a \wedge b \wedge e)(c \wedge d \wedge e)}{(a \wedge c \wedge e)(b \wedge d \wedge e)}.
	$$
	We immediately see that for all $\alpha,\beta,\gamma,\delta,\epsilon \in \mathbb{R}^\times$,
	$$
		F(\alpha a, \beta b, \gamma c, \delta d, \epsilon e) = F(a,b,c,d,e),
	$$
	which shows that $F$ only depends on $\bar{a},\bar{b},\bar{c},\bar{d}$ 
	and $\bar{e}$ in $\mathbb{P}^2(\mathbb{R})$.
	Furthermore, if $T: \mathbb{R}^3 \to \mathbb{R}^3$ is a linear bijection then
	$$
		F(Ta,Tb,Tc,Td,Te) = \frac{(\det T)^2}{(\det T)^2} F(a,b,c,d,e),
	$$
	and $F(a,b,c,d,e)$ is a so-called projective invariant of five points.
	
	So far we have not used that the projective points 
	$\bar{a},\bar{b},\bar{c},$ and $\bar{d}$ lie on a common line. 
	We shall now show, with our requirements on $a,b,c,d,e$,
	that $F$ does not depend on $e$.
	Assume that $e' = e + \alpha a + \beta b$. 
	Then $a \wedge b \wedge e' = a \wedge b \wedge e$ and
	$$
		c \wedge d \wedge e' = c \wedge d \wedge e + \alpha c \wedge d \wedge a + \beta c \wedge d \wedge b = c \wedge d \wedge e,
	$$
	since $c \wedge d \wedge a = c \wedge d \wedge b = 0$ 
	(they all lie in the same plane in $\mathbb{R}^3$).
	Similarly, we find that we could have added linear combinations of $c$ and $d$ as well,
	and that the factors in the denominator also do not depend on $e$.
	Hence, we can define the cross ratio
	$$
		D(a,b,c,d) := F(a,b,c,d,e),
	$$
	which then becomes a projective invariant of $\bar{a},\bar{b},\bar{c},\bar{d}$
	with a similar geometric interpretation as before 
	whenever these points lie on a common line.
	
	Let us now have a look at quadrics in $\mathbb{P}^2(\mathbb{R})$.
	Let $P(x_1,x_2,x_3)$ denote a homogeneous quadratic polynomial\footnote{I.e.
	a quadratic form, but here we think of it as a polynomial rather than
	an extra structure of the vector space.}
	in the (ordinary, commutative) polynomial ring $\mathbb{R}[x_1,x_2,x_3]$.
	From the homogeneity follows that if $(\alpha,\beta,\gamma)$
	is a zero of the polynomial $P$, then also $t(\alpha,\beta,\gamma)$ 
	is a zero, for all $t \in \mathbb{R}$.
	Hence, if we (as usual) interpret the triple $(x_1,x_2,x_3)$
	as the point $x = x_1e_1 + x_2e_2 + x_3e_3 \in \mathbb{R}^3$ then we see
	that zeros of $P$ can be interpreted as points in $\mathbb{P}^2(\mathbb{R})$.
	The zero set $Z(P)$ of a homogeneous quadratic polynomial $P$
	is called a \emph{quadric}.

	\begin{exmp}
		Consider the homogeneous quadratic polynomial in three variables
		$$
			P(x,y,z) = x^2 + y^2 - z^2
		$$
		Note that the intersection of the quadric $Z(P)$ with the
		plane $\{(x,y,z) \in \mathbb{R}^3 : z=1\}$ is a unit circle,
		but by choosing some other plane we may get an ellipse, a hyperbola,
		or even a parabola or a line.
	\end{exmp}
	
	It is easy to see that five distinct points in $\mathbb{P}^2(\mathbb{R})$
	determine a homogeneous quadratic polynomial in three variables,
	up to a multiple, if we demand that the five points should be zeros
	of the polynomial.
	Assume that we are given five points $a,b,c,d,e$ in $\mathbb{R}^3$,
	and form
	$$
		P(x) = (a \wedge b \wedge e)(c \wedge d \wedge e)(a \wedge c \wedge x)(b \wedge d \wedge x)
			 - (a \wedge b \wedge x)(c \wedge d \wedge x)(a \wedge c \wedge e)(b \wedge d \wedge e)
	$$
	$P$ is then a homogeneous quadratic polynomial in $(x_1,x_2,x_3)$,
	where $x = x_1e_1 + x_2e_2 + x_3e_3$.
	One verifies by direct substitution that
	$$
		P(a) = P(b) = P(c) = P(d) = P(e) = 0,
	$$
	and hence that $Z(P)$ is a quadric containing 
	$\bar{a},\bar{b},\bar{c},\bar{d},$ and $\bar{e}$.

\subsubsection{Pascal's theorem}

	There is an alternative way of finding a quadric through
	five points, which is equivalent to a classical theorem in
	plane projective geometry due to Pascal.
	
	\begin{thm}[Pascal] \label{thm_pascal}
		Let $a,b,c,d,e,f$ be points on a quadric in $\mathbb{P}^2(\mathbb{R})$.
		Then
		\begin{equation} \label{pascals_eqn}
			\big((a \wedge b) \vee (d \wedge e)\big) \wedge
			\big((b \wedge c) \vee (e \wedge f)\big) \wedge
			\big((c \wedge d) \vee (f \wedge a)\big) = 0
		\end{equation}
	\end{thm}
	\begin{proof}
		Form the homogeneous quadratic polynomial $P(x)$
		which is obtained from the expression \eqref{pascals_eqn}
		when e.g. $f$ is replaced with $x$.
		One easily verifies (e.g. using Exercise \ref{exc_projective_intersection})
		that $P(a) = P(b) = P(c) = P(d) = P(e) = 0$,
		from which the theorem follows.
	\end{proof}

	\noindent
	The geometric interpretation of the expression \eqref{pascals_eqn}
	is that the three intersection points of the line
	$a \wedge b$ with the line $d \wedge e$,
	the line $b \wedge c$ with $e \wedge f$, 
	resp. $c \wedge d$ with $f \wedge a$,
	are all situated on one and the same line.
	See Figure \ref{fig_pascal}.

	\begin{figure}[ht]
		\centering
		\psfrag{T_a}{$a$}
		\psfrag{T_b}{$b$}
		\psfrag{T_c}{$c$}
		\psfrag{T_d}{$d$}
		\psfrag{T_e}{$e$}
		\psfrag{T_f}{$f$}
		\includegraphics{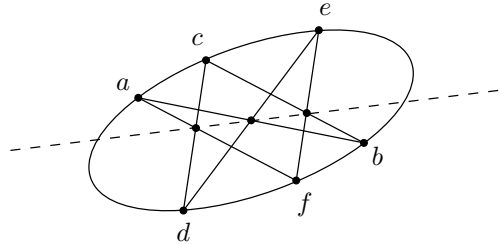}
		\caption{A geometric interpretation of Pascal's theorem.}
		\label{fig_pascal}
	\end{figure}
	
\subsubsection{Polarizations in quadrics}
	
	Let $Q$ be a quadric in $\mathbb{P}^2$ and choose a point 
	$p \in \mathbb{P}^2$ according to Figure \ref{fig_polar_line}.
	Draw the tangent lines $L_1$ and $L_2$ to $Q$ through $p$
	and form the line $L$ through the tangent points $t_1$ and $t_2$.
	Then $L$ is called the \emph{polar line} to $p$ with respect to $Q$.
	In classical projective geometry there are a number of
	theorems related to this notion.
	We will now show how geometric algebra can be used to obtain
	simple proofs and more algebraic formulations of some
	of these theorems.

	\begin{figure}[ht]
		\centering
		\psfrag{T_p}{$p$}
		\psfrag{T_t1}{$t_1$}
		\psfrag{T_t2}{$t_2$}
		\psfrag{T_L1}{$L_1$}
		\psfrag{T_L2}{$L_2$}
		\psfrag{T_L}{$L$}
		\psfrag{T_Q}{$Q$}
		\includegraphics{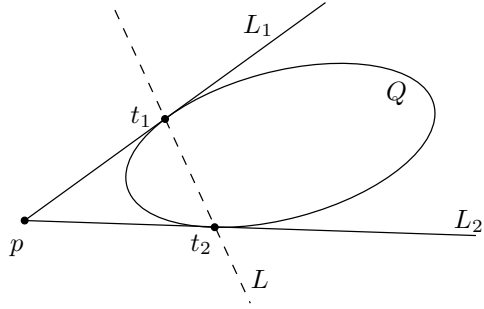}
		\caption{The polar line $L$ to the point $p$ with respect to the quadric $Q$.}
		\label{fig_polar_line}
	\end{figure}
	
	Let $\mathcal{G} = \mathcal{G}(\mathbb{R}^{s,t})$, 
	where $n = s+t$, and let $T$ denote a self-adjoint outermorphism from 
	$\mathcal{G}$ to $\mathcal{G}$, i.e. $T = T^*$.
	Consider the map $f: \mathcal{G} \to \mathbb{R}$,
	$f(x) := x * T(x)$ and form\footnote{Note that $f$ is actually
	a homogeneous quadratic polynomial in $2^n$ variables.}
	$$
		Q := \{x \in \mathcal{G} : f(x) = 0\}, \quad \textrm{and} \quad Q^m := \mathcal{G}^m \cap Q.
	$$
	We then have $f(x + h) = f(x) + 2T(x)*h + h*T(h)$ so that $f'(x)$, 
	the derivative of $f$ at the point $x \in \mathcal{G}$, 
	is given by the linear map
	$$
		\begin{array}{rcl}
			\mathcal{G} & \to & \mathbb{R} \\
			h & \mapsto & 2T(x)*h
		\end{array}
	$$
	The \emph{tangent space} $Q_x$ to the surface $Q$ at the point $x$
	is then defined as the linear space
	$$
		Q_x := \ker f'(x) \subseteq \mathcal{G}.
	$$
	We also define $Q^m_x := \mathcal{G}^m \cap Q_x$, which is the
	tangent space of $Q^m$ at $x$.
	
	\begin{defn}
		The \emph{polar} to $x$ with respect to $Q$ 
		is defined by
		$$
			\Pol_Q(x) := T(x)\dual = T(x)I^{-1}.
		$$
		We note in particular that $\overline{\Pol_Q(x)} = \left( \overline{T(x)} \right)^\perp$
		if $x \in \mathcal{B}^*$ is a nonzero blade.
	\end{defn}
	
	We have a number of results regarding the polar.
	
	\begin{thm}
		If $x \in Q^1$ and $h_1,\ldots,h_k \in Q_x^1$ then
		$x \wedge h_1 \wedge \ldots \wedge h_k \in Q^{k+1}$, i.e.
		$$
			T(x \wedge h_1 \wedge \ldots \wedge h_k) * (x \wedge h_1 \wedge \ldots \wedge h_k) = 0.
		$$
	\end{thm}
	\begin{proof}
		We have
		$$
			\big(T(x) \wedge T(h_1) \wedge \ldots \wedge T(h_k)\big) * (h_k \wedge \ldots \wedge h_1 \wedge x)
			= \det \left[
			\begin{array}{cc}
				T(x)*x & \ldots \\
				T(x)*h_1 & \ldots \\
				\vdots \\
				T(x)*h_k & \ldots \\
			\end{array}
			\right]
		$$
		which is zero since $f(x) = 0$ and
		$2T(x)*h_j = f'(x)(h_j) = 0$ for $j=1,\ldots,k$.
	\end{proof}
	
	If $x \in Q^m$ then its polar $\Pol_Q(x)$ 
	is an element of $Q^{n-m}$ because of the following
	
	\begin{thm}
		We have
		$$
			T(x)\dual * T\big( T(x)\dual \big) = \tau \det T \ x*T(x),
		$$
		where $\tau = (-1)^{mn-m}I^2$ and $x \in \mathcal{G}^m$.
	\end{thm}
	\begin{proof}
		Assume that $x \in \mathcal{G}^m$ and recall that
		$$
			T^* \circ T\dual = (\det T) \id
		$$
		holds for all outermorphisms $T$. When $T^*=T$ we then obtain
		\begin{eqnarray*}
			\lefteqn{ T(x)\dual * T(T(x)\dual) = (T(x)I^{-1}) * T^*(T\dual(xI^{-1})) }\\
			&=& (T(x)I^{-1}) * ((\det T) xI^{-1}) = \det T \thinspace \langle T(x)I^{-1}xI^{-1} \rangle_0 \\
			&=& (-1)^{m(n-m)} I^{-2} \det T \thinspace \langle T(x)x \rangle_0,
		\end{eqnarray*}
		which proves the theorem.
	\end{proof}
	
	\noindent
	In other words, $f(\Pol_Q(x)) = \tau \det T \thinspace f(x)$.
	We also have a reciprocality theorem:
	
	\begin{thm}
		If $x,y \in \mathcal{G}^m$ then
		$$
			x \wedge T(y)\dual = y \wedge T(x)\dual.
		$$
	\end{thm}
	\begin{proof}
		From Proposition \ref{prop_dual} we have
		$$
			x \wedge T(y)\dual = (x \liprod T(y))\dual = (x * T(y))\dual
			= (y * T(x))\dual = (y \liprod T(x))\dual = y \wedge T(x)\dual,
		$$
		where we used that $T(x),T(y) \in \mathcal{G}^m$.
	\end{proof}
	
	Finally, we have the following duality theorem:
	
	\begin{thm}
		$\big( T(x \wedge y) \big)\dual = I^2 T(x)\dual \vee T(y)\dual$.
	\end{thm}
	\begin{proof}
		This follows immediately from the definition of the
		meet product and the fact that $T$ is an outermorphism.
	\end{proof}
	
	\begin{exc}
		Interpret the above theorems geometrically in the case $\mathbb{P}^2$.
	\end{exc}

%% file: clifford_discrete.tex
In this section we consider applications of Clifford algebras
in discrete geometry and combinatorics.
We also extend the definition of Clifford algebra and
geometric algebra to infinite sets and infinite-dimensional
vector spaces.

\subsection{Simplicial complexes}

A \emph{simplicial complex} $K$ on a finite set $V$ is
a subset of $\mathscr{P}(V)$ which is closed under taking subsets,
i.e. such that for all $A,B \in \mathscr{P}(V)$
$$
	A \subseteq B \in K \quad \Rightarrow \quad A \in K.
$$

In the following,
let $\cl(V)$ denote the Clifford algebra over the set $V$,
where $v^2 = 1$, $\forall v \in V$, and where we use $R = \mathbb{Z}$ as scalars.
Define the \emph{chain group} 
$C(K) := \bigoplus_K \mathbb{Z}$ and note that
$C(K)$ is an abelian subgroup of $\cl(V)$ (w.r.t addition).
The sets $A \in K$ with signs are called \emph{(oriented) simplices} 
and $|A|-1$ is the \emph{dimension} of the simplex $A$.
Hence, we write
$$
	C(K) = \bigoplus_{d \ge -1} C_d(K),
	\quad \textrm{where} \quad
	C_d(K) := \bigoplus_{\genfrac{}{}{0pt}{}{A \in K}{|A|-1 = d}} \mathbb{Z}.
$$
An element of $C_d(K) \subseteq \cl^{d+1}(V)$ is called a \emph{$d$-chain}.
The Clifford algebra structure of $\cl(V)$ handles the
orientation of the simplices, so that e.g. 
a 1-simplex $v_0v_1$, thought of as a line segment directed from the point $v_0$ to $v_1$,
expresses its orientation as $v_0v_1 = -v_1v_0$.

Let $s_V := \sum\limits_{v \in V} v \in \cl^1(V)$ and define
the \emph{boundary map} $\partial_V$ by
$$
	\setlength\arraycolsep{2pt}
	\begin{array}{rccl}
		\partial_V\!: 	& \cl(V) 	& \to 		& \cl(V), \\
						& x 		& \mapsto 	& \partial_V(x) := s_V \liprod x.
	\end{array}
$$
Then $\partial_V^2 = 0$, since
\begin{equation} \label{boundary_squared}
	\partial_V \circ \partial_V (x) = s_V \liprod (s_V \liprod x) = (s_V \wedge s_V) \liprod x = 0,
\end{equation}
and the action of $\partial_V$ 
on a $k$-simplex $v_0v_1 \ldots v_k$ is given by
$$
	\partial_V(v_0 v_1 \ldots v_k) 
		= \sum_{i=0}^k v_i \liprod (v_0 v_1 \ldots v_k)
		= \sum_{i=0}^k (-1)^i v_0 v_1 \ldots \check{v}_i \ldots v_k.
$$
Note that $C(K)$ is invariant under the action of $\partial_V$
by the definition of a simplicial complex.
We also have, for $x,y \in \cl$, (see Exercise \ref{exc_anti_derivation})
\begin{eqnarray*}
	\partial_V(xy) &=& \partial_V(x)y + x^\star \partial_V(y), \\
	\partial_V(x \wedge y) &=& \partial_V(x) \wedge y + x^\star \wedge \partial_V(y), \\
	\partial_V: \cl^m(V) &\to& \cl^{m-1}(V), \\
	\partial_V: C_d(K) &\to& C_{d-1}(K).
\end{eqnarray*}

\begin{exmp}
	Consider a set of three vertices, $V = \{v_1,v_2,v_3\}$,
	and a simplicial complex $K = \{\varnothing, v_1,v_2,v_3, v_1v_2, v_2v_3\}$.
	An example of a 1-chain is $x = v_1v_2 + v_2v_3$, with boundary
	$\partial_V x = v_2 - v_1 + v_3 - v_2 = v_3 - v_1$.
	Applying $\partial_V$ again, we explicitly obtain $\partial_V^2 x = 1-1 = 0$.
	See Figure \ref{fig_complex_ex}.

	\begin{figure}[ht]
		\centering
		\psfrag{T_v1}{$v_1$}
		\psfrag{T_v2}{$v_2$}
		\psfrag{T_v3}{$v_3$}
		\psfrag{T_v1v2}{$v_1 v_2$}
		\psfrag{T_v2v3}{$v_2 v_3$}
		\includegraphics{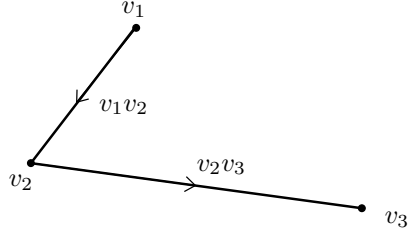}
		\caption{A 1-chain with nontrivial boundary in a simplicial complex $K$.}
		\label{fig_complex_ex}
	\end{figure}
\end{exmp}

\subsubsection{Chain maps and complex morphisms}

Let $K$ and $L$ be simplicial complexes on the finite sets $V$ and $W$, respectively.
A grade preserving $\mathbb{Z}$-linear map 
$F: C(K) \to C(L)$ is called a \emph{chain map}
if the following diagram commutes
$$
	\setlength\arraycolsep{2pt}
	\begin{array}{ccccc}
		C(K)					& \xrightarrow{F} 	& C(L) \\
		\partial_V \downarrow	&					& \downarrow \partial_W \\
		C(K) 					& \xrightarrow{F}	& C(L)
	\end{array}
$$
i.e. if $\partial_W \circ F = F \circ \partial_V$.

A natural way to obtain a chain map is through the following construction.
Let $f: V \to W$ be a map which takes simplices of $K$ to simplices of $L$,
i.e. such that
$$
	\{ f(a) \in W : a \in A \} =: f(A) \in L
$$
whenever $A \in K$.
Then $f$ is called a \emph{complex morphism}.

We can extend a complex morphism $f$ to an 
outermorphism $f_\wedge: \cl(V) \to \cl(W)$
such that $f_\wedge: C(K) \to C(L)$.

\begin{lem}
	$f_\wedge^*(s_W) = s_V$
\end{lem}
\begin{proof}
	For any $v \in V$ we have
	$$
		v * f_\wedge^*(s_W) = f_\wedge(v) * s_W = 1,
	$$
	while $x * f_\wedge^*(s_W) = 0$ for any $x \in \cl^{d \neq 1}(V)$
	because outermorphisms preserve grades.
	From this we conclude that $f_\wedge^*(s_W) = s_V$.
\end{proof}

\begin{cor}
	$f_\wedge$ is a chain map.
\end{cor}
\begin{proof}
	Consider the diagram
	$$
		\setlength\arraycolsep{2pt}
		\begin{array}{ccc}
			C(K)					& \xrightarrow{f_\wedge} 	& C(L) \\
			\partial_V\downarrow	&							& \downarrow \partial_W \\
			C(K) 					& \xrightarrow{f_\wedge}	& C(L)
		\end{array}
	$$
	By Theorem \ref{thm_hestenes}, we have
	$$
		\partial_W f_\wedge x = s_W \liprod (f_\wedge x) = f_\wedge\big( (f_\wedge^* s_W) \liprod x \big)
		= f_\wedge(s_V \liprod x) = f_\wedge \partial_V x
	$$
	for all $x \in C(K)$.
\end{proof}

\subsubsection{An index theorem and Sperner's lemma}

Consider an arbitrary linear map $F: \cl(V) \to \cl(W)$ 
such that $F\partial_V = \partial_W F$
(typically the chain map $f_\wedge$ induced by some $f$)
and define two linear mappings,
the \emph{content} and the \emph{index}, by
$$
	\setlength\arraycolsep{2pt}
	\begin{array}{rccll}
		\Cont: 	& \cl(V) 	& \to 		& \mathbb{Z} 	\\
				& x 		& \mapsto 	& W * (Fx),  & \textrm{and,} \\[5pt]
		\Ind: 	& \cl(V) 	& \to 		& \mathbb{Z} 	\\
				& x 		& \mapsto 	& (Wa) * (F\partial_V x),
	\end{array}
$$
where $a$ is a fixed element in $W$.

\begin{thm}
	$\Cont = \Ind$.
\end{thm}
\begin{proof}
	Since $F$ is a chain map, we have for all $x \in \cl(V)$
	\begin{eqnarray*}
		Wa * F \partial_V x &=& Wa * \partial_W F x = Wa * (s_W \liprod Fx) \\
		&=& (Wa \wedge s_W) * Fx = W * Fx,
	\end{eqnarray*}
	which proves the theorem.
\end{proof}

As an application we will prove Sperner's lemma (in the plane for simplicity).

\begin{figure}[t]
	\centering
	\psfrag{T_A}{$A$}
	\psfrag{T_B}{$B$}
	\psfrag{T_C}{$C$}
	\psfrag{T_a}{$a$}
	\psfrag{T_b}{$b$}
	\psfrag{T_c}{$c$}
	\psfrag{T_v}{$v$}
	\psfrag{T_vp}{$v'$}
	\psfrag{T_vpp}{$v''$}
	\includegraphics{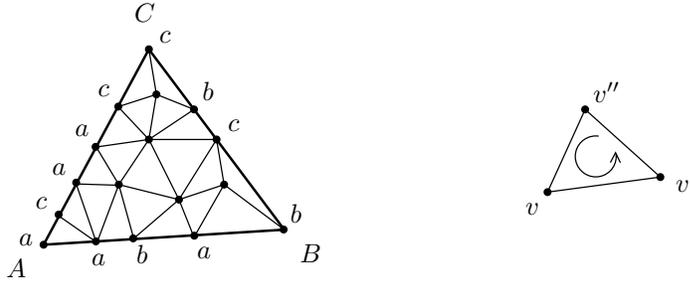}
	\caption{A triangulation of the triangle $ABC$, 
	with a labeling of the nodes of its sides.
	Positive orientation means with respect to the plane of the page, 
	as in the triangle on the right.}
	\label{fig_triangulation}
\end{figure}

\begin{thm}[Sperner's lemma in the plane]
	Consider a triangulation of a triangle $ABC$ and label
	every node in the triangulation 
	with one of the symbols $a,b$ or $c$ in such a manner that
	$A,B,C$ is labeled $a,b,c$ respectively and so that every
	node on the side $AB$ is labeled with $a$ or $b$,
	every node on side $BC$ is labeled with $b$ or $c$,
	and those on $AC$ with $a$ or $c$.
	A triangle in the triangulation is called \emph{complete} if its
	nodes nodes have distinct labels (i.e. $a$, $b$ and $c$).
	Sperner's lemma states that there exists an odd number
	of complete triangles in the triangulation.
\end{thm}
\begin{proof}
	To prove this, let $V$ denote the set of nodes in the triangulation.
	To each triangle in the triangulation we associate the simplex
	$v \wedge v' \wedge v''$ in $\cl(V)$, where the order is
	chosen so that $v,v',v''$ is positively oriented as in Figure \ref{fig_triangulation}.
	Let $x$ denote the sum of all these simplices.
	
	Now, we let $f: V \to W:=\{a,b,c\}$ be a map defined by
	the labeling, and extend it to an outermorphism (chain map)
	$f_\wedge=F: \cl(V) \to \cl(W)$.
	We now see that the triangulation has a complete triangle if
	$\Cont x \neq 0$, i.e. if $W * Fx \neq 0$.
	But by the index theorem, $\Cont = \Ind$,
	so we only have to prove that 
	$\Ind x = Wa * F(\partial_V x) \neq 0$,
	or equivalently that 
	$abca * \partial_W Fx = bc * \partial_W Fx \neq 0$.
	Note that $\partial_V x$ is just the boundary chain of the triangulation.
	Hence, the only terms in $\partial_W Fx$ we need to
	consider are those corresponding to the nodes on the segment $BC$.
	But since $b \wedge b = c \wedge c = 0$ we conclude that
	(see Figure \ref{fig_triangle-seg})
	$$
		\Ind x = bc * (b \wedge c + c \wedge b + b \wedge c + c \wedge b + \ldots + b \wedge c)
		= bc * bc = -1,
	$$
	so that $W * Fx = -1$.
	This proves that there is an odd number of complete triangles.
\end{proof}

\begin{figure}[t]
	\centering
	\psfrag{T_A}{$A$}
	\psfrag{T_B}{$B$}
	\psfrag{T_C}{$C$}
	\psfrag{T_a}{$a$}
	\psfrag{T_b}{$b$}
	\psfrag{T_c}{$c$}
	\includegraphics{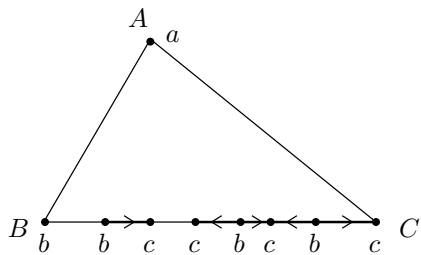}
	\caption{The nodes on the segment $BC$ of a triangulation of $ABC$.}
	\label{fig_triangle-seg}
\end{figure}

\noindent
By using induction on the dimension it is now easy to prove
Sperner's lemma in higher dimensions.

\begin{exc}
	Formulate and prove Sperner's lemma in general.
\end{exc}

\subsubsection{Homology}

Another example where chain maps are useful is in homology theory,
which is an invaluable tool in e.g. topology.

Again, let $K$ be a complex on $V$,
and denote the boundary map on $V$ by $\partial_V$. 
We define the \emph{cycles} $Z(K) = \ker \partial_V$
and the \emph{boundaries} $B(K) = \im \partial_V$.

Since $\partial_V^2 = 0$ we have $B(K) \subseteq Z(K)$.
Hence, we can introduce the \emph{homology}
$$
	H(K) = Z(K)/B(K).
$$
We also define $d$-cycles, $d$-boundaries, and the
corresponding $d$th homology group,
\begin{eqnarray*}
	Z_d(K) &=& Z(K) \cap C_d(K), \\
	B_d(K) &=& B(K) \cap C_d(K), \\
	H_d(K) &=& Z_d(K)/B_d(K).
\end{eqnarray*}

The fundamental theorem of finitely generated abelian groups
(see e.g. \cite{lang}) states that
every finitely generated abelian group is a direct sum of cyclic subgroups.
Thus,
$$
	H_d(K) = F_d \oplus T_d,
$$
where $F_d \cong \mathbb{Z}^r$ for some $r \in \mathbb{N}$ called
the \emph{rank} of $H_d(K)$ or the \emph{Betti-number} denoted $\beta_d$,
and $T_d$ is the \emph{torsion part}:
for all $t \in T_d$ there is a positive integer $n$ such that $nt = 0$.

\begin{thm}
	Every chain map $F: C(K) \to C(L)$ induces a homomorphism
	$F_*: H(K) \to H(L)$ on the homology level.
	Moreover, $F_*: H_d(K) \to H_d(L)$ for all $d \ge 0$.
\end{thm}
\begin{proof}
	Consider the commutative diagram
	$$
		\setlength\arraycolsep{2pt}
		\begin{array}{ccccc}
			C(K)		& \xrightarrow{\partial_V} 	& C(K) \\
			F\downarrow &							& \downarrow F \\
			C(L) 		& \xrightarrow{\partial_W}	& C(L)
		\end{array}
	$$
	We have for any $b \in B(K)$ that $b = \partial_V x$,
	for some $x \in V$, so that
	$$
		Fb = F\partial_V x = \partial_W Fx \in B(L).
	$$
	Also, $z \in Z(K)$ implies
	$
		\partial_W Fz = F \partial_V z = F0 = 0,
	$
	and $Fz \in Z(L)$.
	Hence, the map
	$$
		\setlength\arraycolsep{2pt}
		\begin{array}{ccc}
			H(K) &\to& H(L) \\
			z+B(K) &\mapsto& Fz + B(L)
		\end{array}
	$$
	is well defined.
	Lastly, $F_*: H_d(K) \to H_d(L)$ follows since $F$ is grade preserving.
\end{proof}

\begin{rem}
	From Exercise \ref{exc_wedge_commutativity} 
	and \eqref{boundary_squared} it follows
	that a more general boundary map $\partial_p := s_p \liprod$, 
	with $s_p := \sum_{A \in K, |A| = p} A$ and $p$ odd,
	would also give rise to a homology on $K$.
\end{rem}

\subsection{The number of spanning trees in a graph}

We consider a finite connected simple graph 
(i.e. no loops or multi-edges) $G$ with a set of edges $E$ and vertices $V$.
A subset $T$ of $E$ is called a \emph{spanning tree} if it is connected,
has no cycles, and contains every vertex in $V$.
This is equivalent with having no cycles and having $|V|-1$ edges.
In this subsection we will use Clifford algebra to prove
a well known result about the number of spanning trees in $G$
(see e.g. \cite{petersson} and references therein for recent generalizations).

Let us first choose a total ordering on $V$, denoted $\prec$.
We introduce the Clifford algebras $\cl(V)$ and $\cl(E)$ with
signatures $1$. If $e = \{v',v\} \in E$ with $v' \prec v$, then
we define $\delta(e) := v - v'$, and extend $\delta$ to an
outermorphism
$$
	\delta: \cl(E) \to \cl(V).
$$
We will need two simple lemmas.

\begin{lem}
	If the subgraph of $G$ induced by $F \subseteq E$
	contains a cycle \\ $v_0v_1 \ldots v_kv_0$, then $\delta F = 0$.
\end{lem}
\begin{proof}
	The identity
	$$
		(v_0 - v_1) + (v_1 - v_2) + \ldots + (v_{k-1} - v_k) + (v_k - v_0) = 0
	$$
	implies $(v_0 - v_1) \wedge (v_1 - v_2) \wedge \ldots \wedge (v_k - v_0) = 0$,
	and $\delta F = 0$.
\end{proof}

\begin{lem}
	If $F \subseteq E$, $\delta F \neq 0$ and $|F| = |V|-1$,
	then there exists an $\epsilon = \pm 1$, such that
	$v \wedge \delta F = \epsilon V$ for every $v$ in $V$.
\end{lem}
\begin{proof}
	Note that, if $e = \{v,v'\}$ is an edge in $F$, then we have
	$(v-v') \wedge \delta F = \delta(e \wedge F) = 0$.
	Then, since the conditions imply that the subgraph induced
	by $F$ is connected and contains every vertex of $V$,
	we conclude that $v \wedge \delta F$ does not depend on $v \in V$.
	
	Fix a vertex $v$ and 
	let $v_1,v_2,\ldots,v_s$ be the neighbourhood of $v$ in the subgraph induced by $F$, i.e.
	all the vertices in $V$
	such that $e_1 := \{v,v_1\}, e_2:= \{v,v_2\}, \ldots, e_s := \{v,v_s\}$
	are edges in $F$. Denote the set of remaining edges in $F$ by $F'$.
	Then
	$$
		v \wedge (v_1 - v) \wedge (v_2 - v) \wedge \ldots \wedge (v_s-v) = v \wedge v_1 \wedge v_2 \wedge \ldots \wedge v_s,
	$$
	so that, up to some signs $\epsilon_k$,
	$$
		v \wedge \delta F = \epsilon_1 v \wedge \delta(e_1 e_2 \ldots e_s) \wedge \delta F' = \epsilon_1 \epsilon_2 v \wedge v_2 \wedge \ldots \wedge v_s \wedge v_1 \wedge \delta F'.
	$$
	By repeating this argument with $v_1$ in place of $v$, and so on,
	we obviously obtain $v \wedge \delta F = \epsilon V$, 
	for some $\epsilon \in \{-1,1\}$ independent of $v$.
	This proves the lemma.
\end{proof}

Before we state the theorem we define the outermorphism
``Laplace of $G$'' by
$$
	\Delta = \delta \circ \delta^*: \cl(V) \to \cl(V),
$$
where $\delta^*: \cl(V) \to \cl(E)$ is defined 
by the property $\delta^* x * y = x * \delta y$ for all $x \in \cl(V)$, $y \in \cl(E)$.
The adjugate $\Delta^{\adj}$ of $\Delta$ was
earlier defined by $\Delta^{\adj} = \Delta^{*\mathbf{c}}$,
where $\Delta^* = (\delta \circ \delta^*)^* = \delta \circ \delta^* = \Delta$
and $\Delta\dual(x) = \Delta(xV)V^{-1}$.
Hence,
$$
	\Delta^{\textup{adj}}(x) = \Delta(xV)V^\dagger \quad \forall x \in \cl(V).
$$

\begin{thm}[Kirchoff's matrix tree theorem]
	For all $u,v \in V$ we have
	$$
		u * \Delta^{\textup{adj}} v = N,
	$$
	where $N$ denotes the number of spanning trees in $G$.
\end{thm}

\begin{proof}
	In the proof we will use the following simple rules for $v \in V$,
	$x,y,z \in \cl(V)$:
	\begin{enumerate}
		\item[\textperiodcentered] 
			$\delta(F^\dagger) = (\delta F)^\dagger$
		\item[\textperiodcentered] 
			$x * y = \langle xy \rangle_0 = \langle (xy)^\dagger \rangle_0 
			= \langle y^\dagger x^\dagger \rangle_0 = y^\dagger * x^\dagger = y * x$
		\item[\textperiodcentered] 
			$x * (yz) = \langle xyz \rangle_0 = \langle zxy \rangle_0 = (xy)*z$
		\item[\textperiodcentered] 
			$vx = v \liprod x + v \wedge x$
	\end{enumerate}
	We have
	$$
		u * \Delta^{\textup{adj}} v = u \thinspace * \thinspace \delta \circ \delta^*(vV)V^\dagger 
		= V^\dagger u \thinspace * \thinspace \delta \circ \delta^* (vV) = \delta^*(V^\dagger u) * \delta^*(vV).
	$$
	Denote $n = |V|$.
	Now, since $\delta^*$ is grade preserving, 
	and $V^\dagger u$ and $vV$ are of grade $n-1$, we conclude that
	$\delta^*(V^\dagger u)$ and $\delta^*(vV)$ also have grade $n-1$.
	By expanding these in a basis of subsets of $E$ we obtain
	\begin{eqnarray*}
		u * \Delta^{\textup{adj}} v 
			&=& \left( \sum_{\genfrac{}{}{0pt}{}{F \subseteq E}{|F| = n-1}} \left( \delta^*(V^\dagger u) * F \right) F^\dagger \right)
			  * \left( \sum_{\genfrac{}{}{0pt}{}{F \subseteq E}{|F| = n-1}} \left( \delta^*(vV) * F^\dagger \right) F \right) \\
			&=& \sum_{\genfrac{}{}{0pt}{}{F \subseteq E}{|F| = n-1}} \left( (V^\dagger u) * \delta F \right)   \left( (vV) * \delta F^\dagger \right) \\
			&=&	\sum_{\genfrac{}{}{0pt}{}{F \subseteq E}{|F| = n-1}} \left( V^\dagger * (u\delta F) \right)   \left( V^\dagger * (v\delta F) \right) \\
			&=&	\sum_{\genfrac{}{}{0pt}{}{F \subseteq E}{|F| = n-1}} \left\langle V^\dagger (u \wedge \delta F) \right\rangle_0   \left\langle V^\dagger (v \wedge \delta F) \right\rangle_0 \\
			&=&	\sum_{\textrm{$F$ spanning tree}} \left\langle V^\dagger \epsilon_F V \right\rangle_0   \left\langle V^\dagger \epsilon_F V \right\rangle_0
			\ = \ N.
	\end{eqnarray*}
\end{proof}

\begin{rem}
	Note that the boundary map $\delta$ on edges is induced by
	the boundary map $\partial_V$ defined previously.
	In general, we can consider the induced structure
	$$
		\setlength\arraycolsep{2pt}
		\begin{array}{ccccc}
			\cl(\binom{V}{k})	&\hookleftarrow	& \cl^1(\binom{V}{k})	& \leftrightarrow	& \cl^k(V) \\[5pt]
			\downarrow \delta 	&				& \delta \downarrow		&					& \downarrow \partial_V \\[5pt]
			\cl(\binom{V}{k-1})	&\hookleftarrow	& \cl^1(\binom{V}{k-1})	& \leftrightarrow	& \cl^{k-1}(V)
		\end{array}
	$$
	(or similarly for $\partial_p$),
	where $\binom{V}{k}$ denotes the set of $k$-subsets of $V$.
\end{rem}

\subsection{Fermionic operators}

	Here we consider a finite set $E = \{e_1,\ldots,e_n\}$ 
	and its corresponding Clifford algebra with signature 1
	and scalars $\mathbb{F}$, which we denote
	by $\mathcal{F} = \cl(E,\mathbb{F},1)$.
	Define the \emph{fermion creation and annihilation operators}
	acting on $\psi \in \mathcal{F}$
	$$
		c_i^\dagger(\psi) := e_i \wedge \psi, \qquad
		c_i(\psi) := e_i \liprod \psi.
	$$
	Again, by the properties of the inner and outer products, 
	we find, for any $\psi \in \mathcal{F}$,
	\begin{eqnarray*}
	&	c_i c_j(\psi) = (e_i \wedge e_j) \liprod \psi = -c_j c_i(\psi), \\
	&	c_i^\dagger c_j^\dagger(\psi) = e_i \wedge e_j \wedge x = -c_j^\dagger c_i^\dagger(\psi), \\
	&	c_i c_j^\dagger(\psi) = e_i \liprod (e_j \wedge \psi) = \delta_{ij}\psi - e_j \wedge (e_i \liprod \psi) 
		= \delta_{ij}\psi - c_j^\dagger c_i(\psi).
	\end{eqnarray*}
	We summarize these properties by the \emph{anticommutation relations}
	\begin{equation} \label{fermion_anticommutators}
		\{c_i,c_j\} = 0, \quad
		\{c_i,c_j^\dagger\} = \delta_{ij}, \quad
		\{c_i^\dagger,c_j^\dagger\} = 0,
	\end{equation}
	where $\{A,B\} := AB + BA$ is called the \emph{anticommutator}.

	With $\mathbb{F} = \mathbb{C}$,
	the inner product $\langle \phi, \psi \rangle_\mathcal{F} := \bar{\phi}^\dagger * \psi$ 
	($\bar{\psi}$ denotes the complex conjugate
	acting on the scalar coefficients, and $\psi^\dagger$ is the reverse)
	makes $\mathcal{F}$ into a complex $2^n$-dimensional Hilbert space.
	We also find that
	$$
		\langle \phi, c_i \psi \rangle_\mathcal{F} 
		= \bar{\phi}^\dagger * (e_i \liprod \psi) = \overline{\phi^\dagger \wedge e_i} * \psi
		= \langle c_i^\dagger \phi, \psi \rangle_\mathcal{F},
	$$
	so that $c_i^\dagger$ is the adjoint of $c_i$ 
	with respect to this inner product.
	In quantum mechanics, this space $\mathcal{F}$ is usually called the
	\emph{fermionic Fock space of $n$ particles}
	and the physical interpretation is that each of the operators
	$c_i^\dagger$ and $c_i$ creates resp. removes a particle of type/location $i$.
	Because $(c_i^\dagger)^2 = 0$, there cannot be more than one 
	particle of the same type at the same location\footnote{This 
	characterizes \emph{fermionic} particles in contrast to \emph{bosonic} particles.}.
	Any state $\psi \in \mathcal{F}$ can be written as a linear combination
	$$
		\psi = \left( \alpha_0 + \sum_i \alpha_i c_i^\dagger + \sum_{i<j} \alpha_{ij} c_i^\dagger c_j^\dagger
		+ \ldots + \alpha_{1 \ldots n} c_1^\dagger \ldots c_n^\dagger \right)\psi_0
	$$
	of particle creations acting on an empty, or zero-particle, state $\psi_0 = 1$.
	The scalar coefficients $\alpha_\mathbf{i}$ determine the probabilities 
	$p_\mathbf{i} = |\alpha_\mathbf{i}|^2 / \|\psi \|_\mathcal{F}^2$ 
	of measuring the physical state $\psi$ in a definite particle configuration 
	$e_\mathbf{i} = c_{i_1}^\dagger c_{i_2}^\dagger \ldots c_{i_k}^\dagger \psi_0$.
	There is also a \emph{number operator},
	$n_F := \sum_i c_i^\dagger c_i$,
	which simply counts the number of particles present in a state.

	\begin{exc}
		Show that $n_F\psi = k\psi$ for 
		$\psi \in \mathcal{F}_k := \cl^k(E,\mathbb{C},1)$,
		both by using the anticommutation relations \eqref{fermion_anticommutators},
		and in an alternative way by using the properties of the
		inner and outer products.
	\end{exc}

\subsection{Infinite-dimensional Clifford algebra}

	This far we have only defined the Clifford algebra 
	$\cl(X,R,r)$ of a \emph{finite} set $X$, resulting
	in the \emph{finite-dimen\-sional} geometric algebra 
	$\mathcal{G}(V)$ whenever $R$ is a field.
	In order for this combinatorial construction to 
	qualify as a complete generalization of $\mathcal{G}$,
	we would like to be able to define the corresponding
	Clifford algebra of an infinite-dimensional vector space, something which
	was possible for $\mathcal{G}$ in Definition \ref{def_g}.
	
	The treatment of $\cl$ in the previous subsections 
	is in a form which easily generalizes to an infinite $X$.
	Reconsidering Definition \ref{def_cl}, we now have two possibilites: 
	either we consider the set $\mathscr{P}(X)$ of all subsets of $X$,
	or just the set $\mathscr{F}(X)$ of all \emph{finite} subsets.
	With the tensor-algebraic definition in mind,
	we choose to first consider only $\mathscr{F}(X)$,
	and then the generalized case $\mathscr{P}(X)$ in the next subsection.
	We therefore define, for an arbitrary set $X$, ring $R$, and signature $r \!: X \to R$,
	\begin{equation}
		\cl_\mathscr{F}(X,R,r) := \bigoplus_{\mathscr{F}(X)} R.
	\end{equation}
	Elements in $\cl_\mathscr{F}$ are then finite linear combinations
	of finite subsets of $X$.
	
	Our problem now is to define a Clifford product for $\cl_\mathscr{F}$.
	This can be achieved just as in the finite case if only we can find 
	a map $\tau \!: \mathscr{F}(X) \times \mathscr{F}(X) \to R$
	satisfying the conditions in Lemma \ref{lem_finite_tau}.
	We show below that this is possible.
	
	We call a map $\tau \!: \mathscr{F}(X) \times \mathscr{F}(X) \to R$ \emph{grassmannian on} $X$ if it
	satisfies (\emph{i})-(\emph{v}) in Lemma \ref{lem_finite_tau}, with 
	$\mathscr{P}(X)$ replaced by $\mathscr{F}(X)$.
	
	\begin{thm}
		For any given $X,R,r$ there exists a grassmannian map on $\mathscr{F}(X)$.
	\end{thm}
	\begin{proof}[*Proof]
		We know that there exists such a map for any finite $X$.
		Let $Y \subseteq X$ and assume $\tau' \!: \mathscr{F}(Y) \times \mathscr{F}(Y) \to R$
		is grassmannian on $Y$.
		If there exists $z \in X \smallsetminus Y$ we can, by proceeding as in the
		proof of Lemma \ref{lem_finite_tau}, extend $\tau'$ to a grassmannian map 
		$\tau \!: \mathscr{F}(Y\cup\{z\}) \times \mathscr{F}(Y\cup\{z\}) \to R$
		on $Y\cup\{z\}$ such that $\tau|_{\mathscr{F}(Y) \times \mathscr{F}(Y)} = \tau'$.
		
		We will now use \emph{transfinite induction}, or the \emph{Hausdorff maximality theorem}\footnote{This 
		theorem should actually be regarded as an axiom of set theory since it is equivalent to the Axiom of Choice.},
		to prove that $\tau$ can be extended to all of $\mathscr{F}(X) \subseteq \mathscr{P}(X)$.
		Note that if $\tau$ is grassmannian on $Y \subseteq X$ then $\tau$
		is also a relation $\tau \subseteq \mathscr{P}(X) \times \mathscr{P}(X) \times R$.
		Let
		\begin{equation}
			\mathcal{H} := \Big\{ (Y,\tau) \in \mathscr{P}(X) \times \mathscr{P}\big( \mathscr{P}(X) \times \mathscr{P}(X) \times R \big) 
				: \textrm{$\tau$ is grassmannian on $Y$} \Big\}.
		\end{equation}
		Then $\mathcal{H}$ is partially ordered by
		\begin{equation}
			(Y,\tau) \leq (Y',\tau') \qquad \textrm{iff} \qquad 
				Y \subseteq Y' \quad \textrm{and} \quad \tau'|_{\mathscr{F}(Y) \times \mathscr{F}(Y)} = \tau.
		\end{equation}
		By the Hausdorff maximality theorem, there exists a maximal totally 
		ordered ``chain" $\mathcal{K} \subseteq \mathcal{H}$.
		Put $Y^* := \bigcup_{(Y,\tau) \in \mathcal{K}} Y$. We want to define a 
		grassmannian map $\tau^*$ on $Y^*$, for if we succeed in that, we find 
		$(Y^*,\tau^*) \in \mathcal{H} \cap \mathcal{K}$ and can conclude that $Y^* = X$
		by maximality and the former result.
		
		Take finite subsets $A$ and $B$ of $Y^*$. Each of the finite elements
		in $A \cup B$ lies in some $Y$ such that $(Y,\tau) \in \mathcal{K}$.
		Therefore, by the total ordering of $\mathcal{K}$, there exists one
		such $Y$ containing $A \cup B$.
		Put $\tau^* (A,B) := \tau(A,B)$, where $(Y,\tau)$ is this chosen element in $\mathcal{K}$.
		$\tau^*$ is well-defined since if $A \cup B \subseteq Y$ and $A \cup B \subseteq Y'$
		where $(Y,\tau), (Y',\tau') \in \mathcal{K}$ then $Y \subseteq Y'$ or $Y' \subseteq Y$
		and $\tau,\tau'$ agree on $(A,B)$.
		It is easy to verify that this $\tau^*$ is grassmannian on $Y^*$,
		since for each $A,B,C \in \mathscr{F}(X)$ there exists $(Y,\tau) \in \mathcal{K}$ 
		such that $A \cup B \cup C \subseteq Y$.
	\end{proof}
	
	We have shown that there exists a map $\tau \!: \mathscr{F}(X) \times \mathscr{F}(X) \to R$
	with the properties in Lemma \ref{lem_finite_tau}.
	We can then define the Clifford product on $\cl_\mathscr{F}(X)$ as usual by
	$AB := \tau(A,B) A \symdiff B$ for $A,B \in \mathscr{F}(X)$ and linear extension.
	Since only finite subsets are included, most of the previous
	constructions for finite-dimensional $\cl$ carry over to $\cl_\mathscr{F}$.
	For example, the decomposition into graded subspaces remains but now extends towards infinity,
	\begin{equation} \label{infinite_grade_decomposition}
		\cl_\mathscr{F} = \bigoplus_{k=0}^{\infty} \cl_\mathscr{F}^k.
	\end{equation}
	Furthermore, Proposition \ref{prop_reverse} still holds, so the reverse and
	all other involutions behave as expected.
	On the other hand, note that we cannot talk about a pseudoscalar in this context.

	Let us now see how $\cl_\mathscr{F}$ can be applied to the 
	setting of an in\-finite-dimen\-sion\-al vector space $V$ 
	over a field $\mathbb{F}$ and with a quadratic form $q$.
	By the Hausdorff maximality theorem one can actually find a (necessarily infinite) orthogonal
	basis $E$ for this space in the sense that any vector in $V$
	can be written as a finite linear combination of elements in $E$ and
	that $\beta_q(e,e') = 0$ for any pair of disjoint elements $e,e' \in E$.
	We then have
	\begin{equation}
		\mathcal{G}(V,q) \cong \cl_\mathscr{F}(E,\mathbb{F},q|_E),
	\end{equation}
	which is proved just like in the finite-dimensional case, using
	Proposition \ref{prop_universality}.

	Let us consider an application of the infinite-dimensional
	Clifford algebra $\cl_\mathscr{F}$. Given a vector space $V$,
	define the \emph{full simplicial complex algebra}
	\begin{equation}
		\mathscr{C}(V) := \cl_\mathscr{F}(V,R,1),
	\end{equation}
	where we forget about the vector space structure of $V$ and treat
	individual points $\dot{v} \in V$ as orthogonal basis 1-vectors 
	in $\cl_\mathscr{F}^1$ with $\dot{v}^2 = 1$.
	The dot indicates that we think of $v$ as a point rather than a vector.
	A basis $(k+1)$-blade in $\mathscr{C}(V)$ consists of a product
	$\dot{v}_0 \dot{v}_1 \ldots \dot{v}_k$ of individual points and represents
	a (possibly degenerate) oriented $k$-simplex embedded in $V$. 
	This simplex is given by the convex hull
	\begin{equation}
		\textrm{Conv}\{v_0,v_1,\ldots,v_k\} := \left\{ \sum_{i=0}^k \alpha_i v_i \in V : \alpha_i \geq 0, \sum_{i=0}^k \alpha_i = 1 \right\}.
	\end{equation}
	Hence, an arbitrary element in $\mathscr{C}(V)$ is a linear
	combination of simplices and can therefore represent a simplicial complex embedded in $V$.
	Restricting to a chosen embedded simplex $K \subseteq \mathscr{F}(V)$,
	we return to the context of the chain group $C(K)$ introduced previously.
	However, the difference here is that the chain group of \emph{any} 
	simplicial complex properly embedded in $V$ can be found as a
	subgroup of $\mathscr{C}(V)$ -- even infinite such complexes 
	(composed of finite-dimensional simplices).
	
	We can define a boundary map for the full simplicial complex algebra,
	\begin{displaymath}
	\setlength\arraycolsep{2pt}
	\begin{array}{l}
		\partial_V \!: \mathscr{C}(V) \to \mathscr{C}(V), \\[5pt]
		\displaystyle \quad \partial_V(x) := \sum_{\dot{v} \in V} \dot{v} \liprod x.
	\end{array}
	\end{displaymath}
	Note that this is well-defined since only a finite number of points $\dot{v}$ 
	can be present in any fixed $x$.
	Also, $\partial_V^2 = 0$ follows by analogy with $s_V$, or simply
	\begin{equation}
		\partial_V^2(x) 
			= \sum_{\dot{u} \in V} \dot{u} \liprod \left( \sum_{\dot{v} \in V} \dot{v} \liprod x \right)
			= \sum_{\dot{u}, \dot{v} \in V} \dot{u} \liprod (\dot{v} \liprod x)
			= \sum_{\dot{u}, \dot{v} \in V} (\dot{u} \wedge \dot{v}) \liprod x = 0.
	\end{equation}
	
	We can also assign a \emph{geometric measure} $\sigma$ to embedded simplices,
	by mapping a $k$-simplex to a corresponding $k$-blade in $\mathcal{G}(V)$
	representing the directed volume of the simplex.
	Define $\sigma \!: \mathscr{C}(V) \to \mathcal{G}(V)$ by
	\begin{displaymath}
	\setlength\arraycolsep{2pt}
	\begin{array}{rcl}
		\sigma(1) &:=& 0, \\[5pt]
		\sigma(\dot{v}) &:=& 1, \\[5pt]
		\sigma(\dot{v}_0 \dot{v}_1 \ldots \dot{v}_k) &:=& \frac{1}{k!} (v_1-v_0) \wedge (v_2-v_0) \wedge \cdots \wedge (v_k-v_0),
	\end{array}
	\end{displaymath}
	and extending linearly. One can verify that this is well-defined and that the
	geometric measure of a boundary is zero, i.e. $\sigma \circ \partial_V = 0$.
	One can take this construction even further and arrive at
	``discrete" equivalents of differentials, integrals and Stokes' theorem.
	See \cite{svensson} or \cite{naeve_svensson} for more on this.

\subsection{*Generalized infinite-dimensional Clifford algebra}

	We define, for an arbitrary set $X$, ring $R$, and signature $r \!: X \to R$,
	\begin{equation}
		\cl(X,R,r) := \bigoplus_{\mathscr{P}(X)} R.
	\end{equation}
	Hence, elements of $\cl(X)$ are finite linear combinations of 
	(possibly infinite) subsets of $X$.
	The following theorem asserts that it is possible to extend $\tau$ all the
	way to $\mathscr{P}(X)$ even in the infinite case. 
	We therefore have a Clifford product also on $\cl(X)$.
	
	\begin{thm} \label{thm_infinite_tau}
		For any set $X$ there exists a map $|\cdot|_2 \!: \mathscr{P}\big(\mathscr{P}(X)\big) \to \mathbb{Z}_2$
		such that
		\begin{displaymath}
		\setlength\arraycolsep{2pt}
		\begin{array}{rl}
			i)   & |\mathcal{A}|_2 \equiv |\mathcal{A}| \pmod{2} 
				\qquad \textrm{for finite $\mathcal{A} \subseteq \mathscr{P}(X)$}, \\[5pt]
			ii)  & |\mathcal{A} \symdiff \mathcal{B}|_2 = |\mathcal{A}|_2 + |\mathcal{B}|_2 \pmod{2}
		\end{array}
		\end{displaymath}
		Furthermore, for any commutative ring $R$ with unit, and signature $r\!: X \to R$
		such that $r(X)$ is contained in a finite and multiplicatively closed subset of $R$,
		there exists a map  $\tau \!: \mathscr{P}(X) \times \mathscr{P}(X) \to R$ such that
		properties \emph{(\emph{i})-(\emph{v})} in Lemma \ref{lem_finite_tau} hold, plus
		\begin{displaymath}
			vi) \ \tau(A,B) = (-1)^{\left|\binom{A}{2}\right|_2 + \left|\binom{B}{2}\right|_2 + \left|\binom{A \bigtriangleup B}{2}\right|_2}\ \tau(B,A)
				\quad \forall\ A,B \in \mathscr{P}(X).
		\end{displaymath}
	\end{thm}
	
	\noindent
	Here, $\binom{A}{n}$ denotes the set of all subsets of $A$ with $n$ elements.
	Note that for a finite set $A$, $\big|\binom{A}{n}\big| = \binom{|A|}{n}$ so that for example
	$\big|\binom{A}{1}\big| = |A|$ (in general, $\textrm{card}\ \binom{A}{1} = \textrm{card}\ A$)
	and $\big|\binom{A}{2}\big| = \frac{1}{2}|A|(|A|-1)$.
	This enables us to extend the basic involutions $\star$, $\dagger$ and $\cliffconj$
	to infinite sets as
	\begin{displaymath}
	\setlength\arraycolsep{2pt}
	\begin{array}{rcl}
		A^\star   &:=& (-1)^{\left|\binom{A}{1}\right|_2}\ A, \\[5pt]
		A^\dagger &:=& (-1)^{\left|\binom{A}{2}\right|_2}\ A,
	\end{array}
	\end{displaymath}
	and because $\big|\binom{A \bigtriangleup B}{1}\big|_2 = \big|\binom{A}{1}\big|_2 + \big|\binom{B}{1}\big|_2 \pmod{2}$
	still holds, we find that they satisfy the fundamental properties
	\eqref{grade_inv_property}-\eqref{cliffconj_property} for all elements of $\cl(X)$.
	The extra requirement (\emph{vi}) on $\tau$ was necessary here since we 
	cannot use Proposition \ref{prop_reverse} for infinite sets.
	Moreover, we can no longer write the decomposition \eqref{infinite_grade_decomposition} 
	since it goes beyond finite grades.
	We do have even and odd subspaces, though, defined by
	\begin{equation} \label{infinite_evenodd_subspaces}
		\cl^{\pm} := \{ x \in \cl : x^\star = \pm x \},
	\end{equation}
	and $\cl^+$ and $\cl_\mathscr{F}$ (with this $\tau$) are both subalgebras of $\cl$.
	
	It should be emphasized that $\tau$ needs not be zero on intersecting infinite sets
	(a rather trivial solution), but if e.g. $r \!: X \to \{\pm 1\}$ we can
	also demand that $\tau \!: \mathscr{P}(X) \times \mathscr{P}(X) \to \{\pm 1\}$.
	A proof of the first part of
	Theorem \ref{thm_infinite_tau} is given in Appendix \ref{app_extension}.

	Although this construction provides a way to handle combinatorics
	of infinite sets, we should emphasize that its applicability in
	the context of infinite dimensional vector spaces is limited.
	Let us sketch an intuitive picture of why this is so.
	
	With a countable basis $E = \{e_i\}_{i=1}^\infty$, an
	infinite basis blade in $\cl$ could be thought of as an infinite product
	$A = e_{i_1} e_{i_2} e_{i_3} \ldots = \prod_{k=1}^\infty e_{i_k}$.
	A change of basis to $E'$ would turn each $e \in E$ into a finite linear combination
	of elements in $E'$, e.g. $e_j = \sum_k \beta_{jk} e_k'$. However,
	this would require $A$ to be an \emph{infinite sum} of basis blades in $E'$, which is not allowed.
	Note that this is no problem in $\cl_\mathscr{F}$ since a
	basis blade $A = \prod_{k=1}^N e_{i_k}$ is a finite product
	and the change of basis therefore results in a finite sum.

	This completes our excursion to infinite-dimensional Clifford algebras
	(we refer to \cite{de_la_Harpe,Plymen,Wene} and references therein for more on this topic).
	In the following sections we will always assume that $X$ is finite
	and $V$ finite-dimensional.

%% file: clifford_isomorphisms.tex
	In this section we establish an extensive set of relations and isomorphisms between real and complex
	geometric algebras of varying signature. This eventually leads to an identification
	of these algebras as matrix algebras over $\mathbb{R}$, $\mathbb{C}$, or the
	quaternions $\mathbb{H}$. The complete listing of such identifications is usually 
	called the \emph{classification} of geometric algebras.

\subsection{Matrix algebra classification}
	
	We have seen that the even subspace $\mathcal{G}^{+}$ of $\mathcal{G}$ constitutes
	a subalgebra. The following proposition shows that this subalgebra actually
	is the geometric algebra of a space of one dimension lower.
	
	\begin{prop} \label{prop_iso_even}
		We have the algebra isomorphisms
		\begin{displaymath}
		\setlength\arraycolsep{2pt}
		\begin{array}{c}
			\mathcal{G}^{+}(\mathbb{R}^{s,t}) \cong \mathcal{G}(\mathbb{R}^{s,t-1}), \\[5pt]
			\mathcal{G}^{+}(\mathbb{R}^{s,t}) \cong \mathcal{G}(\mathbb{R}^{t,s-1}),
		\end{array}
		\end{displaymath}
		for all $s,t$ for which the expressions make sense.
	\end{prop}
	\begin{proof}
		Take an orthonormal basis $\{e_1,\ldots,e_s,\epsilon_1,\ldots,\epsilon_t\}$ of $\mathbb{R}^{s,t}$ 
		such that $e_i^2 = 1$, $\epsilon_i^2 = -1$, and
		a corresponding basis $\{\underline{e}_1,\ldots,\underline{e}_s,\underline{\epsilon}_1,\ldots,\underline{\epsilon}_{t-1}\}$ of $\mathbb{R}^{s,t-1}$.
		Define $f\!: \mathbb{R}^{s,t-1} \to \mathcal{G}^{+}(\mathbb{R}^{s,t})$ by mapping
		\begin{displaymath}
		\setlength\arraycolsep{2pt}
		\begin{array}{ccl}
			\underline{e}_i &\mapsto& e_i \epsilon_t, \quad i=1,\ldots,s, \\[3pt]
			\underline{\epsilon}_i &\mapsto& \epsilon_i \epsilon_t, \quad i=1,\ldots,t-1,
		\end{array}
		\end{displaymath}
		and extending linearly. We then have
		\begin{displaymath}
		\setlength\arraycolsep{2pt}
		\begin{array}{rcl}
			f(\underline{e}_i) f(\underline{e}_j) &=& -f(\underline{e}_j) f(\underline{e}_i), \\[3pt]
			f(\underline{\epsilon}_i) f(\underline{\epsilon}_j) &=& -f(\underline{\epsilon}_j) f(\underline{\epsilon}_i)
		\end{array}
		\end{displaymath}
		for $i \neq j$, and
		\begin{displaymath}
		\setlength\arraycolsep{2pt}
		\begin{array}{c}
			f(\underline{e}_i) f(\underline{\epsilon}_j) = -f(\underline{\epsilon}_j) f(\underline{e}_i), \\[3pt]
			f(\underline{e}_i)^2 = 1, \quad f(\underline{\epsilon}_i)^2 = -1 
		\end{array}
		\end{displaymath}
		for all reasonable $i,j$.
		By Proposition \ref{prop_universality} (universality) we can extend $f$ to a homomorphism $F\!: \mathcal{G}(\mathbb{R}^{s,t-1}) \to \mathcal{G}^{+}(\mathbb{R}^{s,t})$.
		Since $\dim \mathcal{G}(\mathbb{R}^{s,t-1}) = 2^{s+t-1} = 2^{s+t}/2 = \dim \mathcal{G}^{+}(\mathbb{R}^{s,t})$ and
		$F$ is easily seen to be surjective, we have that $F$ is an isomorphism.
		
		For the second statement, we take a corresponding basis 
		$\{\underline{e}_1,\ldots,\underline{e}_t,\underline{\epsilon}_1,\ldots \\ \ldots,\underline{\epsilon}_{s-1}\}$ of $\mathbb{R}^{t,s-1}$
		and define $f\!: \mathbb{R}^{t,s-1} \to \mathcal{G}^{+}(\mathbb{R}^{s,t})$ by
		\begin{displaymath}
		\setlength\arraycolsep{2pt}
		\begin{array}{ccl}
			\underline{e}_i &\mapsto& \epsilon_i e_s, \quad i=1,\ldots,t, \\[3pt]
			\underline{\epsilon}_i &\mapsto& e_i e_s, \quad i=1,\ldots,s-1.
		\end{array}
		\end{displaymath}
		Proceeding as above, we obtain the isomorphism.
	\end{proof}
	
	\begin{cor}
		It follows immediately that
		\begin{displaymath}
		\setlength\arraycolsep{2pt}
		\begin{array}{c}
			\mathcal{G}(\mathbb{R}^{s,t}) \cong \mathcal{G}(\mathbb{R}^{t+1,s-1}), \\[5pt]
			\mathcal{G}^{+}(\mathbb{R}^{s,t}) \cong \mathcal{G}^{+}(\mathbb{R}^{t,s}).
		\end{array}
		\end{displaymath}
	\end{cor}
	
	\noindent
	In the above and further on we use the notation 
	$\mathcal{G}(\mathbb{F}^{0,0}) := \cl(\varnothing,\mathbb{F},\varnothing) = \mathbb{F}$
	for completeness.

	\begin{exmp}
		We have already seen explicitly that 
		$\mathcal{G}^+(\mathbb{R}^2) \cong \mathbb{C} \cong \mathcal{G}(\mathbb{R}^{0,1})$
		and that the even subalgebra of the space algebra is
		$\mathcal{G}^+(\mathbb{R}^3) \cong \mathcal{G}(\mathbb{R}^{0,2})$ 
		i.e. the quaternion algebra (see Example \ref{exmp_space_algebra_decomp}).
		Moreover, since the pseudoscalar $k$ in $\mathbb{H}$ is an imaginary unit,
		we also see that $\mathcal{G}^+(\mathbb{R}^{0,2}) \cong \mathcal{G}^+(\mathbb{R}^{2,0})$.
	\end{exmp}
	
	The property of geometric algebras that leads to their eventual classification
	as matrix algebras is that they can be split up into tensor products of
	geometric algebras of lower dimension.
	
	\begin{prop} \label{prop_iso_tensor}
		We have the algebra isomorphisms
		\begin{displaymath}
		\setlength\arraycolsep{2pt}
		\begin{array}{c}
			\mathcal{G}(\mathbb{R}^{n+2,0}) \cong \mathcal{G}(\mathbb{R}^{0,n}) \otimes \mathcal{G}(\mathbb{R}^{2,0}), \\[5pt]
			\mathcal{G}(\mathbb{R}^{0,n+2}) \cong \mathcal{G}(\mathbb{R}^{n,0}) \otimes \mathcal{G}(\mathbb{R}^{0,2}), \\[5pt]
			\mathcal{G}(\mathbb{R}^{s+1,t+1}) \cong \mathcal{G}(\mathbb{R}^{s,t}) \otimes \mathcal{G}(\mathbb{R}^{1,1}),
		\end{array}
		\end{displaymath}
		for all $n$, $s$ and $t$ for which the expressions make sense.
	\end{prop}
	\begin{proof}
		For the first expression, take orthonormal bases 
		$\{e_i\}$ of $\mathbb{R}^{n+2}$, $\{\underline{\epsilon}_i\}$ of $\mathbb{R}^{0,n}$ 
		and $\{\overline{e}_i\}$ of $\mathbb{R}^2$.
		Define a mapping $f\!: \mathbb{R}^{n+2} \to \mathcal{G}(\mathbb{R}^{0,n}) \otimes \mathcal{G}(\mathbb{R}^2)$ by
		\begin{displaymath}
		\setlength\arraycolsep{1pt}
		\begin{array}{ccrcll}
			e_j &\ \mapsto\ & \underline{\epsilon}_j &\otimes& \overline{e}_1 \overline{e}_2, &\quad j=1,\ldots,n, \\[3pt]
			e_j &\ \mapsto\ & 1 &\otimes& \overline{e}_{j-n}, &\quad j=n+1,n+2,
		\end{array}
		\end{displaymath}
		and extend to an algebra homomorphism $F$ using the universal property.
		Since $F$ maps onto a set of generators for $\mathcal{G}(\mathbb{R}^{0,n}) \otimes \mathcal{G}(\mathbb{R}^2)$
		it is clearly surjective. Furthermore, $\dim \mathcal{G}(\mathbb{R}^{n+2}) = 2^{n+2} = \dim \mathcal{G}(\mathbb{R}^{0,n}) \otimes \mathcal{G}(\mathbb{R}^2)$,
		so $F$ is an isomorphism.
		
		The second expression is proved similarly. For the third expression, take orthonormal bases 
		$\{e_1,\ldots,e_{s+1},\epsilon_1,\ldots,\epsilon_{t+1}\}$ of $\mathbb{R}^{s+1,t+1}$,
		$\{\underline{e}_1,\ldots,\underline{e}_s,\underline{\epsilon}_1,\ldots,\underline{\epsilon}_t\}$ of $\mathbb{R}^{s,t}$
		and $\{\overline{e},\overline{\epsilon}\}$ of $\mathbb{R}^{1,1}$, where
		$e_i^2 = 1, \epsilon_i^2 = -1$ etc. Define $f\!: \mathbb{R}^{s+1,t+1} \to \mathcal{G}(\mathbb{R}^{s,t}) \otimes \mathcal{G}(\mathbb{R}^{1,1})$ by
		\begin{displaymath}
		\setlength\arraycolsep{1pt}
		\begin{array}{ccrcll}
			e_j &\ \mapsto\ & \underline{e}_j &\otimes& \overline{e} \overline{\epsilon}, &\quad j=1,\ldots,s, \\[3pt]
			\epsilon_j &\ \mapsto\ & \underline{\epsilon}_j &\otimes& \overline{e} \overline{\epsilon}, &\quad j=1,\ldots,t, \\[3pt]
			e_{s+1} &\ \mapsto\ & 1 &\otimes& \overline{e}, \\[3pt]
			\epsilon_{t+1} &\ \mapsto\ & 1 &\otimes& \overline{\epsilon}.
		\end{array}
		\end{displaymath}
		Proceeding as above, we can extend $f$ to an algebra isomorphism.
	\end{proof}
	
	We can also relate certain real geometric algebras to complex equivalents.
	
	\begin{prop} \label{prop_iso_real_complex}
		If $s+t$ is odd and $I^2 = -1$ then
		\begin{displaymath}
			\mathcal{G}(\mathbb{R}^{s,t}) 
				\cong \mathcal{G}^{+}(\mathbb{R}^{s,t}) \otimes \mathbb{C} 
				\cong \mathcal{G}(\mathbb{C}^{s+t-1}).
		\end{displaymath}
	\end{prop}
	\begin{proof}
		Since $s+t$ is odd, the pseudoscalar $I$ commutes with all elements.
		This, together with the property $I^2 = -1$, makes it a good candidate for a scalar imaginary.
		Define $F\!: \mathcal{G}^{+}(\mathbb{R}^{s,t}) \otimes \mathbb{C} \to \mathcal{G}(\mathbb{R}^{s,t})$ 
		by linear extension of
		\begin{displaymath}
		\setlength\arraycolsep{1pt}
		\begin{array}{ccccll}
			E &\otimes& 1 &\ \mapsto\ & E & \quad \in \mathcal{G}^+, \\[3pt]
			E &\otimes& i &\ \mapsto\ & EI & \quad \in \mathcal{G}^-,
		\end{array}
		\end{displaymath}
		for even basis blades $E$. $F$ is easily seen to be an injective algebra homomorphism. 
		Using that the dimensions of these algebras are equal, we have an isomorphism.
		
		For the second isomorphism, note that Proposition \ref{prop_iso_even} gives
		$\mathcal{G}^{+}(\mathbb{R}^{s,t}) \otimes \mathbb{C} \cong \mathcal{G}(\mathbb{R}^{s,t-1}) \otimes \mathbb{C}$.
		Finally, the order of complexification is unimportant since all nondegenerate
		complex quadratic forms are equivalent.
	\end{proof}
	
	\begin{cor}
		It follows immediately that, for these conditions,
		\begin{displaymath}
			\mathcal{G}(\mathbb{R}^{s,t}) \cong \mathcal{G}(\mathbb{R}^{p,q-1}) \otimes \mathbb{C}
		\end{displaymath}
		for any $p \geq 0$, $q \geq 1$ such that $p+q=s+t$.
	\end{cor}
	
	One important consequence of the tensor algebra isomorphisms in Proposition \ref{prop_iso_tensor}
	is that geometric algebras experience a kind of periodicity over 8 real 
	dimensions in the underlying vector space.
	
	\begin{prop} \label{prop_iso_periodicity}
		For all $n \geq 0$, there are periodicity isomorphisms
		\begin{displaymath}
		\setlength\arraycolsep{2pt}
		\begin{array}{c}
			\mathcal{G}(\mathbb{R}^{n+8,0}) \cong \mathcal{G}(\mathbb{R}^{n,0}) \otimes \mathcal{G}(\mathbb{R}^{8,0}), \\[5pt]
			\mathcal{G}(\mathbb{R}^{0,n+8}) \cong \mathcal{G}(\mathbb{R}^{0,n}) \otimes \mathcal{G}(\mathbb{R}^{0,8}), \\[5pt]
			\mathcal{G}(\mathbb{C}^{n+2}) \cong \mathcal{G}(\mathbb{C}^{n}) \otimes_\mathbb{C} \mathcal{G}(\mathbb{C}^{2}).
		\end{array}
		\end{displaymath}
	\end{prop}
	\begin{proof}
		Using Proposition \ref{prop_iso_tensor} 
		repeatedly, we obtain
		\begin{displaymath}
		\setlength\arraycolsep{2pt}
		\begin{array}{rcl}
			\mathcal{G}(\mathbb{R}^{n+8,0})
				&\cong& \mathcal{G}(\mathbb{R}^{0,n+6}) \otimes \mathcal{G}(\mathbb{R}^{2,0}) \\[5pt]
				&\cong& \mathcal{G}(\mathbb{R}^{n,0}) \otimes \mathcal{G}(\mathbb{R}^{0,2}) \otimes \mathcal{G}(\mathbb{R}^{2,0}) \otimes \mathcal{G}(\mathbb{R}^{0,2}) \otimes \mathcal{G}(\mathbb{R}^{2,0}) \\[5pt]
				&\cong& \mathcal{G}(\mathbb{R}^{n,0}) \otimes \mathcal{G}(\mathbb{R}^{8,0}),
		\end{array}
		\end{displaymath}
		and analogously for the second statement.
		
		For the last statement we take orthonormal bases $\{e_i\}$ of $\mathbb{C}^{n+2}$, $\{\underline{e}_i\}$ of $\mathbb{C}^n$ and $\{\overline{e}_i\}$ of $\mathbb{C}^2$.
		Define a mapping $f\!: \mathbb{C}^{n+2} \to \mathcal{G}(\mathbb{C}^n) \otimes_\mathbb{C} \mathcal{G}(\mathbb{C}^2)$ by
		\begin{displaymath}
		\setlength\arraycolsep{1pt}
		\begin{array}{ccrcll}
			e_j &\ \mapsto\ & i\ \underline{e}_j &\otimes_\mathbb{C}& \overline{e}_1 \overline{e}_2, &\quad j=1,\ldots,n, \\[3pt]
			e_j &\ \mapsto\ & 1 &\otimes_\mathbb{C}& \overline{e}_{j-n}, &\quad j=n+1,n+2,
		\end{array}
		\end{displaymath}
		and extend to an algebra isomorphism as usual.
	\end{proof}
	
	\begin{thm} \label{thm_classification_real}
		We obtain the classification of real geometric algebras as matrix algebras,
		given by Table \ref{table_classification_real} together with the periodicity
		\begin{displaymath}
			\mathcal{G}(\mathbb{R}^{s+8,t}) \cong \mathcal{G}(\mathbb{R}^{s,t+8}) \cong \mathcal{G}(\mathbb{R}^{s,t}) \otimes \mathbb{R}^{16 \times 16}.
		\end{displaymath}
	\end{thm}
	
	\begin{sidewaystable}\centering
		\begin{displaymath}
		\setlength\arraycolsep{2pt}
		\begin{array}{|c||c|c|c|c|c|c|c|c|c|}
			\hline &&&&&&&&& \\[-10pt]
			8	& \mathbb{R}[16]						& \mathbb{R}[16] \oplus \mathbb{R}[16]	& \mathbb{R}[32]						& \mathbb{C}[32]						& \mathbb{H}[32]						& \mathbb{H}[32] \oplus \mathbb{H}[32]	& \mathbb{H}[64]						& \mathbb{C}[128]						& \mathbb{R}[256] \\
			\hline &&&&&&&&& \\[-10pt]
			7	& \mathbb{R}[8] \oplus \mathbb{R}[8]	& \mathbb{R}[16]						& \mathbb{C}[16]						& \mathbb{H}[16]						& \mathbb{H}[16] \oplus \mathbb{H}[16]	& \mathbb{H}[32]						& \mathbb{C}[64]						& \mathbb{R}[128]						& \mathbb{R}[128] \oplus \mathbb{R}[128] \\
			\hline &&&&&&&&& \\[-10pt]
			6	& \mathbb{R}[8]							& \mathbb{C}[8]							& \mathbb{H}[8]							& \mathbb{H}[8] \oplus \mathbb{H}[8]	& \mathbb{H}[16]						& \mathbb{C}[32]						& \mathbb{R}[64]						& \mathbb{R}[64] \oplus \mathbb{R}[64]	& \mathbb{R}[128] \\
			\hline &&&&&&&&& \\[-10pt]
			5	& \mathbb{C}[4]							& \mathbb{H}[4]							& \mathbb{H}[4] \oplus \mathbb{H}[4]	& \mathbb{H}[8]							& \mathbb{C}[16]						& \mathbb{R}[32]						& \mathbb{R}[32] \oplus \mathbb{R}[32]	& \mathbb{R}[64]						& \mathbb{C}[64] \\
			\hline &&&&&&&&& \\[-10pt]
			4	& \mathbb{H}[2]							& \mathbb{H}[2] \oplus \mathbb{H}[2]	& \mathbb{H}[4]							& \mathbb{C}[8]							& \mathbb{R}[16]						& \mathbb{R}[16] \oplus \mathbb{R}[16]	& \mathbb{R}[32]						& \mathbb{C}[32]						& \mathbb{H}[32] \\
			\hline &&&&&&&&& \\[-10pt]
			3	& \mathbb{H} \oplus \mathbb{H}			& \mathbb{H}[2]							& \mathbb{C}[4]							& \mathbb{R}[8]							& \mathbb{R}[8] \oplus \mathbb{R}[8]	& \mathbb{R}[16]						& \mathbb{C}[16]						& \mathbb{H}[16]						& \mathbb{H}[16] \oplus \mathbb{H}[16] \\
			\hline &&&&&&&&& \\[-10pt]
			2	& \mathbb{H}							& \mathbb{C}[2]							& \mathbb{R}[4]							& \mathbb{R}[4] \oplus \mathbb{R}[4]	& \mathbb{R}[8]							& \mathbb{C}[8]							& \mathbb{H}[8]							& \mathbb{H}[8] \oplus \mathbb{H}[8]	& \mathbb{H}[16] \\
			\hline &&&&&&&&& \\[-10pt]
			1	& \mathbb{C}							& \mathbb{R}[2]							& \mathbb{R}[2] \oplus \mathbb{R}[2]	& \mathbb{R}[4]							& \mathbb{C}[4]							& \mathbb{H}[4]							& \mathbb{H}[4] \oplus \mathbb{H}[4]	& \mathbb{H}[8]							& \mathbb{C}[16] \\
			\hline &&&&&&&&& \\[-10pt]
			0	& \mathbb{R}							& \mathbb{R} \oplus \mathbb{R}			& \mathbb{R}[2]							& \mathbb{C}[2]							& \mathbb{H}[2]							& \mathbb{H}[2] \oplus \mathbb{H}[2]	& \mathbb{H}[4]							& \mathbb{C}[8]							& \mathbb{R}[16] \\
			\hline
			\hline
				& 0 & 1 & 2 & 3 & 4 & 5 & 6 & 7 & 8 \\
			\hline
		\end{array}
		\end{displaymath}
		\caption{The algebra $\mathcal{G}(\mathbb{R}^{s,t})$ in the box (s,t), where $\mathbb{F}[N] = \mathbb{F}^{N \times N}$. \label{table_classification_real}}
	\end{sidewaystable}
	
	\begin{proof}
		Start with the following easily verified isomorphisms 
		(see Examples \ref{exmp_complex_numbers}, \ref{exmp_quaternions}
		and Exercise \ref{exc_twodim_isomorphism}):
		\begin{displaymath}
		\setlength\arraycolsep{1pt}
		\begin{array}{l}
			\mathcal{G}(\mathbb{R}^{1,0}) \cong \mathbb{R} \oplus \mathbb{R}, \\[3pt]
			\mathcal{G}(\mathbb{R}^{0,1}) \cong \mathbb{C}, \\[3pt]
			\mathcal{G}(\mathbb{R}^{2,0}) \cong \mathbb{R}^{2 \times 2}, \\[3pt]
			\mathcal{G}(\mathbb{R}^{0,2}) \cong \mathbb{H}.
		\end{array}
		\end{displaymath}
		We can now work out the cases $(n,0)$ and $(0,n)$ for $n=0,1,\ldots,7$
		in a criss-cross fashion using Proposition \ref{prop_iso_tensor} and the
		tensor isomorphisms
		(see Exercises \ref{exc_iso1}-\ref{exc_iso3})
		\begin{displaymath}
		\setlength\arraycolsep{1pt}
		\begin{array}{l}
			\mathbb{C} \otimes_\mathbb{R} \mathbb{C} \cong \mathbb{C} \oplus \mathbb{C}, \\[3pt]
			\mathbb{C} \otimes_\mathbb{R} \mathbb{H} \cong \mathbb{C}^{2 \times 2}, \\[3pt]
			\mathbb{H} \otimes_\mathbb{R} \mathbb{H} \cong \mathbb{R}^{4 \times 4}.
		\end{array}
		\end{displaymath}
		With $\mathcal{G}(\mathbb{R}^{1,1}) \cong \mathcal{G}(\mathbb{R}^{2,0})$ 
		and Proposition \ref{prop_iso_tensor} we can then work our way
		through the whole table diagonally. The periodicity follows from Proposition \ref{prop_iso_periodicity}
		and 
		\begin{equation*}
			\mathcal{G}(\mathbb{R}^{8,0}) 
			\cong \mathbb{H} \otimes \mathbb{R}^{2 \times 2} \otimes \mathbb{H} \otimes \mathbb{R}^{2 \times 2} 
			\cong \mathbb{R}^{16 \times 16}. \qedhere
		\end{equation*}
	\end{proof}
	
	Because all nondegenerate complex quadratic forms on $\mathbb{C}^n$ are equivalent, 
	the complex version of the above theorem turns out to be much simpler.
	
	\begin{thm} \label{thm_classification_complex}
		We obtain the classification of complex geometric algebras as matrix algebras,
		given by
		\begin{displaymath}
		\setlength\arraycolsep{2pt}
		\begin{array}{rcl}
			\mathcal{G}(\mathbb{C}^{0}) &\cong& \mathbb{C}, \\[5pt]
			\mathcal{G}(\mathbb{C}^{1}) &\cong& \mathbb{C} \oplus \mathbb{C},
		\end{array}
		\end{displaymath}
		together with the periodicity
		\begin{displaymath}
			\mathcal{G}(\mathbb{C}^{n+2}) \cong \mathcal{G}(\mathbb{C}^{n}) \otimes_\mathbb{C} \mathbb{C}^{2 \times 2}.
		\end{displaymath}
	\end{thm}

	\noindent
	In other words,
		\begin{displaymath}
		\setlength\arraycolsep{2pt}
		\begin{array}{rcl}
			\mathcal{G}(\mathbb{C}^{2k}) &\cong& \mathbb{C}^{2^k \times 2^k}, \\[5pt]
			\mathcal{G}(\mathbb{C}^{2k+1}) &\cong& \mathbb{C}^{2^k \times 2^k} \oplus \mathbb{C}^{2^k \times 2^k},
		\end{array}
		\end{displaymath}
	for $k=0,1,2,\ldots$

	\begin{proof}
		The isomorphism $\mathcal{G}(\mathbb{C}^n) \cong \mathcal{G}(\mathbb{R}^n) \otimes \mathbb{C}$ gives
		\begin{displaymath}
		\setlength\arraycolsep{1pt}
		\begin{array}{l}
			\mathcal{G}(\mathbb{C}^0) \cong \mathbb{C} \\[3pt]
			\mathcal{G}(\mathbb{C}^1) \cong (\mathbb{R} \oplus \mathbb{R}) \otimes \mathbb{C} \cong \mathbb{C} \oplus \mathbb{C} \\[3pt]
			\mathcal{G}(\mathbb{C}^2) \cong \mathbb{R}^{2 \times 2} \otimes \mathbb{C} \cong \mathbb{C}^{2 \times 2}.
		\end{array}
		\end{displaymath}
		Then use Proposition \ref{prop_iso_periodicity} for periodicity.
	\end{proof}

	\begin{exmp} \label{exmp_pauli_algebra}
		A concrete representation of the space algebra $\mathcal{G}(\mathbb{R}^3)$
		as a matrix algebra is obtained by considering the so-called
		\emph{Pauli matrices}
		\begin{equation} \label{pauli_matrices}
			\sigma_1 = \left[ \begin{array}{cc}
			0 & 1  \\[3pt] 1 & 0
			\end{array} \right], \quad
			\sigma_2 = \left[ \begin{array}{cc}
			0 & -i  \\[3pt] i & 0
			\end{array} \right], \quad
			\sigma_3 = \left[ \begin{array}{cc}
			1 & 0  \\[3pt] 0 & -1
			\end{array} \right].
		\end{equation}
		These satisfy 
		$\sigma_1^2 = \sigma_2^2 = \sigma_3^2 = 1$,
		$\sigma_j\sigma_k = -\sigma_k\sigma_j$, $j \neq k$,
		and $\sigma_1\sigma_2\sigma_3 = i$,
		and hence identifies the space algebra with the
		\emph{Pauli algebra} $\mathbb{C}^{2 \times 2}$ through the isomorphism
		$\rho: \mathcal{G}(\mathbb{R}^3) \to \mathbb{C}^{2 \times 2}$,
		defined by $\rho(e_k) := \sigma_k$.
		An arbitrary element 
		$x = \alpha + \boldsymbol{a} + \boldsymbol{b}I + \beta I \in \mathcal{G}$ 
		is then represented as the matrix
		\begin{equation} \label{pauli_element}
			\rho(x) = \left[ \begin{array}{cc}
			\alpha + a_3 + (\beta + b_3)i	& a_1 + b_2 + (b_1 - a_2)i  \\[3pt]
			a_1 - b_2 + (b_1 + a_2)i		& \alpha - a_3 + (\beta - b_3)i
			\end{array} \right].
		\end{equation}
	\end{exmp}
	
	The periodicity of geometric algebras actually has a number of
	far-reaching consequences. One example is \emph{Bott periodicity}, which simply
	put gives a periodicity in the homotopy groups $\pi_k$ of the unitary, orthogonal and symplectic groups. 
	An example is the following

	\begin{thm} \label{thm_bott_so}
		For $n \ge k+2$ and all $k \ge 1$ we have
		$$
			\pi_k(\Ogrp(n)) \cong \pi_k(\SO(n)) \cong \left\{
				\begin{array}{ll}
					\{0\} 			& \textrm{if}\ k \equiv 2,4,5,6 \pmod{8} \\
					\mathbb{Z}_2 	& \textrm{if}\ k \equiv 0,1 \pmod{8} \\
					\mathbb{Z} 		& \textrm{if}\ k \equiv 3,7 \pmod{8} \\
				\end{array}
			\right.
		$$
	\end{thm}

	\noindent
	See \cite{lawson_michelsohn} for a proof using K-theory,
	or \cite{nakahara} and references therein for more examples.
	
	\begin{exc}
		Complete the proof of the second statement of Proposition \ref{prop_iso_even}.
	\end{exc}

	\begin{exc}
		Prove the second expression of Proposition \ref{prop_iso_tensor}.
	\end{exc}

	\begin{exc}
		Complete the proof of the second isomorphism in Proposition \ref{prop_iso_real_complex}.
	\end{exc}

	\begin{exc}
		Complete the proof of Proposition \ref{prop_iso_periodicity}.
	\end{exc}

	\begin{exc} \label{exc_iso1}
		Prove that $\mathbb{C} \oplus \mathbb{C} \cong \mathbb{C} \otimes_\mathbb{R} \mathbb{C}$.\\
		\emph{Hint:} Consider the map
		\begin{eqnarray*}
			(1,0) &\mapsto& \frac{1}{2}(1 \otimes 1 + i \otimes i), \\
			(0,1) &\mapsto& \frac{1}{2}(1 \otimes 1 - i \otimes i).
		\end{eqnarray*}
	\end{exc}

	\begin{exc} \label{exc_iso2}
		Prove that $\mathbb{C} \otimes_\mathbb{R} \mathbb{H} \cong \mathbb{C}^{2 \times 2}$.\\
		\emph{Hint:} 
		Consider $\mathbb{H}$ as a $\mathbb{C}$-module under left scalar
		multiplication, and define an $\mathbb{R}$-bilinear map
		$\Phi\!: \mathbb{C} \times \mathbb{H} \to \textup{Hom}_\mathbb{C}(\mathbb{H},\mathbb{H})$
		by setting $\Phi_{z,q}(x) := z x \overline{q}$.
		This extends (by the universal property of the tensor product; 
		see Theorem \ref{thm_tensor_universality})
		to an $\mathbb{R}$-linear map 
		$\Phi\!: \mathbb{C} \otimes_{\mathbb{R}} \mathbb{H} \to \textup{Hom}_\mathbb{C}(\mathbb{H},\mathbb{H}) \cong \mathbb{C}^{2 \times 2}$.
	\end{exc}

	\begin{exc} \label{exc_iso3}
		Prove that $\mathbb{H} \otimes_\mathbb{R} \mathbb{H} \cong \mathbb{R}^{4 \times 4}$.\\
		\emph{Hint:}
		Consider the $\mathbb{R}$-bilinear map 
		$\Psi\!: \mathbb{H} \times \mathbb{H} \to \textup{Hom}_\mathbb{R}(\mathbb{H},\mathbb{H})$
		given by setting $\Psi_{q_1,q_2}(x) := q_1 x \overline{q}_2$.
	\end{exc}

\subsection{Graded tensor products and the mother algebra}

	Let us also consider a different characterization of geometric algebras.
	
	Let $\mathcal{G}(V_1,q_1)$ and $\mathcal{G}(V_2,q_2)$ be two
	geometric algebras, and form their tensor product
	(considered not as an algebra, but as an $\mathbb{F}$-module)
	$$
		T = \mathcal{G}(V_1,q_1) \otimes \mathcal{G}(V_2,q_2).
	$$
	We shall now introduce a new multiplication on $T$ through
	$$
		\begin{array}{ccl}
		T \times T &\to& T \\
		(x \otimes y\ ,\ x' \otimes y') &\mapsto& (-1)^{\delta(y)\delta(x')} (xx') \otimes (yy')
		\end{array}
	$$
	for elements such that
	\begin{eqnarray*}
		x,x' &\in& \mathcal{G}^+(V_1,q_1) \cup \mathcal{G}^-(V_1,q_1), \\
		y,y' &\in& \mathcal{G}^+(V_2,q_2) \cup \mathcal{G}^-(V_1,q_1),
	\end{eqnarray*}
	and where $\delta(z) = 0$ if $z \in \mathcal{G}^+$ 
	and $\delta(z) = 1$ if $z \in \mathcal{G}^-$.
	This multiplication map is then extended to 
	the other elements through bilinearity.

	We immediately note that
	\begin{eqnarray*}
		(v_1 \otimes 1 + 1 \otimes v_2)^2 
		&=& (v_1 \otimes 1)(v_1 \otimes 1) + (v_1 \otimes 1)(1 \otimes v_2) \\
		&& +\ (1 \otimes v_2)(v_1 \otimes 1) + (1 \otimes v_2)(1 \otimes v_2) \\
		&=& v_1^2 \otimes 1 + v_1 \otimes v_2 - v_1 \otimes v_2 + 1 \otimes v_2^2 \\
		&=& (v_1^2 + v_2^2) 1 \otimes 1,
	\end{eqnarray*}
	so that, if we as usual identify $1 \otimes 1$ with $1 \in \mathbb{F}$,
	we obtain
	$$
		(v_1 \otimes 1 + 1 \otimes v_2)^2 = v_1^2 + v_2^2 = q(v_1) + q(v_2).
	$$
	Hence, if we introduce the vector space
	$$
		V = \{ v_1 \otimes 1 + 1 \otimes v_2 : v_1 \in V_1, v_2 \in V_2 \} \cong V_1 \oplus V_2,
	$$
	and the quadratic form $q: V \to \mathbb{F}$,
	$$
		q(v_1 \otimes 1 + 1 \otimes v_2) := q_1(v_1) + q_2(v_2),
	$$
	then we find that $\mathcal{G}(V,q)$ becomes a Clifford
	algebra which is isomorphic to $T = \mathcal{G}(V_1,q_1) \hat{\otimes} \mathcal{G}(V_2,q_2)$,
	where the symbol $\hat{\otimes}$ signals that we have
	defined a special so-called $\mathbb{Z}_2$-graded multiplication
	on the usual tensor product space.
	
	As a consequence, we have the following
	
	\begin{prop} \label{prop_iso_graded_tensor}
		For all $s,t,p,q \geq 0$, there is a graded tensor algebra isomorphism
		\begin{displaymath}
			\mathcal{G}(\mathbb{R}^{s+p,t+q}) \cong \mathcal{G}(\mathbb{R}^{s,t})\ \hat{\otimes}\ \mathcal{G}(\mathbb{R}^{p,q}).
		\end{displaymath}
	\end{prop}
	
	\begin{cor}
		It follows immediately that
		\begin{displaymath}
			\mathcal{G}(\mathbb{R}^{s,t}) \cong 
				\underbrace{\mathcal{G}(\mathbb{R}^{1,0})\ \hat{\otimes}\ \ldots\ \hat{\otimes}\ \mathcal{G}(\mathbb{R}^{1,0})}_{s\ \textrm{factors}}\ \hat{\otimes}\ 
				\underbrace{\mathcal{G}(\mathbb{R}^{0,1})\ \hat{\otimes}\ \ldots\ \hat{\otimes}\ \mathcal{G}(\mathbb{R}^{0,1})}_{t\ \textrm{factors}}.
		\end{displaymath}
	\end{cor}

	When working with a geometric algebra over a mixed-signature space,
	where the quadratic form is neither positive nor negative definite,
	it can often be a good idea to embed the algebra in a larger one.

	Consider $\mathcal{G}(\mathbb{R}^{s,t,u})$ generated by an orthonormal basis 
	$$
		\{ e_1^+,\ldots,e_s^+, e_1^-,\ldots,e_t^-, e_1^0,\ldots,e_u^0 \}
	$$
	with $(e_i^+)^2 = 1$, $(e_j^-)^2 = -1$, and $(e_k^0)^2 = 0$.
	Introduce the \emph{mother algebra} 
	$\mathcal{G}(\mathbb{R}^{n,n})$, with $n = s+t+u$,
	and an orthonormal basis
	$$
		\{ f_1^+,\ldots,f_n^+, f_1^-,\ldots,f_n^- \}
	$$
	with $(f_i^+)^2 = 1$ and $(f_j^-)^2 = -1$.
	
	We now define 
	$\Phi: \mathcal{G}(\mathbb{R}^{s,t,u}) \to \mathcal{G}(\mathbb{R}^{n,n})$
	on $\mathbb{R}^{s,t,u}$ by setting
	$$
		\Phi(e_i^+) := f_i^+, \quad \Phi(e_j^-) := f_j^-, \quad \Phi(e_k^0) := f_{s+t+k}^+ - f_{s+t+k}^-, 
	$$
	and extending linearly. Note that
	$$
		\Phi(e_i^+)^2 = (e_i^+)^2, \quad \Phi(e_j^-)^2 = (e_j^-)^2, \quad \Phi(e_k^0)^2 = (e_k^0)^2.
	$$
	By universality, $\Phi$ extends to all of $\mathcal{G}(\mathbb{R}^{s,t,u})$
	and since $\Phi$ is injective $\mathbb{R}^{s,t,u} \to \mathbb{R}^{n,n}$,
	it follows that $\Phi$ is an injective homomorphism of geometric algebras.

	Similarly, it also follows that every finite-dimensional geometric
	algebra $\mathcal{G}(\mathbb{R}^{s,t,u})$ is embedded in the
	\emph{infinite-dimensional mother algebra} $\cl_{(\mathcal{F})}(\mathbb{Z},\mathbb{R},r)$,
	where $r(k) := (k \ge 0) - (k < 0)$, $k \in \mathbb{Z}$.

	\begin{exc}
		Verify that the graded multiplication on 
		$T$ introduced above is associative.
	\end{exc}

%% file: clifford_groups.tex
	One of the reasons that geometric algebras appear naturally
	in many areas of mathematics and physics is the fact that they contain
	a number of important groups. These are groups under the geometric product
	and thus lie embedded within the multiplicative group of invertible elements in $\mathcal{G}$.
	In this section we will discuss the properties of various embedded groups and their relation
	to other familiar transformation groups such as the orthogonal and Lorentz groups.
	The introduced notion of a rotor will be seen to be essential for the
	description of the orthogonal groups, and will later play a fundamental
	role for understanding the concept of spinors.

	Throughout this section we will always assume that our scalars are real numbers
	unless otherwise stated. This is reasonable both from a geometric
	viewpoint and from the fact that e.g. many common complex groups can be represented
	by groups embedded in real geometric algebras.
	Furthermore, we assume that $\mathcal{G}$ is nondegenerate so that we are
	working with a vector space of type $\mathbb{R}^{s,t}$. The corresponding
	groups associated to this space will be denoted $\SO(s,t)$ etc.

\subsection{Groups in $\mathcal{G}$ and their actions on $\mathcal{G}$}

	An obvious group contained in $\mathcal{G}$ is of course
	the group of all invertible elements of $\mathcal{G}$, 
	which we denote by $\mathcal{G}^\times$.
	All other groups we will discuss in this section are actually
	subgroups of $\mathcal{G}^\times$, or directly related to such subgroups.
	Let us introduce the subgroups we will be discussing 
	in the following definition.

	\begin{defn} \label{def_groups}
		We identify the following groups embedded in $\mathcal{G}$:
		\begin{displaymath}
		\setlength\arraycolsep{1.5pt}
		\begin{array}{lcll}
			\mathcal{G}^\times &:=& \{ x \in \mathcal{G} : \exists y \in \mathcal{G} : xy = yx = 1 \} \ \ 
				& \textrm{\emph{the group of all invertible elements}} \\[5pt]
			\tilde{\Gamma} &:=& \{ x \in \mathcal{G}^\times : x^\star V x^{-1} \subseteq V \}
				& \textrm{\emph{the Lipschitz group}} \\[5pt]
			\Gamma &:=& \{ v_1 v_2 \ldots v_k \in \mathcal{G} : v_i \in V^\times \}
				& \textrm{\emph{the versor group}} \\[5pt]
			\textrm{Pin} &:=& \{ x \in \Gamma : x x^\dagger = \pm 1 \}
				& \textrm{\emph{the group of unit versors}} \\[5pt]
			\textrm{Spin} &:=& \textrm{Pin} \cap \mathcal{G}^+
				& \textrm{\emph{the group of even unit versors}} \\[5pt]
			\textrm{Spin}^+ &:=& \{ x \in \textrm{Spin} : x x^\dagger = 1 \}
				& \textrm{\emph{the rotor group}}
		\end{array}
		\end{displaymath}
		where $V^\times := \{ v \in V : v^2 \neq 0 \}$ is the set of invertible vectors.
	\end{defn}
	
	\noindent
	One of the central, and highly non-trivial, results of this section 
	is that the versor group $\Gamma$ and Lipschitz group $\tilde{\Gamma}$ actually are equal.
	Therefore, $\Gamma$ is also called the Lipschitz group in honor of its creator.
	Sometimes it is also given the name \emph{Clifford group}, but we will, in accordance
	with other conventions, use that name to denote the finite group generated
	by an orthonormal basis under Clifford multiplication.
	
	The versor group $\Gamma$ is the smallest group which contains $V^\times$. 
	Its elements are finite products of invertible vectors and are called \emph{versors}.
	As seen from Definition \ref{def_groups}, 
	the Pin, Spin, and rotor groups are all subgroups of this group.
	These subgroups are generated by unit vectors, and in the case
	of Spin and Spin$^+$, only an \emph{even} number of such vector factors can be present.
	The elements of Spin$^+$ are called \emph{rotors} and, as we will see, these
	groups are intimately connected to orthogonal groups and rotations.

	\begin{exmp} \label{exmp_plane_rotation}
		Consider a rotation by an angle $\varphi$ in a plane. The relevant algebra is 
		the plane algebra $\mathcal{G}(\mathbb{R}^2)$, and we have seen that
		the rotation can be
		represented as the map $\mathbb{R}^2 \to \mathbb{R}^2$,
		\begin{equation} \label{plane_rotation_revisited}
			\boldsymbol{v} \mapsto \boldsymbol{v} e^{\varphi I} = e^{-\frac{\varphi}{2}I} \boldsymbol{v} e^{\frac{\varphi}{2} I}.
		\end{equation}
		The invertible element 
		$R_\varphi := e^{\frac{\varphi}{2}I} = \cos \frac{\varphi}{2} + \sin \frac{\varphi}{2} I$ 
		is clearly in the even subalgebra, so $R_\varphi^* = R_\varphi$ 
		and we also see from \eqref{plane_rotation_revisited} that
		$R_\varphi \mathbb{R}^2 R_\varphi^{-1} \subseteq \mathbb{R}^2$, hence
		$R_\varphi \in \tilde{\Gamma}$. 
		The set of elements of this form,
		i.e. the even unit multivectors
		$\{ e^{\varphi I} : \varphi \in \mathbb{R} \} \cong \textup{U}(1)$,
		is obviously the group of unit complex numbers.
		Furthermore, we can always write e.g.
		$$
			R_\varphi = e_1\left( \cos \frac{\varphi}{2}\thinspace e_1 + \sin \frac{\varphi}{2}\thinspace e_2 \right),
		$$
		i.e. as a product of two unit vectors, and
		we find that this is the group of rotors in two
		dimensions; $\Spin(2) = \Spin^+(2) \cong \textup{U}(1)$.
		However, note that because of the factor $\frac{1}{2}$ in $R_\varphi$,
		the map from the rotors to the actual
		rotations is a two-to-one map $\textup{U}(1) \to \SO(2) \cong \textup{U}(1)$.
	\end{exmp}
	
	In order to understand how groups embedded in a geometric algebra are related to
	more familiar groups of linear transformations, it is necessary 
	to study how groups in $\mathcal{G}$ can act on the 
	vector space $\mathcal{G}$ itself and on the embedded underlying vector space $V$.
	The following are natural candidates for such actions.
	
	\begin{defn} \label{def_actions}
		Derived from the geometric product, we have the following canonical actions:
		\begin{displaymath}
		\setlength\arraycolsep{2pt}
		\begin{array}{llcll}
			L\!: & \mathcal{G} &\to& \textrm{End}\ \mathcal{G} 
				& \quad\textrm{\emph{left action}} \\[2pt]
				& x &\mapsto& L_x\!: y \mapsto xy \\[5pt]
			R\!: & \mathcal{G} &\to& \textrm{End}\ \mathcal{G} 
				& \quad\textrm{\emph{right action}} \\[2pt]
				& x &\mapsto& R_x\!: y \mapsto yx \\[5pt]
			\Ad\!: & \mathcal{G}^\times &\to& \textrm{End}\ \mathcal{G} 
				& \quad\textrm{\emph{adjoint action}} \\[2pt]
				& x &\mapsto& \Ad_x\!: y \mapsto xyx^{-1} \\[5pt]
			\tAd\!: & \mathcal{G}^\times &\to& \textrm{End}\ \mathcal{G} 
				& \quad\textrm{\emph{twisted adjoint action}} \\[2pt]
				& x &\mapsto& \tAd_x\!: y \mapsto x^\star yx^{-1}
		\end{array}
		\end{displaymath}
		where $\textrm{End}\ \mathcal{G}$ are the (vector space) 
		endomorphisms of $\mathcal{G}$.
	\end{defn}
	
	\noindent
	Note that the left and right actions $L$ and $R$ are algebra 
	homomorphisms while $\Ad$ and $\tAd$ are group homomorphisms (Exercise \ref{exc_actions}).
	These actions give rise to canonical representations of the groups embedded in $\mathcal{G}$.
	Although, using the expansion \eqref{blade_expansion} one can verify that $\textrm{Ad}_x$ is always
	an outermorphism while in general $\tAd_x$ is not,
	the twisted adjoint action takes the graded structure of $\mathcal{G}$ into account and will be
	seen to play a more important role than the normal adjoint action in geometric algebra.
	Note, however, that these actions agree on the subgroup of even invertible elements 
	$\mathcal{G}^\times \cap \mathcal{G}^+$.

	\begin{exmp}	
		As an application, let us recall that, 
		because the algebra $\mathcal{G}$ is assumed to
		be finite-dimen\-sion\-al, left inverses are always right inverses and \emph{vice versa}.
		This can be seen as follows. First note that the left and right actions are injective.
		Namely, assume that $L_x = 0$. Then $L_x(y) = 0\ \forall y$ and in particular $L_x(1) = x = 0$.
		Suppose now that $xy = 1$ for some $x,y \in \mathcal{G}$.
		But then $L_x L_y = \id$, so that $L_y$ is a right inverse to $L_x$. 
		Now, using the dimension theorem 
		\begin{displaymath}
			\dim \ker L_y + \dim \im L_y = \dim \mathcal{G}
		\end{displaymath}
		with $\ker L_y = 0$, we can conclude that $L_y$ is also a left inverse to $L_x$.
		Hence, $L_y L_x = \id$, so that $L_{yx-1} = 0$, and $yx = 1$.
	\end{exmp}

	\begin{exc} \label{exc_groups}
		Verify that the groups defined in Definition \ref{def_groups} really
		are groups under the geometric product.
	\end{exc}

	\begin{exc} \label{exc_actions}
		Verify that the actions $L$ and $R$ are $\mathbb{R}$-algebra homomorphisms, while
		$\Ad$ and $\tAd : \mathcal{G}^\times \to \textup{GL}(\mathcal{G})$ are group homomorphisms.
		Also, show that $\Ad_x$ is an outermorphism for fixed $x \in \mathcal{G}^\times$,
		while 
		$$
			\left( \tAd_x\big|_V \right)_\wedge(y) = \Ad_x(y^\star)
		$$
		(i.e. $\tAd_x\big|_\mathcal{G}$ is \emph{not} an outermorphism)
		for $x \in \mathcal{G}^\times \cap \mathcal{G}^-$ and $y \in \mathcal{G}$.
	\end{exc}

\subsection{The Cartan-Dieudonn\'e Theorem}
	
	Let us begin with studying the properties of the twisted adjoint action.
	For $v \in V^\times$ we obtain
	\begin{equation}
		\tAd_v (v) = v^\star v v^{-1} = -v,
	\end{equation}
	and if $w \in V$ is orthogonal to $v$,
	\begin{equation}
		\tAd_v (w) = v^\star w v^{-1} = -vwv^{-1} = w vv^{-1} = w.
	\end{equation}
	Hence, $\tAd_v$ acts on $V$ as a reflection along $v$,
	in other words in the hyperplane orthogonal to $v$.
	Such a transformation is a linear isometry, i.e. an element of
	the group of \emph{orthogonal transformations} of $(V,q)$,
	\begin{displaymath}
		\textrm{O}(V,q) := \{f\!: V \to V : f\ \textrm{linear bijection s.t.}\ q \circ f = q \}.
	\end{displaymath}
	Obviously, 
	$q(\tAd_v(u)) = \tAd_v(u)^2 = vuv^{-1}vuv^{-1} = u^2vv^{-1} = q(u)$ 
	for all $u \in V$.

	For a general versor $x = u_1 u_2 \ldots u_k \in \Gamma$ we have
	\begin{equation} \label{versor_action}
	\setlength\arraycolsep{2pt}
	\begin{array}{rcl}
		\tAd_x (v) 
			&=& (u_1 \ldots u_k)^\star v (u_1 \ldots u_k)^{-1} 
			= u_1^\star \ldots u_k^\star v u_k^{-1} \ldots u_1^{-1} \\[5pt]
			&=& \tAd_{u_1} \circ \ldots \circ \tAd_{u_k} (v),
	\end{array}
	\end{equation}
	i.e. $\tAd_x$ is a product of reflections in hyperplanes.
	It follows that the twisted adjoint action (restricted to act only on $V$
	which is clearly invariant) gives a homomorphism 
	from the versor group into the orthogonal group $\Ogrp(V,q)$.

	In Example \ref{exmp_plane_rotation} above we saw that any 
	rotation in two dimensions can be written in terms of a rotor $R_\varphi$
	which is a product of two unit vectors, i.e. as a product of two reflections.
	In general, we have the following fundamental theorem regarding the orthogonal group.
	
	\begin{thm}[Cartan-Dieudonn\'e] \label{thm_cartan_dieudonne}
		Every orthogonal transformation on a non-degen\-erate space 
		$(V,q)$ is a product of reflections in hyperplanes.
		The number of reflections required is at most equal to the dimension of $V$.
	\end{thm}
	
	\noindent We shall give a constructive proof below, 
	where the bound on the number of reflections is $n = \dim V$ for definite signatures,
	and $2n$ for mixed signature spaces.
	For the optimal case we refer to e.g. \cite{artin}, or \cite{grove}.
	
	\begin{cor}
		The homomorphism
		$\widetilde{\textrm{\emph{Ad}}}\!: \Gamma \to \textrm{\emph{O}}(V,q)$ 
		is surjective.
	\end{cor}
	\begin{proof}
		We know that any $f \in \textrm{O}(V,q)$ can be written
		$f = \tAd_{v_1} \circ \ldots \circ \tAd_{v_k}$
		for some invertible vectors $v_1,\ldots,v_k$, $k \leq 2n$.
		But then $f = \tAd_{v_1 \ldots v_k}$, where $v_1 v_2 \ldots v_k \in \Gamma$.
	\end{proof}

	We divide the proof of the Cartan-Dieudonn\'e Theorem into several steps,
	and we will see that the proof holds for general fields $\mathbb{F}$
	with the usual assumption that $\charop \mathbb{F} \neq 2$.
	We will also find it convenient to introduce the following
	notation for the twisted adjoint action:
	$$
		\underline{U} := \tAd_U, \quad \textrm{i.e.} \quad \underline{U}(x) = U^\star x U^{-1}.
	$$
	
	\begin{lem} \label{lem_cartan_invertibility}
		If $x,y \in V^{\times}$ and $x^2 = y^2$ then either
		$x + y$ or $x - y$ is in $V^{\times}$.
	\end{lem}
	\begin{proof}
		Assume to the contrary that
		\begin{eqnarray*}
			0 &=& (x+y)^2 = x^2 + y^2 + xy + yx  \qquad \textrm{and} \\
			0 &=& (x-y)^2 = x^2 + y^2 - xy - yx.
		\end{eqnarray*}
		Thus $2(x^2+y^2) = 0$, but $2 \neq 0$ in $\mathbb{F}$ then
		implies $2x^2 = x^2+y^2 = 0$, which is a contradiction.
	\end{proof}
	
	\noindent
		If the quadratic form $q$ is either positive or negative
		definite then, under the same conditions, 
		either $x=y$, $x=-y$, or \emph{both} $x \pm y \in V^{\times}$ 
		(see Exercise \ref{exc_degenerate_sum}).
	
	\begin{lem} \label{lem_cartan_action}
		With the conditions in Lemma \ref{lem_cartan_invertibility} it holds that
		\begin{eqnarray*}
			x + y \in V^{\times} &\Rightarrow& \underline{(x+y)}(x) = -y \qquad \textrm{and} \\
			x - y \in V^{\times} &\Rightarrow& \underline{(x-y)}(x) =  y.
		\end{eqnarray*}
	\end{lem}
	\begin{proof}
		If $x+y \in V^\times$ we have
		\begin{eqnarray*}
			&& (x+y)^\star x(x+y)^{-1} = -(x^2+yx)(x+y)^{-1} = -(y^2+yx)(x+y)^{-1} \\
			&& \qquad = -y(y+x)(x+y)^{-1} = -y.
		\end{eqnarray*}
		The other implication is obtained by replacing $y$ with $-y$.
	\end{proof}
	
	\noindent
		Note that
		when $x^2 = y^2 \neq 0$ and $x+y \in V^\times$ then we also see that
		$\underline{y(x+y)}(x) = \underline{y}(-y) = y$.
	
	\begin{lem} \label{lem_cartan_sequences}
		Let $e_1,\ldots,e_n$ and $f_1,\ldots,f_n$ be two orthogonal
		sequences in $V$ such that $e_i^2 = f_i^2$, $i=1,\ldots,n$. 
		Then there exist $v_1,\ldots,v_r \in V^{\times}$, where $r \le n$,
		such that
		$$
			\underline{v_1 \ldots v_r}(e_j) = \tau_j f_j, \quad j=1,\ldots,n,
		$$
		for suitable $\tau_j \in \{-1,1\}$.
	\end{lem}
	\begin{proof}
		We shall construct $u_1,\ldots,u_n$ such that
		$u_i = 1$ or $u_i \in V^\times$ and such that for
		each $m \le n$ the conditions
		$$
		\begin{array}{rrcll}
			i)  & \underline{u_m \ldots u_1}(e_j) &=& \tau_j f_j, & \quad 1 \le j \le m \\[5pt]
			ii) & \underline{u_m}(f_k) &=& f_k, & \quad 1 \le k < m
		\end{array}
		$$
		are satisfied.
		
		We start with $m=1$:
		If $e_1 = f_1$ we choose $u_1 = 1$. 
		If $e_1 \neq f_1$ and $e_1 - f_1 \in V^\times$ 
		we choose $u_1 := e_1 - f_1$ and obtain $\underline{u_1}(e_1) = f_1$.
		If $e_1 \neq f_1$ and $e_1 - f_1 \notin V^\times$ 
		then we must have $e_1 + f_1 \in V^\times$,
		so we choose $u_1 := e_1 + f_1$ and obtain $\underline{u_1}(e_1) = -f_1$.
		
		Now, assume we have constructed $u_1,\ldots,u_m$ 
		such that $m \le n-1$ and such that 
		the conditions $(i)$ and $(ii)$ are satisfied.
		We would like to construct $u_{m+1}$:
		
		If $\underline{u_m \ldots u_1}(e_{m+1}) = f_{m+1}$ 
		then we choose $u_{m+1} := 1$ and hence also satisfy 
		the corresponding condition $(ii)$.
		
		If $\underline{u_m \ldots u_1}(e_{m+1}) - f_{m+1} \in V^\times$,
		we choose $u_{m+1} := \underline{u_m \ldots u_1}(e_{m+1}) - f_{m+1}$ 
		and obtain $\underline{u_{m+1} \ldots u_1}(e_{m+1}) = f_{m+1}$,
		which proves $(i)$ for the case $j=m+1$.
		When $k < m+1$ we have
		\begin{eqnarray*}
			f_k * u_{m+1} &=& f_k * \left( \underline{u_m \ldots u_1}(e_{m+1}) - f_{m+1} \right) \\
			&=& \left( \tau_k \underline{u_m \ldots u_1}(e_k) \right) * \underline{u_m \ldots u_1}(e_{m+1}) = \tau_k e_k * e_{m+1} = 0.
		\end{eqnarray*}
		But then we must have $\underline{u_{m+1}}(f_k) = f_k$ which gives $(ii)$.
		Now let $j \le m$. Then, by $(i)$ and $(ii)$, we see that
		$$
			\underline{u_{m+1} \ldots u_1}(e_j) 
			= \underline{u_{m+1}}\left( \underline{u_m \ldots u_1}(e_j) \right)
			= \underline{u_{m+1}}( \tau_j f_j ) = \tau_j f_j,
		$$
		which proves $(i)$ also for this case.
		
		It remains to consider the case when
		$\underline{u_m \ldots u_1}(e_{m+1}) + f_{m+1} \in V^\times$
		but none of the earlier cases hold.
		This case is investigated in a similar way, or is immediately
		realized by replacing $e_{m+1}$ with $-e_{m+1}$ (exercise).
		
		Proceeding by induction, this proves the lemma.
	\end{proof}
	
	\noindent
	Note that	
	when we have constructed $u_1,\ldots,u_{n-1}$
	it then also follows that $\underline{u_{n-1} \ldots u_1}(e_n) = \tau_n f_n$
	since $\underline{u_{n-1} \ldots u_1} \in \Ogrp(V)$ 
	and $e_1,\ldots,e_n$ resp. $f_1,\ldots,f_n$ are orthogonal bases.

	We are now ready to prove Theorem \ref{thm_cartan_dieudonne}.
	Let $f \in \Ogrp(V,q)$ and put $f_j := f(e_j)$ for $j=1,\ldots,n$,
	where $\{e_1,\ldots,e_n\}$ is some orthogonal basis for $V$.
	Then $e_i^2 = f_i^2\ \forall i$ 
	and $\{f_1,\ldots,f_n\}$ is also an orthogonal basis.
	Let $u_1,\ldots,u_n$ be as constructed in Lemma \ref{lem_cartan_sequences}
	(i.e. $u_i = 1$ or $u_i \in V^\times$)
	and let $u_{n+k} := (\tau_k = 1) + (\tau_k = -1)f_k$, $1 \le k \le n$.
	Then we find $\underline{u_{2n} \ldots u_{n+1}}(f_k) = \tau_k f_k$ for all $k=1,\ldots,n$.
	This shows that $\underline{u_{2n} \ldots u_1} = f$, which proves the theorem.
	
	\begin{exc}
		Show that for $f \in \Ogrp(V,q)$ (as defined above) 
		$f_\wedge(x)*f_\wedge(y) = x*y\ \forall x,y \in \mathcal{G}$,
		$f^* = f^{-1}$, and $\det f \in \{1,-1\}$.
	\end{exc}
	
	\begin{exc}
		Show that $\tAd_v$ is a reflection along $v$,
		in an alternative way, by using
		the formula for a projection in Section \ref{subsec_proj_rej}.
		Also, verify that $\det \tAd_v = -1$.
	\end{exc}
	
	\begin{exc} \label{exc_degenerate_sum}
		Show that $(x+y)^2(x-y)^2 = -4(x \wedge y)^2$ for all $x,y \in V$ s.t. $x^2 = y^2$, and
		construct an example of two linearly independent vectors $x,y$
		in a nondegenerate space
		such that $x^2 = y^2 \neq 0$ but $x+y \notin V^{\times}$. \\
	\end{exc}

	\begin{exc}
		Find $v_1,v_2,v_3 \in (\mathbb{R}^3)^\times$ such that
		$$
			\tAd_{v_3v_2v_1} = f
		$$
		where $f \in \Ogrp(3)$ is defined by 
		$f(e_1) = e_3$, $f(e_2) = \frac{4}{5}e_1 + \frac{3}{5}e_2$
		and $f(e_3) = \frac{3}{5}e_1 - \frac{4}{5}e_2$.
	\end{exc}

\subsection{The Lipschitz group}
\label{sec_lipschitz}
	
	We saw above that the twisted adjoint action maps the versor
	group onto the group of orthogonal transformations of $V$.
	The largest group in $\mathcal{G}$ for which $\tAd$
	forms a representation on $V$, i.e. leaves $V$ invariant,
	is by Definition \ref{def_groups} the Lipschitz group $\tilde{\Gamma}$.
	Hence, (or explicitly from \eqref{versor_action}), 
	we see that $\Gamma \subseteq \tilde{\Gamma}$.
	
	We will now introduce an important function on $\mathcal{G}$, conventionally
	called the \emph{norm function},
	\begin{equation} \label{norm_function}
	\setlength\arraycolsep{2pt}
	\begin{array}{c}
		N\!: \mathcal{G} \to \mathcal{G}, \\[5pt]
		N(x) := x^\cliffconj x.
	\end{array}
	\end{equation}
	The name is a bit misleading since $N$ is not even
	guaranteed to take values in $\mathbb{R}$. 
	However, for some special cases of algebras
	it does act as a natural norm (squared) and we will see that it can be extended
	in many lower-dimensional algebras where it will act as a kind of determinant.
	Our first main result for this function is that it acts 
	similarly to a determinant on $\tilde{\Gamma}$.
	This will help us prove that $\Gamma = \tilde{\Gamma}$.
	
	\begin{lem} \label{lem_grade_commute}
		Assume that $\mathcal{G}$ is nondegenerate.
		If $x \in \mathcal{G}$ and $x^\star v = vx$ for all $v \in V$ then
		$x$ must be a scalar.
	\end{lem}
	\begin{proof}
		From Proposition \ref{prop_trix2} we have that 
		$v \liprod x = 0$ for all $v \in V$.
		This means that, for a $k$-blade, 
		$(v_1 \wedge \cdots \wedge v_{k-1} \wedge v_k) * x = (v_1 \wedge \cdots \wedge v_{k-1}) * (v_k \liprod x) = 0$
		whenever $k \geq 1$. The nondegeneracy of the scalar product implies that $x$ must have grade 0.
	\end{proof}
	
	\begin{thm} \label{thm_norm}
		The norm function is a group homomorphism $N\!: \tilde{\Gamma} \to \mathbb{R}^\times$.
	\end{thm}
	\begin{proof}
		First note that if $xx^{-1}=1$ then also $x^\star (x^{-1})^\star = 1$ and $(x^{-1})^\dagger x^\dagger = 1$, 
		hence $(x^\star)^{-1}=(x^{-1})^\star$ and $(x^\dagger)^{-1}=(x^{-1})^\dagger$.
		
		Now take $x \in \tilde{\Gamma}$. Then $x^\star v x^{-1} \in V$ for all $v \in V$
		and therefore 
		\begin{equation}
			x^\star v x^{-1} = (x^\star v x^{-1})^\dagger = (x^{-1})^\dagger v x^\cliffconj.
		\end{equation}
		This means that $x^\dagger x^\star v = v x^\cliffconj x$, or $N(x)^\star v = v N(x)$.
		By Lemma \ref{lem_grade_commute} we find that $N(x) \in \mathbb{R}$.
		The homomorphism property now follows easily, since for $x,y \in \tilde{\Gamma}$,
		\begin{equation}
			N(xy) = (xy)^\cliffconj xy = y^\cliffconj x^\cliffconj x y = y^\cliffconj N(x) y = N(x)N(y).
		\end{equation}
		Finally, because $1 = N(1) = N(xx^{-1}) = N(x)N(x^{-1})$, we must have that $N(x)$ is nonzero.
	\end{proof}
	
	\begin{lem} \label{lem_lipschitz_surjective}
		We have a homomorphism $\widetilde{\textrm{\emph{Ad}}}\!: \tilde{\Gamma} \to \textrm{\emph{O}}(V,q)$
		with kernel $\mathbb{R}^\times$.
	\end{lem}
	\begin{proof}
		We first prove that $\tAd_x$ is orthogonal for $x \in \tilde{\Gamma}$.
		Note that, for $v \in V$,
		\begin{equation}
		\setlength\arraycolsep{2pt}
		\begin{array}{rcl}
			N(\tAd_x(v)) &=& N(x^\star v x^{-1}) 
				= (x^\star v x^{-1})^\cliffconj x^\star v x^{-1} \\[5pt]
				&=& (x^{-1})^\cliffconj v^\cliffconj x^\dagger x^\star v x^{-1}
				= (x^{-1})^\cliffconj v^\cliffconj N(x)^\star v x^{-1} \\[5pt]
				&=& (x^{-1})^\cliffconj v^\cliffconj v x^{-1} N(x)^\star
				= N(v) N(x^{-1}) N(x) = N(v).
		\end{array}
		\end{equation}
		Then, since $N(v) = v^\cliffconj v = -v^2$, we have that $\tAd_x(v)^2 = v^2$.
		
		Now, if $\tAd_x = \id$ then $x^\star v = v x$ for all $v \in V$
		and by Lemma \ref{lem_grade_commute}
		we must have $x \in \mathbb{R} \cap \tilde{\Gamma} = \mathbb{R}^\times$.
	\end{proof}
	
	\noindent We finally obtain 
	
	\begin{thm} \label{thm_lipschitz}
		It holds that $\Gamma = \tilde{\Gamma}$.
	\end{thm}
	\begin{proof}
		We saw earlier that $\Gamma \subseteq \tilde{\Gamma}$.
		Take $x \in \tilde{\Gamma}$. By the above lemma we have $\tAd_x \in \textrm{O}(V,q)$.
		Using the corollary to Theorem \ref{thm_cartan_dieudonne} we then find that
		$\tAd_x = \tAd_y$ for some $y \in \Gamma$.
		Then $\tAd_{xy^{-1}} = \id$, and $xy^{-1} = \lambda \in \mathbb{R}^\times$.
		Hence, $x = \lambda y \in \Gamma$.
	\end{proof}

	\begin{exc}
		Let $\mathcal{G}^{\times \pm} := \mathcal{G}^\times \cap \mathcal{G}^\pm$ and 
		show that the Lipschitz group $\Gamma$ also can be
		defined through
		$$
			\Gamma = \{ x \in \mathcal{G}^{\times +} \cup \mathcal{G}^{\times -} \ : \ \Ad_x(V) \subseteq V \}.
		$$
	\end{exc}
	
\subsection{Properties of Pin and Spin groups}
	
	From the discussion in the previous subsections followed that $\tAd$ gives a
	surjective homomorphism from the versor/Lipschitz group $\Gamma$ to the orthogonal
	group. The kernel of this homomorphism is the set of invertible scalars.
	Because the Pin and Spin groups consist of normalized versors (i.e. $N(x) = \pm 1$)
	we find the following
	
	\begin{thm} \label{thm_double_cover}
		The homomorphisms
		\begin{displaymath}
		\setlength\arraycolsep{2pt}
		\begin{array}{llcl}
			\tAd\!: & \textrm{\emph{Pin}} (s,t) &\to& \textup{O} (s,t) \\[5pt]
			\tAd\!: & \textrm{\emph{Spin}}(s,t) &\to& \textup{SO}(s,t) \\[5pt]
			\tAd\!: & \textrm{\emph{Spin}}^+(s,t) &\to& \textup{SO}^+(s,t)
		\end{array}
		\end{displaymath}
		are surjective with kernel $\{\pm 1\}$.
	\end{thm}
	
	\noindent
	The homomorphism onto the \emph{special orthogonal group},
	\begin{displaymath}
		\textrm{SO}(V,q) := \{f \in \textrm{O}(V,q) : \det f = 1 \},
	\end{displaymath}
	follows since it
	is generated by an \emph{even} number of reflections. $\textrm{SO}^+$ denotes
	the connected component of SO containing the identity, which will be explained shortly.
	
	In other words, Theorem \ref{thm_double_cover} shows that
	the Pin and Spin groups are two-sheeted coverings of the orthogonal groups.
	Furthermore, we have the following relations between these groups.
	
	Take a unit versor $\psi = u_1 u_2 \ldots u_k \in \textrm{Pin}(s,t)$.
	If $\psi$ is odd then we can always multiply by a unit vector $e$ so that
	$\psi = \pm \psi ee$ and $\pm \psi e \in \textrm{Spin}(s,t)$.
	Furthermore, when the signature is euclidean we have 
	$\psi \psi^\dagger = u_1^2 u_2^2 \ldots u_k^2 = 1$ for all unit versors.
	The same holds for \emph{even} unit versors in anti-euclidean spaces since the signs cancel out.
	Hence, $\textrm{Spin} = \textrm{Spin}^+$ unless there is mixed signature. But in that case we can
	find two orthogonal unit vectors $e_+,e_-$ such that $e_+^2 = 1$ and $e_-^2 = -1$.
	Since $e_+e_-(e_+e_-)^\dagger = -1$ we then have that $\psi = \psi(e_+e_-)^2$, where
	$\psi e_+e_- (\psi e_+e_-)^\dagger = 1$ if $\psi \psi^\dagger = -1$.
	
	Summing up, we have that, for mixed signature $s,t \geq 1$ 
	and any pair of orthogonal vectors $e_+,e_-$ such that
	$e_+^2 = 1$, $e_-^2 = -1$,
	\begin{displaymath}
	\setlength\arraycolsep{2pt}
	\begin{array}{rcl}
		\textrm{Pin}(s,t) &=& \textrm{Spin}^+(s,t) \cdot \{ 1, e_+, e_-, e_+ e_- \}, \\[5pt]
		\textrm{Spin}(s,t) &=& \textrm{Spin}^+(s,t) \cdot \{ 1, e_+ e_- \},
	\end{array}
	\end{displaymath}
	while for euclidean and anti-euclidean signatures,
	\begin{displaymath}
		\textrm{Pin}(s,t) = \textrm{Spin}^{(+)}(s,t) \cdot \{ 1, e \},
	\end{displaymath}
	for any $e \in V$ such that $e^2 = \pm 1$.
	From the isomorphism $\mathcal{G}^+(\mathbb{R}^{s,t}) \cong \mathcal{G}^+(\mathbb{R}^{t,s})$ we also have
	the signature symmetry 
	$$
		\textrm{Spin}^{(+)}(s,t) \cong \textrm{Spin}^{(+)}(t,s).
	$$
	In all cases,
	\begin{displaymath}
		\Gamma(s,t) = \mathbb{R}^\times \cdot \textrm{Pin}(s,t).
	\end{displaymath}
	
	From these considerations it is sufficient to study the properties of the
	rotor groups in order to understand the Pin, Spin and orthogonal groups.
	Fortunately, it turns out that the rotor groups have very convenient topological
	features.
	
	\begin{thm} \label{thm_spin_connected}
		The groups $\textrm{\emph{Spin}}^+(s,t)$ are pathwise connected for $s \geq 2$ or $t \geq 2$.
	\end{thm}
	\begin{proof}
		Pick a rotor $R \in \textrm{Spin}^+(s,t)$, where $s$ or $t$ is greater than one.
		Then $R = v_1 v_2 \ldots v_{2k}$ with an even number of $v_i \in V$
		such that $v_i^2=1$ and an even number such that $v_i^2=-1$.
		Note that for any two invertible vectors $a,b$ we have
		$ab = aba^{-1}a = b'a$, where $b'^2 = b^2$.
		Hence, we can rearrange the vectors so that those with positive square come first, i.e.
		\begin{equation}
			R = a_1 b_1 \ldots a_p b_p a_1' b_1' \ldots a_q' b_q' = R_1 \ldots R_p R_1' \ldots R_q',
		\end{equation}
		where $a_i^2 = b_i^2 = 1$ and $R_i = a_i b_i = a_i * b_i + a_i \wedge b_i$
		(similarly for $a_i'^2 = b_i'^2 = -1$)
		are so called \emph{simple rotors} which are connected to either 1 or -1.
		This holds because $1 = R_i R_i^\dagger = (a_i * b_i)^2 - (a_i \wedge b_i)^2$,
		so we can (see Exercise \ref{exc_simple_rotor}) write 
		$R_i = \pm e^{\phi_i a_i \wedge b_i}$ for some $\phi_i \in \mathbb{R}$.
		Depending on the signature of the plane associated to $a_i \wedge b_i$, 
		i.e. on the sign of $(a_i \wedge b_i)^2 \in \mathbb{R}$,
		the set $e^{\mathbb{R} a_i \wedge b_i} \subseteq \textrm{Spin}^+$
		forms either a circle, a line or a hyperbola. In any case, it goes through the unit element.
		Finally, since $s>1$ or $t>1$ we can connect -1 to 1 with for example
		the circle $e^{\mathbb{R}e_1e_2}$, where $e_1, e_2$ are two orthonormal basis elements with
		the same signature.
	\end{proof}
	
	\noindent
	Continuity of the map $\tAd$ (which follows from the continuous (smooth)
	operation of taking products and inverses in $\mathcal{G}^\times$)
	now implies that the set of rotations
	represented by rotors, i.e. $\textrm{SO}^+$, 
	forms a connected subgroup containing the identity.
	For euclidean and lorentzian signatures, we have an even simpler situation:
	
	\begin{thm} \label{thm_spin_simply_conn}
		The groups $\textrm{\emph{Spin}}^+(s,t)$ are simply connected 
		for $(s,t) = (n,0)$, $(0,n)$, $(1,n)$ or $(n,1)$, where $n \geq 3$.
		Hence, these are the \emph{universal covering groups} of $\textrm{\emph{SO}}^+(s,t)$.
	\end{thm}
	\begin{proof}
		This follows from the fact that, for these signatures,
		$$
			\pi_1\big(\textrm{SO}^+(s,t)\big) = \mathbb{Z}_2 = \ker \tAd|_{\Spin^+(s,t)}
		$$
		(in general $\pi_1\big(\textrm{SO}^+(s,t)\big) 
		= \pi_1\big(\textrm{SO}(s)\big) \times \pi_1\big(\textrm{SO}(t)\big)$;
		cp. e.g. Theorem \ref{thm_bott_so} for the euclidean case),
		together with the fact that the rotor groups are connected,
		and the properties of the universal covering groups
		(see e.g. Theorem VII.6.4 in \cite{simon} and references therein).
	\end{proof}
	
	This sums up the the situation nicely for higher-dimensional euclidean
	and lorentzian spaces: The Pin group, which is a double-cover
	of the orthogonal group, consists of 
	two (euclidean case) or four (lorentzian case) simply connected components.
	These components are copies of the rotor group.

	\begin{exc} \label{exc_simple_rotor}
		Verify that if $R = ab = a*b + a \wedge b$, with $a,b \in \mathcal{G}^1$,
		and $RR^\dagger = 1$, then $R = \sigma e^{\phi a \wedge b}$ for some 
		$\phi \in \mathbb{R}$ and $\sigma \in \{\pm 1\}$.
	\end{exc}

\subsection{The bivector Lie algebra $\mathcal{G}^2$}

	We have reduced the study of the orthogonal groups to
	a study of the rotor groups, which in all higher dimensions
	are connected, smooth Lie groups.
	An understanding of such groups is provided by their local
	structure, i.e. their corresponding tangent space Lie algebra, which
	in this case is simply the space $\mathcal{G}^2$ of bivectors
	with a Lie product given by the commutator bracket $[\cdot,\cdot]$
	(recall Exercise \ref{exc_bivectors_closed}).
	To see this, take any smooth curve $R: \mathbb{R} \to \Spin^+$, $t \to R(t)$ 
	through the unit element; $R(0) = 1$. 
	We claim that the tangent $R'(0)$ is a bivector. Namely,
	differentiating the identity $R(t)R(t)^\dagger = 1$ at $t=0$ gives 
	\begin{equation} \label{rotor_prime_reverse}
		R'(0)^\dagger = -R'(0).
	\end{equation}
	Furthermore, since $R(t) \in \mathcal{G}^+$ for all $t$, we must have 
	$R'(0) \in \bigoplus_{k=0,1,\ldots} \mathcal{G}^{2+4k}$.
	We write $R'(0) = B + Z$, where $B \in \mathcal{G}^2$ and $Z$ contains
	no grades lower than 6. 
	Also, because $R(t) \in \tilde{\Gamma}$ along the curve, we have for any $v \in V$
	$$
		w(t) := R(t)vR(t)^\dagger \in V  \qquad \forall t.
	$$
	Again, differentiating this expression and using 
	\eqref{rotor_prime_reverse} yields
	$$
		V \ni w'(0) = R'(0)v - vR'(0) = -2 v \liprod R'(0)
		= -2 v \liprod B - 2 v \liprod Z.
	$$
	Inspecting
	the grades involved in this equation, it follows
	that $v \liprod Z = 0$ for all $v \in V$, and by Lemma \ref{lem_grade_commute},
	that $Z = 0$. 

	The commutator product of $\mathcal{G}^2$ is 
	inherited from the geometric product of $\Spin^+$,
	which through the twisted adjoint action corresponds to taking products of orthogonal transformations.
	In fact, there is a canonical Lie algebra isomorphism between
	the bivector space $\mathcal{G}^2$ 
	and the algebra $\so(V,q)$ of antisymmetric transformations of $V$
	(which is the Lie algebra of $\SO(V,q)$), given by
	(cp. Exercise \ref{exc_func_bivec_corresp})
	\begin{eqnarray*}
		\so & \cong & \mathfrak{spin} = \mathcal{G}^2 \\
		f = \ad_B & \leftrightarrow & B = \frac{1}{2} \sum_{i,j} e^i*f(e^j)\ e_i \wedge e_j,
	\end{eqnarray*}
	where $\{e_i\}_i$ is any general basis of $V$.
	One verifies, by expanding in the geometric product, that
	(Exercise \ref{exc_spin_commutators})
	\begin{equation} \label{spin_commutators}
		\frac{1}{2} [e_i \wedge e_j, e_k \wedge e_l] 
		= (e_j*e_k) e_i \wedge e_l - (e_j*e_l) e_i \wedge e_k + (e_i*e_l) e_j \wedge e_k - (e_i*e_k) e_j \wedge e_l.
	\end{equation}
	
	\begin{rem}
		Actually, this bivector Lie algebra is more general than it might first
		seem. Namely, one can show (see e.g. \cite{doran_et_al} or \cite{doran_lasenby})
		that the Lie algebra $\mathfrak{gl}$ of the general linear group can be represented as
		a bivector algebra. From the fact that any finite-dimensional Lie algebra
		has a faithful finite-dimensional representation
		(Ado's Theorem for characteristic zero, Iwasawa's Theorem for nonzero characteristic,
		see e.g. \cite{jacobson})
		it then follows that \emph{any} finite-dimensional real or complex Lie algebra can be represented 
		as a bivector algebra.
	\end{rem}
	
	Recall that (see Exercise \ref{exc_exponential}), for any choice of norm on $\mathcal{G}$, 
	the exponential defined by
	$e^x := \sum_{k=0}^\infty \frac{x^k}{k!}$
	converges and satisfies
	$$
			e^x e^{-x} = e^{-x} e^x = 1
	$$
	and $xe^x = e^xx$ for any $x \in \mathcal{G}$.
	The following holds for any signature.
	
	\begin{thm} \label{thm_exp_in_spin}
		For any bivector $B \in \mathcal{G}^2$ we have that $\pm e^B \in \textrm{\emph{Spin}}^+$.
	\end{thm}
	\begin{proof}
		It is obvious that $\pm e^B$ is an even multivector and that $e^B (e^B)^\dagger = e^B e^{-B} = 1$.
		Hence, it is sufficient to prove that $\pm e^B \in \Gamma$, or by Theorem \ref{thm_lipschitz},
		that $e^B V e^{-B} \subseteq V$. 
		
		Let $W := \{ y \in \mathcal{G} : v*y=0\ \forall v \in V \} = \oplus_{k \neq 1} \mathcal{G}^k$
		and choose $x \in V$ and $y \in W$.
		Define a map $f: \mathbb{R} \to \mathbb{R}$ by
		$$
			f(t) := (e^{tB} x e^{-tB}) * y.
		$$
		$f$ is obviously a real analytic function with derivatives
		$$
			f'(t) = (Be^{tB}xe^{-tB} - e^{tB}xe^{-tB}B) * y = [B, e^{tB}xe^{-tB}] * y,
		$$
		$$
			f''(t) = [B, Be^{tB}xe^{-tB} - e^{tB}xe^{-tB}B] * y = \left[ B, [B, e^{tB}xe^{-tB}] \right] * y,
		$$
		etc.
		It follows that $f^{(k)}(0) = \ad_B^{k}(x) * y = 0$ 
		for all $k \ge 0$, since $\ad_B: V \to V$ (see Exercise \ref{exc_bivec_adj}).
		Hence, $f(t) = 0\ \forall t \in \mathbb{R}$ and $e^{tB}xe^{-tB} \in V$.
	\end{proof}
	
	\noindent
	In other words, $\pm e^{\mathcal{G}^2} \subseteq \Spin^+$.
	Actually, the converse inclusion holds for (anti-) euclidean and lorentzian spaces.
	
	\begin{thm} \label{thm_spin_is_exp}
		For $(s,t) = (n,0)$, $(0,n)$, $(1,n)$ or $(n,1)$, we have 
		\begin{displaymath}
			\textrm{\emph{Spin}}^+(s,t) = \pm e^{\mathcal{G}^2(\mathbb{R}^{s,t})},
		\end{displaymath}
		i.e. any rotor can be written as (minus) the exponential of a bivector.
		The minus sign is only required in the lower-dimensional cases 
		$(0,0)$, $(1,0)$, $(0,1)$, $(1,1)$, $(1,2)$, $(2,1)$, $(1,3)$ and $(3,1)$.
	\end{thm}
	
	\noindent
	The proof can be found in \cite{riesz}. Essentially, it relies on the fact
	that any isometry of an euclidean or lorentzian space can be generated
	by a single infinitesimal transformation. This holds for these spaces \emph{only},
	so that for example $\textrm{Spin}^+(2,2) \neq \pm e^{\mathcal{G}^2(\mathbb{R}^{2,2})}$,
	where for instance the rotor
	$$
		\pm e_1e_2e_3e_4 e^{\beta(e_1e_2 + 2e_1e_4 + e_3e_4)}, \qquad \beta>0, 
	$$
	cannot be reduced to a single exponential; 
	see \cite{lounesto} and \cite{riesz} pp. 150-152.

	Similarly, we see below that even though one can always decompose a bivector into a
	sum of at most $\frac{1}{2} \cdot \dim V$ basis blades, 
	the decomposition is particularly nice only in the 
	euclidean and lorentzian cases.

	\begin{thm}	\label{thm_bivec_decomp}
		Given any bivector $B \in \mathcal{G}^2(V)$,
		there exists a basis $\{f_1,\ldots,f_n\}$ of $V$ such that
		\begin{equation} \label{bivec_decomp}
			B = \beta_{12} f_1 \wedge f_2 + \beta_{34} f_3 \wedge f_4 + \ldots + \beta_{2k-1,2k} f_{2k-1} \wedge f_{2k},
		\end{equation}
		where $2k \le n$.
		If $V$ is (anti-)euclidean then the basis can be chosen orthonormal,
		while if $V$ is (anti-)lorentzian then 
		either the $\{f_i\}_i$ can be chosen to form an orthonormal basis, or we need to
		choose $f_1 = e_0+e_1$, where $\{e_0,e_1,f_2,\ldots,f_n\}$ 
		is an orthonormal basis of $\mathbb{R}^{1,n}$.
		For other signatures,
		such an \emph{orthogonal} decomposition is not always possible.
	\end{thm}
	\begin{proof}
		We will first prove the orthogonal decomposition for the
		euclidean case, and then observe that this also gives the
		general (weaker) statement.
		For the lorentzian case and the last statement of the
		theorem we refer to \cite{riesz}.
		
		We start with a general observation:
		If $T: \mathbb{R}^n \to \mathbb{R}^n$ is linear then there
		exists a subspace $M \subseteq \mathbb{R}^n$ such that
		$T(M) \subseteq M$ and $1 \le \dim M \le 2$.
		To see this we start by extending $T: \mathbb{C}^n \to \mathbb{C}^n$
		and therefore realize that there exists $\lambda \in \mathbb{C}$
		and $z \in \mathbb{C}^n \smallsetminus \{0\}$ such that
		$Tz = \lambda z$ (this follows from the fact that the polynomial
		$\det (\lambda \id - T)$ has a zero in $\mathbb{C}$, by the
		fundamental theorem of algebra).
		Now, write $\lambda = \sigma + i\tau$ and $z = x + iy$,
		where $\sigma,\tau \in \mathbb{R}$ and $x,y \in \mathbb{R}^n$.
		Let $M = \Span_\mathbb{R} \{x,y\}$.
		We now have that
		$$
			Tx + iTy = T(x+iy) = (\sigma + i\tau)(x + iy)
			= (\sigma x - \tau y) + i(\tau x + \sigma y),
		$$
		which gives $Tx = \sigma x - \tau y$ and
		$Ty = \tau x + \sigma y$.
		Hence, $T(M) \subseteq M$ and $\dim M \le 2$.
		
		Let us now assume that the signature is euclidean, i.e. $V = \mathbb{R}^n$,
		and prove that, given any bivector $B \in \mathcal{G}^2$,
		there exists an orthonormal basis $\{f_j\}_{j=1}^n$ such that
		the decomposition \eqref{bivec_decomp} holds.
		Introduce the \emph{rank} of $B$, $\textup{rank}(B)$, 
		as the least integer $m$ such that there exists a
		basis $\{e_1,\ldots,e_n\}$ in which $B$ can be written
		\begin{equation} \label{bivec_rank}
			B = \sum_{1 \le i<j \le m} \beta_{ij} e_i \wedge e_j.
		\end{equation}
		We prove the theorem by induction on $\textup{rank}(B)$,
		i.e. for each fixed $m \ge 2$ we show that for every 
		$B \in \mathcal{G}^2$ with $\textup{rank}(B) = m$
		there exists an orthonormal basis $\{f_1,\ldots,f_n\}$
		such that 
		\eqref{bivec_decomp} holds.
		
		For $m = 2$ the statement is trivial 
		(choose an orhonormal basis for the plane $\overline{e_1 \wedge e_2}$), so assume that $m > 2$
		and put $W = \mathbb{R}e_1 + \ldots + \mathbb{R}e_m = \overline{e_1 \wedge \ldots \wedge e_m}$
		(where $\{e_i\}_i$ is the basis in \eqref{bivec_rank}).
		But then the map $W \ni x \mapsto x \liprod B \in W$
		has a nontrivial invariant subspace $M$ of dimension $\le 2$.
		
		First assume that $\dim M = 1$.
		Then there exists an orthonormal basis $\{f_1,\ldots,f_m,e_{m+1},\ldots,e_n\}$ of $V$
		such that $\overline{f_1 \wedge \ldots \wedge f_m} = \overline{e_1 \wedge \ldots \wedge e_m}$,
		$M = \mathbb{R}f_1$, and $B = \sum\limits_{1 \le i<j \le m} \beta_{ij}' f_if_j$.
		But then
		$$
			\lambda f_1 = f_1 \liprod B = \beta_{12}'f_2 + \beta_{13}'f_3 + \ldots + \beta_{1m}'f_m,
		$$
		implying $\beta_{1j}' = 0 = \lambda$ for $1 < j \le m$.
		This means that at most $m-1$ basis vectors are required in
		the expression for $B$, contradicting the assumption that $\textup{rank}(B)=m$.
		
		Hence, $\dim M = 2$ and we can assume that
		$M = \mathbb{R}f_1 + \mathbb{R}f_2$, where once again
		$\{f_1,\ldots,f_m,e_{m+1},\ldots,e_n\}$ denotes 
		an orthonormal basis of $V$ such that
		$B = \sum\limits_{1 \le i<j \le m} \beta_{ij}' f_if_j$. We can write
		$$
			\alpha f_1 + \beta f_2 = f_1 \liprod B = \sum_{j=2}^m \beta_{1j}'f_j,
		$$
		so that $\alpha = 0$ and $\beta = \beta_{12}'$.
		Furthermore, if $f_2 \liprod B = \gamma f_1 + \delta f_2$
		then $\delta = f_2 \liprod (f_2 \liprod B) = (f_2 \wedge f_2) \liprod B = 0$
		and
		$$
			\gamma = f_1 \liprod (f_2 \liprod B) = (f_1 \wedge f_2) \liprod B
			= -(f_2 \wedge f_1) \liprod B = -f_2 \liprod (f_1 \liprod B) = -\beta_{12}'.
		$$
		Thus, $f_1 \liprod B = \beta_{12}'f_2$ and $f_2 \liprod B = -\beta_{12}'f_1$.
		If we let $B' := B - \beta_{12}'f_1f_2$ we therefore have
		$f_1 \liprod B' = f_2 \liprod B' = 0$ and, writing
		$B' = \sum\limits_{1 \le i<j \le m} \beta_{ij}'' f_if_j$,
		we find
		$$
			0 = f_1 \liprod B' = \sum_{1<j} \beta_{1j}'' f_j \quad \Rightarrow \quad \beta_{1j}'' = 0 \quad \forall\ 1 < j \le m,
		$$
		$$
			0 = f_2 \liprod B' = \sum_{2<j} \beta_{2j}'' f_j \quad \Rightarrow \quad \beta_{2j}'' = 0 \quad \forall\ 2 < j \le m.
		$$
		Hence, $B'$ can be expressed solely in terms of the $m-2$ 
		basis vectors $f_3,\ldots,f_m$, i.e. $\textup{rank}(B') \le m-2$.
		The induction assumption now implies existence of an
		orthonormal basis on the form $\{f_1,f_2,g_3,\ldots,g_m,e_{m+1},\ldots,e_n\}$
		such that 
		$B' = \sum\limits_{j=3,5,7,\ldots,k<m} \gamma_{j,j+1} g_jg_{j+1}$,
		and hence that
		$$
			B = \beta_{12}'f_1f_2 + \gamma_{34}g_3g_4 + \gamma_{56}g_5g_6 + \ldots + \gamma_{k,k+1}g_kg_{k+1},
		$$
		which proves the theorem in the euclidean case.
		
		For the weaker but signature-independent statement, 
		note that a decomposition of the form \eqref{bivec_decomp} 
		in terms of a general basis $\{f_j\}_{j=1}^n$ only concerns
		the exterior algebra associated to the outer product.
		Hence, given a bivector $B = \sum_{1\le i<j \le n} \beta_{ij} e_i \wedge e_j$, 
		there is no loss in generality in 
		temporarily switching to euclidean signature (where e.g. $e_i * e_j = \delta_{ij}$), 
		with respect to which
		one can find an orthonormal basis $\{f_j = \sum_k \alpha_{jk}e_k \}_{j=1}^n$ 
		such that \eqref{bivec_decomp} holds. 
		Once we have found such a basis we can switch back to the original
		signature, and realize that $\{f_j\}_{j=1}^n$ is of course still a basis 
		(not necessarily orthogonal)
		and that the expression \eqref{bivec_decomp} reduces to the
		original expression for $B$ upon expressing $f_j = \sum_k \alpha_{jk}e_k$.
	\end{proof}

	\begin{rem}
		Let us see why the additional requirement is necessary
		in the lorentzian case.
		Namely, assume that $\{e_0,e_1,\ldots,e_n\}$ is an orthonormal
		basis of $\mathbb{R}^{1,n}$ with $e_0^2 = 1$.
		Consider $B := (e_0 + e_1) \wedge e_2$, and assume that $B$ can also be written
		$$
			B = \beta_1 f_1f_2 + \beta_3 f_3f_4 + \ldots + \beta_{2p-1} f_{2p-1}f_{2p}, \qquad 2p \le n+1
		$$
		for some other orthonormal basis $\{f_1,f_2,\ldots,f_{n+1}\}$.
		Since
		\begin{eqnarray*}
			0 = B \wedge B &=& 2\beta_1\beta_3 f_1f_2f_3f_4 + 2\beta_1\beta_5 f_1f_2f_5f_6 + \ldots \\
			&&	\ldots + 2\beta_{2p-3}\beta_{2p-1} f_{2p-3}f_{2p-2}f_{2p-1}f_{2p},
		\end{eqnarray*}
		we must have $\beta_{2s+1}\beta_{2t+1} = 0$ for all $s,t$.
		But then only one term in the sum can survive, and we can assume, e.g.
		$$
			(e_0 + e_1) \wedge e_2 = \beta f_1f_2, \qquad \beta \neq 0.
		$$
		Squaring both sides, we find
		$$
			0 = -(e_0+e_1)^2 e_2^2 = -\beta^2 f_1^2 f_2^2,
		$$
		so that either $f_1^2 = 0$ or $f_2^2 = 0$, 
		but this contradicts the signature of the space.
		\qed
	\end{rem}
	
	Another way to state Theorem \ref{thm_bivec_decomp} 
	in the euclidean and lorentzian case
	is that every bivector $B$ can be written as a sum of commuting 2-blades:
	\begin{equation} \label{bivector_decomposition}
		B = B_1 + B_2 + \ldots + B_k,
	\end{equation}
	such that 
	$B_i \in \mathcal{B}_2$, $B_iB_j = B_jB_i\ \forall i,j$, $\dim \bar{B}_i = 2$,
	$\bar{B}_i \perp \bar{B}_j\ \forall i \neq j$, and $k \le n/2$.
	It then follows that every rotor $e^B$ can be written
	$e^{B_1}e^{B_2} \ldots e^{B_k}$, where each factor is 
	a simple rotor
	(as discussed in the proof of Theorem \ref{thm_spin_connected}).
	The decomposition \eqref{bivector_decomposition} is unique
	unless $B_i^2 = B_j^2$ for some $i \neq j$.

	\begin{exc} \label{exc_spin_commutators}
		Verify \eqref{spin_commutators} and observe that, 
		with $L_{ij} := \frac{1}{2}e_i \wedge e_j$ 
		and the metric tensor $g_{ij} := e_i * e_j$
		(recall that $\{e_i\}_i$ is any general basis),
		this leads to the conventional commutation relations
		$$
			[L_{ij},L_{kl}] = g_{jk}L_{il} - g_{jl}L_{ik} + g_{il}L_{jk} - g_{ik}L_{jl}
		$$
		for the (anti-hermitian) generators of $\mathfrak{so}(s,t)$.
	\end{exc}
	
	\begin{exc} \label{exc_bivector_blade}
		Use Theorem \ref{thm_bivec_decomp} to prove that a bivector
		$B \in \mathcal{G}^2(V)$ is a blade if and only if $B^2$ is
		a scalar (regardless of the signature of $V$).
	\end{exc}

	\begin{exc} \label{exc_blade_counterexample}
		Construct an example of a homogeneous multivector which
		squares to a nonzero scalar even though it is not a blade.
	\end{exc}

	\begin{exc} \label{exc_killing_form}
		Show that the \emph{Killing form},
		$$
			K(A,B) := \tr_{\mathcal{G}^2} (\ad_A \ad_B), \qquad A,B \in \mathcal{G}^2,
		$$
		of the Lie algebra $\mathfrak{spin} = \mathcal{G}^2$ simplifies to
		$$
			K(A,B) = c \thinspace A*B,
		$$
		and determine the constant $c$.
		(Note that for (anti-)euclidean signatures, $K \leq 0$,
		reflecting the fact that the corresponding Spin groups are compact.)
		\\
		\emph{Hint:} Use reciprocal bases and Exercise \ref{exc_basis_identities}.
	\end{exc}

%% file: clifford_examples.tex
\subsection{Examples in lower dimensions}

	We finish this section by considering various examples in lower
	dimensions, and work out the respective $\Pin$, $\Spin$ and rotor
	groups in each case.
	We also see that we can use various involutions and extended versions 
	of the norm function \eqref{norm_function} to determine the 
	corresponding group of invertible elements.

	It will be convenient to use the following characterization of the
	rotor groups in lower dimensions as simply being the groups of \emph{even unit multivectors}.
	
	\begin{prop} \label{prop_rotors_as_even_units}
		If $\dim V \le 5$ then
		$$
			\Spin^+ = \{ R \in \mathcal{G}^+ : RR^\dagger = 1 \}.
		$$
	\end{prop}
	\begin{proof}
		By definition, the rotor group is 
		a subset of the right hand side
		in all dimensions, so let us consider any even grade
		multivector $R$ such that $RR^\dagger = 1$.
		Then $R$ is invertible and $R^{-1} = R^\dagger \in \mathcal{G}^+$.
		We claim that $RVR^\dagger \subseteq V$, so 
		$R \in \tilde{\Gamma}$ and hence also a rotor.
		The claim follows because for any $v \in V$ 
		the expression $RvR^\dagger$ is both odd and self-reversing,
		and hence in $\mathcal{G}^1 \oplus \mathcal{G}^5$.
		This proves the claim and hence the theorem for $\dim V < 5$.
		In the case $\dim V = 5$, assume to the contrary that
		$$
			RvR^\dagger = w + \alpha I,
		$$
		with $w \in V$ and $\alpha \neq 0$. 
		Then, since the pseudoscalar commutes with everything,
		$$
			\alpha = \langle RvR^\dagger I^{-1} \rangle_0 
			= \langle RvI^{-1}R^\dagger \rangle_0 = \langle v I^{-1} R^\dagger R \rangle_0 
			= \langle v I^{-1} \rangle_0 = 0,
		$$
		which is a contradiction.
	\end{proof}
	
	We will also find it useful to extend our set of grade-based
	involutions with the following generalized ones.
	
	\begin{defn}
		We define
		\begin{displaymath}
		\setlength\arraycolsep{2pt}
		\begin{array}{rcl}
			[A]							&:=& (-1)^{(A \neq \varnothing)} A, \\[5pt]
			[A]_{i_1,i_2,\ldots,i_k} 	&:=& (-1)^{(|A| \in \{i_1,i_2,\ldots,i_k\})} A.
		\end{array}
		\end{displaymath}
		for $A \in \mathscr{P}(X)$ and extend linearly to $\cl(X,R,r)$.
	\end{defn}
	
	\begin{exc}
		Show that Proposition \ref{prop_rotors_as_even_units}
		cannot be extended to $\dim V = 6$.
	\end{exc}

\subsubsection{The euclidean line}
	
	Let us first consider the algebra of the euclidean line,
	$\mathcal{G}(\mathbb{R}^{1}) = \textrm{Span}_\mathbb{R} \{1,e\}$,
	where $e$ is the basis element, $e^2 = 1$.
	This is a commutative algebra with pseudoscalar $e$. 
	Because the only unit vectors are $\{\pm e\}$,
	the unit versors resp. rotors form the discrete groups
	\begin{displaymath}
	\setlength\arraycolsep{2pt}
	\begin{array}{rcl}
		\textrm{Pin}(1,0) &=& \{ 1, -1, e, -e \} \cong \mathbb{Z}_2 \times \mathbb{Z}_2, \\[5pt]
		\textrm{Spin}^{(+)}(1,0) &=& \{ 1, -1 \} \cong \mathbb{Z}_2.
	\end{array}
	\end{displaymath}
	Note that $\underline{\pm 1}$ 
	is the identity map, while $\underline{\pm e}$ 
	is the unique reflection of the line.
	
	One easily finds the multiplicative group 
	$\mathcal{G}^\times(\mathbb{R}^1)$ by considering the norm function,
	which with a one-dimensional vector space is given by
	$$
		N_1(x) := x^\cliffconj x = x^\star x.
	$$
	For an arbitrary element $x = \alpha + \beta e$ then
	$$
		N_1(x) = (\alpha - \beta e)(\alpha + \beta e) = \alpha^2 - \beta^2 \in \mathbb{R}.
	$$
	When $N_1(x) \neq 0$ we find that $x$ has an inverse $x^{-1} = \frac{1}{N_1(x)} x^\star = \frac{1}{\alpha^2-\beta^2}(\alpha - \beta e)$.
	Hence,
	$$
		\mathcal{G}^\times(\mathbb{R}^{1}) = \{ x \in \mathcal{G} : N_1(x) \neq 0 \} 
			= \{ \alpha + \beta e \in \mathcal{G} \ : \ \alpha^2 \neq \beta^2 \}.
	$$
	Note also that $N_1(xy) = x^\star y^\star x y = N_1(x)N_1(y)$ for all $x,y \in \mathcal{G}$
	since the algebra is commutative.
	
\subsubsection{The anti-euclidean line}
	
	The algebra of the anti-euclidean line,
	$\mathcal{G}(\mathbb{R}^{0,1}) = \textrm{Span}_\mathbb{R} \{1,i\} \cong \mathbb{C}$,
	is of course
	also a commutative algebra, but, unlike the previous example, this is even a field
	since every nonzero element is invertible. The norm function is an
	actual norm (squared) in this case,
	$$
		N_1(\alpha + \beta i) = (\alpha - \beta i)(\alpha + \beta i) = \alpha^2 + \beta^2 \in \mathbb{R}^+,
	$$
	namely the modulus of the complex number. 
	We have already noted that the grade involution
	represents the complex conjugate and, as above, $x^{-1} = \frac{1}{N_1(x)} x^\star$.
	The relevant groups are
	\begin{displaymath}
	\setlength\arraycolsep{2pt}
	\begin{array}{rcl}
		\mathcal{G}^\times(\mathbb{R}^{0,1}) &=& \mathcal{G} \!\smallsetminus\! \{0\} \cong \mathbb{C}^\times, \\[5pt]
		\textrm{Pin}(0,1) &=& \{ 1, -1, i, -i \} \cong \mathbb{Z}_4, \\[5pt]
		\textrm{Spin}^{(+)}(0,1) &=& \{ 1, -1 \} \cong \mathbb{Z}_2.
	\end{array}
	\end{displaymath}
	
\subsubsection{The degenerate line}

	We also include the simplest example of a degenerate algebra, 
	$\mathcal{G}(\mathbb{R}^{0,0,1})$,
	just to see what happens in such a situation.
	Let the vector $n$ span a one-dimensional space with quadratic form $q=0$.
	Then
	$$
		\mathcal{G}(\mathbb{R}^{0,0,1}) = \textrm{Span}_\mathbb{R} \{1,n\} \cong \bigwedge\nolimits^* \mathbb{R}^1,
	$$
	and $n^2=0$. The norm function depends only on the scalar part in this case,
	$$
		N_1(\alpha + \beta n) = (\alpha - \beta n)(\alpha + \beta n) = \alpha^2 \in \mathbb{R}^+.
	$$
	Hence, an element is invertible if and only if the scalar part is nonzero;
	$$
		\mathcal{G}^\times(\mathbb{R}^{0,0,1}) = \{ \alpha + \beta n \in \mathcal{G} \ : \ \alpha \neq 0 \}.
	$$
	Since no vectors are invertible, we are left with only the empty product
	in the versor group, i.e.
	$
			\Gamma = \{1\}.
	$
	Note, however, that for $\alpha \neq 0$
	$$
		(\alpha + \beta n)^\star n (\alpha + \beta n)^{-1} = (\alpha - \beta n) n {\textstyle \frac{1}{\alpha^2}} (\alpha - \beta n) = n,
	$$
	so the Lipschitz group in this case is
	$$
		\tilde{\Gamma} = \mathcal{G}^\times \neq \Gamma.
	$$
	This shows that the assumption on nondegeneracy was necessary in the
	discussion about the Lipschitz group in Section \ref{sec_lipschitz}.
	
\subsubsection{The euclidean plane}
	
	Let $\{e_1, e_2\}$ be an orthonormal basis of $\mathbb{R}^2$ and consider
	the plane algebra
	$$
		\mathcal{G}(\mathbb{R}^2) = \textrm{Span}_\mathbb{R} \{1,\ e_1,\ e_2,\ I=e_1e_2\}.
	$$
	We recall that the even subalgebra is
	$\mathcal{G}^+(\mathbb{R}^{2}) \cong \mathcal{G}(\mathbb{R}^{0,1}) \cong \mathbb{C}$
	and that the rotor group 
	corresponds to the group of unit complex numbers,
	$$
		\textrm{Spin}^{(+)}(2,0) = e^{\mathbb{R}I} \cong \textrm{U}(1).
	$$
	Furthermore,
	$$
		\underline{e^{\varphi I}}(e_1)
		= e^{\varphi I} e_1 e^{-\varphi I} = e_1 e^{-2\varphi I} = e_1(\cos 2\varphi - I \sin 2\varphi)
		= \cos 2\varphi \thinspace e_1 - \sin 2\varphi \thinspace e_2,
	$$
	so a rotor on the form $\pm e^{-\frac{\varphi}{2} I}$ 
	represents a counter-clockwise\footnote{Assuming, 
	of course, that $e_1$ points at 3 o'clock and $e_2$ at 12 o'clock.}
	rotation in the plane by an angle $\varphi$.
	The Pin group is found by picking any unit vector $\boldsymbol{e} \in S^1$:
	$$
		\textrm{Pin}(2,0) = e^{\mathbb{R}I}\ \bigsqcup\ e^{\mathbb{R}I} \boldsymbol{e} \ \ \cong \ \ U(1)\ \bigsqcup\ S^1,
	$$
	i.e. two copies of a unit circle, 
	where we write $\bigsqcup$ to emphasize a \emph{disjoint} union.
	An element $\pm e^{-\frac{\varphi}{2} I}\boldsymbol{e} \in S^1$ of the second circle corresponds 
	to a reflection along $\boldsymbol{e}$ (which here serves as a reference axis)
	followed by a rotation by $\varphi$ in the plane.
	
	Let us next 
	determine the full multiplicative group
	of the plane algebra.
	For algebras over two-dimensional spaces we use the original norm function
	$$
		N_2(x) := x^\cliffconj x
	$$
	since it satisfies $N_2(x)^\cliffconj = N_2(x)$ for all $x \in \mathcal{G}$.
	The properties of the involutions in Table \ref{table_involutions} then require
	this to be a scalar, so we have a map
	$$
		N_2 \!: \mathcal{G} \to \mathcal{G}^0 = \mathbb{R}.
	$$
	For an arbitrary element $x = \alpha + a_1e_1 + a_2e_2 + \beta I \in \mathcal{G}$, 
	we explicitly find
	$$
	\setlength\arraycolsep{2pt}
	\begin{array}{rcl}
		N_2(x) &=& (\alpha - a_1e_1 - a_2e_2 - \beta I)(\alpha + a_1e_1 + a_2e_2 + \beta I) \\[5pt]
			&=& \alpha^2 - a_1^2 - a_2^2 + \beta^2.
	\end{array}
	$$
	Furthermore, $N_2(x^\cliffconj) = N_2(x)$ and 
	$N_2(xy) = y^{\cliffconj}N_2(x)y = N_2(x)N_2(y)$ for all $x,y \in \mathcal{G}$.
	Proceeding as in the one-dimensional case, we find that $x$ has
	an inverse $x^{-1} = \frac{1}{N_2(x)} x^\cliffconj$ if and only if $N_2(x) \neq 0$, i.e.
	$$
	\setlength\arraycolsep{2pt}
	\begin{array}{rcl}
		\mathcal{G}^\times(\mathbb{R}^{2}) &=& \{ x \in \mathcal{G} : N_2(x) \neq 0 \} \\[5pt]
			&=& \{ \alpha + a_1e_1 + a_2e_2 + \beta I \in \mathcal{G} \ : \ \alpha^2 + \beta^2 \neq a_1^2 + a_2^2 \}.
	\end{array}
	$$
		
\subsubsection{The anti-euclidean plane}

	In the case of the anti-euclidean plane, $\mathcal{G}(\mathbb{R}^{0,2}) \cong \mathbb{H}$,
	the norm function $N_2$ has 
	similar properties as in the euclidean algebra, except that it now
	once again represents the square of an actual norm,
	namely the \emph{quaternion norm},
	$$
		N_2(x) = \alpha^2 + a_1^2 + a_2^2 + \beta^2.
	$$
	Just as for the complex numbers then, all nonzero elements are invertible,
	$$
		\mathcal{G}^\times(\mathbb{R}^{0,2}) = \{ x \in \mathcal{G} : N_2(x) \neq 0 \} = \mathcal{G} \smallsetminus \{0\}.
	$$
	The even subalgebra is also in this case isomorphic to the complex
	numbers, and the 
	Pin, Spin and rotor groups are analogous to
	the euclidean case.

\subsubsection{The lorentzian plane}

	With the standard basis $\{e_0,e_1\}$ of $\mathbb{R}^{1,1}$,
	the two-dimensional lorentzian algebra is given by
	$$
		\mathcal{G}(\mathbb{R}^{1,1}) = \textrm{Span}_\mathbb{R} \{1,\ e_0,\ e_1,\ I=e_1e_0\},
	$$
	where $e_0$ is interpreted as a time direction in physical applications,
	while $e_1$ defines a spatial direction.
	In general, a vector (or blade) $v$ is called \emph{timelike} if $v^2>0$, 
	\emph{spacelike} if $v^2<0$, and \emph{lightlike} or \emph{null} if $v^2=0$.

	The multiplicative group 
	is like previously given by
	$$
	\setlength\arraycolsep{2pt}
	\begin{array}{rcl}
		\mathcal{G}^\times(\mathbb{R}^{2}) &=& \{ x \in \mathcal{G} : N_2(x) \neq 0 \} \\[5pt]
			&=& \{ \alpha + a_0e_0 + a_1e_1 + \beta I \in \mathcal{G} \ : \ \alpha^2 - a_0^2 + a_1^2 - \beta^2 \neq 0 \}.
	\end{array}
	$$
	However, the pseudoscalar\footnote{The choice of orientation of $I$ here
	corresponds to the fact that space-time diagrams are conventionally
	drawn with the time axis vertical.}
	$I = e_1e_0$ squares to the identity in this case and
	the even subalgebra is therefore $\mathcal{G}^+(\mathbb{R}^{1,1}) \cong \mathcal{G}(\mathbb{R}^1)$.
	This has as an important consequence that the rotor group is fundamentally
	different from the (anti-) euclidean case,
	$$
		\textrm{Spin}^+(1,1) = \{ R = \alpha + \beta I \in \mathcal{G}^+ \ : \ \alpha^2 - \beta^2 = 1 \} 
			= \pm e^{\mathbb{R}I},
	$$
	(recall Exercise \ref{exc_simple_rotor}, or Theorem \ref{thm_spin_is_exp}).
	Thus, the set of rotors can be understood as a pair of 
	disjoint hyperbolas passing through the points 1 and -1, respectively.
	The rotations that are represented by rotors of this form are called \emph{Lorentz boosts}.
	Note the hyperbolic nature of these rotations,
	\begin{eqnarray*}
		e^{\alpha I} e_0 e^{-\alpha I} &=& e_0 e^{-2\alpha I} = e_0(\cosh 2\alpha - I \sinh 2\alpha) \\
			&=& \cosh 2\alpha \thinspace e_0 + \sinh 2\alpha \thinspace e_1, \\
		e^{\alpha I} e_1 e^{-\alpha I} &=& \cosh 2\alpha \thinspace e_1 + \sinh 2\alpha \thinspace e_0.
	\end{eqnarray*}
	Hence, a rotor $\pm e^{\frac{\alpha}{2}I}$ transforms, or \emph{boosts}, timelike vectors 
	by a hyperbolic angle $\alpha$ in the positive spacelike direction (see Figure \ref{fig_boost}).

	\begin{figure}[t]
		\centering
		\psfrag{T_e0}{$e_0$}
		\psfrag{T_e1}{$e_1$}
		\psfrag{T_e}{$\underline{R}(e_0)$}
		\psfrag{T_cone}{The light-cone}
		\includegraphics{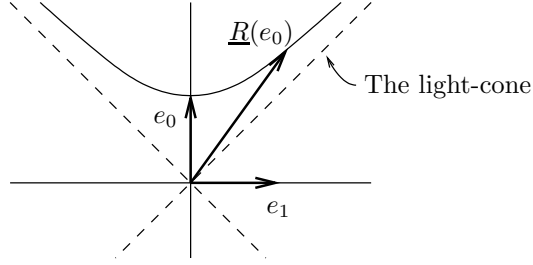}
		\caption{A Lorentz-boosted vector $\underline{R}(e_0)$,
		where $R = \pm e^{\alpha I/2}$ and $\alpha > 0$.}
		\label{fig_boost}
	\end{figure}

	The Spin group consists of four hyperbolas and the Pin group of eight,
	\begin{displaymath}
	\setlength\arraycolsep{2pt}
	\begin{array}{rcl}
		\textrm{Spin}(1,1) &=& \pm e^{\mathbb{R}I}\ \bigsqcup\ \pm e^{\mathbb{R}I}I, \\[5pt]
		\textrm{Pin}(1,1) &=& \pm e^{\mathbb{R}I}\ \bigsqcup\ \pm e^{\mathbb{R}I}e_0\ \bigsqcup\ \pm e^{\mathbb{R}I}e_1\ \bigsqcup\ \pm e^{\mathbb{R}I}I,
	\end{array}
	\end{displaymath}
	Note that the action of the rotor $I = e_1e_0$ 
	in the second pair of components of the Spin group
	is a reflection both in the time and the space direction;
	$$
		\underline{I}(x) = IxI^{-1} = x^\star = (-\id)_\wedge(x), \quad x \in \mathcal{G}.
	$$
	Hence, it preserves the total orientation of the space, but it cannot
	be smoothly transformed into to an identity transformation because of the hyperbolic nature of the space,
	or intuitively, because spacelike and timelike vectors cannot be boosted across the light-cone
	(reflecting the physical interpretation that nothing can travel faster than light).
	The two additional pairs of components of the Pin group correspond
	to a reflection in the time resp. space direction, followed by a boost.

\subsubsection{The space algebra}
	
	We write an arbitrary element $x \in \mathcal{G}(\mathbb{R}^3)$ 
	of the space algebra as 
	\begin{equation} \label{arbitrary_space_element}
		x = \alpha + \boldsymbol{a} + \boldsymbol{b}I + \beta I,
	\end{equation}
	where $\alpha,\beta \in \mathbb{R}$ and $\boldsymbol{a},\boldsymbol{b} \in \mathbb{R}^3$.
	The norm function $N_2$ does not take values in $\mathbb{R}$ in this algebra,
	but due to the properties of the Clifford conjugate we have 
	$N_2(x) = N_2(x)^\cliffconj \in \mathcal{G}^0 \oplus \mathcal{G}^3$.
	This subalgebra is isomorphic to 
	$\mathbb{C}$, and its corresponding complex
	conjugate is given by $[x]_3$ (or $[x]$). 
	Using these properties, we can construct a real-valued map 
	$N_3 \!: \mathcal{G} \to \mathbb{R}^+$ by taking the complex modulus,
	\begin{equation} \label{N3_definition}
		N_3(x) := [N_2(x)]_3 N_2(x) = [x^\cliffconj x] x^\cliffconj x.
	\end{equation}
	Plugging in \eqref{arbitrary_space_element} we obtain
	$$
		N_2(x) = \alpha^2 - \boldsymbol{a}^2 + \boldsymbol{b}^2 - \beta^2 + 2(\alpha\beta - \boldsymbol{a}*\boldsymbol{b})I,
	$$
	and
	\begin{equation} \label{N3_expression_space}
		N_3(x) = (\alpha^2 - \boldsymbol{a}^2 + \boldsymbol{b}^2 - \beta^2)^2 + 4(\alpha\beta - \boldsymbol{a}*\boldsymbol{b})^2.
	\end{equation}
	Although $N_3$ takes values in $\mathbb{R}^+$, it is not a norm\footnote{Actually, 
	it is not possible to find a norm $|\cdot|$ on $\mathcal{G}(\mathbb{R}^{s,t})$, $s+t \ge 3$,
	satisfying $|xy|=|x||y|$, 
	as will be clear from Theorem \ref{thm_hurwitz} and the remark following it.}
	on $\mathcal{G}$ since there are nonzero elements with $N_3(x)=0$.
	It does however have the multiplicative property
	\begin{eqnarray}
		N_3(xy) &=& [N_2(xy)]_3 N_2(xy) = [N_2(x)]_3 [N_2(y)]_3 N_2(x) N_2(y) \nonumber \\
				&=& N_3(x)N_3(y), \label{N3_factorization}
	\end{eqnarray}
	for all $x,y \in \mathcal{G}$, since
	$N_2(x)$ commutes with all of $\mathcal{G}$, and
	$$
		N_2(xy) = (xy)^\cliffconj xy = y^\cliffconj N_2(x) y = N_2(x)N_2(y).
	$$
	We also observe from \eqref{N3_expression_space} that $N_3(x^\cliffconj) = N_3(x)$.
	The expression \eqref{N3_definition} singles out the invertible elements
	as those elements \eqref{arbitrary_space_element} for which $N_3(x) \neq 0$, i.e.
	$$
		\mathcal{G}^\times(\mathbb{R}^{3})
			= \{ x \in \mathcal{G} \ : \ (\alpha^2 - \boldsymbol{a}^2 + \boldsymbol{b}^2 - \beta^2)^2 + 4(\alpha\beta - \boldsymbol{a}*\boldsymbol{b})^2 \neq 0 \},
	$$
	and
	$
		x^{-1} = \frac{1}{N_3(x)} [x^\cliffconj x] x^\cliffconj.
	$
	
	The rotor group of the space algebra is, according to the isomorphism 
	$\mathcal{G}^+(\mathbb{R}^{3,0}) \cong \mathbb{H}$ and
	Proposition \ref{prop_rotors_as_even_units},
	the group of unit quaternions 
	(note that the reverse here acts as the quaternion conjugate),
	\begin{equation} \label{space_rotor_group}
		\textrm{Spin}^{(+)}(3,0) = \{ \alpha + \boldsymbol{b}I \in \mathcal{G}^+ : \alpha^2 + \boldsymbol{b}^2 = 1 \} 
			= e^{\mathcal{G}^2(\mathbb{R}^3)}. 
	\end{equation}
	The exponentiation of the bivector algebra in the second identity
	followed from Theorem \ref{thm_spin_is_exp}.
	An arbitrary rotor $R$ can according to \eqref{space_rotor_group} be written
	in the polar form $R = e^{\varphi\boldsymbol{n}I}$, 
	where $\boldsymbol{n} \in S^2$ is a unit vector,
	and represents (cp. e.g. \eqref{plane_rotation}-\eqref{plane_rotation_rotor})
	a rotation by an angle $-2\varphi$ in the plane 
	$\boldsymbol{n}I = -\boldsymbol{n}\dual$,
	i.e. $2\varphi$ clockwise around the axis $\boldsymbol{n}$.
	
	The Pin group consists of two copies of the rotor group,
	$$
		\textrm{Pin}(3,0) = e^{\mathcal{G}^2(\mathbb{R}^3)}\ \bigsqcup\ e^{\mathcal{G}^2(\mathbb{R}^3)} \boldsymbol{n},
	$$
	for any unit vector $\boldsymbol{n} \in S^2$. 
	The Pin group can be understood topologically as a pair of
	unit 3-spheres $S^3$ lying in the even and odd subspaces, respectively.
	The odd one includes a reflection and corresponds to 
	the non-orientation-preserving part of O$(3)$.

	By the isomorphism $\rho$ in
	Example \ref{exmp_pauli_algebra} we can also represent the
	spatial rotor group in terms of the complex matrices
	$$
		\rho(R) = \left[ \begin{array}{cc}
		\alpha + b_3i	& b_2 + b_1i  \\[3pt]
		-b_2 + b_1i		& \alpha - b_3i
		\end{array} \right].
	$$
	We find that the hermitian conjugate of such a matrix is
	$\rho(R)^\dagger = \rho(R)^{-1}$ and $\det_\mathbb{C} \rho(R) = 1$,
	and that these form the \emph{special unitary group} of $\mathbb{C}^2$,
	i.e.
	$$
		\Spin^{(+)}(3,0) \cong \SU(2).
	$$
	This is the representation of the spatial rotor group which is most
	often encountered in physics.

	\begin{exc} \label{exc_noninvertible_vec_bivec}
		Show that if $x \in \mathcal{G}^1(\mathbb{R}^3) \oplus \mathcal{G}^2(\mathbb{R}^3)$
		is noninvertible, then it is on the form
		$$
			x = \alpha \boldsymbol{e}(1+\boldsymbol{f}),
		$$
		where $\alpha \in \mathbb{R}$ and 
		$\boldsymbol{e},\boldsymbol{f} \in S^2$ are orthogonal unit vectors.
	\end{exc}
	
\subsubsection{The spacetime algebra}
	
	As a four-dimensional example we consider the \emph{spacetime algebra} (STA),
	which is the geometric algebra of \emph{Minkowski spacetime}, $\mathbb{R}^{1,3}$.
	By convention\footnote{This 
	choice of notation is motivated by the Dirac representation of the STA
	in terms of so called gamma matrices which 
	is conventionally used by physicists.},
	we denote an orthonormal basis of the Minkowski space
	by $\{\gamma_0,\gamma_1,\gamma_2,\gamma_3\}$, where $\gamma_0$ is timelike
	and the other $\gamma_i$ are spacelike.
	The STA expressed in this basis is
	\begin{displaymath} \label{spacetime_algebra}
	\setlength\arraycolsep{2pt}
	\begin{array}{l}
		\mathcal{G}(\mathbb{R}^{1,3}) = \\[3pt]
		\quad \textrm{Span}_\mathbb{R} \{1,\ 
			\gamma_0,\gamma_1,\gamma_2,\gamma_3,\ 
			\boldsymbol{e}_1,\boldsymbol{e}_2,\boldsymbol{e}_3, 
			\boldsymbol{e}_1I,\boldsymbol{e}_2I,\boldsymbol{e}_3I,\ 
			\gamma_0I,\gamma_1I,\gamma_2I,\gamma_3I,\ I\},
	\end{array}
	\end{displaymath}
	where the pseudoscalar is $I := \gamma_0\gamma_1\gamma_2\gamma_3$ 
	and we set $\boldsymbol{e}_i := \gamma_i\gamma_0$, $i=1,2,3$. The form of the STA basis
	chosen above emphasizes the duality which exists between the graded subspaces.
	It also hints that the even subalgebra of the STA is the space algebra. 
	This is actually the case since there is an
	isomorphism $\mathcal{G}^+(\mathbb{R}^{1,3}) \cong \mathcal{G}(\mathbb{R}^{3,0})$,
	but we can also verify this explicitly by noting that 
	$\boldsymbol{e}_i^2 = 1$ and 
	$\boldsymbol{e}_i \boldsymbol{e}_j = -\boldsymbol{e}_j \boldsymbol{e}_i$, $i \neq j$, as well as
	$\boldsymbol{e}_1\boldsymbol{e}_2\boldsymbol{e}_3 = I$.
	Hence, the timelike (positive square) blades $\{\boldsymbol{e}_i\}$ form a basis of
	a 3-dimensional euclidean space called the \emph{relative space} to $\gamma_0$.
	Given any timelike vector $a$ we can find a similar relative space spanned by the bivectors
	$\{b \wedge a\}$ for $b \in \mathbb{R}^{1,3}$. These spaces all generate the
	\emph{relative space algebra} $\mathcal{G}^+$,
	and only the 
	precise split between the vector and bivector part
	of this relative algebra depends on the chosen reference vector $a$.
	
	Using boldface to denote relative space elements, 
	an arbitrary multivector $x \in \mathcal{G}$ can be written
	\begin{equation} \label{arbitrary_STA_element}
		x = \alpha + a + \boldsymbol{a} + \boldsymbol{b}I + bI + \beta I,
	\end{equation}
	where $\alpha,\beta \in \mathbb{R}$, $a,b \in \mathbb{R}^{1,3}$ and $\boldsymbol{a},\boldsymbol{b}$ 
	in relative space $\mathbb{R}^3$.
	As in previous examples, we would like to find the invertible elements.
	Looking at the norm function $N_2 \!: \mathcal{G} \to \mathcal{G}^0 \oplus \mathcal{G}^3 \oplus \mathcal{G}^4$,
	it is not obvious that we can extend this to a real-valued function on $\mathcal{G}$.
	Fortunately, we have for $X = \alpha + bI + \beta I \in \mathcal{G}^0 \oplus \mathcal{G}^3 \oplus \mathcal{G}^4$
	that
	\begin{equation} \label{G_0_3_4_norm}
		X[X]_{3,4} = [X]_{3,4}X = (\alpha - bI - \beta I)(\alpha + bI + \beta I) = \alpha^2 - b^2 + \beta^2 \in \mathbb{R}.
	\end{equation}
	Hence, we can define a map $N_4 \!: \mathcal{G} \to \mathbb{R}$ by
	\begin{equation} \label{N4_definition}
		N_4(x) := [N_2(x)]_{3,4} N_2(x) = [x^\cliffconj x] x^\cliffconj x.
	\end{equation}
	Plugging in \eqref{arbitrary_STA_element} into $N_2$, we obtain after some simplifications 
	\begin{equation} \label{N2_STA_expression}
	\setlength\arraycolsep{2pt}
	\begin{array}{rcl}
		N_2(x) &=& \alpha^2 - a^2 - \boldsymbol{a}^2 + \boldsymbol{b}^2 + b^2 - \beta^2 \\[5pt]
		&&	+\ 2(\alpha b - \beta a - a \liprod \boldsymbol{b} + b \liprod \boldsymbol{a} - a \liprod \boldsymbol{a}\dual - b \liprod \boldsymbol{b}\dual)I \\[5pt]
		&&	+\ 2(\alpha \beta - a*b - \boldsymbol{a}*\boldsymbol{b})I
	\end{array}
	\end{equation}
	and hence, by \eqref{G_0_3_4_norm},
	\begin{equation} \label{N4_expression}
	\setlength\arraycolsep{2pt}
	\begin{array}{rcl}
		N_4(x) &=& (\alpha^2 - a^2 - \boldsymbol{a}^2 + \boldsymbol{b}^2 + b^2 - \beta^2)^2 \\[5pt]
		&&	-\ 4(\alpha b - \beta a - a \liprod \boldsymbol{b} + b \liprod \boldsymbol{a} - a \liprod \boldsymbol{a}\dual - b \liprod \boldsymbol{b}\dual)^2 \\[5pt]
		&&	+\ 4(\alpha \beta - a*b - \boldsymbol{a}*\boldsymbol{b})^2.
	\end{array}
	\end{equation}
	We will prove some rather non-trivial statements about this norm function
	where we need that $[xy]x=x[yx]$ for all $x,y \in \mathcal{G}$.
	This is a quite general property of this involution.
	
	\begin{lem} \label{lem_nonscalar_conj}
		In any Clifford algebra $\cl(X,R,r)$ (even when $X$ is infinite), we have
		\begin{displaymath}
			[xy]x = x[yx] \quad \forall x,y \in \cl.
		\end{displaymath}
	\end{lem}
	\begin{proof}
		By linearity, we can take $y=A \in \mathscr{P}(X)$ and expand $x$ in 
		coordinates $x_B \in R$ as $x = \sum_{B \in \mathscr{P}(X)} x_B B$.
		We obtain
		\begin{displaymath}
		\setlength\arraycolsep{2pt}
		\begin{array}{rcl}
			x[Ax]
				&=& \sum_{B,C} x_B x_C\ B[AC] \\[5pt]
				&=& \sum_{B,C} x_B x_C\ \big((A \symdiff C = \varnothing) - (A \symdiff C \neq \varnothing)\big)\ BAC \\[5pt]
				&=& \sum_{B,C} x_B x_C\ \big((A=C) - (A \neq C)\big)\ BAC \\[5pt]
				&=& \sum_B x_B x_A\ BAA - \sum_{C \neq A} \sum_B x_B x_C\ BAC
		\end{array}
		\end{displaymath}
		and
		\begin{displaymath}
		\setlength\arraycolsep{2pt}
		\begin{array}{rcl}
			[xA]x
				&=& \sum_{B,C} x_B x_C\ [BA]C \\[5pt]
				&=& \sum_{B,C} x_B x_C\ \big((B=A) - (B \neq A)\big)\ BAC \\[5pt]
				&=& \sum_C x_A x_C\ AAC - \sum_{B \neq A} \sum_C x_B x_C\ BAC \\[5pt]
				&=& x_A^2\ AAA + \sum_{C \neq A} x_A x_C\ AAC - \sum_{B \neq A} x_B x_A \underbrace{BAA}_{AAB} \\[5pt]
				&& \quad - \sum_{B \neq A} \sum_{C \neq A} x_B x_C\ BAC \\[5pt]
				&=& x_A^2\ AAA - \sum_{B \neq A} \sum_{C \neq A} x_B x_C\ BAC \\[5pt]
				&=& x[Ax].
		\end{array}
		\end{displaymath}
	\end{proof}
	
	\noindent
	We now have the following
	
	\begin{lem} \label{lem_norm_conj_4}
		$N_4(x^\cliffconj) = N_4(x)$ for all $x \in \mathcal{G}(\mathbb{R}^{1,3})$.
	\end{lem}
	
	\noindent
		(Note that this identity is not at all obvious from the expression \eqref{N4_expression}.)
	
	\begin{proof}
		Using Lemma \ref{lem_nonscalar_conj} we have that
		$$
			N_4(x^\cliffconj) = [xx^\cliffconj]xx^\cliffconj = x[x^\cliffconj x]x^\cliffconj.
		$$
		Since $N_4$ takes values in $\mathbb{R}$, this must be a scalar, so that
		$$
			N_4(x^\cliffconj) = \langle x[x^\cliffconj x]x^\cliffconj \rangle_0 
				= \langle [x^\cliffconj x]x^\cliffconj x \rangle_0 = \langle N_4(x) \rangle_0 = N_4(x),
		$$
		where we used the symmetry of the scalar product.
	\end{proof}
	
	\begin{lem} \label{lem_conj_034}
		For all $X,Y \in \mathcal{G}^0 \oplus \mathcal{G}^3 \oplus \mathcal{G}^4$ we have
		\begin{displaymath}
			[XY] = [Y][X].
		\end{displaymath}
	\end{lem}
	\begin{proof}
		Take arbitrary elements $X = \alpha + bI + \beta I$ and $Y = \alpha' + b'I + \beta' I$. Then
		\begin{displaymath}
		\setlength\arraycolsep{2pt}
		\begin{array}{rcl}
			[XY]
				&=& [(\alpha + bI + \beta I)(\alpha' + b'I + \beta' I)]\\[5pt]
				&=& \alpha \alpha' - \alpha b'I - \alpha \beta' I - bI\alpha' + b*b' - b \wedge b' + b \beta' - \beta I \alpha' - \beta b' - \beta \beta'
		\end{array}
		\end{displaymath}
		and
		\begin{displaymath}
		\setlength\arraycolsep{2pt}
		\begin{array}{rcl}
			[Y][X]
				&=& (\alpha' - b'I - \beta' I)(\alpha - bI - \beta I)\\[5pt]
				&=& \alpha' \alpha - \alpha' bI - \alpha' \beta I - b'I\alpha + b'*b + b' \wedge b - b' \beta - \beta' I \alpha + \beta' b - \beta' \beta.
		\end{array}
		\end{displaymath}
		Comparing these expressions we find that they are equal.
	\end{proof}
	
	\noindent
	We can now prove that $N_4$ actually acts as a determinant on the STA.
	
	\begin{thm} \label{thm_norm_prod_4}
		The norm function $N_4$ satisfies the product property
		\begin{displaymath}
			N_4(xy) = N_4(x) N_4(y) \quad \forall x,y \in \mathcal{G}(\mathbb{R}^{1,3}).
		\end{displaymath}
	\end{thm}
	
	\begin{proof}
		Using that $N_4(xy)$ is a scalar, and that $N_2$ takes values in 
		$\mathcal{G}^0 \oplus \mathcal{G}^3 \oplus \mathcal{G}^4$, we obtain
		\begin{displaymath}
		\setlength\arraycolsep{2pt}
		\begin{array}{rcl}
			N_4(xy) 
				&=& \langle [(xy)^\cliffconj xy](xy)^\cliffconj xy \rangle_0
				= \langle [y^\cliffconj x^\cliffconj x y] y^\cliffconj x^\cliffconj x y \rangle_0 \\[5pt]
				&=& 
				\langle x^\cliffconj x y [y^\cliffconj x^\cliffconj x y] y^\cliffconj \rangle_0 
				= \langle x^\cliffconj x [y y^\cliffconj x^\cliffconj x] y y^\cliffconj \rangle_0 \\[5pt]
				&=& \langle N_2(x) [N_2(y^\cliffconj) N_2(x)] N_2(y^\cliffconj) \rangle_0 \\[5pt]
				&=& \langle N_2(x) [N_2(x)] [N_2(y^\cliffconj)] N_2(y^\cliffconj) \rangle_0 \\[5pt]
				&=& \langle N_4(x) N_4(y^\cliffconj) \rangle_0 = N_4(x) N_4(y^\cliffconj),
		\end{array}
		\end{displaymath}
		where we applied Lemma \ref{lem_nonscalar_conj} and then Lemma \ref{lem_conj_034}.
		Finally, Lemma \ref{lem_norm_conj_4} gives the claimed identity.
	\end{proof}

	From \eqref{N4_definition} we find that the 
	multiplicative group of the STA is given by
	$$
		\mathcal{G}^\times(\mathbb{R}^{1,3}) = \{ x \in \mathcal{G} : N_4(x) \neq 0 \}
	$$
	and the inverse of $x \in \mathcal{G}^\times$ is
	$$
		x^{-1} = \frac{1}{N_4(x)} [x^\cliffconj x] x^\cliffconj.
	$$
	Note that the above theorems regarding $N_4$ only rely on the commutation properties
	of the different graded subspaces and not on the actual signature
	and field of the vector space. 
	
	Let us now turn our attention to the rotor group of the STA.
	The reverse equals the Clifford conjugate on the even subalgebra (it also corresponds
	to the Clifford conjugate defined on the relative space),
	so we find from \eqref{N2_STA_expression} that the rotor group is
	$$
	\setlength\arraycolsep{2pt}
	\begin{array}{l}
		\textrm{Spin}^{+}(1,3) 
				= \{ x \in \mathcal{G}^+ : N_2(x) = x^\cliffconj x = 1  \} \\[5pt]
		\quad	= \{ \alpha + \boldsymbol{a} + \boldsymbol{b}I + \beta I \in \mathcal{G}^+ \ : \ \alpha^2 - \boldsymbol{a}^2 + \boldsymbol{b}^2 - \beta^2 = 1 \ \textrm{and} \ \alpha\beta = \boldsymbol{a}*\boldsymbol{b} \} \\[5pt]
		\quad	= \pm e^{\mathcal{G}^2(\mathbb{R}^{1,3})} \cong \textrm{SL}(2,\mathbb{C}).
	\end{array}
	$$
	The last isomorphism is related to the Dirac representation of the STA, while
	the exponentiation identity was obtained from Theorem \ref{thm_spin_is_exp} and gives a better picture
	of what the rotor group looks like. Namely, any rotor $R$ can be written 
	$R = \pm e^{\boldsymbol{a}+\boldsymbol{b}I}$
	for some relative vectors $\boldsymbol{a},\boldsymbol{b}$,
	or by Theorem \ref{thm_bivec_decomp}, either
	$$
		R = e^{\alpha\boldsymbol{e}} e^{\varphi\boldsymbol{e}I},
	$$
	where $\alpha,\varphi \in \mathbb{R}$ and 
	$\boldsymbol{e} = f_0 \wedge f_1$ is a timelike unit blade, or 
	$$
		R = \pm e^{n \wedge f_2} = \pm e^{\alpha\boldsymbol{e}(1+\boldsymbol{f})}
		= \pm 1 \pm \alpha\boldsymbol{e}(1+\boldsymbol{f}),
	$$
	with a null vector $n$ and anticommuting timelike unit blades $\boldsymbol{e},\boldsymbol{f}$
	(cp. Exercise \ref{exc_noninvertible_vec_bivec}).
	A simple rotor of the form $e^{\boldsymbol{b}I}$ corresponds to a rotation in the spacelike plane
	$\boldsymbol{b}\dual$ with angle $2|\boldsymbol{b}|$ 
	(which is a rotation also in relative space),
	while $e^{\boldsymbol{a}}$ corresponds to a hyperbolic rotation in the timelike plane $\boldsymbol{a}$,
	i.e. a boost in the relative space direction $\boldsymbol{a}$ with 
	hyperbolic angle $2|\boldsymbol{a}|$.
	
	Picking a timelike and a spacelike unit reference vector, e.g. $\gamma_0$ and $\gamma_1$,
	we obtain the Spin and Pin groups,
	\begin{displaymath}
	\setlength\arraycolsep{2pt}
	\begin{array}{rcl}
		\textrm{Spin}(1,3) &=& \pm e^{\mathcal{G}^2(\mathbb{R}^{1,3})}\ \bigsqcup\ \pm e^{\mathcal{G}^2(\mathbb{R}^{1,3})} \gamma_0\gamma_1, \\[5pt]
		\textrm{Pin}(1,3) &=& \pm e^{\mathcal{G}^2(\mathbb{R}^{1,3})}\ \bigsqcup\ \pm e^{\mathcal{G}^2(\mathbb{R}^{1,3})} \gamma_0\ \bigsqcup\ \pm e^{\mathcal{G}^2(\mathbb{R}^{1,3})} \gamma_1\ \bigsqcup\ \pm e^{\mathcal{G}^2(\mathbb{R}^{1,3})} \gamma_0\gamma_1,
	\end{array}
	\end{displaymath}
	The Pin group forms a double-cover of the so called \emph{Lorentz group} O(1,3).
	Since the STA rotor group is connected, we find that O(1,3) has four connected components.
	Two of these are not in Spin and correspond to a single inversion of time resp. space.
	The Spin group covers the subgroup of \emph{proper} Lorentz transformations 
	preserving the total orientation,
	while the rotor group covers the connected, \emph{proper orthochronous} Lorentz group
	$\SO^+(1,3)$, which also preserves the direction of time.
	The physical interpretations of the spacetime algebra will be further discussed in
	Section \ref{sec_STA}.

\subsubsection{*The Dirac algebra}
	
	Due to historic reasons, the \emph{Dirac algebra} 
	$\mathcal{G}(\mathbb{C}^{4}) \cong \mathcal{G}(\mathbb{R}^{4,1})$
	is actually the representation of the STA which is most
	commonly used in physical applications. 
	The relation between these algebras is observed by noting that the
	pseudoscalar in $\mathcal{G}(\mathbb{R}^{4,1})$ commutes with all
	elements and squares to minus the identity. By Proposition \ref{prop_iso_real_complex}
	we have that the Dirac algebra is the complexification of the STA,
	$$
		\mathcal{G}(\mathbb{R}^{4,1}) \cong \mathcal{G}(\mathbb{R}^{1,3}) \otimes \mathbb{C} \cong \mathcal{G}(\mathbb{C}^4) \cong \mathbb{C}^{4 \times 4}.
	$$
	We construct this isomorphism explicitly by taking bases
	$\{\gamma_0,\gamma_1,\gamma_2,\gamma_3\}$ of $\mathbb{R}^{1,3}$ as above,
	and $\{e_0,\ldots,e_4\}$ of $\mathbb{R}^{4,1}$ such that $e_0^2 = -1$ and the other $e_j^2 = 1$.
	We write $\mathcal{G}_5 := \mathcal{G}(\mathbb{R}^{4,1})$ and 
	$\mathcal{G}_4^\mathbb{C} := \mathcal{G}(\mathbb{R}^{1,3}) \otimes \mathbb{C}$, 
	and use the convention that Greek indices run from 0 to 3.
	The isomorphism $F \!: \mathcal{G}_5 \to \mathcal{G}_4^\mathbb{C}$
	is given by the following 1-to-1 correspondence of basis elements:
	\begin{displaymath}
	\begin{array}{rccccccc}
		\mathcal{G}_4^\mathbb{C}:
		& 1 \otimes 1
		& \gamma_\mu \otimes 1
		& \gamma_\mu \wedge \gamma_\nu \otimes 1
		& \gamma_\mu \wedge \gamma_\nu \wedge \gamma_\lambda \otimes 1
		& I_4 \otimes 1
		& 1 \otimes i \\[5pt]
		\mathcal{G}_5:
		& 1
		& e_\mu e_4
		& -e_\mu \wedge e_\nu
		& -e_\mu \wedge e_\nu \wedge e_\lambda \thinspace e_4
		& e_0 e_1 e_2 e_3
		& I_5 \\[5pt]
		x^\cliffconj\ \textrm{in}\ \mathcal{G}_4^\mathbb{C}:
		& + & - & - & + & + & + \\[5pt]
		[x]\ \textrm{in}\ \mathcal{G}_4^\mathbb{C}:
		& + & - & - & - & - & + \\[5pt]
		\overline{x}\ \textrm{in}\ \mathcal{G}_4^\mathbb{C}:
		& + & + & + & + & + & - \\[5pt]
	\end{array}
	\end{displaymath}
	The respective pseudoscalars are $I_4 := \gamma_0\gamma_1\gamma_2\gamma_3$
	and $I_5 := e_0e_1e_2e_3e_4$.
	We have also noted the correspondence between involutions in
	the different algebras. Clifford conjugate in $\mathcal{G}_4^\mathbb{C}$
	corresponds to reversion in $\mathcal{G}_5$, 
	the $[\thinspace\cdot\thinspace]$-involution becomes the $[\thinspace\cdot\thinspace]_{1,2,3,4}$-involution,
	while complex conjugation in $\mathcal{G}_4^\mathbb{C}$ corresponds 
	to grade involution in $\mathcal{G}_5$. In other words,
	$$
		F(x^\dagger) = F(x)^\cliffconj,
		\quad F([x]_{1,2,3,4}) = [F(x)],
		\quad F(x^\star) = \overline{F(x)}.
	$$
	
	We can use the correspondence above to find a norm function on $\mathcal{G}_5$.
	Since $N_4 \!: \mathcal{G}(\mathbb{R}^{1,3}) \to \mathbb{R}$ 
	is actually independent of the choice of field,
	we have that the complexification of $N_4$ satisfies
	\begin{displaymath}
	\setlength\arraycolsep{2pt}
	\begin{array}{rcl}
		N_4^\mathbb{C} \!: \mathcal{G}(\mathbb{C}^{4}) &\to& \mathbb{C}, \\[5pt]
			x &\mapsto& [x^\cliffconj x] x^\cliffconj x.
	\end{array}
	\end{displaymath}
	Taking the modulus of this complex number, we arrive at a real-valued map
	$N_5 \!: \mathcal{G}(\mathbb{R}^{4,1}) \to \mathbb{R}$ with
	\begin{displaymath}
	\setlength\arraycolsep{2pt}
	\begin{array}{rrl}
		N_5(x) &:=& \overline{N^{\mathbb{C}}_4\big(F(x)\big)} N^{\mathbb{C}}_4\big(F(x)\big) \\[5pt]
			&=& \overline{ [F(x)^\cliffconj F(x)] F(x)^\cliffconj F(x) }\ [F(x)^\cliffconj F(x)] F(x)^\cliffconj F(x) \\[5pt]
			&=& \big[ [x^\dagger x]_{1,2,3,4} x^\dagger x \big]_5 [x^\dagger x]_{1,2,3,4} x^\dagger x \\[5pt]
			&=& \big[ [x^\dagger x]_{1,4} x^\dagger x \big] [x^\dagger x]_{1,4} x^\dagger x.
	\end{array}
	\end{displaymath}
	In the final steps we noted that $x^\dagger x \in \mathcal{G}^0 \oplus \mathcal{G}^1 \oplus \mathcal{G}^4 \oplus \mathcal{G}^5$
	and that $\mathbb{C} \subseteq \mathcal{G}_4^\mathbb{C}$ corresponds
	to $\mathcal{G}^0 \oplus \mathcal{G}^5 \subseteq \mathcal{G}_5$.
	Furthermore, since $N^{\mathbb{C}}_4(xy) = N^{\mathbb{C}}_4(x) N^{\mathbb{C}}_4(y)$, we have
	\begin{displaymath}
	\setlength\arraycolsep{2pt}
	\begin{array}{rrl}
		N_5(xy) 
			&=& \overline{N^{\mathbb{C}}_4\big(F(x)F(y)\big)}\ N^{\mathbb{C}}_4\big(F(x)F(y)\big) \\[5pt]
			&=& \overline{N^{\mathbb{C}}_4\big(F(x)\big)}\ \overline{N^{\mathbb{C}}_4\big(F(y)\big)}\ N^{\mathbb{C}}_4\big(F(x)\big)\ N^{\mathbb{C}}_4\big(F(y)\big) \\[5pt]
			&=& N_5(x) N_5(y)
	\end{array}
	\end{displaymath}
	for all $x,y \in \mathcal{G}$.
	The invertible elements of the Dirac algebra are then as usual
	$$
		\mathcal{G}^\times(\mathbb{R}^{4,1}) = \{ x \in \mathcal{G} : N_5(x) \neq 0 \}
	$$
	and the inverse of $x \in \mathcal{G}^\times$ is
	$$
		x^{-1} = \frac{1}{N_5(x)} \big[ [x^\dagger x]_{1,4} x^\dagger x \big] [x^\dagger x]_{1,4} x^\dagger.
	$$
	The above strategy could also have been used to obtain the expected 
	result for $N_3$ on $\mathcal{G}(\mathbb{R}^{3,0}) \cong \mathcal{G}(\mathbb{C}^2)$
	(with a corresponding isomorphism $F$):
	$$
		N_3(x) := \overline{N^{\mathbb{C}}_2\big(F(x)\big)} N^{\mathbb{C}}_2\big(F(x)\big) = [x^\cliffconj x] x^\cliffconj x.
	$$

\subsection{*Norm functions and factorization identities}

	The norm functions 
	\begin{displaymath}
	\setlength\arraycolsep{2pt}
	\begin{array}{rcl}
		N_0(x) &:=& x, \\[5pt]
		N_1(x) &=& x^\cliffconj x, \\[5pt]
		N_2(x) &=& x^\cliffconj x, \\[5pt]
		N_3(x) &=& [x^\cliffconj x] x^\cliffconj x, \\[5pt]
		N_4(x) &=& [x^\cliffconj x] x^\cliffconj x, \\[5pt]
		N_5(x) &=& \big[ [x^\dagger x]_{1,4} x^\dagger x \big] [x^\dagger x]_{1,4} x^\dagger x
	\end{array}
	\end{displaymath}
	constructed above (where we added $N_0$ for completeness)
	all satisfy
	$$
		N_k\!: \mathcal{G}(V) \to \mathbb{F},
	$$
	where $\dim V = k$,
	and the product property
	\begin{equation} \label{norm_function_property}
		N_k(xy) = N_k(x) N_k(y)
	\end{equation}
	for all $x,y \in \mathcal{G}(V)$. 
	Furthermore,
	because these functions only involve products and involutions,
	and the proofs of the above identities only rely on commutation properties in the
	respective algebras, they even hold for arbitrary Clifford algebras $\cl(X,R,r)$
	with $|X|=k=0,1,\ldots,5$, respectively, 
	and the corresponding groups of invertible elements are
	$$
		\cl^\times(X,R,r) = \{ x \in \cl : \textrm{$N_k(x) \in R$ invertible} \}.
	$$
	
	For matrix algebras, a similar product property is satisfied by the determinant.
	On the other hand, we have the following theorem for matrices.
	
	\begin{thm} \label{thm_determinant}
		Assume that $d\!: \mathbb{R}^{n \times n} \to \mathbb{R}$ is continuous
		and satisfies
		\begin{equation} \label{determinant_property}
			d(AB) = d(A)d(B)
		\end{equation}
		for all $A,B \in \mathbb{R}^{n \times n}$.
		Then $d$ must be either $0$, $1$, $|\det|^\alpha$ or $(\textrm{\emph{sign}} \circ \det) |\det|^\alpha$
		for some $\alpha > 0$.
	\end{thm}
	
	\noindent
	In other words, we must have that $d = d_1 \circ \det$, where $d_1\!: \mathbb{R} \to \mathbb{R}$
	is continuous and satisfies $d_1(\lambda \mu) = d_1(\lambda) d_1(\mu)$.
	This $d_1$ is uniquely determined e.g. by whether $d$ takes negative values, together with
	the value of $d(\lambda\id)$ for any $\lambda > 1$.
	This means that the determinant is the $unique$ real-valued function on real matrices with
	the product property \eqref{determinant_property}.
	A proof of this theorem can be found in the appendix (Theorem \ref{thm_uniqueness_of_det}).
	
	Now, looking at Table \ref{table_classification_real}, we see that 
	$\mathcal{G}(\mathbb{R}^{k,k}) \cong \mathbb{R}^{2^k \times 2^k}$ for $k=0,1,2,\ldots$
	From the above theorem we then know that there are \emph{unique}\footnote{Actually,
	the functions are either $\det$ or $|\det|$. $N_2$ and $N_4$ constructed previously
	are smooth, however, so they must be equal to $\det$.}
	continuous functions $N_{2k}\!: \mathcal{G}(\mathbb{R}^{k,k}) \to \mathbb{R}$
	such that $N_{2k}(xy) = N_{2k}(x) N_{2k}(y)$ and $N_{2k}(\lambda) = \lambda^{2^k}$.
	These are given by the determinant on the corresponding matrix algebra,
	and the corresponding multiplicative group $\mathcal{G}^\times$
	is the general linear group $\GL(2^k,\mathbb{R})$.

	\begin{exmp}
		The product property \eqref{norm_function_property} of
		the norm functions leads to interesting factorization
		identities on rings. An example is $N_2$ for anti-euclidean signatures (quaternions),
		\begin{equation} \label{lagrange_identity}
		\setlength\arraycolsep{2pt}
		\begin{array}{l}
			(x_1^2 + x_2^2 + x_3^2 + x_4^2)(y_1^2 + y_2^2 + y_3^2 + y_4^2) \\[5pt]
			\quad	= (x_1y_1 - x_2y_2 - x_3y_3 - x_4y_4)^2
					+ (x_1y_2 + x_2y_1 + x_3y_4 - x_4y_3)^2 \\[5pt]
			\qquad +\ (x_1y_3 - x_2y_4 + x_3y_1 + x_4y_2)^2
					+ (x_1y_4 + x_2y_3 - x_3y_2 + x_4y_1)^2,
		\end{array}
		\end{equation}
		which holds for all $x_j,y_k \in R$ and
		is called the \emph{Lagrange identity}.
		These types of identities can be used to prove theorems in number
		theory. Using \eqref{lagrange_identity}, one can for example prove that every integer
		can be written as a sum of four squares of integers. Or, in other
		words, that every integer is the norm (squared) of an integral quaternion.
		See e.g. \cite{herstein} for a proof.
	\end{exmp}

%% file: clifford_conformal.tex
\subsection{Linearization of the Euclidean group}

The \emph{euclidean group} $\mathbb{E}_n^+$ on $\mathbb{R}^n$ consists of
all orientation-preserving isometries, i.e. maps 
$f: \mathbb{R}^n \to \mathbb{R}^n$ such that $|f(x)-f(y)| = |x-y|$
for all $x,y \in \mathbb{R}^n$.
An element of $\mathbb{E}_n^+$ can be shown to be 
a composition of rotations
(orthogonal maps with determinant 1) and translations.
As we have seen, rotations $\underline{R} \in \SO(n)$ 
can be represented by rotors $R \in \Spin(n)$ through
$$
	\mathbb{R}^n \ni x \mapsto \underline{R}(x) = R^\star x R^{-1} = RxR^\dagger \in \mathbb{R}^n,
$$
while translations are of the form
$$
	\mathbb{R}^n \ni x \mapsto T_a(x) = x + a \in \mathbb{R}^n, \quad a \in \mathbb{R}^n.
$$
The euclidean group is sometimes slightly cumbersome to 
work with due to the fact that it not composed exclusively of linear transformations.
We will therefore embed $\mathbb{R}^n$ in a different geometric
algebra where the euclidean isometries are transformed into linear maps.

Let $\{e_1\ldots,e_n\}$ be an orthonormal basis of $\mathbb{R}^n$
and introduce a new symbol $e$ such that $\{e_1\ldots,e_n,e\}$
is an orthonormal basis of $\mathbb{R}^{n,0,1}$, where $e^2 = 0$.
Define a map
$$
	\begin{array}{ccc}
		\mathbb{R}^n & \xrightarrow{\rho} & \mathcal{G}(\mathbb{R}^{n,0,1}) \\
		x & \mapsto & 1 + ex
	\end{array}
$$
and let $W$ denote the image of the map $\rho$.
We extend a rotation $\underline{R}$ to $W$ as
$$
	\underline{R}(1 + ex) = R(1 + ex)R^\dagger = RR^\dagger + RexR^\dagger 
	= 1 + eRxR^\dagger = 1 + e\underline{R}(x),
$$
and hence have $\underline{R}: W \to W$ and 
$\underline{R} \circ \rho = \rho \circ \underline{R}$ on $\mathbb{R}^n$.

Let us now consider the translation map $T_a(x) = x + a$.
For this we first introduce
$$
	A = e^{-\frac{1}{2}ea} = 1 - \frac{1}{2}ea
$$
and a slightly modified version of the grade involution $x \to x^\star$
by demanding that $e^\star = e$, 
while as usual $e_i^\star = -e_i$, for $i=1,\ldots,n$.
We then obtain
$$
	A^\star = 1 - \frac{1}{2}ea
	\quad \textrm{and} \quad
	A^\dagger = 1 - \frac{1}{2}ea = A^{-1}.
$$
Hence,
\begin{eqnarray*}
	\underline{A}(\rho(x)) &=& A^\star (1 + ex) A^{-1} = (1 + \frac{1}{2}ea)(1 + ex)(1 + \frac{1}{2}ea) \\
	&=& (1 + ex + \frac{1}{2}ea)(1 + \frac{1}{2}ea) = 1 + \frac{1}{2}ea + ex + \frac{1}{2}ea \\
	&=& 1 + e(x + a) = \rho(x+a) = \rho(T_a(x)),
\end{eqnarray*}
i.e. $\underline{A} \circ \rho = \rho \circ T_a$ on $\mathbb{R}^n$.

This means that an isometry $S = T_a \circ \underline{R}: \mathbb{R}^n \to \mathbb{R}^n$
has the rotor representation
$\underline{AR} = \underline{A} \circ \underline{R}: W \to W$,
with $AR = e^{-\frac{1}{2}ea} e^{-\frac{1}{2}B} \in \mathcal{G}^\times(\mathbb{R}^{n,0,1})$.
Also note that an arbitrary point $x = T_x(0) \in \mathbb{R}^n$ has the 
representation 
$$
	\rho(x) = \underline{e^{-\frac{1}{2}ex}}(\rho(0)) 
	= (e^{-\frac{1}{2}ex})^\star 1(e^{-\frac{1}{2}ex})^\dagger = e^{ex}
$$
in $W \subseteq \mathcal{G}^\times(\mathbb{R}^{n,0,1})$.

\begin{exc}
	Repeat the above construction, but replacing the null vector
	$e$ by a null 2-blade in $\mathbb{R}^{2,1}, \mathbb{R}^{1,2}$, or $\mathbb{R}^{0,2}$.
	Explain what simplifications can be made and show that the
	construction extends to the full euclidean group $\mathbb{E}_n$
	(not necessarily orientation-preserving).
\end{exc}

\subsection{Conformal algebra}

For now, we refer to Chapter 10 in \cite{doran_lasenby}.
For an application of the conformal algebra to so-called \emph{twistors}, 
see \cite{arcaute_lasenby_doran}.

%% file: clifford_reps.tex
	In this section we will use the classification of geometric algebras as
	matrix algebras, which was developed in Section \ref{sec_isomorphisms},
	to work out the representation theory of these algebras.
	Since one can find representations of geometric algebras in many
	areas of mathematics and physics, this leads to a number of interesting
	applications. In this section
	we will consider two main examples in detail, namely normed division algebras
	and vector fields on higher-dimensional spheres.
	Another important application of the representation theory for geometric
	algebras is to \emph{spinor spaces}, which will be treated in Section \ref{sec_spinors}.
	
	\begin{defn} \label{def_representation}
		For $\mathbb{K} = \mathbb{R}$, $\mathbb{C}$ or $\mathbb{H}$, we define a 
		\emph{$\mathbb{K}$-representation} of $\mathcal{G}(V,q)$
		as an $\mathbb{R}$-algebra homomorphism
		\begin{displaymath}
			\rho\!: \mathcal{G}(V,q) \to \textrm{End}_\mathbb{K}(W),
		\end{displaymath}
		where $W$ is a finite-dimensional vector space over $\mathbb{K}$.
		$W$ is called a \emph{$\mathcal{G}(V,q)$-module} over $\mathbb{K}$.
	\end{defn}
	
	\noindent
	Note that a vector space over $\mathbb{C}$ or $\mathbb{H}$ can be considered as a real vector space together
	with operators $J$ or $I,J,K$ in $\textrm{End}_\mathbb{R}(W)$ that anticommute and 
	square to minus the identity.
	In the definition above we assume that these operators commute with $\rho(x)$
	for all $x \in \mathcal{G}$, so that $\rho$ can be said to respect the $\mathbb{K}$-structure
	of the space $W$. When talking about the dimension of the module $W$
	we will always refer to its dimension as a real vector space.
	
	The standard strategy when studying representation theory is to look
	for irreducible representations.
	
	\begin{defn} \label{def_reducible_rep}
		A representation $\rho$ is called \emph{reducible} if $W$ can be written
		as a direct sum of proper (not equal to $0$ or $W$) invariant subspaces, i.e.
		\begin{displaymath}
			W = W_1 \oplus W_2 \quad \textrm{and} \quad \rho(x)(W_j) \subseteq W_j \quad \forall\ x \in \mathcal{G}.
		\end{displaymath}
		In this case we can write $\rho = \rho_1 \oplus \rho_2$, where $\rho_j(x) := \rho(x)|_{W_j}$.
		A representation is called \emph{irreducible} if it is not reducible.
	\end{defn}
	
	\begin{rem}
		The traditional definition of an irreducible representation is that it
		does not have any proper invariant subspaces 
		(the above notion is then called \emph{indecomposable}).
		However, because $\mathcal{G}$ is
		generated by a finite group (the Clifford group) one can verify that these 
		two definitions are equivalent in this case (see Exercise \ref{exc_irrep}).
	\end{rem}
	
	\begin{prop} \label{prop_irreps}
		Every $\mathbb{K}$-representation $\rho$ of a geometric algebra $\mathcal{G}(V,q)$ 
		can be split up into a direct sum $\rho = \rho_1 \oplus \ldots \oplus \rho_m$
		of irreducible representations.
	\end{prop}
	\begin{proof}
		This follows directly from the definitions and the fact that $W$ is finite-dimensional.
	\end{proof}
	
	\begin{defn} \label{def_equivalent_reps}
		Two $\mathbb{K}$-representations $\rho_j\!: \mathcal{G}(V,q) \to \textrm{End}_\mathbb{K}(W_j)$,
		$j=1,2$, are said to be \emph{equivalent} if there exists a $\mathbb{K}$-linear 
		isomorphism $F\!: W_1 \to W_2$ such that
		\begin{displaymath}
			F \circ \rho_1(x) \circ F^{-1} = \rho_2(x) \quad \forall\ x \in \mathcal{G}.
		\end{displaymath}
	\end{defn}
	
	\begin{thm} \label{thm_matrix_reps}
		Up to equivalence, the only irreducible representations of the matrix
		algebras $\mathbb{K}^{n \times n}$ and $\mathbb{K}^{n \times n} \oplus \mathbb{K}^{n \times n}$ are
		\begin{displaymath}
			\rho\!: \mathbb{K}^{n \times n} \to \textrm{\emph{End}}_\mathbb{K}(\mathbb{K}^n)
		\end{displaymath}
		and
		\begin{displaymath}
			\rho_{1,2}\!: \mathbb{K}^{n \times n} \oplus \mathbb{K}^{n \times n} \to \textrm{\emph{End}}_\mathbb{K}(\mathbb{K}^n)
		\end{displaymath}
		respectively, where $\rho$ is the defining representation 
		(ordinary matrix multiplication) and
		\begin{displaymath}
		\setlength\arraycolsep{2pt}
		\begin{array}{c}
			\rho_1(x,y):=\rho(x), \\
			\rho_2(x,y):=\rho(y).
		\end{array}
		\end{displaymath}
	\end{thm}
		\noindent
		This follows from the classical fact that the algebras 
		$\mathbb{K}^{n \times n}$ are \emph{simple} 
		(i.e. have no proper two-sided ideals) and that simple algebras have only
		one irreducible representation up to equivalence. See e.g. \cite{lang} for details.
	
	\begin{thm} \label{thm_representations}
		From the above, together with the classification of real geometric algebras,
		follows the table of representations in Table \ref{table_representations},
		where $\nu_{s,t}$ is the number of inequivalent irreducible representations
		and $d_{s,t}$ is the (real) dimension of an irreducible representation 
		for $\mathcal{G}(\mathbb{R}^{s,t})$.
		The cases for $n>8$ are obtained using the periodicity
		\begin{equation} \label{rep_periodicity}
		\setlength\arraycolsep{2pt}
		\begin{array}{lcl}
			\nu_{m+8k} &=& \nu_m, \\[5pt]
			d_{m+8k} &=& 16^k d_m.
		\end{array}
		\end{equation}
	\end{thm}
	
	\begin{table}[ht]
		\begin{displaymath}
		\begin{array}{c|l|c|c|l|c|c}
				n & \mathcal{G}(\mathbb{R}^{n,0}) & \nu_{n,0} & d_{n,0} & \mathcal{G}(\mathbb{R}^{0,n}) & \nu_{0,n} & d_{0,n} \\
			\hline
			&&&&&&\\[-1.8ex]
			0 & \mathbb{R}												& 1 & 1 	& \mathbb{R}												& 1 & 1 \\
			1 & \mathbb{R} \oplus \mathbb{R}							& 2 & 1 	& \mathbb{C}												& 1 & 2 \\
			2 & \mathbb{R}^{2 \times 2}									& 1 & 2 	& \mathbb{H}												& 1 & 4 \\
			3 & \mathbb{C}^{2 \times 2}									& 1 & 4 	& \mathbb{H} \oplus \mathbb{H}								& 2 & 4 \\
			4 & \mathbb{H}^{2 \times 2}									& 1 & 8 	& \mathbb{H}^{2 \times 2}									& 1 & 8 \\
			5 & \mathbb{H}^{2 \times 2} \oplus \mathbb{H}^{2 \times 2}	& 2 & 8 	& \mathbb{C}^{4 \times 4}									& 1 & 8 \\
			6 & \mathbb{H}^{4 \times 4}									& 1 & 16 	& \mathbb{R}^{8 \times 8}									& 1 & 8 \\
			7 & \mathbb{C}^{8 \times 8}									& 1 & 16 	& \mathbb{R}^{8 \times 8} \oplus \mathbb{R}^{8 \times 8}	& 2 & 8 \\
			8 & \mathbb{R}^{16 \times 16}								& 1 & 16 	& \mathbb{R}^{16 \times 16}									& 1 & 16
		\end{array}
		\end{displaymath}
		\caption{Number and dimension of irreducible representations
		of euclidean and anti-euclidean geometric algebras. \label{table_representations}}
	\end{table}
	
	\noindent
	Note that the cases when there are two inequivalent irreducible
	representations of $\mathcal{G}(\mathbb{R}^{s,t})$ correspond
	exactly to the cases when the pseudoscalar $I$ is central and 
	squares to one.
	Furthermore, the two representations $\rho_{\pm}$ are characterized by
	their value on the pseudoscalar,
	$$
		\rho_{\pm}(I) = \pm \id.
	$$
	It is clear that these two representations are inequivalent, 
	since if $F\!: W_1 \to W_2$ is an isomorphism, and $\rho(I) = \sigma\id_{W_1}$,
	then also $F \circ \rho(I) \circ F^{-1} = \sigma\id_{W_2}$.

	As one could expect, the corresponding representation theory 
	over the complex field is simpler:
	
	\begin{thm} \label{thm_complex_representations}
	Let $\nu_n^{\mathbb{C}}$ denote the number of inequivalent 
	irreducible complex representations of $\mathcal{G}(\mathbb{C}^n)$,
	and let $d_n^{\mathbb{C}}$ denote the complex dimension of such a representation.
	Then, if $n=2k$ is even,
	$$
		\nu_n^{\mathbb{C}} = 1 \quad \textrm{and} \quad d_n^{\mathbb{C}} = 2^k,
	$$
	while if $n=2k+1$ is odd,
	$$
		\nu_n^{\mathbb{C}} = 2 \quad \textrm{and} \quad d_n^{\mathbb{C}} = 2^k.
	$$
	\end{thm}
	
	\noindent
	The inequivalent representations in the odd case are characterized
	by their value on the \emph{complex volume element} $I_\mathbb{C} := i^kI$,
	which in that case is central, and always squares to one.
	
	We will now consider the situation when the representation space $W$
	is endowed with an inner product.
	Note that if $W$ is a vector space over $\mathbb{K}$ with an inner product,
	then we can always find a \emph{$\mathbb{K}$-invariant} inner product on $W$, i.e.
	such that the operators $J$ or $I,J,K$ are orthogonal.
	Namely, let $\langle\cdot,\cdot\rangle_\mathbb{R}$ be an inner product on $W$
	and put
	\begin{equation} \label{C_H_inv_inner_prod}
		\langle x,y \rangle_\mathbb{C} := \sum_{\Gamma \in \{\id,J\}} \langle \Gamma x, \Gamma y \rangle_\mathbb{R}, \quad
		\langle x,y \rangle_\mathbb{H} := \sum_{\Gamma \in \{\id,I,J,K\}} \langle \Gamma x, \Gamma y \rangle_\mathbb{R}.
	\end{equation}
	Then $\langle Jx,Jy \rangle_\mathbb{K} = \langle x,y \rangle_\mathbb{K}$ 
	and $\langle Jx,y \rangle_\mathbb{K} = -\langle x,Jy \rangle_\mathbb{K}$, etc.
	
	In the same way, when $V$ is euclidean or anti-euclidean, we can 
	for a given representation $\rho\!: \mathcal{G}(V) \to \textrm{End}_\mathbb{K}(W)$
	find an inner product such that $\rho$ acts orthogonally with unit vectors, 
	i.e. such that $\langle \rho(e)x,\rho(e)y \rangle = \langle x,y \rangle$ for
	all $x,y \in W$ and $e \in V$ with $e^2 = \pm 1$.
	We construct such an inner product by averaging a, possibly $\mathbb{K}$-invariant,
	inner product $\langle \cdot,\cdot \rangle_\mathbb{K}$ over the Clifford group.
	Also for general signatures,
	take an orthonormal basis $E$ of $V$ and put
	\begin{equation} \label{rho_inv_inner_prod}
		\langle x,y \rangle := \sum_{\Gamma \in \mathcal{B}_E} \langle \rho(\Gamma) x, \rho(\Gamma) y \rangle_\mathbb{K}.
	\end{equation}
	We then have that
	\begin{equation}
		\langle \rho(e_i)x,\rho(e_i)y \rangle = \langle x,y \rangle
	\end{equation}
	for all $e_i \in E$, and
	if $e_i \neq e_j$ in $E$ have the same signature,
	$$
	\setlength\arraycolsep{2pt}
	\begin{array}{lcl}
		\langle \rho(e_i)x,\rho(e_j)y \rangle 
			= \langle \rho(e_i)\rho(e_i)x,\rho(e_i)\rho(e_j)y \rangle
			= \pm \langle x,\rho(e_i)\rho(e_j)y \rangle \\[3pt]
			\quad = \mp \langle x,\rho(e_j)\rho(e_i)y \rangle
			= \mp \langle \rho(e_j)x,\rho(e_j)\rho(e_j)\rho(e_i)y \rangle \\[3pt]
			\quad = - \langle \rho(e_j)x,\rho(e_i)y \rangle.
	\end{array}
	$$
	Thus, 
	in the (anti-)euclidean case,
	if $e = \sum_i a_i e_i$ and $\sum_i a_i^2 = 1$, we obtain
	$$
		\langle \rho(e)x,\rho(e)y \rangle = \sum_{i,j} a_i a_j \langle \rho(e_i)x,\rho(e_i)y \rangle = \langle x,y \rangle.
	$$
	Hence, this inner product has the desired property.
	Also note that, for $v \in V = \mathbb{R}^{n,0}$, we have
	\begin{equation} \label{rho_acts_symmetric}
		\langle \rho(v)x,y \rangle = \langle x,\rho(v)y \rangle,
	\end{equation}
	while for $V = \mathbb{R}^{0,n}$,
	\begin{equation} \label{rho_acts_antisymmetric}
		\langle \rho(v)x,y \rangle = - \langle x,\rho(v)y \rangle,
	\end{equation}
	i.e. $\rho(v)$ is symmetric for euclidean spaces and antisymmetric for
	anti-euclidean spaces.

	\begin{exc} \label{exc_irrep}
		Prove that a representation $\rho$ of $\mathcal{G}$ is irreducible
		(according to Definition \ref{def_reducible_rep})
		if and only if it does not have any proper invariant subspaces. \\
		\emph{Hint:} You may assume that the representation space has an invariant inner product.
	\end{exc}
	
\subsection{Examples}
	
	We are now ready for some examples which illustrate how representations
	of geometric algebras can appear in various contexts and how their
	representation theory can be used to prove important theorems.

\subsubsection{Normed division algebras}
	
	Our first example concerns the possible dimensions of normed division algebras. 
	A \emph{normed division algebra} is an algebra $\mathcal{A}$ over $\mathbb{R}$ 
	(not necessarily associative) with unit and a norm $|\cdot|$ such that 
	\begin{equation} \label{norm_div_algebra}
		|xy|=|x||y|
	\end{equation}
	for all $x,y \in \mathcal{A}$
	and such that every nonzero element is invertible. We will prove the following
	
	\begin{thm}[Hurwitz' Theorem] \label{thm_hurwitz}
		If $\mathcal{A}$ is a finite-dimensional normed division algebra 
		over $\mathbb{R}$, then its dimension is either 1, 2, 4 or 8.
	\end{thm}
	
	\begin{rem}
		This corresponds uniquely to $\mathbb{R}$, $\mathbb{C}$, $\mathbb{H}$, and the octonions $\mathbb{O}$, respectively.
		The proof of uniqueness requires some additional steps, see e.g. \cite{baez}.
	\end{rem}
	
	Let us first consider what restrictions that the 
	requirement \eqref{norm_div_algebra} puts on the norm.
	Assume that $\mathcal{A}$ has dimension $n$.
	For every $a \in \mathcal{A}$ we have a linear transformation
	\begin{displaymath}
	\setlength\arraycolsep{2pt}
	\begin{array}{lccl}
		L_a\!:	& \mathcal{A}	& \to 		& \mathcal{A}, \\
				& x 			& \mapsto 	& ax
	\end{array}
	\end{displaymath}
	given by left multiplication by $a$. When $|a|=1$ we then have
	\begin{equation}
		|L_a x|=|ax|=|a||x|=|x|,
	\end{equation}
	i.e. $L_a$ preserves the norm. Hence, it maps the unit sphere $S:=\{x \in \mathcal{A}: |x|=1\}$
	in $\mathcal{A}$ into itself. Furthermore, since every element in $\mathcal{A}$ is invertible,
	we can for each pair $x,y \in S$ find an $a \in S$ such that $L_a x = ax = y$.
	Now, these facts imply a large amount of symmetry of $S$. In fact, we have the following
	
	\begin{lem} \label{lem_operator_symmetry}
		Assume that $V$ is a finite-dimensional normed vector space.
		Let $S_V$ denote the unit sphere in $V$. If, for every $x,y \in S_V$,
		there exists an operator $L \in \textrm{\emph{End}}(V)$ such that
		$L(S_V) \subseteq S_V$ and $L(x)=y$, then $V$ must be an inner product space.
	\end{lem}
	\begin{proof}[*Proof]
		We will need the following fact:
		Every compact subgroup $G$ of $\textrm{GL}(n)$ preserves some inner product on $\mathbb{R}^n$.
		This can be shown by picking a Haar-measure $\mu$ on $G$ and averaging any
		inner product $\langle \cdot,\cdot \rangle$ on $\mathbb{R}^n$ over $G$ using this measure,
		(analogous to \eqref{C_H_inv_inner_prod},\eqref{rho_inv_inner_prod} where the group is finite)
		\begin{equation}
			\langle x,y \rangle_G := \int_G \langle gx,gy \rangle\ d\mu(g).
		\end{equation}
		
		Now, let $G$ be the group of linear transformations on $V \cong \mathbb{R}^n$ 
		which preserve its norm $|\cdot|$.
		$G$ is compact in the finite-dimensional operator norm topology, since $G = \bigcap_{x \in V} \{L \in \textrm{End}(\mathbb{R}^n) : |Lx|=|x|\}$
		is closed and bounded by 1. Furthermore, $L \in G$ is injective and therefore an isomorphism.
		The group structure is obvious. Hence, $G$ is a compact subgroup of $\textrm{GL}(n)$.
		
		From the above we know that there exists an inner product $\langle \cdot,\cdot \rangle$ on $\mathbb{R}^n$
		which is preserved by $G$. Let $|\cdot|_\circ$ denote the norm associated to this inner
		product, i.e. $|x|_\circ^2 = \langle x,x \rangle$.
		Take a point $x \in \mathbb{R}^n$ with $|x|=1$ and rescale the inner 
		product so that also $|x|_\circ = 1$.
		Let $S$ and $S_\circ$ denote the unit spheres associated to $|\cdot|$ and $|\cdot|_\circ$,
		respectively. By the conditions in the lemma, there is for every $y \in S$
		an $L \in G$ such that $L(x)=y$.
		But $G$ also preserves the norm $|\cdot|_\circ$, so $y$ must also lie in $S_\circ$.
		Hence, $S$ is a subset of $S_\circ$. However, being unit spheres associated to norms,
		$S$ and $S_\circ$ are both homeomorphic
		to the standard sphere $S^{n-1}$, so we must have that they are equal.
		Therefore, the norms must be equal.
	\end{proof}
	
	We now know that our normed division algebra $\mathcal{A}$ has some
	inner product $\langle \cdot,\cdot \rangle$ such that $\langle x,x \rangle = |x|^2$.
	We call an element $a \in \mathcal{A}$ \emph{imaginary} if $a$ is orthogonal to the
	unit element, i.e. if $\langle a,1_\mathcal{A} \rangle = 0$.
	Let $\textrm{Im}\ \mathcal{A}$ denote the $(n-1)$-dimensional subspace of imaginary elements.
	We will observe that $\textrm{Im}\ \mathcal{A}$ acts on $\mathcal{A}$ in a special way.
	
	Take a curve $\gamma\!: (-\epsilon,\epsilon) \to S$ on the unit sphere such that
	$\gamma(0)=1_\mathcal{A}$ and $\gamma'(0)=a \in \textrm{Im}\ \mathcal{A}$.
	(Note that $\textrm{Im}\ \mathcal{A}$ is the tangent space to $S$ at the unit element.)
	Then, because the product in $\mathcal{A}$ is continuous, $\forall x,y \in \mathcal{A}$
	$$
	\setlength\arraycolsep{2pt}
	\begin{array}{rcl}
		\frac{d}{dt}\big|_{t=0} L_{\gamma(t)} x 
			&=& {\displaystyle \lim_{h \to 0}}\ \frac{1}{h}\big(L_{\gamma(h)} x - L_{\gamma(0)} x\big) \\[8pt]
			&=& {\displaystyle \lim_{h \to 0}}\ \frac{1}{h}\big(\gamma(h) - \gamma(0)\big)x
			= \gamma'(0)x = ax = L_a x
	\end{array}
	$$
	and
	$$
	\setlength\arraycolsep{2pt}
	\begin{array}{rcl}
		0 &=& \frac{d}{dt}\big|_{t=0} \langle x,y \rangle 
			= \frac{d}{dt}\big|_{t=0} \langle L_{\gamma(t)}x,L_{\gamma(t)}y \rangle \\[8pt]
			&=& \langle \frac{d}{dt}\big|_{t=0} L_{\gamma(t)}x,L_{\gamma(0)}y \rangle + \langle L_{\gamma(0)}x,\frac{d}{dt}\big|_{t=0} L_{\gamma(t)}y \rangle \\[8pt]
			&=& \langle L_a x,y \rangle + \langle x,L_a y \rangle.
	\end{array}
	$$
	Hence, $L_a^* = -L_a$ for $a \in \textrm{Im}\ \mathcal{A}$. 
	If, in addition, $|a|=1$ we have that $L_a \in \textrm{O}(\mathcal{A},|\cdot|^2)$, so
	$L_a^2 = - L_a L_a^* = -\id$. For an arbitrary imaginary element $a$ we obtain by rescaling
	\begin{equation} \label{normed_div_square}
		L_a^2 = -|a|^2\id.
	\end{equation}
	This motivates us to consider the geometric algebra $\mathcal{G}(\textrm{Im}\ \mathcal{A},q)$
	with quadratic form $q(a) := -|a|^2$.
	By \eqref{normed_div_square} and the universal property of geometric algebras (Proposition \ref{prop_universality}) 
	we find that $L$ extends to a representation of $\mathcal{G}(\textrm{Im}\ \mathcal{A},q)$ on $\mathcal{A}$,
	\begin{equation} \label{normed_div_rep}
		\hat{L}\!: \mathcal{G}(\textrm{Im}\ \mathcal{A},q) \to \textrm{End}(\mathcal{A}),
	\end{equation}
	i.e. a representation of $\mathcal{G}(\mathbb{R}^{0,n-1})$ on $\mathbb{R}^n$.
	The representation theory then demands that $n$ is 
	a multiple of $d_{0,n-1}$. By studying Table \ref{table_representations}
	and taking periodicity \eqref{rep_periodicity} into account
	we find that this is only possible for $n=1,2,4,8$.
	\qed

	\begin{rem}
		Not only $\mathbb{R}$, $\mathbb{C}$, and $\mathbb{H}$, but
		also the octonion algebra $\mathbb{O}$ can be constructed explicitly
		in terms of Clifford algebras, however, one cannot use simply
		the geometric product in this case since the octonion product is
		non-associative.
		With a choice of \emph{octonionic structure}, e.g.
		$$
			C := e_1e_2e_4 + e_2e_3e_5 + e_3e_4e_6 + e_4e_5e_7 + e_5e_6e_1 + e_6e_7e_2 + e_7e_1e_3 \in \mathcal{G}^3,
		$$
		we can identify $\mathbb{O} := \mathcal{G}^0 \oplus \mathcal{G}^1
		\subseteq \mathcal{G}(\mathbb{R}^{0,7})$, with a multiplication given by
		$$
			a \diamond b := \langle ab(1-C) \rangle_{0,1},
		$$
		or alternatively, $\mathbb{O} := \mathcal{G}^1(\mathbb{R}^8)$ with
		$$
			a \diamond b := \langle a e_8 b (1-CI_7)(1-I) \rangle_1,
		$$
		where $I_7 := e_1 \ldots e_7$, 
		and $e_8$ in this case represents the unit of $\mathbb{O}$.
		We refer to \cite{lounesto} for more on this construction,
		and to \cite{baez} for more on octonions in general.
	\end{rem}

\subsubsection{Vector fields on spheres}
	
	In our next example we consider the $N$-dimensional unit spheres $S^N$
	and use representations of geometric algebras to construct vector fields on them.
	The number of such vector fields that can be found gives us information
	about the topological features of these spheres.
	
	\begin{thm}[Radon-Hurwitz] \label{thm_vec_field}
		On $S^N$ there exist $n_N$ pointwise linearly independent vector fields, where, if we write $N$ uniquely as
		$$
			N+1 = (2t+1)2^{4a+b}, \quad t,a \in \mathbb{N},\ b \in \{0,1,2,3\},
		$$
		then
		$$
			n_N = 8a + 2^b - 1.
		$$
		For example,
		\begin{displaymath}
		\begin{array}{c|ccccccccccccccccc}
			N	& 0 & 1 & 2 & 3 & 4 & 5 & 6 & 7 & 8 & 9 &10 &11 &12 &13 &14 &15 &16 \\
			\hline
			n_N	& 0 & 1 & 0 & 3 & 0 & 1 & 0 & 7 & 0 & 1 & 0 & 3 & 0 & 1 & 0 & 8 & 0
		\end{array}
		\end{displaymath}
	\end{thm}

	\begin{cor}
		$S^1$, $S^3$ and $S^7$ are parallelizable.
	\end{cor}

	\begin{rem}
		The number of vector fields constructed in this way 
		is actually the maximum number of possible such fields on $S^N$.
		This is a much deeper result proven by Adams \cite{adams} using algebraic topology.
	\end{rem}
	
	Our main observation is that if $\mathbb{R}^{N+1}$ is a $\mathcal{G}(\mathbb{R}^{0,n})$-module
	then we can construct $n$ pointwise linearly independent 
	vector fields on the unit sphere
	$$
		S^N = \{x \in \mathbb{R}^{N+1} : \langle x,x \rangle=1\}. 
	$$
	Namely, suppose we have a representation $\rho$ of $\mathcal{G}(\mathbb{R}^{0,n})$ on $\mathbb{R}^{N+1}$.
	Take an inner product $\langle\cdot,\cdot\rangle$ on $\mathbb{R}^{N+1}$ such that
	the action of $\rho$ is orthogonal and pick any basis $\{e_1,\ldots,e_n\}$ of $\mathbb{R}^{0,n}$.
	We can now define a collection of smooth vector fields $\{V_1,\ldots,V_n\}$ on $\mathbb{R}^{N+1}$ by
	$$
		V_i(x) := \rho(e_i)x, \quad i=1,\ldots,n.
	$$
	According to the observation \eqref{rho_acts_antisymmetric} this action is
	antisymmetric, so that
	$$
		\langle V_i(x),x \rangle = \langle \rho(e_i)x,x \rangle = -\langle x,\rho(e_i)x \rangle = 0.
	$$
	Hence, $V_i(x) \in T_x S^N$ for $x \in S^N$. By restricting to $S^N$ we therefore
	have $n$ tangent vector fields. It remains to show that these are pointwise
	linearly independent.
	Take $x \in S^N$ and consider the linear map
	$$
	\setlength\arraycolsep{2pt}
	\begin{array}{rccl}
		i_x\!: 	& \mathbb{R}^{0,n} &\to& T_x S^N \\[3pt]
				& v &\mapsto& i_x(v) := \rho(v)x
	\end{array}
	$$
	Since the image of $i_x$ is $\textrm{Span}_\mathbb{R} \{V_i(x)\}$ it is sufficient
	to prove that $i_x$ is injective. But if $i_x(v)=\rho(v)x=0$ 
	for some $v \in \mathbb{R}^{0,n}$
	then also $v^2 x = \rho(v)^2 x = 0$, so we must have $v=0$.
	
	Now, for a fixed $N$ we want to find as many vector fields as possible, so we seek the
	highest $n$ such that $\mathbb{R}^{N+1}$ is a $\mathcal{G}(\mathbb{R}^{0,n})$-module.
	From the representation theory we know that this requires that $N+1$ is a multiple
	of $d_{0,n}$. Furthermore, since $d_{0,n}$ is a power of 2 we obtain the maximal
	such $n$ when $N+1=p 2^m$, where $p$ is odd and $d_{0,n} = 2^m$.
	Using Table \ref{table_representations} and the periodicity \eqref{rep_periodicity}
	we find that if we write $N+1 = p2^{4a+b}$, with $0 \leq b \leq 3$, then $n = 8a + 2^b - 1$.
	This proves the theorem.
	\qed

%% file: clifford_spinors.tex
	This section is currently incomplete.
	We refer to the references below
	for more complete treatments.

	In general, a \emph{spinor} $\Psi$ is an object upon which a rotor $R$ acts
	by single-sided action like left multiplication $\Psi \mapsto R\Psi$,
	instead of two-sided action like e.g. 
	$\Psi \mapsto \underline{R}(\Psi) = R\Psi R^\dagger$.
	More precisely,
	a \emph{spinor space} is a representation space for
	the rotor group, but such that the action does not factor
	through the twisted adjoint action $\tAd\!: \Spin^+ \to \SO^+$.
	The most commonly considered types of spinor spaces
	are spinor modules and irreducible Spin representations,
	as well as ideal spinors, spinor operators,
	and more generally, mixed-action spinors.

\subsection{Spin representations}

	Recall the representation theory of geometric algebras of Section \ref{sec_reps}.
	Given $\mathcal{G}(\mathbb{R}^{s,t})$, there is either one or two
	inequivalent irreducible $\mathbb{K}$-representations of real dimension $d_{s,t}$,
	where $\mathbb{K}$ is either $\mathbb{R}$, $\mathbb{C}$, or $\mathbb{H}$.
	This can be read off from Table \ref{table_classification_real} together with periodicity.
	We call such an irreducible representation space,
	$$
		\mathcal{S}_{s,t} := \mathbb{K}^{d_{s,t} / \dim_\mathbb{R} \mathbb{K}} \cong \mathbb{R}^{d_{s,t}},
	$$
	the corresponding \emph{spinor module}\footnote{Sometimes called a \emph{space of pinors}.
	Beware that the terminology regarding spinor spaces varies a lot in the
	literature, and can unfortunately be quite confusing.}
	of $\mathcal{G}(\mathbb{R}^{s,t})$.
	Furthermore,
	because the Pin, Spin and rotor groups are contained in $\mathcal{G}$,
	this will also provide us with a representation space for those groups.

	\begin{prop} \label{prop_pinor_reps}
		The representation of the group $\Pin(s,t)$, 
		obtained by restricting an irreducible real representation
		$\rho: \mathcal{G}(\mathbb{R}^{s,t}) \to \End_\mathbb{R}(\mathcal{S}_{s,t})$
		to $\Pin \subseteq \mathcal{G}$, is irreducible.
		Furthermore, whenever there are two inequivalent such representations $\rho$,
		their corresponding restrictions to $\Pin$ are also
		inequivalent.
	\end{prop}
	
	\begin{proof}
		The first statement follows from the fact that  
		the standard basis of $\mathcal{G}$
		is contained in the Pin group.
		The second follows by recalling the role of the pseudoscalar
		$I \in \Pin$ in singling out inequivalent irreducible representations.
	\end{proof}

	\begin{defn}
		The \emph{real spinor representation} of the group $\Spin(s,t)$
		is the homomorphism
		$$
		\begin{array}{rccc}
			\Delta_{s,t}\!: & \Spin(s,t) & \to & \GL(\mathcal{S}_{s,t}) \\
							& R & \mapsto & \rho(R)
		\end{array}
		$$
		given by restricting an irreducible real representation
		$\rho: \mathcal{G}(\mathbb{R}^{s,t}) \to \End_\mathbb{R}(\mathcal{S}_{s,t})$
		to $\Spin \subseteq \mathcal{G}^+ \subseteq \mathcal{G}$.
		Analogously, a real spinor representation of the rotor
		group is obtained by restricting to $\Spin^+(s,t)$.
	\end{defn}

	\begin{prop} \label{prop_spinor_reps}
		Depending on $s$ and $t$, the representation $\Delta_{s,t}$ decomposes as
		$$
		\setlength\arraycolsep{2pt}
		\begin{array}{lcl}
			\Delta_{s,t} &=& \Delta^+ \oplus \Delta^-, \\[5pt]
			\Delta_{s,t} &=& \Delta^0 \oplus \Delta^0, \ \textrm{or} \\[5pt]
			\Delta_{s,t} &=& \Delta^0,
		\end{array}
		$$
		where $\Delta^{+}$, $\Delta^{-}$, $\Delta^{0}$ 
		in the respective
		cases denote inequivalent irreducible representations of $\Spin(s,t)$. 
		This decomposition of $\Delta_{s,t}$ is independent of which irreducible representation
		$\rho$ is chosen in the definition
		(in the case that there are two inequivalent such representations).
		We write $\mathcal{S}_{s,t} = \mathcal{S}_\mathbb{K}^+ \oplus \mathcal{S}_\mathbb{K}^-$,
		$\mathcal{S}_{s,t} = \mathcal{S}_\mathbb{K}^0 \oplus \mathcal{S}_\mathbb{K}^0$, and 
		$\mathcal{S}_{s,t} = \mathcal{S}_\mathbb{K}^0$,
		for the respective representation spaces,
		where $\mathbb{K}$ tells whether the representation is real, complex or quaternionic.
		The different cases of $(s,t)$ modulo 8 are listed in Table \ref{table_real_spinor_reps}.
	\end{prop}

	\begin{sidewaystable}\centering
		\begin{displaymath}
		\setlength\arraycolsep{2pt}
		\begin{array}{|c||c|c|c|c|c|c|c|c|}
			\hline &&&&&&&& \\[-10pt]
			7	& \mathcal{S}^{0}_\mathbb{R}									& \mathcal{S}^{0}_\mathbb{C} 									& \mathcal{S}^{0}_\mathbb{H}									& \mathcal{S}^{+}_\mathbb{H} \oplus \mathcal{S}^{-}_\mathbb{H}	& \mathcal{S}^{0}_\mathbb{H}									& \mathcal{S}^{0}_\mathbb{C} \oplus \mathcal{S}^{0}_\mathbb{C}	& \mathcal{S}^{0}_\mathbb{R} \oplus \mathcal{S}^{0}_\mathbb{R}	& \mathcal{S}^{+}_\mathbb{R} \oplus \mathcal{S}^{-}_\mathbb{R}	\\
			\hline &&&&&&&& \\[-10pt]
			6	& \mathcal{S}^{0}_\mathbb{C}									& \mathcal{S}^{0}_\mathbb{H}									& \mathcal{S}^{+}_\mathbb{H} \oplus \mathcal{S}^{-}_\mathbb{H}	& \mathcal{S}^{0}_\mathbb{H}									& \mathcal{S}^{0}_\mathbb{C} \oplus \mathcal{S}^{0}_\mathbb{C}	& \mathcal{S}^{0}_\mathbb{R} \oplus \mathcal{S}^{0}_\mathbb{R}	& \mathcal{S}^{+}_\mathbb{R} \oplus \mathcal{S}^{-}_\mathbb{R}	& \mathcal{S}^{0}_\mathbb{R}									\\
			\hline &&&&&&&& \\[-10pt]
			5	& \mathcal{S}^{0}_\mathbb{H}									& \mathcal{S}^{+}_\mathbb{H} \oplus \mathcal{S}^{-}_\mathbb{H}	& \mathcal{S}^{0}_\mathbb{H}									& \mathcal{S}^{0}_\mathbb{C} \oplus \mathcal{S}^{0}_\mathbb{C}	& \mathcal{S}^{0}_\mathbb{R} \oplus \mathcal{S}^{0}_\mathbb{R}	& \mathcal{S}^{+}_\mathbb{R} \oplus \mathcal{S}^{-}_\mathbb{R}	& \mathcal{S}^{0}_\mathbb{R}									& \mathcal{S}^{0}_\mathbb{C}									\\
			\hline &&&&&&&& \\[-10pt]
			4	& \mathcal{S}^{+}_\mathbb{H} \oplus \mathcal{S}^{-}_\mathbb{H}	& \mathcal{S}^{0}_\mathbb{H}									& \mathcal{S}^{0}_\mathbb{C} \oplus \mathcal{S}^{0}_\mathbb{C}	& \mathcal{S}^{0}_\mathbb{R} \oplus \mathcal{S}^{0}_\mathbb{R}	& \mathcal{S}^{+}_\mathbb{R} \oplus \mathcal{S}^{-}_\mathbb{R}	& \mathcal{S}^{0}_\mathbb{R}									& \mathcal{S}^{0}_\mathbb{C}									& \mathcal{S}^{0}_\mathbb{H}									\\
			\hline &&&&&&&& \\[-10pt]
			3	& \mathcal{S}^{0}_\mathbb{H}									& \mathcal{S}^{0}_\mathbb{C} \oplus \mathcal{S}^{0}_\mathbb{C}	& \mathcal{S}^{0}_\mathbb{R} \oplus \mathcal{S}^{0}_\mathbb{R}	& \mathcal{S}^{+}_\mathbb{R} \oplus \mathcal{S}^{-}_\mathbb{R}	& \mathcal{S}^{0}_\mathbb{R}									& \mathcal{S}^{0}_\mathbb{C}									& \mathcal{S}^{0}_\mathbb{H}									& \mathcal{S}^{+}_\mathbb{H} \oplus \mathcal{S}^{-}_\mathbb{H}	\\
			\hline &&&&&&&& \\[-10pt]
			2	& \mathcal{S}^{0}_\mathbb{C} \oplus \mathcal{S}^{0}_\mathbb{C}	& \mathcal{S}^{0}_\mathbb{R} \oplus \mathcal{S}^{0}_\mathbb{R}	& \mathcal{S}^{+}_\mathbb{R} \oplus \mathcal{S}^{-}_\mathbb{R}	& \mathcal{S}^{0}_\mathbb{R}									& \mathcal{S}^{0}_\mathbb{C}									& \mathcal{S}^{0}_\mathbb{H}									& \mathcal{S}^{+}_\mathbb{H} \oplus \mathcal{S}^{-}_\mathbb{H}	& \mathcal{S}^{0}_\mathbb{H}									\\
			\hline &&&&&&&& \\[-10pt]
			1	& \mathcal{S}^{0}_\mathbb{R} \oplus \mathcal{S}^{0}_\mathbb{R}	& \mathcal{S}^{+}_\mathbb{R} \oplus \mathcal{S}^{-}_\mathbb{R}	& \mathcal{S}^{0}_\mathbb{R}									& \mathcal{S}^{0}_\mathbb{C}									& \mathcal{S}^{0}_\mathbb{H}									& \mathcal{S}^{+}_\mathbb{H} \oplus \mathcal{S}^{-}_\mathbb{H}	& \mathcal{S}^{0}_\mathbb{H}									& \mathcal{S}^{0}_\mathbb{C} \oplus \mathcal{S}^{0}_\mathbb{C}	\\
			\hline &&&&&&&& \\[-10pt]
			0	& \mathcal{S}^{+}_\mathbb{R} \oplus \mathcal{S}^{-}_\mathbb{R}	& \mathcal{S}^{0}_\mathbb{R}									& \mathcal{S}^{0}_\mathbb{C}									& \mathcal{S}^{0}_\mathbb{H}									& \mathcal{S}^{+}_\mathbb{H} \oplus \mathcal{S}^{-}_\mathbb{H}	& \mathcal{S}^{0}_\mathbb{H}									& \mathcal{S}^{0}_\mathbb{C} \oplus \mathcal{S}^{0}_\mathbb{C}	& \mathcal{S}^{0}_\mathbb{R} \oplus \mathcal{S}^{0}_\mathbb{R}	\\
			\hline
			\hline
				& 0 & 1 & 2 & 3 & 4 & 5 & 6 & 7 \\
			\hline
		\end{array}
		\end{displaymath}
		\caption{The decomposition of the spinor module $\mathcal{S}_{s,t}$ in the box (s,t) modulo 8. \label{table_real_spinor_reps}}
	\end{sidewaystable}

	\begin{proof}
		This follows from Proposition \ref{prop_pinor_reps} and the identification
		$$
		\begin{array}{cccccl}
			\Spin(s,t+1) &\subseteq& \mathcal{G}^+(\mathbb{R}^{s,t+1}) &\subseteq& \mathcal{G}(\mathbb{R}^{s,t+1}) &\to \End \mathcal{S}_{s,t+1} \\
			\rotatebox{90}{$\hookrightarrow$} & & \rotatebox{90}{$\cong$} \\
			\Pin(s,t) &\subseteq& \mathcal{G}(\mathbb{R}^{s,t}) & & &\to \End \mathcal{S}_{s,t}
		\end{array}
		$$
		induced by the isomorphism $F$ in the proof of Proposition \ref{prop_iso_even}.
	\end{proof}

	\begin{exmp}
		The spinor module of the space algebra $\mathcal{G}(\mathbb{R}^3)$
		is irreducible w.r.t. the spatial rotor group $\Spin(3,0)$,
		$$
			\mathcal{S}_{3,0} = \mathbb{C}^2 = \mathcal{S}^0_\mathbb{H},
		$$
		while the spinor module of the spacetime algebra $\mathcal{G}(\mathbb{R}^{1,3})$
		decomposes into a pair of equivalent four-dimensional 
		irreducible representations of $\Spin(1,3)$,
		$$
			\mathcal{S}_{1,3} = \mathbb{H}^2 = \mathcal{S}^0_\mathbb{C} \oplus \mathcal{S}^0_\mathbb{C}.
		$$
		There is also an interesting decomposition in eight dimensions,
		$$
			\mathcal{S}_{8,0} = \mathcal{S}_{0,8} = \mathbb{R}^{16} 
			= \mathcal{S}^+_\mathbb{R} \oplus \mathcal{S}^-_\mathbb{R},
		$$
		which is related to the notion of \emph{triality}
		(see e.g. \cite{lounesto} and references therein).
	\end{exmp}
	
	\begin{defn}
		The \emph{complex spinor representation} of $\Spin(s,t)$
		is the homomorphism
		$$
		\begin{array}{rccc}
			\Delta_{s,t}^\mathbb{C}\!: & \Spin(s,t) & \to & \GL_\mathbb{C}(\mathcal{S}) \\
							& R & \mapsto & \rho(R)
		\end{array}
		$$
		obtained by restricting an irreducible complex representation
		$\rho: \mathcal{G}(\mathbb{C}^{s+t}) \to \End_\mathbb{C}(\mathcal{S})$
		to $\Spin(s,t) \subseteq \mathcal{G}^+(\mathbb{R}^{s,t}) \subseteq \mathcal{G}(\mathbb{R}^{s,t}) \otimes \mathbb{C} \cong \mathcal{G}(\mathbb{C}^{s+t})$.
	\end{defn}
	
	\begin{prop}
		When $n$ is odd, the representation $\Delta_{n}^\mathbb{C}$
		is irreducible and independent of which irreducible representation $\rho$
		is chosen in the definition. When $n$ is even, there is a decomposition
		$$
			\Delta_{n}^\mathbb{C} = \Delta^+_\mathbb{C} \oplus \Delta^-_\mathbb{C}
		$$
		into a direct sum of two inequivalent irreducible complex representations
		of $\Spin(n)$.
	\end{prop}

	In physics, spinors of the complex representation $\Delta_{n}^\mathbb{C}$
	are usually called \emph{Dirac spinors}, and $\Delta^+_\mathbb{C}$, $\Delta^-_\mathbb{C}$
	are \emph{left-} and \emph{right-handed Weyl (or chiral) spinor} representations.
	Elements of the real representations $\mathcal{S}_{s,t}$ are called 
	\emph{Majorana spinors}, while those of $\mathcal{S}_\mathbb{K}^+$, $\mathcal{S}_\mathbb{K}^-$
	are \emph{left-} and \emph{right-handed Majorana-Weyl spinors}.

	Note that the spinor representations defined above do not
	descend to representations of $\SO(s,t)$ since
	$\Delta_{s,t}^{(\mathbb{C})}(-1) = -\id$.
	Actually, together with the canonical tensor representations
	\begin{eqnarray*}
		\ad|_{\mathcal{G}^1}: & \ad_{\frac{1}{2} e_i \wedge e_j}(v) &= \frac{1}{2}[e_i \wedge e_j, v] = (e_i \wedge e_j) \riprod v = v_je_i - v_ie_j, \\
		\ad|_{\mathcal{G}^k}: & \ad_{\frac{1}{2} e_i \wedge e_j}(x) &= 
		\sum_{r=1}^k \sum_{i_1,\ldots,i_k} v_{i_1} \wedge \cdots \wedge \ad_{\frac{1}{2} e_i \wedge e_j}(v_{i_r}) \wedge \cdots \wedge v_{i_k},
	\end{eqnarray*}
	where $v= \sum_i v_i e_i$, $x = \sum_{i_1,\ldots,i_k} v_{i_1} \wedge \cdots \wedge v_{i_k} \in \mathcal{G}^k$,
	and $\{e_i\}_{i=1}^n$ is an orthonormal basis of $\mathbb{R}^n$,
	these additional representations provide us with the full set of
	so called \emph{fundamental representations} of the Lie algebras $\mathfrak{so}(n)$
	and the corresponding Lie groups $\Spin(n)$.
	
	We refer to \cite{lawson_michelsohn,goodman_wallach} for more on
	this approach to spinor spaces, and proofs of the above propositions.

\subsection{Ideal spinors}

	The above spaces of spinors can actually be considered as subspaces
	of the geometric algebra $\mathcal{G}$ itself, 
	namely as minimal left ideals of $\mathcal{G}$,
	with the action given by left multiplication on the algebra. 
	For now, we refer to \cite{lounesto} for this approach.

\subsection{Mixed-action spinors}

	A subspace of spinors $\mathcal{S} \subseteq \mathcal{G}$
	can of course be extended from a minimal left ideal of $\mathcal{G}$ 
	to some larger invariant subspace, 
	like a sum of such ideals, or the whole algebra $\mathcal{G}$.
	However, it is possible to consider other, intermediate, subspaces $\mathcal{S}$
	by, instead of only left action, also take advantage of right action.
	For now, we refer to \cite{doran_lasenby,lundholm_geosusy} 
	for examples of this approach.

%% file: clifford_analysis.tex
	This section is currently incomplete. We refer to e.g. 
	\cite{doran_lasenby,gilbert_murray,hestenes_sobczyk}.

%% file: clifford_minkowski.tex
	For now, we refer to chapters 5,7 and 8 in \cite{doran_lasenby}.
	See also e.g. \cite{ablamowicz_sobczyk,hestenes_mechanics}.

%% file: clifford_appendix.tex
\subsection{Notions in algebra} \label{app_algebra}

In order to make the presentation of this course easier to follow
also without a wide mathematical knowledge, we give here a short
summary of basic notions in algebra, most of which will be used in the course.

\subsubsection{Basic notions}

A \emph{binary composition} $*$ on a set $M$ is a map 
$$
	\begin{array}{ccl}
		M \times M & \to & M \\
		(x,y) & \mapsto & x*y
	\end{array}
$$
The composition is called
\begin{description}
\item \emph{associative} if $(x*y)*z = x*(y*z)$ $\forall x,y,z \in M$,
\item \emph{commutative} if $x*y = y*x$ $\forall x,y \in M$.
\end{description}

If there exists an element $e \in M$ such that $e*x = x*e = x$ $\forall x \in M$
then $e$ is called a \emph{unit}.
Units are unique because, if $e$ and $e'$ are units, then $e = e*e' = e'$.

An element $x \in M$ has a \emph{left inverse} $y$ if $y*x=e$,
\emph{right inverse} $z$ if $x*z=e$, and \emph{inverse} $y$ if $y*x=x*y=e$.
If $*$ is associative and if $y$ and $z$ are inverses to $x$ then
$$
	z = z*e = z*(x*y) = (z*x)*y = e*y = y.
$$
In other words, inverses are unique whenever $*$ is associative.

Let $*$ and $\diamond$ be two binary compositions on $M$. 
We say that $*$ is \emph{distributive} over $\diamond$ if
$$
	x*(y \diamond z) = (x*y) \diamond (x*z) \quad \forall x,y,z \in M
$$
and
$$
	(y \diamond z)*x = (y*x) \diamond (z*x) \quad \forall x,y,z \in M.
$$

From these notions one can define a variety of common mathematical
structures as follows.
\begin{description}
	\item[Monoid:] A set with an associative binary composition.

	\item[Group:] Monoid with unit, where every element has an inverse.

	\item[Abelian group:] Commutative group.

	\item[Ring:] A set $R$ with two binary compositions, called 
	\emph{addition} $(+)$ and \emph{multiplication} $(\cdot)$, such that $(R,+)$ is
	an abelian group and $(R,\cdot)$ is a monoid,
	and where multiplication is distributive over addition.
	Furthermore, it should hold that $0 \cdot x = x \cdot 0 = 0$
	for all $x \in R$, where $0$ denotes the additive unit and is called \emph{zero}.

	\item[Ring with unit:] Ring where multiplication has a unit,
	often denoted $1$ and called \emph{one} or \emph{the identity}.

	\item[Characteristic:] The characteristic, $\charop R$, 
	of a ring $R$ with unit $1$ is defined
	to be the smallest integer $n>0$ such that 
	$\underbrace{1 + \ldots + 1}_{\textrm{$n$ terms}} = 0$.
	If no such $n$ exists, then we define $\charop R = 0$.

	\item[Commutative ring:] Ring with commutative multiplication.

	\item[(Skew) Field:] (Non-)Commutative ring with unit, where every nonzero element
	has a multiplicative inverse.

	\item[Module:] An abelian group $(M,\oplus)$ is called
	a \emph{module} over the ring $(R,+,\cdot)$ if we are given a map
	(often called \emph{scalar multiplication})
	$$
	\begin{array}{ccl}
		R \times M & \to & M \\
		(r,m) & \mapsto & rm
	\end{array}
	$$
	such that for all $r,r' \in R$ and $m,m' \in M$
	\begin{description}
	\item{i)} $0_R m = 0_M$
	\item{ii)} $1 m = m$, if $R$ has the multiplicative unit $1$
	\item{iii)} $(r + r')m = (rm) \oplus (r'm)$
	\item{iv)} $r(m \oplus m') = (rm) \oplus (rm')$
	\item{v)} $(r \cdot r')m = r(r'm)$
	\end{description}
\end{description}

	\begin{rem}
		Usually, both the zero $0_R$ in $R$ and the additive unit $0_M$ in $M$ are denoted $0$.
		Also, $\oplus$ is denoted $+$ and is called addition as well.
	\end{rem}

\begin{description}
	\item[Submodule:] If $A$ is a non-empty subset of an $R$-module $B$ then
	$A$ is called a \emph{submodule} of $B$ if
	$$
		x,y \in A \ \Rightarrow \ x+y \in A \qquad \forall x,y
	$$
	and
	$$
		r \in R, x \in A  \Rightarrow rx \in A \qquad \forall r,x,
	$$
	i.e. $A$ is closed under addition and multiplication by scalars.
	(In general, a substructure is a non-empty subset of a structure that is
	closed under the operations of the structure.)
	
	\item[Ideal:] A non-empty subset $J$ of a ring $R$ is called a \emph{(two-sided) ideal} if
	$$
		x,y \in J \Rightarrow x+y \in J \qquad \forall x,y
	$$
	and
	$$
		x \in J, r,q \in R \Rightarrow rxq \in J \qquad \forall x,r,q.
	$$
	$J$ is called a \emph{left-sided} resp. \emph{right-sided ideal} if 
	the latter condition is replaced with
	$$
		x \in J, r \in R \Rightarrow rx \in J \qquad \forall x,r
	$$
	resp.
	$$
		x \in J, r \in R \Rightarrow xr \in J \qquad \forall x,r.
	$$
	
	\item[Vector space:] Module over a field.
	
	\item[Hilbert space:] Vector space over $\mathbb{R}$ or $\mathbb{C}$
	with a hermitian inner product (linear or sesquilinear) such that
	the induced topology of the corresponding norm is complete
	(Cauchy sequences converge).
	
	\item[R-algebra:] 
	An $R$-module $A$ is called an \emph{$R$-algebra} if there is defined an 
	$R$-bilinear map $A \times A \xrightarrow{*} A$, usually called
	multiplication.

	\item[Associative R-algebra:] $R$-algebra with associative multiplication (hence also a ring).

	\item[Lie algebra:] (Typically non-associative) $R$-algebra $A$ with
	multiplication commonly denoted $[\cdot,\cdot]$ s.t.
	\begin{eqnarray*}
		& [x,x] = 0, & \textrm{(antisymmetry)} \\
		& {[x,[y,z]] + [y,[z,x]] + [z,[x,y]] = 0}, & \textrm{(Jacobi identity)}
	\end{eqnarray*}
	for all $x,y,z \in A$.

	\item[Lie group:] A group which is also a differentiable manifold
	(a topological space which is essentially composed of patches of $\mathbb{R}^n$
	that are smoothly glued together) and such that the group operations are smooth.

\end{description}
\begin{exmp}
	\item[$\cdot$] $\mathbb{N}$ with addition is a monoid,
	\item[$\cdot$] $\textup{U}(1)$ (unit complex numbers) with multiplication is an abelian group,
	\item[$\cdot$] $(\mathbb{Z},+,\cdot)$ is a commutative ring, 
	\item[$\cdot$] $\mathbb{Q}$, $\mathbb{R}$, and $\mathbb{C}$ are fields,
	\item[$\cdot$] $\mathbb{H}$ is a skew field,
	\item[$\cdot$] $\mathbb{R}$ is a module over $\mathbb{Z}$ 
	and a (infinite-dimensional) vector space over $\mathbb{Q}$,
	\item[$\cdot$] $\mathbb{C}^n$ is a vector space over $\mathbb{C}$,
	and a Hilbert space with the standard sesquilinear inner product 
	$\bar{x}^{\transp} y = \sum_{j=1}^n \bar{x_j} y_j$,
	\item[$\cdot$] the set $\mathbb{R}^{n \times n}$ of real $n \times n$ matrices 
	with matrix multiplication is a non-commutative associative $\mathbb{R}$-algebra with unit,
	\item[$\cdot$] $\mathbb{R}^3$ with the cross product is a Lie algebra, 
	\item[$\cdot$] all the following classical groups are Lie groups:
	\begin{eqnarray*}
		\textup{GL}(n,\mathbb{K})&=& \{ g \in \mathbb{K}^{n \times n} : \det\nolimits_\mathbb{K} g \neq 0 \}, \\
		\textup{O}(n) &=& \{ g \in \mathbb{R}^{n \times n} : \langle gx,gy \rangle_\mathbb{R} = \langle x,y \rangle_\mathbb{R} \ \forall x,y \in \mathbb{R}^n \}, \\
		\textup{SO}(n)&=& \{ g \in \textup{O}(n) : \det\nolimits_\mathbb{R} g = 1 \}, \\
		\textup{U}(n) &=& \{ g \in \mathbb{C}^{n \times n} : \langle gx,gy \rangle_\mathbb{C} = \langle x,y \rangle_\mathbb{C} \ \forall x,y \in \mathbb{C}^n \}, \\
		\textup{SU}(n)&=& \{ g \in \textup{U}(n) : \det\nolimits_\mathbb{C} g = 1 \}, \\
		\textup{Sp}(n)&=& \{ g \in \mathbb{H}^{n \times n} : \langle gx,gy \rangle_\mathbb{H} = \langle x,y \rangle_\mathbb{H} \ \forall x,y \in \mathbb{H}^n \},
	\end{eqnarray*}
	where $\langle x,y \rangle_\mathbb{K} = \bar{x}^{\transp} y$ denotes the standard hermitian
	inner product on the corresponding space.
\end{exmp}
\begin{description}
	
	\item[Group homomorphism:] A map $\varphi: G \to G'$, where $(G,*)$ and $(G',*')$
	are groups, such that 
	$$
		\varphi(x*y) = \varphi(x)*'\varphi(y) \qquad \forall x,y \in G
	$$
	and $\varphi(e) = e'$.
	
	\item[Module homomorphism:] A map $\varphi: M \to M'$, where $M$ and $M'$
	are $R$-modules, such that 
	$$
		\varphi(rx+y) = r\varphi(x)+\varphi(y) \qquad \forall r \in R, x,y \in M.
	$$
	A map with this property is called $R$-\emph{linear},
	and the set of $R$-module homomorphisms $\varphi: M \to M'$
	(which is itself naturally an $R$-module) is denoted
	$\Hom_R(M,M')$.
	
	\item[R-algebra homomorphism:] An $R$-linear map $\varphi: A \to A'$, 
	where $A$ and $A'$ are $R$-algebras, such that 
	$\varphi(x *_A y) = \varphi(x) *_{A'} \varphi(y)$ $\forall x,y \in A$.
	If $A$ and $A'$ have units $1_A$ resp. $1_{A'}$ then we also require that
	$\varphi(1_A) = 1_{A'}$.

	\item[Isomorphism:] 
	A bijective homomorphism $\varphi: A \to B$ 
	(it follows that also $\varphi^{-1}$ is a homomorphism).
	We say that $A$ and $B$ are isomorphic and write $A \cong B$.

	\item[Endomorphism:] A homomorphism from an object to itself.
	For an $R$-module $M$, we denote $\End_R M := \Hom_R(M,M)$.

	\item[Automorphism:] An endomorphism which is also an isomorphism.

	\item[Relation:] A \emph{relation} on a set $X$ is a subset $R \subseteq X \times X$.
	If $(x,y) \in R$ then we say that $x$ is related to $y$ and write $xRy$.

	A relation $R$ on $X$ is called
	\begin{description}
	\item \emph{reflexive} if $x R x$ $\forall x \in X$,
	\item \emph{symmetric} if $x R y \Rightarrow y R x$ $\forall x,y \in X$,
	\item \emph{antisymmetric} if $x R y$ \& $y R x \Rightarrow x=y$ $\forall x,y \in X$,
	\item \emph{transitive} if $xRy$ \& $yRz \Rightarrow xRz$ $\forall x,y,z \in X$.
	\end{description}

	\item[Partial order:] A reflexive, antisymmetric, and transitive relation.

	\item[Total order:] A partial order such that either $xRy$ or $yRx$ holds $\forall x,y \in X$.

	\item[Equivalence relation:] A reflexive, symmetric, and transitive relation.
\end{description}

\subsubsection{Direct products and sums}

Suppose we are given an $R$-module $M_i$ for each $i$ in some
index set $I$. A function $f: I \to \bigcup_{i \in I} M_i$
is called a \emph{section} if $f(i) \in M_i$ for all $i \in I$.
The \emph{support} of a section $f$, $\supp f$, is the set
$\{i \in I : f(i) \neq 0\}$.

\begin{defn}
	\emph{The (direct) product} $\prod\limits_{i \in I} M_i$ 
	of the modules $M_i$ is the set	of all sections. 
	\emph{The (direct) sum} $\bigsqcup\limits_{i \in I} M_i = \bigoplus\limits_{i \in I} M_i$
	is the set of all sections with finite support.
\end{defn}
\noindent
If $I$ is finite, $I = \{1,2,\ldots,n\}$, then the product and sum
coincide and is also denoted $M_1 \oplus M_2 \oplus \ldots \oplus M_n$.

The $R$-module structure on the direct sum and product is given
by ``pointwise'' summation and $R$-multiplication,
$$
	(f+g)(i) := f(i)+g(i), \quad \textrm{and} \quad (rf)(i) := r(f(i)),
$$
for 
sections $f,g$ and $r \in R$.
If, moreover, the $R$-modules $M_i$ also are $R$-algebras, 
then we can promote $\prod_{i \in I} M_i$ and $\bigoplus_{i \in I} M_i$
into $R$-algebras as well by defining ``pointwise'' multiplication
$$
	(fg)(i) := f(i)g(i).
$$
Also note that we have canonical surjective $R$-module homomorphisms
$$
	\pi_j: \prod_{i \in I} M_i \ni f \mapsto f(j) \in M_j,
$$
and canonical injective $R$-module homomorphisms
$$
	\sigma_j: M_j \ni m \mapsto \sigma_j(m) \in \bigoplus_{i \in I} M_i,
$$
defined by $\sigma_j(m)(i) := (i=j)m$.

Let us consider the case $M_i = R$, for all $i \in I$, 
in a little more detail.
Denote by $X$ an arbitrary set and $R$ a ring with unit. 
If $f: X \to R$ then the support of $f$ is given by
$$
	\supp f = \{ x \in X : f(x) \neq 0 \}.
$$

\begin{defn}
	\emph{The direct sum of $R$ over $X$}, 
	or \emph{the free $R$-module generated by $X$},
	is the set of all functions
	from $X$ to $R$ with finite support, and is denoted 
	$\bigoplus\limits_X R$ or $\bigsqcup\limits_X R$.
\end{defn}

For $x \in X$, we conveniently define $\delta_x: X \to R$ through
$\delta_x(y) = (x=y)$, i.e.
$\delta_x(y) = 0$ if $y \neq x$, and $\delta_x(x) = 1$.
Then, for any $f: X \to R$ with finite support, we have
$$
	f(y) = \sum_{x \in X} f(x) \delta_x(y) \qquad \forall y \in X,
$$
or, with the $R$-module structure taken into account,
$$
	f = \sum_{x \in X} f(x) \delta_x.
$$
This sum is usually written simply 
$f = \sum_{x \in X}f(x)x$ or $\sum_{x \in X}f_x x$,
and is interpreted as a formal sum of the elements in $X$ with coefficients in $R$.

\begin{exc}
	Prove the following universality properties of the product
	and the sum:
	\item{a)} Let $N$ be an $R$-module and let 
	$\rho_j \in \Hom_R(N,M_j)$
	for all $j \in I$.
	Show that there exist unique
	$\tau_j \in \Hom_R(N, \prod_{i \in I} M_i)$, 
	such that the diagram below commutes:
	$$
		\setlength\arraycolsep{0pt}
		\begin{array}{rccl}
		&N \\
		&& \searrow^{\tau_j} \\
		\rho_j & \downarrow & & \prod\nolimits_{i \in I} M_i \\
		&& \swarrow_{\pi_j} \\
		&M_j
		\end{array}
	$$
	\item{b)} Let $\rho_j: M_j \to N$ be in $\Hom_R(M_j,N)$.
	Prove that there exist unique 
	$\tau_j \in \Hom_R(\bigoplus_{i \in I} M_i, N)$, 
	making the diagram below commutative:
	$$
		\setlength\arraycolsep{0pt}
		\begin{array}{rccl}
		&M_j \\
		&& \searrow^{\sigma_j} \\
		\rho_j & \downarrow & & \bigoplus\nolimits_{i \in I} M_i \\
		&& \swarrow_{\tau_j} \\
		&N
		\end{array}
	$$
\end{exc}

\subsubsection{Lists}

Let $X$ be an arbitrary set and define the cartesian products of $X$
inductively through
\begin{eqnarray*}
	X^0 &=& \{ \varnothing \} \quad \textrm{(\emph{the empty list})} \\
	X^1 &=& X \\
	X^{n+1} &=& X \times X^n, \ n>0.
\end{eqnarray*}
Note that $X$ is isomorphic to $X \times X^0$.

\begin{defn}
	$\List X := \bigcup_{n \ge 0} X^n$
	is the set of all finite lists of elements in $X$.
\end{defn}

We have a natural monoid structure on $\List X$ through concatenation
of lists:
$$
	\begin{array}{ccl}
	\List X \times \List X &\to& \List X \\
	(x_1x_2\ldots x_m\ ,\ y_1y_2\ldots y_n) &\mapsto& x_1\ldots x_m y_1\ldots y_n
	\end{array}
$$
This binary composition, which we shall call multiplication of lists,
is obviously associative, and the empty list, which we can call $1$, is the unit.

\subsubsection{The free associative $R$-algebra generated by $X$}

Let $M = \List X$ and define \emph{the free associative $R$-algebra generated by $X$},
or \emph{the noncommutative polynomial ring over the ring $R$ in the variables $X$},
first as an $R$-module by
$$
	R\{X\} := \bigoplus_M R.
$$
We shall now define a multiplication on $R\{X\}$ which extends
our multiplication on $M$, such that $R\{X\}$ becomes a ring.
We let
	$$
		\begin{array}{ccl}
		R\{X\} \times R\{X\} &\to& R\{X\} \\
		(f,g) &\mapsto& fg,
		\end{array}
	$$
	where
	$$(fg)(m) := \sum_{m'm''=m} f(m')g(m'').$$
	That is, we sum over all (finitely many) pairs $(m',m'')$ such that 
	$m' \in \supp f$, $m'' \in \supp g$ and $m=m'm''$.

If we interpret $f$ and $g$ as formal sums over $M$, i.e.
$$
	f = \sum_{m \in M} f_m m, \qquad g = \sum_{m \in M} g_m m,
$$
then we obtain
$$
	fg = \sum_{m',m'' \in M} f_{m'} g_{m''} m'm''.
$$

\begin{rem}
	We could of course also promote $\bigoplus_M R$ to a ring in this way
	if we only required that $M$ is a monoid and $R$ a ring with unit.
\end{rem}

\begin{exc}
	Verify that $R\{X\}$ is a ring with unit.
\end{exc}

\subsubsection{Quotients}

Let $A$ be a submodule of the the $R$-module $B$ and define,
for $x \in B$,
$$
	x + A := \{x+y : y \in A \}.
$$
Sets of the form $x+A$ are called \emph{residue classes} and the
set of these is denoted $B/A$, i.e.
$$
	B/A = \{x+A: x \in B\}.
$$
The residue class $x+A$ with representative $x \in B$ is often denoted $[x]$.

We will now promote $B/A$ to an $R$-module (called a \emph{quotient module}).
Addition is defined by
$$
	\begin{array}{ccl}
	B/A \times B/A &\stackrel{+}{\to}& B/A \\
	(x+A,y+A) &\mapsto& x+y+A.
	\end{array}
$$
We have to verify that addition is well defined, so assume that
$x+A = x'+A$ and $y+A = y'+A$.
Then we must have $x-x' \in A$ and $y-y' \in A$ so that
$x'+y'+A = x+y+(x'-x)+(y'-y)+A$. 
Hence, $x'+y'+A \subseteq x+y+A$ and similarly $x+y+A \subseteq x'+y'+A$.
This shows that $x'+y'+A = x+y+A$.

Multiplication with ring elements is defined similarly by
$$
	\begin{array}{ccl}
	R \times B/A &\to& B/A \\
	(r,x+A) &\mapsto& rx+A.
	\end{array}
$$

\begin{exc}
	Verify that the $R$-multiplication above is well defined and that
	$B/A$ becomes a module over $R$.
\end{exc}

If instead $J$ is an ideal in the ring $R$, we let
$$
	R/J = \{x + J : x \in R\}
$$
and define addition through
$$
	\begin{array}{ccl}
	R/J \times R/J &\stackrel{+}{\to}& R/J \\
	(x+J,y+J) &\mapsto& x+y+J,
	\end{array}
$$
and multiplication through
$$
	\begin{array}{ccl}
	R/J \times R/J &\to& R/J \\
	(x+J,y+J) &\mapsto& xy+J.
	\end{array}
$$
We leave to the reader to verify that the addition and multiplication
defined above are well defined and that they promote $R/J$ to a ring.

\begin{exmp}
	$n\mathbb{Z}$ with $n$ a positive integer is an ideal in $\mathbb{Z}$,
	and $\mathbb{Z}_n := \mathbb{Z} \big/ n\mathbb{Z}$ is a ring
	with characteristic $n$
	(and a field when $n$ is prime).
\end{exmp}

\subsubsection{Tensor products of modules}

Let $A$ and $B$ be $R$-modules and let $D$ be the smallest submodule
of $\bigoplus_{A \times B} R$ containing the set
\begin{eqnarray*}
	&& \Big\{ (a+a',b)-(a,b)-(a',b), \ (a,b+b')-(a,b)-(a,b'), \ (ra,b)-r(a,b), \\
	&& \qquad (a,rb)-r(a,b) \ : \ a,a' \in A, \ b,b' \in B, \ r \in R \Big\}.
\end{eqnarray*}
Then the tensor product $A \otimes_R B$
(or simply $A \otimes B$) of $A$ and $B$ is defined by
$$
	A \otimes_R B := \bigoplus_{A \times B} R \Big/ D.
$$
The point with the above construction is the following

\begin{thm}[Universality] \label{thm_tensor_universality}
	To each bilinear $A \times B \stackrel{\beta}{\to} C$, where
	$A,B$ and $C$ are $R$-modules, there exists a unique linear
	$\lambda: A \otimes B \to C$ such that the diagram below commutes
	$$
		\begin{array}{ccl}
		A \times B & \xrightarrow{\beta} & C \\
		\pi \downarrow & \nearrow_\lambda \\
		A \otimes B
		\end{array}
	$$
	Here, $\pi: A \times B \to A \otimes B$ is the canonical projection
	that satisfies $\pi(a,b) = a \otimes b := (a,b) + D$,
	where $(a,b)$ on the right hand side is an element in $\bigoplus_{A \times B} R$.
\end{thm}

\begin{exc}
	Verify the above theorem.
\end{exc}

\subsubsection{Sums and products of algebras}

Let $A$ and $B$ be $R$-algebras with multiplication maps $*_A$ resp. $*_B$.
We have defined the direct sum $A \oplus B$ as an $R$-algebra
through pointwise multiplication,
$$
	\begin{array}{ccl}
	(A \oplus B) \times (A \oplus B) &\to& A \oplus B \\
	\big( (a,b)\ ,\ (a',b') \big) &\mapsto& (a *_A a' , b *_B b').
	\end{array}
$$
We can also promote the tensor product $A \otimes_R B$ 
into an $R$-algebra by introducing the product
$$
	\begin{array}{ccl}
	(A \otimes_R B) \times (A \otimes_R B) &\to& A \otimes_R B \\
	(a \otimes b\ ,\ a' \otimes b') &\mapsto& (a *_A a') \otimes (b *_B b')
	\end{array}
$$
and extending linearly.
If $1_A$ and $1_B$ are units on $A$ resp. $B$, 
then $(1_A,1_B)$ becomes a unit on $A \oplus B$ and
$1_A \otimes 1_B$ becomes a unit on $A \otimes_R B$.
Note that the above algebras may look similar, but they are very different
since e.g. $1_A \otimes 0 = 0 \otimes 1_B = 0$, 
while $(1_A,1_B) = (1_A,0) + (0,1_B) \neq 0$.

It is easy to see that if $A$ and $B$ are associative $R$-algebras, 
then also $A \oplus B$ and $A \otimes_R B$ are associative.
Furthermore, it holds that
\begin{equation} \label{tensor_algebra_isomorphism}
	R^{n \times n} \otimes_R A^{m \times m} \cong A^{nm \times nm},
\end{equation}
as $R$-algebras,
where 
$n,m$ are positive integers, and $R$ is a ring with unit.

\begin{exc}
	Prove the isomorphism \eqref{tensor_algebra_isomorphism} by
	considering the map
	$$
		\begin{array}{ccl}
		R^{n \times n} \times A^{m \times m} &\xrightarrow{F}& A^{nm \times nm} \\
		(r\ ,\ a) &\mapsto& F(r,a),
		\end{array}
	$$
	with $F(r,a)_{(i,j),(k,l)} := r_{i,k}a_{j,l}$, 
	and $i,k \in \{1,\ldots,n\}$, $j,l \in \{1,\ldots,m\}$.
\end{exc}

\subsubsection{The tensor algebra $\mathcal{T}(V)$} \label{app_tensor_algebra}

The tensor algebra $\mathcal{T}(V)$ of 
a module $V$ over a ring $R$ can
be defined as a direct sum of tensor products
of $V$ with itself,
$$
	\mathcal{T}(V) := \bigoplus_{k=0}^{\infty} \left( \bigotimes\nolimits^k V \right),
$$
where $\bigotimes^0 V := R$, $\bigotimes^1 V := V$, 
$\bigotimes^{k+1} V := (\bigotimes^k V) \otimes_R V$,
and with the obvious associative multiplication 
(concatenation of tensors) extended linearly,
$$
	(v_1 \otimes \ldots \otimes v_j) \otimes (w_1 \otimes \ldots \otimes w_k) 
	:= v_1 \otimes \ldots \otimes v_j \otimes w_1 \otimes \ldots \otimes w_k.
$$
Alternatively, $\mathcal{T}(V)$ is obtained as 
the free associative $R$-algebra generated by either the set $V$ itself,
$$
	\mathcal{T}(V) \cong R\{V\} / \hat{D},
$$
where $\hat{D}$ is the ideal generating linearity in each factor
(similarly to $D$ in the definition of the tensor product above),
or (when $R=\mathbb{F}$ is a field) by simply choosing a basis $E$ of $V$,
$$
	\mathcal{T}(V) \cong \mathbb{F}\{E\}.
$$

\begin{exc}
	Verify the equivalence of these three definitions of $\mathcal{T}(V)$.
\end{exc}

\subsection{Expansion of the inner product} \label{app_inner_prod_expansion}

	Let $\Lambda(n,m)$ denote the set of ordered $m$-subsets
	$\lambda = (\lambda_1 < \lambda_2 < \ldots < \lambda_m)$
	of the set $\{1,2,\ldots,n\}$, with $m \le n$.
	The ordered complement of $\lambda$ is denoted
	$\lambda^c = (\lambda^c_1 < \lambda^c_2 < \ldots < \lambda^c_{n-m})$.
	Furthermore, we denote by $\sgn \lambda$ the sign of the permutation 
	$(\lambda_1, \lambda_2, \ldots, \lambda_m, \lambda^c_1, \lambda^c_2, \ldots, \lambda^c_{n-m})$.
	For an $n$-blade $A = a_1 \wedge \cdots \wedge a_n$ we use the
	corresponding notation
	$A_\lambda = a_{\lambda_1} \wedge a_{\lambda_2} \wedge \cdots \wedge a_{\lambda_m}$,
	etc.
	
	\begin{thm}
		For any $m$-vector $x \in \mathcal{G}^m$ and $n$-blade $A \in \mathcal{B}_n$,
		$m \le n$, we have
		\begin{equation} \label{inner_prod_expansion}
			x \liprod A = \sum_{\lambda \in \Lambda(n,m)} (\sgn \lambda) 
			(x * A_\lambda) A_{\lambda^c}.
		\end{equation}
	\end{thm}
	\begin{proof}
		First note that both sides of \eqref{inner_prod_expansion} are linear in $x$
		and that the l.h.s. is linear and
		alternating in the $a_k$'s. 
		To prove that a multilinear mapping $F: V \times V \times \ldots \times V = V^n \to W$
		is alternating it is enough to show that $F(x_1,\ldots,x_n) = 0$
		whenever $x_s = x_{s+1}$ for $s=1,\ldots,n-1$.
		So for the right hand side, suppose that $a_s = a_{s+1}$
		and consider a fixed $\lambda \in \Lambda(n,m)$.
		If $\{s,s+1\} \subseteq \lambda$ or $\{s,s+1\} \subseteq \lambda^c$,
		then $A_\lambda$ or $A_{\lambda^c}$ vanishes.
		So let us assume e.g. that $s \in \lambda$ and $s+1 \in \lambda^c$,
		say $s = \lambda_i$, $s+1 = \lambda^c_j$.
		We define $\mu \in \Lambda(n,m)$ such that
		$$
		\setlength\arraycolsep{1.9pt}
		\begin{array}{rcccccccccccc}
			(\lambda,\lambda^c) =& (\lambda_1, &\lambda_2, &\ldots, &\lambda_i=s, &\ldots, &\lambda_m, &\lambda^c_1, &\lambda^c_2, &\ldots, &\lambda^c_j=s\!+\!1, &\ldots, &\lambda^c_{n-m}) \\
			&\rotatebox{90}{=}&\rotatebox{90}{=}&&&&\rotatebox{90}{=}&\rotatebox{90}{=}&\rotatebox{90}{=}&&&&\rotatebox{90}{=} \\[-5pt]
			(\mu,\mu^c) =& (\mu_1, &\mu_2, &\ldots, &\mu_i=s\!+\!1, &\ldots, &\mu_m, &\mu^c_1, &\mu^c_2, &\ldots, &\mu^c_j=s, &\ldots, &\mu^c_{n-m}) \\
		\end{array}
		$$
		Then $\mu \in \Lambda(n,m)$, $\sgn \mu = -\sgn \lambda$, and since
		$a_s = a_{s+1}$, we have $A_\lambda = A_\mu$ and $A_{\lambda^c} = A_{\mu^c}$.
		Thus, $(\sgn \lambda) (x*A_\lambda) A_{\lambda^c} + (\sgn \mu) (x*A_\mu) A_{\mu^c} = 0$.
		
		We conclude that both sides of \eqref{inner_prod_expansion} are
		alternating and multilinear in $a_1,\ldots,a_n$.
		Hence, we may without loss of generality assume that
		$x = e_1 \ldots e_m$ and
		$A = e_1 \ldots e_m \ldots e_n$ (otherwise both sides are zero),
		where $\{e_1,\ldots,e_d\}$, $m \le n \le d$,
		denotes an orthogonal basis for $V = \mathcal{G}^1$.
		But then the theorem is obvious, since
		\begin{eqnarray*}
			\sum_{\lambda \in \Lambda(n,m)} \sgn \lambda \thinspace ((e_1 \ldots e_m) * (e_{\lambda_1} \ldots e_{\lambda_m})) e_{\lambda^c_1} \ldots e_{\lambda^c_{n-m}} \\
			\quad = ((e_1 \ldots e_m)*(e_1 \ldots e_m)) e_{m+1} \ldots e_d. 
		\end{eqnarray*}
	\end{proof}
	
\subsection{Extension of morphisms} \label{app_extension}
	
	Let $A$ be an abelian group such that $a+a=0$ for all $a \in A$
	and let $M$ be a monoid with unit $0$.
	A map $\varphi: A \to M$ is called a homomorphism if 
	$\varphi(0) = 0$ and $\varphi(a+b) = \varphi(a) + \varphi(b)$ for all $a,b \in A$.
	We have the following
	
	\begin{thm}	
		If $H$ is a subgroup of $A$ and $\varphi: H \to M$ is a
		homomorphism then there exists an extension of $\varphi$
		to a homomorphism $\psi: A \to M$ such that $\psi|_H = \varphi$.
	\end{thm}
	
	\begin{proof}
		Let $a \in A \smallsetminus H$ and form $H_a := H \cup (a+H)$.
		Define $\varphi_a: H_a \to M$ by
		$\varphi_a(h) := \varphi(h)$ if $h \in H$ and $\varphi_a(a+h) := \varphi(h)$.
		It is easy to see that $H_a$ is a subgroup of $A$ and
		that $\varphi_a$ is a homomorphism that extends $\varphi$.
		
		Let us introduce a partial order $\preceq$ on pairs $(H',\varphi')$
		where $H'$ is a subgroup of $A$ and $\varphi'$ is a 
		homomorphism $H' \to M$ which extends $\varphi$.
		We define $(H',\varphi') \preceq (H'',\varphi'')$ if
		$H' \subseteq H''$ and $\varphi''|_{H'} = \varphi'$.
		By the Hausdorff maximality theorem there exists a maximal chain
		$$
			\mathcal{K} = \{(H^i,\varphi^i) : i \in I \},
		$$
		where $I$ is a total ordered index set w.r.t $\le$, and
		$$
			i \le j \quad \Leftrightarrow \quad (H^i,\varphi^i) \preceq (H^j,\varphi^j).
		$$
		Now, let $H^* := \bigcup_{i \in I} H^i$.
		Then $H^*$ is a subgroup of $A$ such that $H^i \subseteq H^* \ \forall i \in I$,
		since if $a,b \in H^*$ then there exists some $i \in I$
		such that both $a,b \in H^i$, i.e. $a+b \in H^i \subseteq H^*$.
		We also define $\varphi^*: H^* \to M$ by
		$\varphi^*(a) = \varphi^i(a)$ if $a \in H^i$.
		This is well-defined.
		
		Now, if $H^* \neq A$ then we could find $a \in A \smallsetminus H^*$
		and extend $\varphi^*$ to a homomorphism $\varphi^*_a: H^*_a \to M$
		as above. But then $\mathcal{K} \cup \{(H^*_a,\varphi^*_a)\}$
		would be an even greater chain than $\mathcal{K}$, which is a contradiction.
		This proves the theorem.
	\end{proof}

	Now, consider the case when
	$A = (\mathscr{P}(X), \symdiff)$,
	$H = \mathscr{F}(X)$, 
	$M = \{\pm 1\} = \mathbb{Z}_2$ with multiplication,
	and $\varphi(A) = (-1)^{|A|}$ for $A \in \mathscr{F}(X)$.
	We then obtain
	
	\begin{thm}
		There exists a homomorphism 
		$$
			\varphi: \big( \mathscr{P}(X), \symdiff \big) \to \big( \{-1,1\},\cdot \big)
		$$
		such that
		$\varphi(A) = (-1)^{|A|}$ for finite subsets $A \subseteq X$,
		and $\varphi(A \symdiff B) = \varphi(A)\varphi(B)$ for all $A,B \subseteq X$.
	\end{thm}

	\noindent
	In particular, if $A,B \subseteq X$ are disjoint then
	$\varphi(A \cup B) = \varphi(A \symdiff B) = \varphi(A)\varphi(B)$.

\subsection{Matrix theorems} \label{app_matrix_theorems}

	In the following theorem we assume that $R$ is an arbitrary commutative ring and
	\begin{displaymath}
		A = \left[
		\begin{array}{ccc}
			a_{11} & \cdots & a_{1m} \\
			\vdots &		& \vdots \\
			a_{n1} & \cdots & a_{nm} \\
		\end{array}
		\right] \in R^{n \times m}, \qquad
		a_j = \left[
		\begin{array}{c}
			a_{1j} \\
			\vdots \\
			a_{nj} \\
		\end{array}
		\right],
	\end{displaymath}
	i.e. $a_j$ denotes the $j$:th column in $A$.
	If $I \subseteq \{1,\ldots,n\}$ and $J \subseteq \{1,\ldots,m\}$ we
	let $A_{I,J}$ denote the $|I| \times |J|$-matrix minor obtained from $A$
	by deleting the rows and columns not in $I$ and $J$.
	Further, let $k$ denote the rank of $A$, i.e. the highest integer $k$
	such that there exists $I,J$ as above with $|I|=|J|=k$ and $\det A_{I,J} \neq 0$.
	By renumbering the $a_{ij}$:s we can without loss of generality assume 
	that $I=J=\{1,2,\ldots,k\}$.
	
	\begin{thm}[Basis minor] \label{thm_basis_minor}
		If the rank of $A$ is $k$, and
		\begin{displaymath}
			d := \det \left[
			\begin{array}{ccc}
				a_{11} & \cdots & a_{1k} \\
				\vdots &		& \vdots \\
				a_{k1} & \cdots & a_{kk} \\
			\end{array}
			\right] \neq 0,
		\end{displaymath}
		then every $d \cdot a_j$ is a linear combination of $a_1,\ldots,a_k$.
	\end{thm}
	
	\begin{proof}
		Pick $i \in \{1,\ldots,n\}$ and $j \in \{1,\ldots,m\}$ and consider
		the $(k+1)\times(k+1)$-matrix
		\begin{displaymath}
			B_{i,j} := \left[
			\begin{array}{cccc}
				a_{11} & \cdots & a_{1k} & a_{1j} \\
				\vdots &		& \vdots & \vdots \\
				a_{k1} & \cdots & a_{kk} & a_{kj} \\
				a_{i1} & \cdots & a_{ik} & a_{ij} \\
			\end{array}
			\right].
		\end{displaymath}
		Then $\det B_{i,j} = 0$. Expanding $\det B_{i,j}$ along the bottom row
		for fixed $i$ we obtain
		\begin{equation}
			a_{i1} C_1 + \ldots + a_{ik} C_k + a_{ij} d = 0,
		\end{equation}
		where the $C_l$ are independent of the choice of $i$ (but of course dependent on $j$).
		Hence,
		\begin{equation}
			C_1 a_1 + \ldots + C_k a_k + d a_j = 0,
		\end{equation}
		and similarly for all $j$.
	\end{proof}
	
	The following shows that the factorization
	$\det(AB) = \det(A) \det(B)$
	is a unique property of the determinant.
	
	\begin{thm}[Uniqueness of determinant] \label{thm_uniqueness_of_det}
		Assume that $d\!: \mathbb{R}^{n \times n} \to \mathbb{R}$ is continuous
		and satisfies
		\begin{equation}
			d(AB) = d(A)d(B)
		\end{equation}
		for all $A,B \in \mathbb{R}^{n \times n}$.
		Then $d$ must be either $0$, $1$, $|\det|^\alpha$ or $(\textrm{\emph{sign}} \circ \det) |\det|^\alpha$
		for some $\alpha > 0$.
	\end{thm}
	
	\begin{proof}
		First, we have that ($I$ denotes the unit matrix)
		$$
		\setlength\arraycolsep{2pt}
		\begin{array}{rcl}
			d(0) = d(0^2) = d(0)^2, \\[5pt]
			d(I) = d(I^2) = d(I)^2,
		\end{array}
		$$
		so $d(0)$ and $d(I)$ must be either 0 or 1.
		Furthermore,
		$$
		\setlength\arraycolsep{2pt}
		\begin{array}{rcl}
			d(0) = d(0A) = d(0)d(A), \\[5pt]
			d(A) = d(IA) = d(I)d(A),
		\end{array}
		$$
		for all $A \in \mathbb{R}^{n \times n}$,
		which implies that $d=1$ if $d(0)=1$ and $d=0$ if $d(I)=0$.
		We can therefore assume that $d(0)=0$ and $d(I)=1$.
		
		Now, an arbitrary matrix $A$ can be written as
		$$
			A = E_1 E_2 \ldots E_k R,
		$$
		where $R$ is on reduced row-echelon form 
		(reduced as much as possible by Gaussian elimination)
		and $E_i$ are elementary row operations of the form
		$$
		\setlength\arraycolsep{2pt}
		\begin{array}{lcl}
			R_{ij} &:=& (\textrm{swap rows $i$ and $j$}), \\[5pt]
			E_i(\lambda) &:=& (\textrm{scale row $i$ by $\lambda$}),\ \textrm{or} \\[5pt]
			E_{ij}(c) &:=& (\textrm{add $c$ times row $j$ to row $i$}).
		\end{array}
		$$
		Because $R_{ij}^2 = I$, we must have $d(R_{ij}) = \pm 1$. This gives, since 
		$$
			E_i(\lambda) = R_{1i} E_1(\lambda) R_{1i}, 
		$$
		that $d\big(E_i(\lambda)\big) = d\big(E_1(\lambda)\big)$ and
		$$
			d(\lambda I) = d\big(E_1(\lambda) \ldots E_n(\lambda)\big) = d\big(E_1(\lambda)\big) \ldots d\big(E_n(\lambda)\big) = d\big(E_1(\lambda)\big)^n.
		$$
		In particular, we have $d\big(E_1(0)\big) = 0$ and of course $d\big(E_1(1)\big) = d(I) = 1$.
		
		If $A$ is invertible, then $R=I$. Otherwise, $R$ must contain a row of zeros
		so that $R = E_i(0)R$ for some $i$. But then $d(R)=0$ and $d(A)=0$.
		When $A$ is invertible we have $I=AA^{-1}$ and $1=d(I)=d(A)d(A^{-1})$,
		i.e. $d(A) \neq 0$ and $d(A^{-1}) = d(A)^{-1}$. Hence,
		$$
			A \in \textrm{GL}(n) \quad \Leftrightarrow \quad d(A) \neq 0.
		$$
		We thus have that $d$ is completely determined by its values on $R_{ij}$,
		$E_1(\lambda)$ and $E_{ij}(c)$. Note that we have not yet used the
		continuity of $d$, but let us do so now. 
		We can split $\mathbb{R}^{n \times n}$ into three connected components,
		namely $\textrm{GL}^-(n)$, $\det^{-1}(0)$ and $\textrm{GL}^+(n)$,
		where the determinant is less than, equal to, and greater than zero, respectively.
		Since $E_1(1), E_{ij}(c) \in \textrm{GL}^+(n)$ and 
		$E_1(-1), R_{ij} \in \textrm{GL}^-(n)$, we have by continuity of $d$ that
		\begin{equation} \label{sign_of_d}
		\setlength\arraycolsep{2pt}
		\begin{array}{lcl}
			d(R_{ij}) =+1 &\quad \Rightarrow \quad& \textrm{$d$ is $>0$, $=0$, resp. $>0$} \\[5pt]
			d(R_{ij}) =-1 &\quad \Rightarrow \quad& \textrm{$d$ is $<0$, $=0$, resp. $>0$} \\[5pt]
		\end{array}
		\end{equation}
		on these parts.
		Using that $d\big(E_1(-1)\big)^2 = d\big(E_1(-1)^2\big) = d(I) = 1$, we have
		$d\big(E_1(-1)\big) = \pm 1$ and $d\big(E_1(-\lambda)\big) = d\big(E_1(-1)\big) d\big(E_1(\lambda)\big) = \pm d\big(E_1(\lambda)\big)$
		where the sign depends on \eqref{sign_of_d}.
		On $\mathbb{R}^{++} := \{ \lambda \in \mathbb{R} : \lambda > 0 \}$ we have a continuous map
		$d \circ E_1\!: \mathbb{R}^{++} \to \mathbb{R}^{++}$ such that
		$$
			d \circ E_1(\lambda \mu) = d \circ E_1(\lambda) \cdot d \circ E_1(\mu) \quad \forall\ \lambda,\mu \in \mathbb{R}^{++}.
		$$
		Forming $f := \ln \circ\ d \circ E_1 \circ \exp$, we then have a continuous map
		$f\!: \mathbb{R} \to \mathbb{R}$ such that
		$$
			f(\lambda + \mu) = f(\lambda) + f(\mu).
		$$
		By extending linearity from $\mathbb{Z}$ to $\mathbb{Q}$ and $\mathbb{R}$ by continuity,
		we must have that $f(\lambda) = \alpha \lambda$ for some $\alpha \in \mathbb{R}$.
		Hence, $d \circ E_1(\lambda) = \lambda^\alpha$. Continuity also demands that $\alpha > 0$.
		
		It only remains to consider $d \circ E_{ij} \!: \mathbb{R} \to \mathbb{R}^{++}$.
		We have $d \circ E_{ij}(0) = d(I) = 1$ and $E_{ij}(c) E_{ij}(\gamma) = E_{ij}(c + \gamma)$,
		i.e.
		\begin{equation} \label{d_add_to_mult}
			d \circ E_{ij}(c + \gamma) = d \circ E_{ij}(c) \cdot d \circ E_{ij}(\gamma) \quad \forall\ c,\gamma \in \mathbb{R}.
		\end{equation}
		Proceeding as above, $g := \ln \circ\ d \circ E_{ij} \!: \mathbb{R} \to \mathbb{R}$
		is linear, so that $g(c) = \alpha_{ij} c$ for some $\alpha_{ij} \in \mathbb{R}$,
		hence $d \circ E_{ij}(c) = e^{\alpha_{ij} c}$.
		One can verify that the following identity holds for all $i,j$:
		$$
			E_{ji}(-1) = E_i(-1) R_{ij} E_{ji}(1) E_{ij}(-1).
		$$
		This gives $d\big(E_{ji}(-1)\big) = (\pm 1) (\pm 1) d\big(E_{ji}(1)\big) d\big(E_{ij}(-1)\big)$
		and, using \eqref{d_add_to_mult},
		$$
		\setlength\arraycolsep{2pt}
		\begin{array}{rcl}
			d\big(E_{ij}(1)\big) 
				&=& d\big(E_{ji}(2)\big) = d\big(E_{ji}(1+1)\big) = d\big(E_{ji}(1)\big)d\big(E_{ji}(1)\big) \\[5pt]
				&=& d\big(E_{ij}(2)\big)d\big(E_{ij}(2)\big) = d\big(E_{ij}(4)\big),
		\end{array}
		$$
		which requires $\alpha_{ij}=0$.
		
		We conclude that $d$ is completely determined by $\alpha > 0$,
		where $d \circ E_1(\lambda) = \lambda^\alpha$ and $\lambda \geq 0$,
		plus whether $d$ takes negative values or not. This proves the theorem.
	\end{proof}

%% file: clifford.bbl
\newpage


\addcontentsline{toc}{section}{References}
\begin{thebibliography}{99}
	\bibitem{ablamowicz_sobczyk}
	R. Ablamowicz, G. Sobczyk. \textit{Lectures on Clifford (Geometric) Algebras and Applications.} 
	Birkh\"auser Boston, Inc., Boston, MA, 2004.
	\bibitem{adams}
	J. F. Adams. \textit{Vector fields on spheres.} 
	Ann. of Math. 75 (1962), 603--632.
	\bibitem{arcaute_lasenby_doran}
	E. Arcaute, A. Lasenby, C. J. L. Doran. \textit{Twistors in Geometric Algebra.} 
	Adv. Appl. Clifford Algebr. 18 (2008), 373--394.
	\bibitem{artin}
	E. Artin. \textit{Geometric algebra.} 
	Interscience, New York, 1957.
	\bibitem{baez}
	J. C. Baez. \textit{The octonions.} 
	Bull. Amer. Math. Soc. (N.S.) 39 (2002), 145--205.
	\bibitem{benn_tucker}
	I. M. Benn, R. W. Tucker. \textit{An introduction to spinors and geometry with applications in physics.} 
	Adam Hilger, Ltd., Bristol, 1987.
	\bibitem{chrisholm_common}
	J. S. R. Chrisholm, A. K. Common (eds.). \textit{Clifford Algebras and their Applications in Mathematical Physics.} 
	D. Reidel Publishing Co., Dordrecht, 1986.
	\bibitem{clifford}
	W. K. Clifford. \textit{Mathematical papers.} Edited by Robert Tucker, with an introduction by H. J. Steven Smith, 
	Chelsea Publishing Co., New York, 1968.
	\bibitem{de_la_Harpe}
	P. de la Harpe. \textit{The Clifford algebra and the spinor group of a Hilbert space.} 
	Compositio Math. 25 (1972), 245--261.
	\bibitem{doran_et_al}
	C. J. L. Doran, D. Hestenes, F. Sommen, N. Van Acker. \textit{Lie groups as spin groups.} 
	J. Math. Phys. 34 (1993), 3642--3669.
	\bibitem{doran_lasenby}
	C. J. L. Doran, A. N. Lasenby. \textit{Geometric Algebra for Physicists.} 
	Cambridge University Press, Cambridge, 2003.
	\bibitem{geroch}
	R. Geroch. \textit{Spinor structure of space-times in general relativity. I.} 
	J. Math. Phys. 9 (1968), 1739--1744.
	\bibitem{gilbert_murray}
	J. E. Gilbert, M. A. M. Murray \textit{Clifford algebras and Dirac operators in harmonic analysis.} 
	Cambridge University Press, Cambridge, 1991.
	\bibitem{goodman_wallach}
	R. Goodman, N. R. Wallach. \textit{Representations and Invariants of the Classical Groups.} 
	Cambridge University Press, Cambridge, 1998.
	\bibitem{grassmann}
	H. G. Grassmann. \textit{A new branch of mathematics.} The Ausdehnungslehre of 1844 and other works. Translated from the German and with a note by Lloyd C. Kannenberg. With a foreword by Albert C. Lewis.
	Open Court Publishing Co., Chicago, IL, 1995.
	\bibitem{grove}
	L. C. Grove. \textit{Classical groups and geometric algebra.} 
	American Mathematical Society, Providence, RI, 2002.
	\bibitem{gull_lasenby_doran}
	S. F. Gull, A.N. Lasenby, C. J. L Doran. \textit{Imaginary Numbers are not Real - the Geometric Algebra of Spacetime.} 
	Found. Phys. 23 (1993), 1175--1201.
	\bibitem{herstein}
	I. N. Herstein. \textit{Topics in Algebra, 2nd ed.} 
	John Wiley and Sons, 1975.
	\bibitem{hestenes_mechanics}
	D. Hestenes. \textit{New foundations for classical mechanics, 2nd ed.} 
	Kluwer Academic Publishers Group, Dordrecht, 1999.
	\bibitem{hestenes_sobczyk}
	D. Hestenes, G. Sobczyk. \textit{Clifford Algebra to Geometric Calculus.} 
	D. Reidel Publishing Co., Dordrecht, 1984.
	\bibitem{hestenes_STA}
	D. Hestenes. \textit{Space-Time Algebra.} 
	Gordon and Breach, New York, 1966.
	\bibitem{hormander}
	L. H\"ormander. \textit{Riemannian geometry : lectures given during the fall of 1990.} 
	Matematiska institutionen, Lunds universitet, 1990.
	\bibitem{jacobson}
	N. Jacobson. \textit{Lie Algebras.} 
	Dover Publications, Inc., New York, 1979.
	\bibitem{lang}
	S. Lang. \textit{Algebra, 3rd ed.} 
	Addison-Wesley, 1993.
	\bibitem{lasenby_doran_gull}
	A. N. Lasenby, C. J. L Doran and S. F. Gull. \textit{Gravity, gauge theories and geometric algebra.} 
	R. Soc. Lond. Philos. Trans. Ser. A Math. Phys. Eng. Sci. 356 (1998), 487--582.
	\bibitem{lawson_michelsohn}
	B.H. Lawson, Jr., M-L Michelsohn. \textit{Spin Geometry.} 
	Princeton University Press, Princeton, NJ, 1989.
	\bibitem{lounesto}
	P. Lounesto. \textit{Clifford Algebras and Spinors.}
	Cambridge University Press, Cambridge, 1997.
	\bibitem{lundholm_geosusy}
	D. Lundholm. \textit{On the Geometry of Supersymmetric Quantum Mechanical Systems.} 
	J. Math. Phys. 49, 062101 (2008).
	\bibitem{lundholm_thesis}
	D. Lundholm. \textit{Geometric (Clifford) algebra and its applications.} 
	M.Sc. Thesis, Trita-MAT. MA, ISSN 1401-2278; 2006:01, {\tt arXiv:math/0605280}.
	\bibitem{naeve_svensson}
	A. Naeve, L. Svensson. \textit{Discrete Integration and Derivation.} 
	Presented at the 5:th International Conference on Clifford Algebras and their Applications in Mathematical Physics, Ixtapa-Zihuatanejo, Mexico, June 27-July 4, 1999.
	(Online version at {\tt http://kmr.nada.kth.se})
	\bibitem{nakahara}
	M. Nakahara. \textit{Geometry, Topology and Physics, 2nd ed.} 
	Institute of Physics, Bristol, 2003.
	\bibitem{petersson}
	A. Petersson. \textit{Enumeration of spanning trees in simplicial complexes.} 
	Licentiate Thesis, Uppsala Universitet, 2009, U.U.D.M. Report 2009:13.
	\bibitem{Plymen}
	R. J. Plymen. \textit{The Laplacian and the Dirac operator in infinitely many variables.} 
	Compositio Math. 41 (1980), 137--152.
	\bibitem{porteous}
	I. R. Porteous. \textit{Clifford algebras and the classical groups.} 
	Cambridge University Press, Cambridge, 1995.
	\bibitem{riesz}
	M. Riesz. \textit{Clifford Numbers and Spinors.} 
	University of Maryland, 1958, Kluwer Academic Publishers, 1993.
	\bibitem{schubring}
	G. Schubring (Editor). \textit{Hermann Günther Grassmann (1809-1877) : visionary mathematician, scientist and neohumanist scholar.} 
	Kluwer Acad. Publ., Dordrecht, 1996.
	\bibitem{simon}
	B. Simon. \textit{Representations of finite and compact groups.} 
	American Mathematical Society, Providence, RI, 1996.
	\bibitem{sommer}
	G. Sommer (Editor). \textit{Geometric Computing with Clifford Algebras: Theor. Found. and Appl. in Computer Vision and Robotics.} 
	Springer, Berlin, 2001.
	\bibitem{svensson}
	L. Svensson. \textit{Lecture notes from a course in Clifford Algebra.} 
	KTH, 2000.
	\bibitem{svensson_naeve}
	L. Svensson, A. Naeve. \textit{Combinatorial Aspects of Clifford Algebra.} 
	Presented at the International Workshop on Applications of Geometric Algebra, Cambridge, 5-6 Sept. 2002. 
	(Online version at {\tt http://kmr.nada.kth.se})
	\bibitem{Wene}
	G. P. Wene. \textit{The Clifford algebra of an infinite-dimensional space.} 
	J. Math. Phys. 30 (1989), 249--251.
\end{thebibliography}
